%% file: mainfossacsjournal.tex
\documentclass{lmcs}



\usepackage[T1]{fontenc}

\usepackage{graphicx}
\usepackage{booktabs}

\input xy 
\xyoption{all}
\input{preamble}
\input{macros}
\usepackage{multicol}

\begin{document}

\title{Convex Biproducts, Stochastic Matrices and Tape Diagrams}

\author{Filippo Bonchi}
\address{University of Pisa, Pisa, Italy}
\email{filippo.bonchi@unipi.it}

\author{Cipriano Junior Cioffo}
\address{University of Pisa, Pisa, Italy}
\email{ciprianojunior.cioffo@di.unipi.it}







\maketitle

\begin{abstract}
Categories with finite biproducts play a central role in category theory, providing an abstract setting in which additive and linear structures can be studied uniformly. In this paper, we introduce categories with \emph{convex} biproducts, which intuitively restrict the linear structures to convex ones.
We show that, whereas categories with finite biproducts give rise to a matrix calculus based on arbitrary linear combinations, convex biproduct categories instead induce a matrix calculus based on stochastic (more generally, substochastic) matrices. This perspective yields a refined algebraic and compositional framework tailored to probabilistic settings.
We exploit this connection to establish an isomorphism that underpins probabilistic tape diagrams, a graphical formalism for bimonoidal (also known as rig) categories, and we demonstrate its effectiveness by providing a complete axiomatisation of probabilistic Boolean circuits.
\end{abstract}

\input{sections/intro}

\input{sections/pca}
\input{sections/cbproducts}


\input{sections/cmatrix}

\input{sections/syntactic}

\input{sections/probtapes}

\input{sections/probbooleancircuits}

\input{sections/probpartialbooltapes-new}
 \input{sections/finpac}

\input{sections/conclusion}

\section*{Acknowledgements}
This research was partly funded by the Advanced Research + Invention Agency (ARIA) Safeguarded AI Programme. Bonchi is supported by the Ministero dell'Università e della Ricerca of Italy grant PRIN 2022 PNRR No. P2022HXNSC - RAP (Resource Awareness in Programming). This study was carried out within the National Centre on HPC, Big Data and Quantum Computing - SPOKE 10 (Quantum Computing) and received funding from the European Union Next-GenerationEU - National Recovery and Resilience Plan (NRRP) – MISSION 4 COMPONENT 2, INVESTMENT N. 1.4 – CUP N. I53C22000690001.

\bibliographystyle{alphaurl}
\bibliography{references}

\appendix

\input{appendices/coherences}

\input{appendices/appendicerefusiprimaversione.tex}

\input{appendices/appconvbiprodcat.tex}

\input{appendices/appcmatrices.tex}
\input{appendices/apptc.tex}

\input{appendices/appprobtapes.tex}
\input{appendices/appprobboolcircuits.tex}
\input{appendices/appprobpartialbooleantapes-new.tex}

\input{appendices/appfinpac.tex}

\end{document}

%% file: preamble.tex
\usepackage{array}
\usepackage{adjustbox}
\usepackage{tikz-cd}
\usepackage{tikz}
\usetikzlibrary{cd,arrows}
\usepackage{mathpartir}
\usepackage{caption}
\usepackage{subcaption}
\usepackage{amsfonts}
\usepackage{amsmath}
\usepackage[renew-matrix,renew-dots]{nicematrix}
\usepackage{float}
\usepackage{enumitem}
\usepackage{xcolor}
\usepackage{booktabs}
\usepackage{rotating}
\usepackage{multirow}
\usepackage{comment}
\usepackage{graphicx}
\usepackage{braket}
\usepackage{enumitem}
\usepackage{xparse}
\usepackage{xstring}
\usepackage{mathtools}
\usepackage{xifthen}
\usepackage{makecell}
\usepackage{txfonts}
\usepackage{hyperref}
\usepackage{cleveref}
\hypersetup{
    colorlinks,
    linkcolor={black!50!black},
    citecolor={blue!50!black},
    urlcolor={blue!80!black}
}
\usepackage{mathtools}
\usepackage{changepage}

\usepackage{quiver}

\usepackage{todonotes}
\newcommand{\cipr}[1]{\todo[color=green!30,inline,caption={}]{\textbf{C:} #1}}

\usepackage{tabularx}

\usepackage{wrapfig}

%% file: sections/intro.tex
\section{Introduction}\label{sec:intro}

Driven by the growing interest in compositional semantics for probabilistic systems, a substantial body of work~\cite{cho2019disintegration,jacobs2019causal,piedeleu2025boolean,DBLP:journals/corr/abs-2410-10627,DBLP:journals/corr/abs-2502-03477,DBLP:journals/corr/abs-2301-12989,jacobs2021logical,moss2023category,fritz_2020,fritz2021finetti,fritz2023dilations,fritz2018bimonoidal,stein2024probabilistic,perrone2023markov,fritz2023d,fritz2023weakly} makes use of \emph{string diagrams}~\cite{joyal1991geometry,selinger2010survey}, which provide a graphical language for morphisms in the free strict symmetric monoidal category generated by a monoidal signature. These diagrams are typically interpreted in \( \KlD \), the Kleisli category of the subdistribution monad \( \Dis \) (or suitable variants), with monoidal structure given by the cartesian product \( \otimes \).
Much of this research focuses on the role of the \emph{copier} \( \CBcopier \) and the \emph{discharger} \( \CBdischarger \), and their correspondence to marginals and joint distributions.

In contrast, the second monoidal tensor \( \oplus \) in $\KlD$--corresponding to disjoint union--has received comparatively little attention~\cite{fritz2009presentation}, despite its crucial interaction with \( \otimes \) in many probabilistic frameworks~\cite{stein2024probabilistic,introductioneffectus,DBLP:journals/pacmpl/LiellCockS25}. A possible explanation is that standard string diagrams are inherently limited to representing a single monoidal structure at a time. 

To overcome this limitation, we consider \emph{tape diagrams}~\cite{bonchi2023deconstructing}, which can be informally understood as string diagrams of string diagrams: the tensor \( \otimes \) is represented by vertical composition of inner diagrams, while \( \oplus \) is captured by vertical composition of outer diagrams.
\[
\resizebox{!}{3.2em}{$
    \InputIfFileExists{tapes/cipriano/composizione3senzap.tikz}{}{\input{./tikz/tapes/cipriano/composizione3senzap.tikz}}
$}
\]

In the diagram above, the boxes $f,g,h,i$ denote arrows
$f\colon A \to A$, $g\colon A\otimes A \to A\otimes A$, 
$h\colon A\otimes A \to 1$, and $i\colon A \to A$
in a category $\Cat{C}$ equipped with a monoidal product $\otimes$ and unit $1$.
Moreover, $\Cat{C}$ should have a second monoidal product $\oplus$,
which is required to be a \emph{biproduct}, i.e., it serves both as a categorical
product and coproduct: the two \emph{splits} of tapes (on the left of the diagram)
are given by the pairing of the identity on $A\otimes A$, namely
\[
\raisebox{-0.5ex}{$
    \InputIfFileExists{tapes/cipriano/diagapera.tikz}{}{\input{./tikz/tapes/cipriano/diagapera.tikz}}
$} \qquad  \qquad \langle \id{A\otimes A}, \id{A\otimes A}\rangle 
\colon A \otimes A \to (A\otimes A)\oplus (A\otimes A)
\]
while the two \emph{joins} of tapes (on the right) are given by copairings of identities:
\begin{center}
\begin{minipage}[t]{0.48\textwidth}
\centering

    \InputIfFileExists{tapes/cipriano/codiagunoperuno.tikz}{}{\input{./tikz/tapes/cipriano/codiagunoperuno.tikz}}

\vspace{0.5em}
\[
[\id{1},\id{1}] \colon 1\oplus 1 \to 1
\]
\end{minipage}
\hfill
\begin{minipage}[t]{0.48\textwidth}
\centering

    \InputIfFileExists{tapes/cipriano/codiagapera.tikz}{}{\input{./tikz/tapes/cipriano/codiagapera.tikz}}

\vspace{0.5em}
\[
[\id{A\otimes A}, \id{A\otimes A}] \colon (A \otimes A)\oplus (A \otimes A) \to A\otimes A
\]
\end{minipage}
\end{center}

To avoid the restriction to biproducts, tape diagrams have recently been extended
in~\cite{bonchi2025tapediagramsmonoidalmonads} to represent arrows of Kleisli
categories for arbitrary monoidal monads, and thus in particular for $\KlD$.
In tapes for this category, hereafter referred to as probabilistic tapes, joins remain as above, while splits are replaced
by \emph{probabilistic splits}:
\[
    \InputIfFileExists{tapes/cipriano/prob.tikz}{}{\input{./tikz/tapes/cipriano/prob.tikz}}
\]
for $p \in (0,1)$. Intuitively, information flows to the top branch with probability
$p$ and to the bottom branch with probability $1-p$. Formally, the above tape
represents the $\KlD$ arrow $A \to A\oplus A$ that maps any $a \in A$ to $a$ in
the first copy of $A$ with probability $p$, and to $a$ in the second copy of $A$
with probability $1-p$.

\medskip

Our starting observation is that the coproduct \( \oplus \) in \( \KlD \) satisfies an additional universal property, akin to that of a product, which makes it resemble a biproduct. We refer to categories in which \( \oplus \) enjoys this property as \emph{convex biproduct categories} (Definition~\ref{def:convbicat}), and we study their associated monoidal algebras. While in categories with finite biproducts objects are equipped with natural and coherent monoid and \emph{comonoid} structures that provide joins and splits, in convex biproduct categories objects carry instead natural and coherent monoid and \emph{co-pointed convex algebra} structures. These are dual to pointed convex algebras (pca)~\cite{stone1949postulates,bonchi2017power,DBLP:journals/lmcs/BonchiSV22,DBLP:conf/lics/MioSV21} --the Eilenberg-Moore algebras for the monad $\Dis$-- and they provide probabilistic splits.

There is a further striking analogy with the theory of finite biproduct categories: given any category \( \Cat{C} \) enriched over commutative monoids, one can construct the category \( \Cat{Mat}(\Cat{C}) \) of matrices over \( \Cat{C} \)~\cite{mac_lane_categories_1978}, which possesses finite biproducts. Analogously, starting from a category \( \Cat{C} \) enriched over pcas, one can define the convex biproduct category \( \stmat{\Cat{C}} \) of \emph{stochastic matrices} over \( \Cat{C} \) (Proposition~\ref{prop:cmatrixcb}). Composition is defined via matrix multiplication: sums are given by the convex sums of the pca-enrichment and products by composition of arrows in $\Cat{C}$.
We show that this construction yields a functor $\stmat{-}$ from the category $\Cat{PCACat}$ of pca-enriched categories to the category $\Cat{CBCat}$ of convex biproduct categories, which is left adjoint to the forgetful functor \( U \) (Theorem~\ref{thm:matfree}):
\begin{equation}\label{eq:doubleadj}
\xymatrix{
\Cat{Cat} \ar@/^/[rr]^{(-)^+}& \;\;{\tiny{\bot}} & \Cat{PCACat}  \ar@/^/[ll]^{U} \ar@/^/[rr]^{\stmat{-}} &\;{\tiny{\bot}}  & \Cat{CBCat} \ar@/^/[ll]^{U}
}
\end{equation}
The leftmost adjunction is standard (see, e.g.,~\cite{borceux2,villoria2024enriching}): the functor \( (-)^+ \) freely enriches a locally small category \( \Cat{C} \) over pcas. The resulting category \( \Cat{C}^+ \) has arrows given by subdistributions of arrows in \( \Cat{C} \).

Our key result is that the composite of the two adjunctions admits a syntactic presentation in terms of generators and equations. Specifically, we define a functor \( \CatT{-} \colon \Cat{Cat} \to \Cat{CBCat} \) that freely adds to each category \( \Cat{C} \) a monoidal structure with natural monoids, and co-pointed convex algebras. We show that \( \CatT{-} \) is left adjoint to the forgetful functor \( U \colon \Cat{CBCat} \to \Cat{Cat} \) (Theorem~\ref{thm:syntacticadjunction}) and we thus conclude that $\CatT{\Cat{C}}$ is isomorphic to $\stmat{\Cat{C}^+}$ (Corollary \ref{cor:isotapematrices}). 


Importantly, when \( \Cat{C} \) is a category of string diagrams \( \DiagS \), \( \CatT{\DiagS} \) coincides with the category of probabilistic tape diagrams from~\cite{bonchi2025tapediagramsmonoidalmonads}. A result from~\cite{bonchi2025tapediagramsmonoidalmonads} shows that \( \CatT{\DiagS} \) admits two monoidal products, \( \oplus \) and \( \otimes \), forming a \emph{rig category} (also known as a bimonoidal category)~\cite{laplaza_coherence_1972,johnson2021bimonoidal}. The isomorphism \( \CatT{\DiagS} \cong \stmat{\DiagS^+} \) then provides a concrete interpretation of tapes as stochastic matrices whose entries are  subdistributions of string diagrams. Interestingly, $\piu$ is direct sum of matrices and $\otimes$ is (an extended) Kronecker product.

We argue that this isomorphism provides a useful and robust tool for deriving axiomatisations of probabilistic languages based on string diagrams. As a case study, we focus on \emph{probabilistic Boolean circuits with explicit conditioning} introduced in~\cite{piedeleu2025boolean}. As observed in~\cite[Example~30]{bonchi2025tapediagramsmonoidalmonads}, such circuits can be faithfully encoded as tape diagrams over \emph{partial} Boolean circuits, yielding a more expressive and structurally transparent account of probabilistic control. Our development proceeds in two steps. First, we establish a complete axiomatisation of partial Boolean circuits (Theorem~\ref{thm:completenesspartialcircuits}), a result that, to the best of our knowledge, has not appeared in the literature before. This provides a foundational equational theory for the underlying deterministic fragment. Building on this, we extend the axiomatisation by combining the axioms for partial Boolean circuits (Figures~\ref{tab:booleanalgebra} and \ref{tab:partialbooleanalgebra}) with those governing probabilistic tapes (Figure~\ref{fig:tapesax}), together with three additional axioms (Figure~\ref{ax:BooleanTAPES}). We show that the resulting system is complete (Corollary~\ref{cor:completenessPBPtapes}), thereby yielding a full axiomatisation of the probabilistic language under consideration.

\medskip

We conclude this introduction by providing some intuition for the isomorphism 
\(
\CatT{\DiagS} \cong \stmat{\DiagS^+}
\), compactly summarised in Figure~\ref{fig:dictionary}.
Consider the following probabilistic tape diagram and its associated stochastic matrix of subdistributions of string diagrams:
\begin{equation}\label{eq:tapeintro1}
\resizebox{!}{3.2em}{$
    \InputIfFileExists{tapes/cipriano/composizione2.tikz}{}{\input{./tikz/tapes/cipriano/composizione2.tikz}}
$}
\qquad
\qquad
\raisebox{-0.4ex}{$
\begin{pNiceMatrix}[first-col,first-row]
	\rotatebox{90}{$\Lsh$} & AA & A \\
	1 & p\cdot 
    \InputIfFileExists{tapes/cipriano/matcomp2h.tikz}{}{\input{./tikz/tapes/cipriano/matcomp2h.tikz}}
 & 0\cdot \star_{A,1} \\
	AA & (1-p)\cdot 
    \InputIfFileExists{tapes/cipriano/matcomp2idperi.tikz}{}{\input{./tikz/tapes/cipriano/matcomp2idperi.tikz}}
 & 1\cdot (
    \InputIfFileExists{tapes/cipriano/matcomp2j.tikz}{}{\input{./tikz/tapes/cipriano/matcomp2j.tikz}}
+_q 
    \InputIfFileExists{tapes/cipriano/matcomp2k.tikz}{}{\input{./tikz/tapes/cipriano/matcomp2k.tikz}}
)
\end{pNiceMatrix}
$}
\end{equation}
This tape represents an arrow of type $(A\otimes A) \oplus A \to 1\oplus (A\otimes A)$, where information flows from left to right. In the matrix, columns are indexed by the summands $A\otimes A$ and $A$ of the source, while rows are indexed by the summands $1$ and $A\otimes A$ of the target.  Information entering the tape in $A\otimes A$ can reach $1$ with probability $p$ via 
    \InputIfFileExists{tapes/cipriano/matcomp2h.tikz}{}{\input{./tikz/tapes/cipriano/matcomp2h.tikz}}
, or it can reach $A\otimes A$ with probability $1-p$ via 
    \InputIfFileExists{tapes/cipriano/matcomp2idperi.tikz}{}{\input{./tikz/tapes/cipriano/matcomp2idperi.tikz}}
. This explains the two entries in the leftmost column. For the rightmost column, observe that information entering into $A$ cannot reach $1$, as captured by the entry $0\cdot \star_{A,1}$. Instead, with probability $1$, it reaches $A\otimes A$ via 
    \InputIfFileExists{tapes/cipriano/matcomp2j.tikz}{}{\input{./tikz/tapes/cipriano/matcomp2j.tikz}}
 with probability $q$ and 
    \InputIfFileExists{tapes/cipriano/matcomp2k.tikz}{}{\input{./tikz/tapes/cipriano/matcomp2k.tikz}}
 with probability $1-q$, as represented by the entry $1\cdot (
    \InputIfFileExists{tapes/cipriano/matcomp2j.tikz}{}{\input{./tikz/tapes/cipriano/matcomp2j.tikz}}
+_q 
    \InputIfFileExists{tapes/cipriano/matcomp2k.tikz}{}{\input{./tikz/tapes/cipriano/matcomp2k.tikz}}
)$. Each entry of the matrix consists of a scalar $r\in[0,1]$ together with a subdistribution of string diagrams: for instance, in $1\cdot (
    \InputIfFileExists{tapes/cipriano/matcomp2j.tikz}{}{\input{./tikz/tapes/cipriano/matcomp2j.tikz}}
+_q 
    \InputIfFileExists{tapes/cipriano/matcomp2k.tikz}{}{\input{./tikz/tapes/cipriano/matcomp2k.tikz}}
)$ the scalar is $1$ and the distribution is $
    \InputIfFileExists{tapes/cipriano/matcomp2j.tikz}{}{\input{./tikz/tapes/cipriano/matcomp2j.tikz}}
+_q 
    \InputIfFileExists{tapes/cipriano/matcomp2k.tikz}{}{\input{./tikz/tapes/cipriano/matcomp2k.tikz}}
$. Moreover, in each column the scalars sum to $1$, so the matrix is stochastic. In general, however, these sums may be strictly less than $1$, in which case the matrices are more precisely \emph{substochastic}; for brevity, we simply refer to them as stochastic matrices.

In the following tape of type $(A\otimes A)\oplus(A\otimes A) \to (A\otimes A) \oplus A$, we make use of $\Tunit{A}$ which represents the unique arrow from the initial object $0$ to $A$.
\begin{equation}\label{eq:tapeintro2}
\resizebox{!}{3.2em}{$
    \InputIfFileExists{tapes/cipriano/composizione1.tikz}{}{\input{./tikz/tapes/cipriano/composizione1.tikz}}
$}
\qquad 
\qquad
\raisebox{-0.4ex}{$
\begin{pNiceMatrix}[first-col,first-row]
	\rotatebox{90}{$\Lsh$} & AA & AA \\
	AA & 1\cdot 
    \InputIfFileExists{tapes/cipriano/matcomp1idperf.tikz}{}{\input{./tikz/tapes/cipriano/matcomp1idperf.tikz}}
 & 1\cdot 
    \InputIfFileExists{tapes/cipriano/matcomp1g.tikz}{}{\input{./tikz/tapes/cipriano/matcomp1g.tikz}}
  \\
	A & 0\cdot \star_{AA,A} & 0\cdot \star_{AA,A}
\end{pNiceMatrix}
$}
\end{equation}

The composition of the  tape in \eqref{eq:tapeintro2} followed by the one in \eqref{eq:tapeintro1} is drawn --as expected-- as the horizontal composition of the two tapes:
\[
\scalebox{1.3}{
    \InputIfFileExists{tapes/cipriano/composizione4.tikz}{}{\input{./tikz/tapes/cipriano/composizione4.tikz}}
}
\]
By naturality of monoids and co-pcas--more precisely, by the laws \eqref{ax:tapes:diagpnat},\eqref{ax:tapes:codiagnat} and \eqref{ax:tapes:bangnat},\eqref{ax:tapes:cobangnat} in Figure~\ref{fig:tapesax}--the above tape is equivalent to the following one, from which the associated matrix can be more easily derived.
\begin{equation}\label{eq:tapeintro3}
\resizebox{!}{3.2em}{$
    \InputIfFileExists{tapes/cipriano/composizione3.tikz}{}{\input{./tikz/tapes/cipriano/composizione3.tikz}}
$}
\qquad
\raisebox{-0.4ex}{$
\begin{pNiceMatrix}[first-col,first-row]
	\rotatebox{90}{$\Lsh$} & AA & AA \\
	1 & p\cdot 
    \InputIfFileExists{tapes/cipriano/matcomp3fh.tikz}{}{\input{./tikz/tapes/cipriano/matcomp3fh.tikz}}
 & p\cdot 
    \InputIfFileExists{tapes/cipriano/matcomp3gh.tikz}{}{\input{./tikz/tapes/cipriano/matcomp3gh.tikz}}
 \\
	AA & (1-p)\cdot 
    \InputIfFileExists{tapes/cipriano/matcomp3if.tikz}{}{\input{./tikz/tapes/cipriano/matcomp3if.tikz}}
 & (1-p)\cdot 
    \InputIfFileExists{tapes/cipriano/matcomp3gi.tikz}{}{\input{./tikz/tapes/cipriano/matcomp3gi.tikz}}

\end{pNiceMatrix}
$}
\end{equation}
Note that, by multiplying the matrix in \eqref{eq:tapeintro1} with that in \eqref{eq:tapeintro2}, one obtains exactly the matrix in \eqref{eq:tapeintro3}. For instance, the top-left entry of \eqref{eq:tapeintro3} is computed as follows:
\begin{itemize}
\item the product of $p\cdot 
    \InputIfFileExists{tapes/cipriano/matcomp2h.tikz}{}{\input{./tikz/tapes/cipriano/matcomp2h.tikz}}
$ by $1\cdot 
    \InputIfFileExists{tapes/cipriano/matcomp1idperf.tikz}{}{\input{./tikz/tapes/cipriano/matcomp1idperf.tikz}}
 $ yields $p\cdot 
    \InputIfFileExists{tapes/cipriano/matcomp3fh.tikz}{}{\input{./tikz/tapes/cipriano/matcomp3fh.tikz}}
$;
\item the product of $0\cdot \star_{A,1}$ by $0\cdot \star_{AA,A}$ yields $0\cdot \star_{AA,1}$;
\item the sum of $p\cdot 
    \InputIfFileExists{tapes/cipriano/matcomp3fh.tikz}{}{\input{./tikz/tapes/cipriano/matcomp3fh.tikz}}
$ and $0\cdot \star_{AA,1}$ yields $p\cdot 
    \InputIfFileExists{tapes/cipriano/matcomp3fh.tikz}{}{\input{./tikz/tapes/cipriano/matcomp3fh.tikz}}
$.
\end{itemize}
The above products exploit multiplication of scalars and composition of string diagrams, while sums exploit convex sums of subdistributions.

\medskip


\emph{Synopsis.} We begin in Section~\ref{sec:pca} by recalling the notions of pointed convex algebras, pca-enriched categories, and the leftmost adjunction in \eqref{eq:doubleadj}. Section~\ref{sec:cbproducts} introduces convex biproduct categories and their associated monoidal algebraic structures. Furthermore, we illustrate a result that can be almost thought of as an analogue of Fox's theorem for convex biproduct categories. In Section~\ref{sec:cmatrix}, we define stochastic matrices and establish the rightmost adjunction in \eqref{eq:doubleadj}. The syntactic construction \( \CatT{-} \), corresponding to the composite adjunction, is presented in Section~\ref{sec:syntactic}. This construction is then applied to categories of string diagrams in Section~\ref{sec:tapediagrams} to obtain the isomorphism between probabilistic tape diagrams and stochastic matrices. Section~\ref{sec:probboolcircuits} provides a complete axiomatisation for partial Boolean circuits and revisits probabilistic Boolean circuits with explicit conditioning from~\cite{piedeleu2025boolean}. Section~\ref{sec:probbooltapes} illustrates their encoding as probabilistic tapes from~\cite{bonchi2025tapediagramsmonoidalmonads}, and presents an axiomatisation along with the completeness proof. Section~\ref{app:eff} illustrates the relationship of convex biproduct categories with finitely partially additive categories, while Section~\ref{sec:conclusion} discusses further related and future work.

This work is an extended version of the conference papers appearing in~\cite{probbooltapesfossacs,completenessprobbooltaperCONCUR26}. In addition to providing further examples and detailed proofs, it extends~\cite{probbooltapesfossacs,completenessprobbooltaperCONCUR26} in two directions. On the one hand, we present an analogue of Fox's theorem (Proposition~\ref{prop:A}), which is then used to simplify several proofs from~\cite{probbooltapesfossacs}. On the other hand, we provide a comparison between convex biproduct categories and finitely partially additive categories \cite{chophd,manes}, which underpin the theory of \textit{effectuses} ~\cite{introductioneffectus}.

%% file: sections/pca.tex
\section{Background}\label{sec:pca} 
We commence our exposition by recalling several well-known algebraic structures: pointed convex algebras, categories enriched over them and the category of sub-Markov kernels which is the archetypical example of a category with such an enrichment.

\subsection{Pointed Convex Algebras}\label{ssec:pca} A \emph{pointed convex algebra} (pca) \cite{stone1949postulates} consists of a set $X$, a designated element $\star \in X$ and, for all $p$ in the open real interval $(0,1)$, a function $+_p \colon X\times X \to X$ such that, by fixing $\tilde{p}\defeq pq$ and $\tilde{q}\defeq \frac{p(1-q)}{1-pq}$, the following laws hold for all $x_1,x_2,x_3\in X$. 
\begin{equation}\label{eq:pca}(x_1+_q x_2)+_p x_3 =  x_1 +_{\tilde{p}} (x_2+_{\tilde{q}} x_3) \qquad x_1+_px_2=x_2+_{1-p}x_1 \qquad  x_1+_p x_1 = x_1\end{equation}
A morphism of pcas is a function preserving $\star$ and $+_p$. 
We denote by $\Cat{PCA}$ the category of pcas and their morphisms.  

In any pca, $+_p$ can be defined for all $p\in[0,1]$ by setting $x+_1 y \defeq x$ and $x+_0 y \defeq y$. For all $n\in \mathbb{N}$, $p_1, \dots, p_n \in [0,1]$ such that $\sum_i p_i\leq 1$, one defines $\sum_{i=1}^n p_i\cdot (-)_i \colon X^n \to X$ inductively as
%
\begin{equation}\label{eq:def somma n pca}
    \sum_{i=1}^0 p_i\cdot (-)_i \defeq \star \qquad \sum_{i=1}^{n+1} p_i\cdot (-)_i \defeq (-)_{1} +_{p_{1}} \sum_{j=1}^n q_j \cdot (-)_j\text{.}
\end{equation}
where $(-)_j=(-)_{i+1}$ and $q_j$ is $\frac{p_{i+1}}{1-p_1}$ if $p_1 \neq 1$ and $0$ otherwise.
Note that when $n=1$, $\sum_{i=1}^1p_i \cdot (-)_i$ is, by definition, $(-)_1 +_{p_1} \star$. Since this operation will play a crucial role, we name it \emph{multiplication by a scalar} $p\in [0,1]$ and we fix the following notation. \[p\cdot x \defeq x+_p \star\] 

The standard example of a pca is provided by the set $\Dis(X)$ of finitely supported probability subdistributions over some set $X$. Recall that a finitely supported subdistribution over  $X$ is a function $d\colon X \to [0,1]$ such that $\sum_{x\in X}d(x)\leq1$ and $d(x)\neq 0$ for finitely many $x$. The distinguished element $\star$ in $\Dis(X)$ is the null subdistribution, i.e., $\star(x)\defeq 0$ for all $x\in X$; given $d_1,d_2\in \Dis(X)$, for all $p\in(0,1)$, $(d_1+_pd_2)(x) \defeq p\cdot d_1(x)+ (1-p)\cdot d_2(x)$. At this point, it is convenient to fix some extra notation: $\DisF(X)$ is the set of all finitely supported distributions (i.e., $\sum_{x\in X}d(x)=1$) and, for all $x\in X$, $\delta_x$ is the Dirac distribution at $x$ (i.e., $\delta_x(x')=1$ if $x=x'$ and $0$ otherwise).

Interestingly, the pointed convex algebra $\Dis(X)$  enjoys an additional property that will be relevant for our development: cancellativity. A pca $(X,+_p,\star)$ is \emph{cancellative} (or, in the terminology of \cite{sokolova2018termination}, cancellative at $\star$) if the following implication holds for all $x,y\in X$ and $p\in(0,1)$.
\begin{equation*}
p\cdot x = p\cdot y \Rightarrow x=y
\end{equation*}
\begin{remark}\label{rmk:collapsed}
The reader may have noticed that the axioms of a pca resemble the familiar laws of
associativity, commutativity, and idempotency for monoids. A crucial
difference, however, is that the distinguished element $\star$ is \emph{not} required
to act as a unit. Indeed, any pca that additionally satisfies the axiom
\begin{equation}\label{eq:unitstar}
x+_p \star = x
\end{equation}
collapses to a commutative idempotent monoid. More precisely, if one enriches
\eqref{eq:pca} with the condition above, one can easily show that
$
x+_p y = x+_q y$
for all $x,y\in X$ and $p,q\in(0,1)$,
so that the parameter $p$ becomes irrelevant. See, for instance, \cite{DBLP:journals/lmcs/BonchiSV22} for a detailed derivation.
\end{remark}


\subsection{Categories enriched over pcas} A category $\Cat{C}$ is \emph{$\Cat{PCA}$-enriched} if every homset $\Cat{C}[X,Y]$ carries a pca structure and composition of arrows is a pcas morphism, namely that the following equalities hold for all $p\in(0,1)$ and properly typed arrows $e,f,g,h$.
\begin{equation}\label{eq:enr}e; (f+_p g) = (e;f)+_p (e;g) \qquad (f+_pg) ;h= (f;h +_p g;h) \qquad f;\star =\star= \star;f\end{equation}
 It is now convenient to note a few properties of scalar multiplication in $\Cat{PCA}$-enriched categories.
\begin{lemma}\label{lemma:initialproperties2}
In a $\Cat{PCA}$-enriched category the following properties hold:
\begin{enumerate}
\item $(p\cdot f) ; g =p\cdot (f;g)$ and $f;(p\cdot g) = p\cdot ( f;g)$ ; \label{lemma:p;}
\item $p\cdot (q \cdot f) =pq \cdot f$; \label{lemma:pq}
\item $q\cdot\sum_{i=1}^{n}p_i\cdot f_i= \sum_{i=1}^{n} (qp_i)\cdot f_i$;\label{lemma q per somma}
\end{enumerate}
\end{lemma}

A functor $F$ is $\Cat{PCA}$-enriched if it preserves the pca structure of each homset. $\Cat{PCA}$-enriched categories and functors form a category, denoted by $\Cat{PCACat}$.
All categories considered hereafter are tacitly assumed to be locally small: $\Cat{Cat}$ stands  for the category of locally small categories. Note that there is a forgetful functor $U\colon \Cat{PCACat} \to \Cat{Cat}$.

Any category $\Cat{C}$ gives rise to the $\Cat{PCA}$-enriched category $\Cat{C}^+$: objects of  $\Cat{C}^+$ are those of $\Cat{C}$; for all objects $X,Y$, the homset is defined as $\Cat{C}^+[X,Y] \defeq \Dis(\Cat{C}[X,Y])$. For $d_1\colon X \to Y$ and $d_2\colon Y \to Z$, their composition $d_1;d_2 \colon X\to Z$ is defined for all $h\in \Cat{C}[X,Z]$ as $d_1;d_2(h) \defeq \sum_{\{(f,g)\mid f;g=h\}}d_1(f) \cdot d_2(g)$; The identity $\id{X}\colon X \to X$ is given by $\delta_{id_X}$. One can easily see that $\Cat{C}^+$ is a $\Cat{PCA}$-enriched category. Moreover, the assignment $\Cat{C} \mapsto \Cat{C}^+$ gives rise to a functor $(-)^+\colon \Cat{Cat} \to \Cat{PCACat}$ which is left adjoint to the forgetful functor $U$. 

\begin{theorem}[\cite{borceux2,villoria2024enriching}]\label{thm:freeenriched}
 $(-)^+\colon \Cat{Cat} \to \Cat{PCACat}$ is left adjoint to $U\colon \Cat{PCACat} \to \Cat{Cat}$.
\end{theorem}
\begin{proof}
The result follows from more general results about enriching on arbitrary algebraic theories: see for instance \cite[Prop. 6.4.7]{borceux2} or \cite[Cor. 1]{villoria2024enriching}. For later use, it is convenient anyway to briefly illustrate the involved structures. We first show the definition of $F^+\colon \Cat{C}^+ \to \Cat{D}^+$ for a functor $F\colon \Cat{C} \to \Cat{D}$. For all objects $X$, $F^+$ is defined as $F^+(X)\defeq F(X)$; For all $d\in \Cat{C}^+[X,Y]$, $F^+(d)\in \Cat{D}^+[FX, FY]$ is defined for all $g\colon FX \to FY$ as $F^+(d)(g)\defeq \sum_{\{f\in \Cat{C}[X,Y]| F(f)=g\}}d(f)$.  One can easily check that the functor is $\Cat{PCA}$-enriched.

For all categories $\Cat{C}$, the unit $\eta_{\Cat{C}}\colon \Cat{C} \to \Cat{C}^+$ is the functor acting as identity on objects and mapping any arrow $f\in \Cat{C}[X,Y]$ into $\delta_f \in \Cat{C^+}[X,Y]$.

Given a $\Cat{PCA}$-enriched category $\Cat{D}$ and a functor $F\colon \Cat{C} \to U(\Cat{D})$, one can define a $\Cat{PCA}$-enriched functor $F^\sharp \colon \Cat{C}^+ \to \Cat{D}$ as follows: for all objects $X$, $F^\sharp(X)\defeq F(X)$, for all arrows $d\in \Cat{C}^+[X,Y]$, $F^\sharp(d)\defeq \sum_{f\in \Cat{C}[X,Y]}d(f)\cdot {F(f)}$. One can easily check that $\eta_{\Cat{C}} ; F^\sharp=F$.   

\end{proof}

\begin{example}\label{ex:PCA}
Let $\Cat{1}$ be a category with a single object and a single arrow. The $\Cat{PCA}$-enriched category $\Cat{1}^+$ has a single object, arrows are $p\in[0,1]$ and composition is given by multiplication.

For a set $A$, the set of words  over $A$ (hereafter denoted by $A^*$) carries the structure of a category with a single object.  The $\Cat{PCA}$-enriched category $(A^*)^+$ has a single object and arrows are subdistributions over $A^*$. The composition of $d_1,d_2\in \Dis(A^*)$ is defined for all words $w\in A^*$ as $d_1;d_2(w)\defeq \sum_{u;v=w}d_1(u)\cdot d_2(v)$. 
\end{example}

\subsection{The category of sub-Markov kernels}
Our main example of a $\Cat{PCA}$-enriched category is $\KlD$, the Kleisli category of the monad $\subdistr\colon \Sets \to \Sets$ (see e.g.,~\cite{hasuo2007generic}).
Objects are sets; morphisms \(f \colon X \to Y\) are functions \(X \to \subdistr(Y)\). 
  We often write \(f(y \mid x)\) for \(f(x)(y)\), as this number represents the probability that \(f\) returns \(y\) given the input \(x\).
  Identities \(\id{X} \colon X \to \subdistr(X)\) map each element \(x \in X\) to \(\delta_{x}\).
  For two functions \(f \colon X \to \subdistr(Y)\) and \(g \colon Y \to \subdistr(Z)\), their composition in $\KlD$ is defined as \(f ; g (z \mid x) \defeq \sum_{y \in Y} f(y \mid x) \cdot g(z \mid y)\).

\begin{proposition}
$\KlD$ is $\Cat{PCA}$-enriched.
\end{proposition}
\begin{proof}
While this result is well known, it is anyway convenient to illustrate its proof. For all objects $X,Y$ of $\KlD$, the hom-set $\KlD[X,Y]$ is a pointed convex algebra: for all $f,g\colon X \to Y$, $x\in X$, $y\in Y$ and $p\in (0,1)$,
\[f +_p g(y|x)\defeq p\cdot f(y|x) + (1-p)\cdot g(y|x) \qquad \star_{X,Y}(y|x) \defeq 0.\] Note that $\star_{X,Y}$ is the arrow which sends every $x$ into the null subdistribution over $Y$. An easy calculation shows that arrow composition is a morphism of pointed convex algebras, i.e., it preserves the operations $+_p$ and $\star_{X,Y}$. For instance $(f+_pg) ;h= (f;h +_p g;h)$ for all $f,g\colon X \to Y$ and $h\colon Y \to Z$ is obtained as follows:
\begin{align}
    ((f+_pg) ;h)(z|x) &= \sum_{y\in Y} (f+_pg)(y|x) \cdot h(z|y) \tag*{} \\
    &= \sum_{y\in Y} (p\cdot f(y|x) + (1-p)\cdot g(y|x)) \cdot h(z|y) \tag*{}\\
    &= p\cdot \sum_{y\in Y} f(y|x) \cdot h(z|y) + (1-p)\cdot \sum_{y\in Y} g(y|x) \cdot h(z|y) \tag*{}\\
    &= p\cdot (f;h)(z|x) + (1-p)\cdot (g;h)(z|x) \tag*{}\\
    &= (f;h +_p g;h)(z|x)\text{.}\tag*{}
    \end{align}
\end{proof}

In our work, it is fundamental that $\KlD$ carries two symmetric monoidal structures: $ (\KlD, \per, \uno)$ and $(\KlD, \piu, \zero)$. The monoidal product $\per$ is defined on objects as the cartesian product of sets, often denoted by $\times$, with unit the singleton $\uno \defeq \{\bullet\}$; $\piu$ is the disjoint union of sets with unit the empty set $\zero \defeq \{\}$. Hereafter we denote the disjoint union of two sets $X$ and $Y$ by $X\oplus Y \defeq \{(x,0) \mid x\in X\} \cup \{(y,1) \mid y \in Y\}$ where $0$ and $1$ are tags used to distinguish the set of provenance. For arrows  \(f \colon X \to Y\) and \(g \colon X' \to Y'\), $f\per g\colon X\otimes X' \to Y \otimes Y'$ and $f \piu g \colon X\oplus X' \to Y\oplus Y'$ are defined, for all $x\in X$, $x'\in X'$, $y\in Y$, $y'\in Y'$, $u\in X\oplus X'$ and $v\in Y\oplus Y'$, as follows.
\begin{equation}\label{ex:products}
\begin{aligned}
f \per g(y,y' \mid x,x') &\defeq f(y \mid x) \cdot g(y' \mid x')\\
f\piu g(v \mid u) &\defeq 
\begin{cases} 
f(y \mid x) & \text{if } u=(x,0) \text{ and } v=(y,0)\\ 
g(y' \mid x') & \text{if } u=(x',1) \text{ and } v=(y',1) \\  
0 & \text{otherwise} 
\end{cases}
\end{aligned}
\end{equation}
Symmetries $\symmt{X}{Y}\colon X \otimes Y \to Y \otimes X$ and $\symmp{X}{Y}\colon X \oplus Y \to Y \oplus X$ are defined as in $\Sets$. 

Actually, $ (\KlD, \piu, \per, \zero, \uno)$ forms a rig (also known as bimonoidal) category in the sense of \cite{laplaza_coherence_1972}. We will come back to rig categories in Section \ref{sec:tapediagrams}. Until then, we will focus only on the monoidal structure provided by $\oplus$. 

%% file: sections/cbproducts.tex
\section{Convex Biproduct Categories}\label{sec:cbproducts}  
In this section we introduce categories with convex biproducts, study the underlying monoidal algebra and illustrate a result which can almost be thought of as an analogue of Fox's theorem.

\subsection{Convex products}\label{ssec:convexproducts} As is the case for all Kleisli categories of monads on $\Sets$, the category $\KlD$ inherits coproducts from $\Sets$: the operation $\oplus$  in \eqref{ex:products} serves as a coproduct in $\KlD$. Our initial observation is that $\oplus$ satisfies an additional universal property, which we refer to as the \emph{convex product}, described below.


\begin{definition}\label{def:convprod}
Let $X_1$, $X_2$  be two objects of a $\Cat{PCA}$-enriched category $\Cat{C}$. The \emph{convex product} of $X_1$ and $X_2$ is an object $Z$ with two arrows $\pi_1\colon Z\to X_1$ and $\pi_2\colon Z\to X_2$ satisfying the following property: for all $p_1,p_2 \in [0,1]$ such that $p_1+p_2\le 1$ and all arrows $f_1\colon A\to X_1$, $f_2\colon A \to X_2$, there exists a unique arrow $h \colon A \to Z$ 
making the following diagram commute: 
\[\begin{tikzcd}
	& A \\
	{X_1} & Z & {X_2}
	\arrow["{p_1\cdot f_1}"', from=1-2, to=2-1]
	\arrow["h", dashed, from=1-2, to=2-2]
	\arrow["{p_2\cdot f_2}", from=1-2, to=2-3]
	\arrow["{\pi_1}", from=2-2, to=2-1]
	\arrow["{\pi_2}"', from=2-2, to=2-3]
\end{tikzcd}\]
\end{definition}

Similarly, the convex product of $n$ objects $X_1,\ldots,X_n$ is an object $Z$ with arrows $\pi_i\colon Z\to X_i$ for $i=1,\ldots,n$ satisfying the following property: for all $p_1, \dots, p_n \in [0,1]$ where $\sum_{i=1}^n p_i \le 1$ and arrows $f_i\colon A\to X_i$, there exists a unique arrow $h\colon A \to Z$ such that $h;\pi_i = p_i\cdot f_i$ for all $i=1,\ldots,n$. Observe that, by definition, the 0-ary convex product is a final object. Hereafter, we will denote the unique arrow $h$ by $\langle f_1, \dots, f_n \rangle_{\vec{p}}$ where $\vec{p}$ is a compact notation for $p_1, \dots , p_n$.

\begin{lemma}\label{lemma:naryconvexproduct}
Let $\Cat{C}$ be a $\Cat{PCA}$-enriched category with final object and binary convex products. Then, for all $n\in \mathbb{N}$, $\Cat{C}$ has $n$-ary convex products.
\end{lemma}
%

Recall from Remark \ref{rmk:collapsed} that pcas that additionally satisfy the axiom \eqref{eq:unitstar} collapse to commutative idempotent monoids.
Observe that whenever a category $\Cat{C}$ is enriched over such pcas, it holds that $p\cdot f = f$ for all arrows $f$ and $p\in(0,1)$ and thus convex products collapse to standard products:
\begin{proposition}\label{prop:productvsconvex}
Let $\Cat{C}$ be a category enriched over commutative idempotent monoids. Let $X_1,X_2,Z$ be objects of $\Cat{C}$, and let $\pi_1\colon Z \to X_1$ and $\pi_2\colon Z \to X_2$ be arrows of $\Cat{C}$. 
\begin{center}$(Z,\pi_1,\pi_2)$ is a  product  of $X_1$ and $X_2$ iff it is a convex product of $X_1$ and $X_2$.\end{center}
\end{proposition}

By means of the proposition above, the reader can easily see that $\Cat{Rel}$, the category of sets and relations, has convex products, although these are just ordinary products. A  category with more interesting convex products is $\KlD$.

\begin{proposition}
$\KlD$ has convex products.
\end{proposition}
\begin{proof}
Let $X_1$ and $X_2$ be two objects of $\KlD$.
We show that  $X_1\oplus X_2$ with the projections $\pi_i\colon X_1 \oplus X_2 \to X_i$ for $i=1,2$ given by 
\[ \pi_1(x_1|z)\defeq \begin{cases}
    1 & z=\iota_1(x_1)\\
    0 & otherwise
\end{cases} \qquad \pi_2(x_2|z)\defeq \begin{cases}
     1 & z=\iota_2(x_2)\\
        0 & otherwise
\end{cases}\]
 is a convex product of $X_1$ and $X_2$. Above, the functions $\iota_i\colon X_i \to X_1 \oplus X_2$ denote the obvious embeddings.
 
For all arrows $f_1\colon A\to X_1$, $f_2\colon A\to X_2$ and $p_1,p_2 \in (0,1)$ such that $p_1+p_2\le 1$, the arrow $\langle f_1,f_2\rangle_{p_1,p_2}\colon A \to X_1\oplus X_2$ is defined for all $a\in A$ and $z\in X_1\oplus X_2$ as follows:
\[ \langle f_1,f_2\rangle_{p_1,p_2}(z|a) \defeq \begin{cases}
    p_1\cdot f_1(x_1|a) & z=\iota_1(x_1)\\
    p_2\cdot f_2(x_2|a) & z=\iota_2(x_2) \text{.}
\end{cases}\]
Observe that, for $i\in \{1,2\}$, $x_i\in X_i$ and $a\in A$
\begin{align*}
\langle f_1,f_2\rangle_{p_1,p_2} ; \pi_i (x_i | a) = & \sum_{z\in X_1\oplus X_2} \langle f_1,f_2\rangle_{p_1,p_2}(z|a) \cdot \pi_i (x_i | z) \\
= & \langle f_1,f_2\rangle_{p_1,p_2} (\iota_i(x_i)|a)\\
= & p_i\cdot f_i (x_i|a)
\end{align*}
namely, it holds that $(\langle f_1,f_2\rangle_{p_1,p_2} );\pi_i = p_i\cdot f_i$. In order to prove that this is the unique arrow with such a property, take $h'\colon A \to X_1\oplus X_2$ such that $h';\pi_i = p_i\cdot f_i$. For all $x_1\in X_1$ and $a\in A$, it holds that
\begin{align*}
  h'(\iota_1(x_1)|a) 
    & =\sum_{z\in X_1\oplus X_2} h'(z|a) \cdot \pi_1(x_1|z) \tag*{}\\
   &= h' ; \pi_1(x_1|a) \\
    &= p_1\cdot f_1(x_1|a) \tag*{}
\end{align*}
and similarly for $x_2\in X_2$, one obtains that $h'(\iota_2(x_2)|a) = p_2\cdot f_2(x_2|a)$, i.e., $h' = \langle f_1,f_2\rangle_{p_1,p_2}$. 
\end{proof}

\subsection{Convex biproduct categories and their monoidal algebra}\label{ssec:convexbiproduct} In both examples of categories with convex products, $\Cat{Rel}$ and $\KlD$, convex products are carried by coproducts (recall that in $\Rel$ products and coproducts coincide). In this work, we are interested in those categories where convex product and coproduct coincide, as formalised in  the following definition.  

\begin{definition}\label{def:convbicat}
A \emph{convex biproduct category} is a $\Cat{PCA}$-enriched category $\Cat{C}$ with an object $\zero$ which is both initial and final and, for every pair of objects $X_1,X_2$, an object $X_1\oplus X_2$ and morphisms $\pi_i \colon X_1\oplus X_2 \to X_i$ and $\iota_i \colon X_i \to X_1 \oplus X_2$ such that $(X_1\oplus X_2, \iota_1, \iota_2)$ is a coproduct, $(X_1\oplus X_2, \pi_1, \pi_2)$ is a convex product and
\begin{equation}\label{eq:delta}\iota_i; \pi_j = 
\begin{cases}  \id{X_i} & \text{if }i=j\text{,}\\
\star_{X_i,X_j} & \text{otherwise.}
\end{cases} \end{equation}
A morphism of convex biproduct categories is a $\Cat{PCA}$-enriched functor $F\colon \Cat{C} \to \Cat{D}$ preserving finite coproducts. We write $\Cat{CBCat}$ for the category of convex biproduct categories and their morphisms.
\end{definition}

The above definition is obtained from that of \emph{category with finite biproducts} by just replacing products by convex products.  As for categories with finite biproducts, morphisms of convex biproduct categories preserve (convex) products. 

 \begin{proposition}\label{prop:functor 1 e mezzo}
     Let $F\colon \Cat{C}\to \Cat{D}$ be a morphism between convex biproduct categories. 
     $F$ preserves n-ary convex products.
\end{proposition}

One can easily see that both $\Cat{Rel}$ and $\KlD$ are convex biproduct categories by checking \eqref{eq:delta}. In Appendix~\ref{app:examples}, we show that the continuous analogue of $\KlD$, namely substochastic Markov kernels on standard Borel spaces, is also a convex biproduct category.

\medskip

Since any convex biproduct category $\Cat{C}$ has finite coproducts, it carries a symmetric monoidal category $(\Cat{C}, \oplus, \zero)$. Hereafter we will denote the structural isomorphisms by $\alpha_{X,Y,Z}\colon (X\oplus Y)\oplus Z \to X \oplus (Y \oplus Z)$, $\lambda_X\colon 0 \oplus X \to X$, $\rho_X \colon 0 \oplus X \to X$ and the symmetries by $\sigma_{X,Y}\colon X\oplus Y \to Y \oplus X$.

For all objects $X$, we define $\codiag{X}\colon X \oplus X \to X$ as the copairing of the identities $[\id{X},\id{X}]$ and $\cobang{X} \colon 0 \to X$ as the unique arrow from the initial object. Similarly, 
 for all $p\in (0,1)$, we define  $\diagp{X} \colon X \to X \oplus X$ as $\langle id_X, id_X \rangle_{p,1-p}$, and $\bang{X}\colon X \to \zero$ as the unique map to the final object $\zero$. We extend $\diagp{X}$ to arbitrary $p\in [0,1]$: $\diagpX{0\;\;}{X}\defeq \cobang{X}\oplus \id{X}$ and $\diagpX{1\;\;}{X}\defeq \id{X} \oplus \cobang{X}$.

By Fox's theorem \cite{fox1976coalgebras}, $(X, \codiag{X}, \cobang{X})$ forms a \emph{natural and coherent commutative monoid} in the monoidal category  $(\Cat{C}, \oplus, \zero)$, namely the laws in Figures~\ref{fig:monoidax} and \ref{fig:fccoherence} in Appendix~\ref{app:coherence axioms} hold. Similarly $(X,\diagp{X}, \bang{X} )$ forms a \emph{natural and coherent co-pca} in $(\Cat{C}, \oplus, \zero)$, that is the laws in Figures~\ref{fig:co-pca axioms} and \ref{fig:copcacoherence} hold. To prove the latter fact, it is convenient to first illustrate some properties of convex biproduct categories.

\begin{lemma}\label{lemma:initialproperties}
Let $\Cat{C}$ be a convex biproduct category. The following equalities hold for all objects $Z,X,Y$ in $\Cat{C}$, arrows $f,g\colon X \to Y$ and $p\in [0,1]$:
\begin{enumerate}
\begin{multicols}{2}
\item[(1)]\mylabelbis{lemma:starxy}{1} $\star_{X,Y} = \bang{X};\cobang{Y}$
\item[(3)]\mylabelbis{lemma:pi1}{3}  $\pi_1=(\id{X_1}\oplus \bang{X_2});\runit{X_1}$;
\item[(5)]\mylabelbis{lemma:iotai1}{5}  $\iota_1=\Irunit{X_1}; (\id{X_1}\oplus \cobang{X_2})$;
\item[(7)]\mylabelbis{lemma:pf}{7} $p \cdot f =\, \diagp{X}; (f \oplus \bang{X} ) ; \runit{Y} $;
\item[(2)]\mylabelbis{eq: assioma aggiuntivo}{2} $f+_p g = \langle f,g \rangle_{p,1-p};[\id{Y},\id{Y}]$;
\item[(4)]\mylabelbis{lemma:pi2}{4} $\pi_2 =(\bang{X_1}\oplus \id{X_2});\lunit{X_2}$;
\item[(6)]\mylabelbis{lemma:iota2}{6} $\iota_2 = \Ilunit{X_2} ;(\cobang{X_1}\oplus \id{X_2})$;
 \item[(8)]\mylabelbis{lemma:1-pf}{8} $(1-p)\cdot  f=\, \diagp{X} ;  \, (\bang{X}  \oplus f \,)  ; \lunit{Y}$; 
 \end{multicols}
\end{enumerate}
\end{lemma}
\begin{proposition}\label{lemma: copca objects in convbicat}
Let $\Cat{C}$ be a convex biproduct category. Then for every object $X$,  
\begin{enumerate}
 \item the triple  $(X, \codiag{X}, \cobang{X})$ is a natural and coherent monoid, i.e., the axioms in Figures~\ref{fig:monoidax} and \ref{fig:fccoherence} hold;
\item the triple $(X,\diagp{X},\bangp{X})$ is a natural and coherent co-pca, i.e., the axioms in Figures~\ref{fig:co-pca axioms} and \ref{fig:copcacoherence} hold. 
 \end{enumerate}
\end{proposition}

For all sets $X$, co-pcas and monoids in $\KlD$ are illustrated below.
\begin{equation}\label{ex:comonoids}
\arraycolsep=2pt
\begin{array}{cccc}
\begin{array}{rcl}
\diagp{X} \colon  X & \to & X \oplus X \\
x & \mapsto & \delta_{(x,0)}+_p\delta_{(x,1)}
\end{array}
&
\begin{array}{rcl}
\bang{X}  \colon X & \to & \zero\\
x & \mapsto& \star
\end{array}
&
\begin{array}{rcl}
\cobang{X}\colon  \zero &\to& X\\
\text{ }
\end{array}
&
\begin{array}{rcl}
\codiag{X}  \colon  X\oplus X & \to & X\\
(x,i) & \mapsto& \delta_{x}
\end{array}
\end{array}
\end{equation}
Monoids and co-pcas will be useful later to freely generate convex biproduct categories. In particular, morphisms of convex biproduct categories  can be characterised as follows.

\begin{proposition}\label{prop: monoidal functors}
A functor $F\colon\Cat{C}\to \Cat{D}$ is a morphism of convex biproduct categories if and only if it is a strong monoidal functor preserving monoids and co-pcas.
\end{proposition}


\subsection{Almost a convex Fox theorem}\label{ssec:almost}
So far, we have shown that in every convex biproduct category each object carries natural and coherent structures of monoid and co-pca (Proposition~\ref{lemma: copca objects in convbicat}). Moreover, morphisms of convex biproduct categories can be characterised as precisely those monoidal functors that preserve these algebraic structures (Proposition~\ref{prop: monoidal functors}).
At this point, one might ask whether the converse holds: does any monoidal category equipped with natural and coherent monoid and co-pca structures give rise to a convex biproduct category? An affirmative answer would yield a characterisation of convex biproduct categories in the spirit of Fox's theorem~\cite{fox1976coalgebras}.
Unfortunately, this is not the case. As we shall see, an additional condition is required: the projections must be \emph{jointly monic}, i.e., for all $f,g\colon A \to X_1\oplus X_2$ the following implication holds.
\begin{center}
if $f;\pi_1=g;\pi_1$ and $f;\pi_2=g;\pi_2$ then $f=g$.
\end{center}
This condition, established in Proposition~\ref{prop:A}, will nevertheless serve as a useful principle throughout the development of our results.

We begin with the following definition, which fixes coprojections, projections, and the enrichment structure as described in Lemma~\ref{lemma:initialproperties}.
\begin{definition}\label{definition:structural}
Let $(\Cat{C}, \oplus, \zero)$ be a symmetric monoidal category such that every object $X\in\Cat{C}$ carries a natural and coherent monoid $(X,\codiag{X},\cobang{X})$ and co-pca $(X,\diagp{X},\,\bangp{X})$,  i.e., the axioms in Figures~\ref{fig:monoidax} -\ref{fig:copcacoherence} hold. 
For all objects $X_1,X_2$, $\iota_i\colon X_i \to X_1\oplus X_2$ and $\pi_i\colon X_1\oplus X_2 \to X_i$ are defined as
\begin{equation} \iota_1 \defeq \Irunit{X_1};(\id{X_1}\oplus \cobang{X_2})\text{ and } \iota_2 \defeq \Ilunit{X_2};(\,\cobang{X_1}\oplus \id{X_2})\end{equation}
\begin{equation}\label{eq:defpi}\pi_1 \defeq (\id{X_1}\oplus \bang{X_2});\runit{X_1} \quad \text{ and } \quad \pi_2 \defeq (\bang{X_1}\oplus \id{X_2});\lunit{X_2}\end{equation}
and, for all $f,g\colon X\to Y$,  $f +_p g\colon X \to Y$ and $\star_{X,Y}\colon X \to Y$ are defined as
\begin{equation}\label{eq:enrichment} f +_p g \defeq\, \diagp{X};(f \oplus g);\codiag{Y} \quad \text{ and } \quad \star_{X,Y} \defeq \bangp{X};\cobang{Y} \end{equation}
\end{definition}

Using the axioms of monoids and co-pcas (Figures~\ref{fig:monoidax}-\ref{fig:copcacoherence}), one verifies that the operations in~\eqref{eq:enrichment} endow $\Cat{C}$ with a $\Cat{PCA}$-enrichment.

\begin{lemma}\label{lemma:enrichmentpca}
Let $(\Cat{C}, \oplus, \zero)$ be a category as in Definition \ref{definition:structural}.  
Then $\Cat{C}$ is $\Cat{PCA}$-enriched. 
\end{lemma}

Furthermore, by Fox's theorem~\cite{fox1976coalgebras}, $(\Cat{C}, \oplus, 0)$ has finite coproducts; naturality of $\bangp{X}$ implies that $0$ is also a final object. Straightforward algebraic manipulations also establish~\eqref{eq:delta}.

\begin{lemma}\label{lemma:coproducts}
Let $(\Cat{C},\oplus,0)$ be a symmetric monoidal category as in Definition~\ref{definition:structural}. 
Then:
\begin{enumerate} 
\item for all objects $X_1,X_2$, $(X_1\oplus X_2, \iota_1,\iota_2)$ is a coproduct;
\item $0$ is initial object;
\item $0$ is a final object;
\item the equality \eqref{eq:delta} holds.
\end{enumerate} 
\end{lemma}

Now, to conclude that a monoidal category $(\Cat{C}, \oplus, \zero)$ as in Definition \ref{lemma:enrichmentpca} is a convex biproduct category, one should only prove that $(X_1\oplus X_2,\pi_1,\pi_2)$ is a convex product. 

Given $p,q\in [0,1]$ such that $p+q\le 1$,  define $ \diaggen{p,q\,}{X}\colon X \to X\oplus X $  as  
\begin{equation}\label{diagpq generalizzata}
    \diaggen{p,q\,}{X} \defeq\,\diagp{X};(\id{X}\oplus\, \diaggen{q'}{X});(\id{X}\oplus (\id{X} \oplus\bang{X}));(\id{X}\oplus \runit{X})
\end{equation}
where $q'=\frac{q}{1-p}$ if $p\not= 1$ and $q'=0$ otherwise.

\begin{lemma}\label{lemma:diagpq}
Let $(\Cat{C}, \oplus, \zero)$ be a symmetric monoidal category as in Definition~\ref{definition:structural}. 
Then: 
    \begin{enumerate}
        \item\label{lemma:diagpq1} $\diaggen{p,q\,}{X};\pi_1 = p\cdot \id{X}$;
        \item\label{lemma:diagpq2} $\diaggen{p,q\,}{X};\pi_2 = q\cdot \id{X}$;
       \item\label{lemma:diagpq4} $\diaggen{p,q\,}{X};(f_1\oplus f_2); \pi_1= p\cdot f_1$ for all $f_i\colon X \to Y_i$;
       \item\label{lemma:diagpq5} $\diaggen{p,q\,}{X};(f_1\oplus f_2); \pi_2= q\cdot f_2$ for all $f_i\colon X \to Y_i$.
    \end{enumerate}
\end{lemma}

By the last two items of the lemma above, we know that there exists \emph{at least one} arrow $h$, namely $\diaggen{p,q\,}{X};(f_1\oplus f_2)$, such that $h;\pi_1=p\cdot f_1$ and $h;\pi_2=q\cdot f_2$. However, uniqueness is not guaranteed. To obtain convex products, we need to assume that the projections are jointly monic.

\begin{proposition}\label{prop:A}
Let $(\Cat{C},\oplus,0)$ be a symmetric monoidal category as in Definition~\ref{definition:structural}. 
\begin{center}If the projections $\pi_1$ and $\pi_2$ are jointly monic, then $\Cat{C}$ is a convex biproduct category.\end{center}
\end{proposition}
\begin{proof}
By Lemmas~\ref{lemma:enrichmentpca} and \ref{lemma:coproducts}, we only need to prove that $(X_1\oplus X_2, \pi_1,\pi_2)$ is a convex product. 

Let $p,q\in [0,1]$ such that $p+q\le 1$ and $f_1\colon Z\to X_1$ and $f_2\colon Z\to X_2$.
By the last two items of Lemma \ref{lemma:diagpq}, $\diaggen{p,q\,}{Z};(f_1\oplus f_2); \pi_1= p\cdot f_1$ and $\diaggen{p,q\,}{Z};(f_1\oplus f_2); \pi_2= q\cdot f_2$. To prove that this is the unique arrow with such property, suppose that there is an arrow $h\colon Z \to X_1\oplus X_2$ such that $h;\pi_1 = p\cdot f_1$ and $h;\pi_2 = q\cdot f_2$. Then, since $\pi_1$ and $\pi_2$ are jointly monic, it follows that $h =\, \diaggen{p,q\,}{Z};(f_1\oplus f_2)$.
\end{proof}
The converse does not hold: in a convex biproduct category, projections need not be jointly monic.
Finally, Lemma~\ref{lemma:diagpq} yields the following characterisation of mediating morphisms in any convex biproduct category.

\begin{lemma}\label{lemma:diagpq3} 
Let $\Cat{C}$ be a convex biproduct category. For all  $p,q\in [0,1]$ such that $p+q\le 1$ and $f_1\colon Z\to X_1$ and $f_2\colon Z\to X_2$, it holds that
 \[\langle f_1, f_2\rangle_{p,q} =\, \diaggen{p,q\,}{Z}; (f_1 \oplus f_2)\text{.}\]
\end{lemma}

\begin{remark}
The axioms in Figures~\ref{fig:monoidax}, \ref{fig:fccoherence}, \ref{fig:co-pca axioms}, and \ref{fig:copcacoherence} have been placed in the Appendix, since their \emph{strict} counterparts already appear in Tables \ref{fig:freestrictfccat} and \ref{fig:freecopcacat}. From now on, we assume monoidal categories to be strict, meaning that the structural isomorphisms  $\assoc{X}{Y}{Z}$, $\lunit{X}$, $\runit{X}$ are identities. This assumption is made without loss of generality \cite{mac_lane_categories_1978}, it greatly simplifies calculations and allows for diagrammatic notations: see the last three rows of Figure~\ref{fig:tapesax} for a diagrammatic representation of the axioms of natural monoids and co-pcas.
\end{remark}

\begin{table}[t]
\resizebox{\textwidth}{!}{%
\begin{tabular}{cc cc}
    \toprule
    $(\id{ P}\piu \codiag{ P}) ; \codiag{ P} = (\codiag{ P}\piu \id{ P}) ; \codiag{ P}$ & (\newtag{$\codiag{}$-as}{eq:codiag assoc}) &  $(\cobang{ P}\piu \id{ P}) ; \codiag{ P}  = \id{ P}  $ & (\newtag{$\codiag{}$-un}{eq:codiag unital}) \\[0.3em]
    $\cobang{\zero} = \id{\zero} \qquad\codiag{\zero} = \id{\zero}$ & (\newtag{$\cobang{\zero},\codiag{\zero}$-coh}{eq:codiag zero coherence}) & $\sigma_{ P, P};\codiag{ P}=\codiag{ P}$ & (\newtag{$\codiag{}$-sym}{eq: codiag symmetry}) \\[0.3em]
    $\cobang{P \piu Q} = \cobang{P} \piu \cobang{Q}$ & (\newtag{\,$\cobang{}$-coh}{eq:cobang coherence}) & $\codiag{P \piu Q} = (\id{P} \piu \sigma_{Q,P} \piu \id{Q}) ; (\codiag{P}\piu \codiag{Q}) $ & (\newtag{$\codiag{}$-coh}{eq:codiag coherence}) \\[0.3em]
    $\cobang{P};f =\cobang{Q}$ & (\newtag{\,$\cobang{}$-nat}{eq:cobang nat}) & $\codiag{P};f =(f\piu f); \codiag{Q}$ & (\newtag{$\codiag{}$-nat}{eq:codiag nat}) \\
    \bottomrule
\end{tabular}}
\caption{Axioms for natural and coherent monoids in a strict monoidal category.}\label{fig:freestrictfccat}
\end{table}

\begin{table}[t]
\resizebox{\textwidth}{!}{%
\begin{tabular}{cc cc}
    \toprule
    $\diagp{P};(\diagq{P}\piu \id{P})=\, \diagptilde{P};(\id{P}\piu \,\diagqtilde{})$ & (\newtag{$\diagp{}$-as}{eq:diagp assoc}) &  $\tilde{p}= pq\qquad \tilde{q}= \frac{p(1-q)}{1-pq}$ &  \\[0.3em]
    $\diagp{P};\codiag{P}  = \id{ P}$ & (\newtag{$\diagp{}$-idem}{eq:diagp idempotency}) & $\diagp{P}\sigma_{ P, P}=\,\diagpbar{P}$ & (\newtag{$\diagp{}$-sym}{eq:diagp symmetry}) \\[0.3em]
    $\diagp{\zero} = \id{\zero}$ & (\newtag{$\diagp{\zero}$-coh}{eq:diagp zero coherence}) & $\bangp{\zero} = \id{\zero}$ & (\newtag{\,$\bangp{\zero}$-coh}{eq:bangp0 coherence}) \\[0.3em]
    $\bangp{P \piu Q} = \bangp{P} \piu\, \bangp{Q}$ & (\newtag{\,$\bangp{}$-coh}{eq:bangp coherence}) & $\diagp{P \piu Q} = (\diagp{P}\piu\, \diagp{Q});(\id{P} \piu \sigma_{P,Q} \piu \id{Q}) $ & (\newtag{$\diagp{}$-coh}{eq:diagp coherence}) \\[0.3em]
    $f; \bangp{Q}=\bangp{P}$ & (\newtag{\,$\bangp{}$-nat}{eq:bangp nat}) & $f; \diagp{Q}=\,\diagp{P}; (f \piu f)$ & (\newtag{$\diagp{}$-nat}{eq:diagp nat}) \\
    \bottomrule
\end{tabular}}
\caption{Axioms for natural and coherent co-pcas in a strict monoidal category.}\label{fig:freecopcacat}
\end{table}

%
%
%

%% file: sections/cmatrix.tex
%
%
\section{Stochastic Matrices over PCA-enriched categories}\label{sec:cmatrix}    
It is well-known (see, e.g., \cite[Exercises VIII.2.5-6]{mac_lane_categories_1978}) that, from a category enriched over commutative monoids, one can freely generate the finite biproduct category of its matrices. In this section, we show a similar construction for $\Cat{PCA}$-enriched categories: given such a category $\Cat{C}$, one can construct the convex biproduct category $\stmat{\Cat{C}}$ of \emph{stochastic matrices} over $\Cat{C}$.

Objects of $\stmat{\Cat{C}}$ are words in $Ob(\Cat{C})^*$. We will write $\bigoplus_{k=1}^m U_k$ for the word $U_1\dots U_m$ and $\zero$ for the empty word. We will denote objects of $\Cat{C}$ by $U,V$ and those of $\stmat{\Cat{C}}$ by $P,Q$. In $\stmat{\Cat{C}}$, an arrow $M\colon\bigoplus_{k=1}^n U_k\to \bigoplus_{k=1}^m V_k$ is an equivalence class of  $m\times n$ matrices with $(j,i)$-entries given by pairs $ (p_{ji}, f_{ji})$ where $f_{ji}\in\Cat{C}[U_i,V_j]$ and $p_{ji}\in[0,1]$ satisfy $\sum_{j=1}^{m}p_{ji}\le 1$. Two matrices $M$ and $M'$ are equivalent, in symbols $M\equiv M'$, if $p_{ji}\cdot f_{ji}= p'_{ji}\cdot f'_{ji}$ in $\Cat{C}[U_i,V_j]$ for all $i,j$. 
\begin{remark}
The use of pairs $ (p_{ji}, f_{ji})$ as entries of matrices is necessary to specify the constraints $\sum_{j=1}^{m}p_{ji}\le 1$. However, $\equiv$ ensures that one can safely write $M_{ji}= p_{ji}\cdot f_{ji}$.
\end{remark}

The composition of two morphisms $M\colon\bigoplus_{k=1}^n U_k\to \bigoplus_{k=1}^m V_k$ and $M'\colon\bigoplus_{k=1}^m V_k\to \bigoplus_{k=1}^l W_k$ 
is obtained by matrix multiplication $M'M$: for all $u\in \{1,\dots,l\}$ and $i\in \{1,\dots,n\}$ the entry at $(u,i)$ is 
\begin{equation}\label{eq:matrixmult}(M'M)_{ui} \defeq  r_{ui}\cdot \left(\sum_{j=1}^{m} \frac{p'_{uj}p_{ji}}{r_{ui}}\cdot (f_{ji};f'_{uj})\right) = \sum_{j=1}^{m} p'_{uj}p_{ji}\cdot (f_{ji};f'_{uj})\end{equation} 
where the convex sums $\sum_k p_k\cdot (-)_k$ are those provided by the $\Cat{PCA}$ enrichment of $\Cat{C}$, the composition $;$ is in $\Cat{C}$ and
$r_{ui}= (\sum_{j=1}^{m} p'_{uj}p_{ji})$. Simple computations confirm that $\sum_{u=1}^{l}r_{ui}\leq 1$ and $\frac{p'_{uj}p_{ji}}{r_{ui}}\in[0,1]$. Note that the rightmost equality in \eqref{eq:matrixmult} follows from Lemma~\ref{lemma:initialproperties2}.\ref{lemma q per somma}.
For all objects $P=\bigoplus_{k=1}^n U_k$, $\id{P}$ is the $n\times n$ matrix with entries $(\id{P})_{jj}=1\cdot \id{U_j}$ and, for $i\neq j$, $(\id{P})_{ji}=0 \cdot \star$.

\begin{example}\label{ex:matrices}
Recall $\Cat{1}^+$ and $(A^*)^+$ from Example~\ref{ex:PCA}.
In $\stmat{\Cat{1}^+}$ objects are natural numbers and arrows $n \to m$ are the usual $m\times n$ sub-stochastic matrices (i.e., $\sum_{j=1}^mp_{ji}\leq 1$).
In $\stmat{(A^*)^+}$ objects are natural numbers and arrows $n \to m$ are $m\times n$ matrices with entries in $\Dis(A^*)$. Consider for instance the matrices $N\colon 2\to 3$ and $M\colon 2 \to 2$ on the left below.  The composition $M;N\colon 2 \to 3$ is the matrix on the right.
\begin{equation*}
\begin{pmatrix}
    \frac{1}{2} \cdot a & \frac{1}{2} \cdot c\\
    \frac{1}{3} \cdot ab & 0 \cdot \star \\
    0 \cdot \star & \frac{1}{3} \cdot \id{}
\end{pmatrix} 
\begin{pmatrix}
    \frac{1}{2} \cdot a & 1 \cdot c\\
    \frac{1}{2} \cdot ab & 0 \cdot \star \\
\end{pmatrix} 
=
\begin{pmatrix}
    \frac{1}{2} \cdot (\frac{1}{2} \cdot aa +\frac{1}{2}abc) & \frac{1}{2} \cdot ca\\
    \frac{1}{6} \cdot aab & \frac{1}{3} \cdot cab \\
    \frac{1}{6} \cdot ab  & 0 \cdot \star
\end{pmatrix} 
\end{equation*}
\end{example}

%
%
For arbitrary matrices $M$ and $N$, $M\oplus N$ is defined as below on the left
\begin{equation}\label{eq:matsmc}
M \piu N \defeq \begin{pmatrix}
    M & \emptyset\\
    \emptyset & N
\end{pmatrix}
\qquad
\symm{P}{Q}\defeq \begin{pmatrix}
    \emptyset & \id{Q}\\
    \id{P} & \emptyset
\end{pmatrix}
\end{equation}
where $\emptyset$ is the matrix with all entries $ 0\cdot \star$. The symmetry $\symmp{P}{Q}\colon P\oplus Q \to Q \oplus P$ is defined as on the right above. 
%
%
%
%
%
%
%
Every object $P$ carries co-pca and monoid structures. These are defined for all objects $U\in \Cat{C}$ as follows, where $!_U$ (respectively $?_U$) is the unique matrix with $0$ rows (columns).
\begin{equation}\label{eq:matmonpca}
\diagp{U}\defeq\begin{pmatrix}
    p\cdot \id{U} \\
    (1-p)\cdot \id{U}
\end{pmatrix}
\qquad
\bang{U}\defeq !_U
\qquad 
\cobang{U}\defeq ?_U
\qquad
\codiag{U}\defeq \begin{pmatrix}
    1\cdot \id{U} & 1\cdot \id{U}
\end{pmatrix} 
\end{equation}
 For arbitrary objects $P$ of $\stmat{\Cat{C}}$, co-pcas and monoids are defined inductively as:
\begin{equation}\label{eq:indmoncopca}
\begin{array}{c|c}
\begin{array}{ccc}
 \diagp{\zero}\defeq \id{\zero}\; &\; \bang{\zero}\defeq \id{\zero}\; &\; \bang{U \oplus P}\defeq \bang{U}\oplus \bang{P}
 \end{array}
 &
 \begin{array}{ccc}
\codiag{\zero}\defeq \id{\zero}\; &\; \cobang{\zero}\defeq \id{\zero} \;& \;\cobang{U \oplus P}\defeq \cobang{U}\oplus \cobang{P}
\end{array} \\
\diagp{U \oplus P}\defeq (\diagp{U} \piu\, \diagp{P});(\id{U} \piu \symm{U}{P} \piu \id{P})&
\codiag{U\oplus P}\defeq (\id{U} \piu \symm{P}{U} \piu \id{P});(\codiag{U}\piu \codiag{P})
\end{array}
\end{equation}

%
%
%
%
\begin{lemma}\label{lemma:stmat coproduct}
Let $\Cat{C}$ be a $\Cat{PCA}$-enriched category. Every object of $\stmat{\Cat{C}}$ is a natural and coherent monoid object.
\end{lemma}

\begin{lemma}\label{lemma:stmat copca}
Let $\Cat{C}$ be a $\Cat{PCA}$-enriched category. Every object of $\stmat{\Cat{C}}$ is a natural and coherent co-pca object.
\end{lemma}

\begin{proposition}\label{prop:cmatrixcb}
     $\stmat{\Cat{C}}$ is a convex biproduct category.
\end{proposition}
\begin{proof}
    By Lemmas~\ref{lemma:stmat coproduct} and~\ref{lemma:stmat copca}, every object of $\stmat{\Cat{C}}$ carries natural and coherent monoid and a co-pca structures. 
Thanks to Proposition~\ref{prop:A}, we can now prove the statement by simply showing that $\stmat{\Cat{C}}$ has jointly monic projections.    
    

Consider two objects $Q_1=\bigoplus_{k=1}^{m_1} U_k$ and $Q_2=\bigoplus_{k=1}^{m_2} V_k$. Recall that $\pi_1\colon Q_1\piu Q_2 \to Q_1$ and $\pi_2\colon Q_1\piu Q_2 \to Q_2$ are 
$(\id{Q_1}\piu\, \bangp{Q_2})$ and $(\,\bangp{Q_1}\piu \id{Q_2})$ respectively, which by \eqref{eq:matsmc} and \eqref{eq:indmoncopca} are the $(m_1\times (m_1+m_2))$ and $(m_2\times (m_1+m_2))$ matrices displayed below.
\[\pi_1=\begin{pNiceMatrix}
1\cdot{\id{U_1}}  	& \emptyset & \Cdots & \emptyset & \emptyset  	& \emptyset & \Cdots & \emptyset \\
\emptyset  &   & \Ddots & \Vdots & \emptyset  &   &  & \Vdots \\	
\Vdots & \Ddots &   & \emptyset & \Vdots &  &   & \emptyset \\
\emptyset  & \Cdots & \emptyset  & 1\cdot{\id{U_{m_1}}}  & \emptyset  & \Cdots & \emptyset  & \emptyset
\CodeAfter
\line{1-1}{4-4}
\line{1-5}{4-8}
\tikz \draw[thin,dotted] (2-5.east) -- (4-7.north west);
\tikz \draw[thin,dotted] (1-6.east) -- (3-8.north west);
\end{pNiceMatrix}\]
\[ \pi_2=\begin{pNiceMatrix}
\emptyset  	& \emptyset & \Cdots & \emptyset & 1\cdot{\id{V_1}}  	& \emptyset & \Cdots & \emptyset \\
\emptyset  &   & \Ddots & \Vdots & \emptyset  &   &   & \Vdots \\	
\Vdots & \Ddots &   & \emptyset & \Vdots &  &   & \emptyset \\
\emptyset  & \Cdots & \emptyset  & \emptyset  & \emptyset  & \Cdots & \emptyset  & 1\cdot{\id{V_{m_2}}}
\CodeAfter
\tikz \draw[thin,dotted] (2-5.south east) -- (4-7.north west);
\tikz \draw[thin,dotted] (1-6.south east) -- (3-8.north west);
\line{1-1}{4-4}
\line{1-5}{4-8}
\end{pNiceMatrix}\]
%
Consider $M,N \colon \bigoplus_{k=1}^n W_k \to Q_1 \oplus Q_2$.
Both $M$ and $N$ are $(m_1+m_2)\times n$ matrices that can be viewed as block matrices $M=\begin{pmatrix} M_1 \\ M_2 \end{pmatrix}$ and $N=\begin{pmatrix} N_1 \\ N_2 \end{pmatrix}$
where $M_1,N_1$ are $(m_1\times n)$ stochastic matrices and $M_2,N_2$ are $(m_2\times n)$ stochastic matrices. If $M;\pi_1 = N;\pi_1$ and $M;\pi_2 = N;\pi_2$, by matrix multiplication it follows that $M_1=N_1$ and $M_2=N_2$ and hence $M=N$. 
\end{proof}

Just as categories of matrices are freely generated finite biproduct categories \cite{mac_lane_categories_1978}, similarly  $\stmat{\Cat{C}}$ is the convex biproduct category freely generated by $\Cat{C}$. We  illustrate this below.

First, for all $\Cat{PCA}$-enriched functors $F\colon \Cat{C}\to \Cat{D}$, one can define $\stmat{F}\colon \stmat{\Cat{C}}\to \stmat{\Cat{D}}$ as the functor mapping an object $\bigoplus_{k=1}^n U_k$ to $\bigoplus_{k=1}^n F(U_k)$ and mapping a matrix
$M\colon\bigoplus_{k=1}^n U_k\to \bigoplus_{k=1}^m V_k$ with entries $M_{ji}= p_{ji}\cdot f_{ji}$  into the matrix with entries $\stmat{F}(M)_{ji}\defeq p_{ji}\cdot F(f_{ji})$.
\begin{proposition}\label{prop:stmatfun}
    Let $F\colon \Cat{C}\to \Cat{D}$ be a $\Cat{PCA}$-enriched functor. $\stmat{F} \colon\stmat{\Cat{C}}\to \stmat{\Cat{D}}$ is a morphism of convex biproduct categories.
\end{proposition}
Then, it is easy to check that the assignment $\Cat{C} \mapsto \stmat{\Cat{C}}$ and $F \mapsto \stmat{F}$ provides a functor $\stmat{-}:\Cat{PCACat}\to \Cat{CBCat}$ which is left adjoint to the forgetful functor $U\colon \Cat{CBCat} \to \Cat{PCACat}$.
\begin{theorem}\label{thm:matfree}
    $\stmat{-} \colon \Cat{PCACat}\to \Cat{CBCat}$ is left adjoint to $U\colon \Cat{CBCat} \to \Cat{PCACat}$.
\end{theorem}
We conclude this section with a simple observation that will be useful later in Section~\ref{sec:probbooltapes}.
\begin{proposition}\label{prop:counit fullfaithful}
    For any convex biproduct category $\Cat{C}$, the counit of the adjunction above\\ $\epsilon_\Cat{C}\colon \stmat{U(\Cat{C})}\to \Cat{C}$ is a full and faithful morphism of convex biproduct categories.
\end{proposition}

%% file: sections/syntactic.tex
\section{An equational presentation of stochastic matrices}\label{sec:syntactic}    

\begin{table}[t]
\resizebox{\textwidth}{!}{%
$
\begin{array}{c}
\toprule
\begin{array}{c}
    \begin{array}{ccccc}
        {\diagp{U}\colon U \to U\oplus U} &
        {\bang{U} \colon U \to \zero} &
        {\sigma_{U, V}^{\piu} \colon U \piu V \to V \piu U} &
        {\codiag{U}\colon U \piu U \to U} &
        {\cobang{U} \colon \zero \to U}   
    \end{array}    
\\[0.3em]
    \begin{array}{ccccc}
        {id_\zero \colon \zero \to \zero} &
        {id_U \colon U \to U} &
        \inferrule{c \colon U \to V}{\tapeFunct{c}\colon U \to V} &
        \inferrule{\t \colon P \to Q \and \s \colon Q \to R}{\t ; \s \colon P \to R} &
        \inferrule{\t \colon P_1 \to Q_1 \and \s \colon P_2 \to Q_2}{\t \piu \s \colon P_1 \piu P_2 \to Q_1 \piu Q_2}       
    \end{array}
\end{array}
\\[0.5em]
\begin{array}{cc}
\midrule
\begin{array}{c}
(f;g);h=f;(g;h) \qquad id_P;f=f=f;id_Q\\
(f_1\piu f_2) ; (g_1 \piu g_2) = (f_1;g_1) \piu (f_2;g_2)
\end{array} 
&
\begin{array}{c}
id_{\zero}\piu f = f = f \piu id_{\zero} \qquad (f \piu g)\, \piu h = f \piu \,(g \piu h) \\
\sigma_{P, Q}; \sigma_{Q, P}= id_{P \piu Q} \qquad (\gen \piu id_R) ; \sigma_{Q, R} = \sigma_{P,R} ; (id_R \piu \gen)
\end{array}
\end{array}\\
\bottomrule
\end{array}
$
}
\caption{Typing rules (top) and axioms (bottom) for freely generated strict symmetric monoidal categories.}
\label{fig:freestricmmoncatax}
\end{table}

So far, we have seen that given a category, one can first freely enrich it over $\Cat{PCA}$, and then obtain a convex biproduct category by taking its category of stochastic matrices.
Now, we give a syntactic description, in terms of generators and equations, of such a category.

For a category $\Cat{C}$, we consider terms generated by the following grammar
  
\begin{equation}\label{tapesGrammar}
\setlength{\arraycolsep}{3pt}
\begin{array}{rcccccccccccccccccccc}
\t & ::= & \diagp{U} & \mid & \bangp{U} & \mid & \tapeFunct{c} & \mid & \cobang{U} & \mid & 
\codiag{U} & \mid & \id{U} & \mid & \id{\zero} & \mid & \sigma_{U,V}^{\piu} & \mid & \t ; \t & \mid & \t \piu \t
\end{array}
\end{equation}
%
%
where $p\in (0,1)$, $U,V \in \ob{\Cat{C}}$, and $c$ is an arrow in $\Cat{C}$. Terms are typed according to the rules at the top of Table~\ref{fig:freestricmmoncatax}: each type is an arrow $P \to Q$ where $P,Q\in \ob{\Cat{C}}^*$ are regarded as sums of objects of $\Cat{C}$. As expected, we consider only those terms that are typable. For arbitrary $P \in \ob{\Cat{C}}^*$, we define
$\diagp{P},\bangp{P},\cobang{P},\codiag{P}$ as in \eqref{eq:indmoncopca}. Analogous inductive definitions give us $\id{P}$ and $\symmp{P}{Q}$. 

The category $\CatTapeC$ has as set of objects $\ob{\Cat{C}}^*$. Arrows in $\CatTapeC[P,Q]$ are terms modulo the axioms of natural and coherent monoids, co-pcas, strict symmetric monoidal categories (respectively in Tables~\ref{fig:freestrictfccat}, \ref{fig:freecopcacat} and~\ref{fig:freestricmmoncatax} bottom) and the following two axioms.
\begin{equation}\label{eq:TapeFunctAxioms}
\tapeFunct{\id{P}}=\id{P} \qquad \tapeFunct{c;d} = \tapeFunct{c}; \tapeFunct{d} \tag{Tape}
\end{equation}
Identities, composition, symmetries and monoidal product are defined as for terms. It is thus immediate to see that $(\CatTapeC,\oplus, \zero)$ forms a symmetric monoidal category that, like in Definition~\ref{definition:structural}, is equipped with natural and coherent monoid and co-pca structures. To prove that $\CatTapeC$ is a convex biproduct category, by means of Proposition~\ref{prop:A}, we need to prove that projections are jointly monic.


%
%
%
%

We crucially rely on the fact that every object $P$ of $\CatTapeC$ is of the form $\bigoplus_{i=1}^nU_i$ where $\oplus$ is a coproduct by Lemma~\ref{lemma:coproducts}. Thanks to the universal property of coproducts, any $\t\colon\bigoplus_{j=1}^mV_j \to \bigoplus_{i=1}^nU_i$ is the copairing $[\t_1, \dots, \t_m]$ for $\t_j=\iota_j; \t$. One can thus restrict to considering the case of arrows of type $ V_j \to \bigoplus_{i=1}^nU_i$. These arrows enjoy a convenient normal form, illustrated in the next lemma. First, for all $n\in \mathbb{N}$ and $\vec{p}=p_1,\dots ,p_n$ with $p_i\in[0,1]$ such that $\sum_{i=1}^n p_i\leq 1$, we inductively define $\diagpn{\vec{p}\;}{U}{n} \colon U \to \bigoplus_{i=1}^nU$ as follows
\begin{equation}\label{eq: diagpn cap equational presentation}
\diagpn{\vec{p}\;}{U}{0}\defeq \bang{U} \qquad \diagpn{\vec{p}\;}{U}{n+1}\defeq\, \diagpn{\;p_{1}}{U}{}; (\id{U} \oplus \diagpn{\;\vec{q}\;}{U}{n})
\end{equation}
where $\vec{q}=q_1,\dots q_{n}$ for $q_i=0$ if $p_{1}= 1$ and $q_i=\frac{p_{i+1}}{1-p_{1}}$ otherwise. Note that the arrow in \eqref{eq: diagpn cap equational presentation} coincides with the one in \eqref{diagpq generalizzata} for $n=2$.
\begin{lemma}\label{decomposition}
For all $\t\colon U \to \bigoplus_{i=1}^n Q_i$, there exist $\vec{p}=p_1,\dots ,p_n$ and $t_i\colon U \to Q_i$ such that \[\t=\,\diagpn{\;\vec{p}\;}{U}{n}; \bigoplus_{i=1}^n \t_i\text{.}\]
\end{lemma}

The second key property is that the enrichment, defined as in \eqref{eq:enrichment}, is cancellative.

\begin{lemma}\label{cancellativity}[Cancellativity]
For all $r\in (0,1)$, for all $\s,\t \colon P \to Q$, 
if $r\cdot \s = r\cdot \t$ then $\s = \t$.
\end{lemma}

By means of the two lemmas above one can easily conclude that projections, defined as in \eqref{eq:defpi}, are jointly monic.

\begin{lemma}\label{lemma:TC jmono}
$\CatTapeC$ has jointly monic projections.
\end{lemma}

\begin{theorem}\label{thm:TCconvexbiproductcategory}
$\CatTapeC$ is a convex biproduct category. In particular, for all $\vec{p}=p_1,\dots ,p_n$ and $\t_i\colon U \to Q_i$, $\langle \t_1, \dots \t_n\rangle_{\vec{p}} = \diagpn{\;\;\vec{p}}{U}{n}; \bigoplus_{i=1}^n \t_i$.
\end{theorem}
\begin{proof}
    By construction, every object in $\CatTapeC$ is a natural and coherent comonoid and co-pca object, and Lemma~\ref{lemma:TC jmono} implies that the projections are jointly monic. Hence, $\CatTapeC$ is a convex biproduct category by Proposition~\ref{prop:A}. The last part of the statement follows from Lemma~\ref{decomposition} and uniqueness of mediating arrows for convex products.
\end{proof}
The assignment $\Cat{C} \mapsto \CatTapeC$ gives rise to a functor $\CatTapeFUN\colon \Cat{Cat} \to \Cat{CBCat}$ which is left adjoint to the forgetful functor  $U\colon \Cat{CBCat} \to \Cat{Cat} $.

\begin{theorem}\label{thm:syntacticadjunction}
$\CatTapeFUN\colon \Cat{Cat} \to \Cat{CBCat}$ is left adjoint to $U\colon \Cat{CBCat} \to \Cat{Cat} $.
\end{theorem}
\begin{proof}
For a functor $F\colon \Cat{C} \to \Cat{D}$, $\CatTapeF \colon \CatTapeC \to \CatTapeD$ is defined on objects as $\CatTapeF(\bigoplus_{i=1}^nU_i) = \bigoplus_{i=1}^nF(U_i)$.
On arrows, it is defined inductively as follows.
\[
\setlength{\arraycolsep}{2pt} 
\begin{array}{rcl rcl}
  \CatTapeF(\t_1 ; \t_2)      &\defeq& \CatTapeF(\t_1);\CatTapeF(\t_2) &
  \CatTapeF(\t_1 \piu \t_2)   &\defeq& \CatTapeF(\t_1)\oplus\CatTapeF(\t_2) \\
  
  \CatTapeF(\diagp{P})        &\defeq& \diagp{\CatTapeF(P)} &
  \CatTapeF(\codiag{P})       &\defeq& \codiag{\CatTapeF(P)} \\
  
  \CatTapeF(\,\bangp{P})        &\defeq& \bangp{\CatTapeF(P)} &
  \CatTapeF(\cobang{P})       &\defeq& \cobang{\CatTapeF(P)} \\
  
  \CatTapeF(\id{U})           &\defeq& \id{\CatTapeF(U)} &
  \CatTapeF(\id{\zero})       &\defeq& \id{\CatTapeF(\zero)} \\
  
  \CatTapeF(\sigma_{P,Q}^{\piu}) &\defeq& \sigma_{\CatTapeF(P),\CatTapeF(Q)}^{\piu} &
  \CatTapeF(\tapeFunct{c})    &\defeq& \tapeFunct{F(c)}
\end{array}
\]

Since by definition $\CatTapeF$ is a monoidal functor preserving co-pcas and monoids, by Proposition~\ref{prop: monoidal functors}, it is a morphism of convex biproduct categories.
Thus, we have a functor $\CatTapeFUN\colon \Cat{Cat} \to \Cat{CBCat}$. Below, we prove the adjunction.

For all categories $\Cat{C}$, the unit of the adjunction $\eta_{\Cat{C}} \colon \Cat{C} \to \CatTapeC$ is the identity-on-objects functor mapping each arrow $c\colon U\to V$ in $\Cat{C}$ into $\tapeFunct{c}$. The axioms in \eqref{eq:TapeFunctAxioms} force $\eta_{\Cat{C}}$ to be a functor. Naturality of the unit is straightforward.

Now, take a convex biproduct category $\Cat{D}$ and consider a functor $F\colon \Cat{C} \to U(\Cat{D})$. By Proposition~\ref{lemma: copca objects in convbicat}, $\Cat{D}$ is a symmetric monoidal category where every object $X$  carries a  co-pca $(\diagp{X},\bang{X})$ and a natural coherent monoid $(\codiag{X},\cobang{X})$. One can use these structures to define $F^\sharp\colon \CatTapeC \to \Cat{D}$ inductively in the same way as  $\CatTapeF$: e.g., $F^\sharp(\diagp{P})\defeq \,\diagp{F^\sharp(P)}$. The base case $F^\sharp(\tapeFunct{c})=F(c)$ ensures that $\eta_{\Cat{C}}; F^\sharp = F$.
Since by definition $F^\sharp$ is a strict symmetric monoidal functor preserving co-pcas and monoids, by Proposition~\ref{prop: monoidal functors}, it is a morphism of convex biproduct categories. 

For uniqueness, take a morphism of convex biproduct categories $H \colon \CatTapeC \to \Cat{D}$ such that $\eta_{\Cat{C}}; H = F$, i.e.,  $H(\tapeFunct{c})=F(c)$. By Proposition~\ref{prop: monoidal functors}, $H$  preserves co-pcas and monoids. Thus $H=F^\sharp$.
\end{proof}

\begin{corollary}\label{cor:isotapematrices}
For all categories $\Cat{C}$, $\CatTapeC$ is isomorphic to $\stmat{\Cat{C}^+}$ in $\Cat{CBCat}$.
\end{corollary}
\begin{proof}
By composition of adjoints, their uniqueness and Theorems~\ref{thm:freeenriched}, \ref{thm:matfree} and \ref{thm:syntacticadjunction}.
\end{proof}
The isomorphism $\CatTapeC \to \stmat{\Cat{C}^+}$ maps $\tapeFunct{c}$ into the 
$1\times 1$ matrix $\begin{pmatrix}
   1 \cdot c
\end{pmatrix}$ and  $\diagp{U}$, $\bang{U}$, $\cobang{U}$, $\codiag{U}$ into the matrices defined in \eqref{eq:matmonpca}.  Compositions and sums of arrows in $\CatTapeC$ 
are mapped into multiplications and direct sums of matrices as defined in \eqref{eq:matrixmult} and \eqref{eq:matsmc}. The same applies to identities and symmetries. For instance, $(\diagpX{\frac{1}{2}\;\,}{U} \oplus \id{U}) ; (\tapeFunct{a}\oplus \tapeFunct{ab} \oplus \tapeFunct{c}); (\id{U}\oplus \symmp{U}{U}); (\codiag{U}\oplus \id{U})$ is mapped into the matrix $M$ of Example~\ref{ex:matrices}. 

The characterisation of $\stmat{\Cat{C}^+}$ by means of generators and equations provided by the above corollary, paves the way to study further equational theories (such as the one in Section~\ref{sec:probbooltapes}). The following result guarantees that the category obtained by quotienting $\CatTapeC$ by additional axioms is still a convex biproduct category.
%

\begin{proposition}\label{cor: quotient category is convex biproduct }
Let $\Cat{C}$ be a category and let $\sim$ be a congruence relation (w.r.t.\ $;$ and $\oplus$) on $\CatTapeC$ which is cancellative: for every $p\in (0,1)$, $p\cdot f\sim p\cdot g$ then $f\sim g$. Let $\CatTapeC_\sim$ be the category obtained as the quotient of $\CatTapeC$ by $\sim$ and $Q_\sim \colon \CatTapeC\to \CatTapeC_\sim$ be the functor mapping each arrow to its $\sim$-equivalence class.
Then $\CatTapeC_\sim$ is a convex biproduct category and $Q_\sim$ is a morphism of convex biproduct categories.
\end{proposition}

%% file: sections/probtapes.tex
\section{Probabilistic Tape Diagrams}\label{sec:tapediagrams}
Recall that a \emph{monoidal signature} is a tuple $(\sort, \sign, \ari, \coar)$ where $\sort$ is a set of basic sorts, hereafter denoted by $A,B,\dots$, $\sign$ is a set of generators, denoted by $s,t, \dots$, and $\ari,\coar \colon \sign \to \sort^*$ assign to each symbol its arity and coarity (words over $\sort$). From a monoidal signature $(\sort, \sign, \ari, \coar)$, one can freely generate the strict symmetric monoidal category $\DiagS$:  objects are words in $\sort^*$; arrows are \emph{string diagrams} \cite{joyal1991geometry,selinger2010survey}. These can be regarded  as the terms generated by the following grammar (where $A,B \in \sort$ and $s \in \Sigma$)
\begin{equation}\label{stringdiagramGrammar}
\setlength{\arraycolsep}{3pt} 
\renewcommand{\arraystretch}{1.1} 
\begin{array}{rcccccccccccccccc}
c & ::= & \id{A} & \mid & \id{\uno} & \mid & \gen & \mid & \symmt{A}{B} & \mid & c ; c & \mid & c \per c
\end{array}
\end{equation}
modulo the axioms of strict symmetric monoidal categories. A \emph{monoidal theory} $\mathbb{T}=(\sign, E)$ consists of a monoidal signature equipped with a set $E$ of pairs of arrows of $\DiagS$ with same source and target. Let $=_{\mathbb{T}}$ be an equivalence relation on arrows  of $\DiagS$ obtained as the congruence closure (w.r.t. $;$ and $\otimes$) of $E$. We write $\DiagT$ for the category obtained as the quotient of $\DiagS$ by $=_{\mathbb{T}}$.

\medskip

The category $\CatT{\DiagT}$, obtained by applying the construction from Section~\ref{sec:syntactic} to $\DiagT$, coincides by definition with the category of \emph{tape diagrams} from \cite{bonchi2025tapediagramsmonoidalmonads}. Objects of $\CatT{\DiagT}$ are elements of $(\sort^*)^*$ which we often write as \emph{polynomials} $P=\Piu[i=1][n]{\Per[j=1][m_i]{A_{i,j}}}$. 
We will call \emph{monomials} of $P$ the $n$ words $\Per[j=1][m_i]{A_{i,j}}$. For instance, the monomials of $(A \per B) \piu \uno$ are $A \per B$ and $1$. We denote monomials by $U,V,\dots$

Arrows of $\CatT{\DiagT}$ enjoy an intuitive diagrammatic representation specified by the following two layered grammar.
\begin{equation*}\label{tapesDiagGrammar} 
    \setlength{\tabcolsep}{2pt}
    \begin{tabular}{rc c@{$\,\mid\,$}c@{$\,\mid\,$}c@{$\,\mid\,$}c@{$\,\mid\,$}c@{$\,\mid\,$}c@{$\,\mid\,$}c@{$\,\mid\,$}c@{$\,\mid\,$}c@{$\,\mid\,$}c}
        $c$  & $\Coloneqq$ &  $\wire{A}$ & $ 
    \InputIfFileExists{empty.tikz}{}{\input{./tikz/empty.tikz}}
 $ & $ \Cgen{\gen}{A}{B}  $ & $ \Csymm{A}{B} $ & $ 
    \InputIfFileExists{seq_compC.tikz}{}{\input{./tikz/seq_compC.tikz}}
   $ & $  
    \InputIfFileExists{par_compC.tikz}{}{\input{./tikz/par_compC.tikz}}
$ \\
        $\t$ & $\Coloneqq$ & $
    \InputIfFileExists{/tapes/cipriano/pcomonoid.tikz}{}{\input{./tikz//tapes/cipriano/pcomonoid.tikz}}
$ & $\Tcounit{U}$  & $ \Tcirc{c}{U}{V}$ & $\Tunit{U}$  & $\Tmonoid{U}$ \\ 
        && $\Twire{U}$ & $ 
    \InputIfFileExists{empty.tikz}{}{\input{./tikz/empty.tikz}}
 $    & $ \Tsymmp{U}{V} $ & $ 
    \InputIfFileExists{tapes/seq_comp.tikz}{}{\input{./tikz/tapes/seq_comp.tikz}}
  $ & $  
    \InputIfFileExists{tapes/par_comp.tikz}{}{\input{./tikz/tapes/par_comp.tikz}}
$  
    \end{tabular}
\end{equation*} 
The first layer corresponds to \eqref{stringdiagramGrammar} while the second to \eqref{tapesGrammar}; diagrams from the first layer are string diagrams, those from the second are called tapes. Note that (a) string diagram can occur inside tapes, (b) string diagrams have type $U\to V$, for $U,V\in \sort^*$ and vertical composition corresponds to $\per$, (c) tapes have type $P\to Q$ for $P,Q\in (\sort^*)^*$ and vertical composition corresponds to $\piu$. 

The identity $\id\zero$ is rendered as the empty tape $
    \InputIfFileExists{empty.tikz}{}{\input{./tikz/empty.tikz}}
$, while $\id\uno$ is $
    \InputIfFileExists{tapes/empty.tikz}{}{\input{./tikz/tapes/empty.tikz}}
$: a tape filled with the empty string diagram. 
For a monomial $U \!=\! A_1\dots A_n$, $\id U$ is depicted as a tape containing  $n$ wires labelled by $A_i$. For instance, $\id{AB}$ is rendered as $\TRwire{A}{B}$. When clear from the context, we will simply represent it as a single wire  $\Twire{U}$ with the appropriate label.
Similarly, for a polynomial $P = \Piu[i=1][n]{U_i}$, $\id{P}$ is obtained as a vertical composition of tapes, as illustrated below on the left. 
\begin{equation}\label{ex:tape}
\scalebox{0.82}{$ 
    \id{AB \piu \uno \piu C} = \!\!\!\begin{aligned}\begin{gathered} \TRwire{A}{B} \\[-1.8mm] \Twire{\uno} \\[-1.8mm] \Twire{C} \end{gathered}\end{aligned}
    \quad
    \codiag{A\piu B \piu C} = \!\!
    \InputIfFileExists{tapes/examples/codiagApBpC.tikz}{}{\input{./tikz/tapes/examples/codiagApBpC.tikz}}

    \quad
    \cobang{AB \piu B \piu C} = \!
    \InputIfFileExists{tapes/examples/cobangABpBpC.tikz}{}{\input{./tikz/tapes/examples/cobangABpBpC.tikz}}

    \quad
    
    \InputIfFileExists{tapes/cipriano/algt.tikz}{}{\input{./tikz/tapes/cipriano/algt.tikz}}

$} \end{equation}
\noindent The codiagonal $\codiag{U} \colon  U \piu U \!\to\! U$ is represented as a merging of tapes, $\diagp{U}\colon U\to U\oplus U $ as a splitting of tapes labeled with $p\in(0,1)$, the cobang $\cobang{U} \colon \zero \!\to\! U$ is a tape closed on its left boundary, while $\bang{U} \colon U \to \zero $ is closed on the right.
Exploiting the definitions in \eqref{eq:indmoncopca}, we can construct $\codiag{P},\diagp{P},\cobang{P},\bang{P}$ for arbitrary polynomials.
For example, $\codiag{A\piu B \piu C}$ and $\cobang{AB \piu B \piu C}$ are depicted as the second and third diagrams above. The last diagram is the tape for 
$\diagp{A};(\diagq{A}\oplus \id{A}) ; (\bang{A} \oplus \codiag{A})$. 
For an arbitrary $\t \colon P \to Q$, we write $\Tbox{\t}{P}{Q}$. 

The graphical representation embodies several axioms such as \eqref{eq:TapeFunctAxioms} and those of monoidal categories. 
Those axioms that are not implicit in the graphical representation are shown in Figure~\ref{fig:tapesax}. 

\medskip

By Corollary~\ref{cor:isotapematrices}, we know that $\CatT{\DiagT}$ is isomorphic to $\stmat{\DiagT^+}$: tapes represent exactly stochastic matrices of subprobability distributions of string diagrams in $\DiagT$. Figure \ref{fig:dictionary} illustrates  
how the isomorphism translates tape diagrams into such matrices. For instance, the rightmost diagram in \eqref{ex:tape} corresponds to the $1\times 1$ matrix $\begin{pmatrix} (1-p)(1-q) \cdot \wire{A}\end{pmatrix}$ while the following tape diagram on the left represents the $2\times 2$ stochastic matrix on the right.
\[
\resizebox{0.95\linewidth}{!}{$

    \InputIfFileExists{tapes/cipriano/esempiomatrice.tikz}{}{\input{./tikz/tapes/cipriano/esempiomatrice.tikz}}

\quad
\text{\raisebox{1.7ex}{$
\begin{pNiceMatrix}[first-col,first-row]
	\rotatebox{90}{$\Lsh$} & AB & C \\
	DE & p\cdot 
    \InputIfFileExists{par_compCMAT.tikz}{}{\input{./tikz/par_compCMAT.tikz}}
 & q\cdot \Cgen{h}{}{} \\
	F & (1-p)\cdot \star & (1-q)\cdot(
    \InputIfFileExists{seq_compCMAT1.tikz}{}{\input{./tikz/seq_compCMAT1.tikz}}
 +_r 
    \InputIfFileExists{seq_compCMAT2.tikz}{}{\input{./tikz/seq_compCMAT2.tikz}}
)
\end{pNiceMatrix}
$}}
$}
\]

\begin{figure}[t]
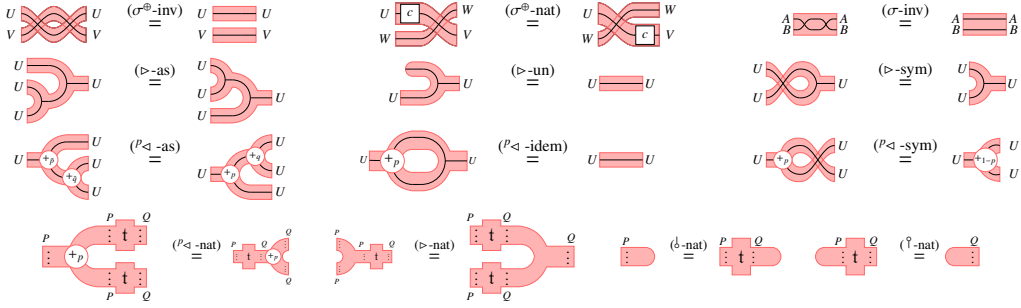

\centering
\setlength{\tabcolsep}{2pt}
\renewcommand{\arraystretch}{0.9}
\resizebox{0.92\linewidth}{!}{%
\begin{tabular}{r@{\;}c@{\;}l @{\qquad} r@{\;}c@{\;}l @{\qquad} r@{\;}c@{\;}l}
    \mylabelbis{ax:tapes:symminv}{\ensuremath{\symmp\text{-inv}}}%
    
    \InputIfFileExists{tapes/ax/symminv_left.tikz}{}{\input{./tikz/tapes/ax/symminv_left.tikz}}
 &$\stackrel{\eqref{ax:tapes:symminv}}{=}$& 
    \InputIfFileExists{tapes/ax/symminv_right.tikz}{}{\input{./tikz/tapes/ax/symminv_right.tikz}}

    &
    \mylabelbis{ax:tapes:symmnat}{\ensuremath{\symmp\text{-nat}}}%
    
    \InputIfFileExists{tapes/ax/symmnat_left.tikz}{}{\input{./tikz/tapes/ax/symmnat_left.tikz}}
 &$\stackrel{\eqref{ax:tapes:symmnat}}{=}$& 
    \InputIfFileExists{tapes/ax/symmnat_right.tikz}{}{\input{./tikz/tapes/ax/symmnat_right.tikz}}

    &
    \mylabelbis{ax:tapes:sigmainv}{\ensuremath{\sigma\text{-inv}}}%
    
    \InputIfFileExists{cb/symm_inv_left.tikz}{}{\input{./tikz/cb/symm_inv_left.tikz}}
 &$\stackrel{\eqref{ax:tapes:sigmainv}}{=}$& 
    \InputIfFileExists{cb/symm_inv_right.tikz}{}{\input{./tikz/cb/symm_inv_right.tikz}}

    \\
    \mylabelbis{ax:tapes:codiagas}{\ensuremath{\codiag{}\text{-as}}}%
    
    \InputIfFileExists{tapes/whiskered_ax/monoid_assoc_left.tikz}{}{\input{./tikz/tapes/whiskered_ax/monoid_assoc_left.tikz}}
 &$\stackrel{\eqref{ax:tapes:codiagas}}{=}$& 
    \InputIfFileExists{tapes/whiskered_ax/monoid_assoc_right.tikz}{}{\input{./tikz/tapes/whiskered_ax/monoid_assoc_right.tikz}}

    &
    \mylabelbis{ax:tapes:codiagun}{\ensuremath{\codiag{}\text{-un}}}%
    
    \InputIfFileExists{tapes/whiskered_ax/monoid_unit_left.tikz}{}{\input{./tikz/tapes/whiskered_ax/monoid_unit_left.tikz}}
 &$\stackrel{\eqref{ax:tapes:codiagun}}{=}$& \Twire{U}
    &
    \mylabelbis{ax:tapes:codiagsym}{\ensuremath{\codiag{}\text{-sym}}}%
    
    \InputIfFileExists{tapes/whiskered_ax/monoid_comm_left.tikz}{}{\input{./tikz/tapes/whiskered_ax/monoid_comm_left.tikz}}
 &$\stackrel{\eqref{ax:tapes:codiagsym}}{=}$& \Tmonoid{U}
    \\
    \mylabelbis{ax:tapes:diagpas}{\ensuremath{\diagp{}{}\text{-as}}}%
    
    \InputIfFileExists{tapes/cipriano/monoid_assoc_left.tikz}{}{\input{./tikz/tapes/cipriano/monoid_assoc_left.tikz}}
 &$\stackrel{\eqref{ax:tapes:diagpas}}{=}$& 
    \InputIfFileExists{tapes/cipriano/monoid_assoc_right.tikz}{}{\input{./tikz/tapes/cipriano/monoid_assoc_right.tikz}}

    &
    \mylabelbis{ax:tapes:diagpidem}{\ensuremath{\diagp{}\text{-idem}}}%
    \scalebox{0.8}{
    \InputIfFileExists{tapes/cipriano/idempotency.tikz}{}{\input{./tikz/tapes/cipriano/idempotency.tikz}}
} &$\stackrel{\eqref{ax:tapes:diagpidem}}{=}$& \Twire{U}
    &
    \mylabelbis{ax:tapes:diagpsym}{\ensuremath{\diagp{}{}\text{-sym}}}%
    
    \InputIfFileExists{tapes/cipriano/monoid_comm_left.tikz}{}{\input{./tikz/tapes/cipriano/monoid_comm_left.tikz}}
 &$\stackrel{\eqref{ax:tapes:diagpsym}}{=}$& 
    \InputIfFileExists{tapes/cipriano/pcomonoidbar.tikz}{}{\input{./tikz/tapes/cipriano/pcomonoidbar.tikz}}

    \\
    \multicolumn{9}{c}{%
        \scalebox{0.85}{%
        \begin{tabular}{r@{\;}c@{\;}l @{\quad} r@{\;}c@{\;}l @{\quad} r@{\;}c@{\;}l @{\quad} r@{\;}c@{\;}l}
            \mylabelbis{ax:tapes:diagpnat}{\ensuremath{\diagp{}\text{-nat}}}%
            
    \InputIfFileExists{tapes/cipriano/pcanatleft.tikz}{}{\input{./tikz/tapes/cipriano/pcanatleft.tikz}}
 &$\stackrel{\eqref{ax:tapes:diagpnat}}{=}$& 
    \InputIfFileExists{tapes/cipriano/pcanatright.tikz}{}{\input{./tikz/tapes/cipriano/pcanatright.tikz}}

            &
            \mylabelbis{ax:tapes:codiagnat}{\ensuremath{\codiag{}\text{-nat}}}%
            
    \InputIfFileExists{tapes/cipriano/monoidnatnewright.tikz}{}{\input{./tikz/tapes/cipriano/monoidnatnewright.tikz}}
 &$\stackrel{\eqref{ax:tapes:codiagnat}}{=}$& 
    \InputIfFileExists{tapes/cipriano/monoidnatnewleft.tikz}{}{\input{./tikz/tapes/cipriano/monoidnatnewleft.tikz}}

            &
            \mylabelbis{ax:tapes:bangnat}{\ensuremath{\,\bangp{}\text{-nat}}}%
            
    \InputIfFileExists{tapes/cipriano/bangnatleft.tikz}{}{\input{./tikz/tapes/cipriano/bangnatleft.tikz}}
 &$\stackrel{\eqref{ax:tapes:bangnat}}{=}$& 
    \InputIfFileExists{tapes/cipriano/bangnatright.tikz}{}{\input{./tikz/tapes/cipriano/bangnatright.tikz}}

            &
            \mylabelbis{ax:tapes:cobangnat}{\ensuremath{\,\cobang{}\text{-nat}}}%
            
    \InputIfFileExists{tapes/cipriano/cobangnatleft.tikz}{}{\input{./tikz/tapes/cipriano/cobangnatleft.tikz}}
 &$\stackrel{\eqref{ax:tapes:cobangnat}}{=}$& 
    \InputIfFileExists{tapes/cipriano/cobangnatright.tikz}{}{\input{./tikz/tapes/cipriano/cobangnatright.tikz}}

        \end{tabular}
        }%
    }
\end{tabular}
}%
\caption{Axioms for probabilistic tape diagrams.}
\label{fig:tapesax}
\end{figure}

\medskip 
It is now crucial to recall from \cite{bonchi2025tapediagramsmonoidalmonads} that $\per$ can be defined not only on string diagrams but also on tapes: for $\t_1 \colon P \to Q$, $\t_2 \colon R \to S$,  $\t_1 \per \t_2 \defeq \LW{P}{\t_2} ; \RW{S}{\t_1} $ where $\LW{P}{\cdot}$, $\RW{S}{\cdot}$ are the left and right whiskerings, defined as in Table~\ref{tab:producttape}. For objects $P = \Piu[i]{U_i}$ and $Q = \Piu[j]{V_j}$, 
\begin{equation}\label{def:productPolynomials} 
P \per Q \defeq \Piu[i]{\Piu[j]{U_iV_j}}\text{.}
\end{equation}

Theorem 27 in \cite{bonchi2025tapediagramsmonoidalmonads} guarantees that $(\CatT{\DiagT} ,\otimes, \uno)$ is a symmetric monoidal category. Such a category is monoidally enriched over $\Cat{PCA}$, i.e., for all $p\in(0,1)$ and properly typed $\t,\s_1,\s_2$:
\begin{equation}\label{eq:monoidalenrichment} \t \otimes (\s_1+_p \s_2)= (\t \otimes \s_1)+_p(\t \otimes \s_2) \qquad  (\s_1+_p \s_2)  \otimes \t= ( \s_1 \otimes \t )+_p(\s_2 \otimes  \t) \qquad \star \otimes \t= \star = \t\otimes \star \end{equation}
 Most importantly, the two monoidal categories $(\CatT{\DiagT}, \piu, \zero)$ and $(\CatT{\DiagT}, \otimes, \uno)$ interact through the laws of (right strict) rig categories (see \cite{laplaza_coherence_1972,johnson2021bimonoidal}).

\begin{theorem}[From \cite{bonchi2025tapediagramsmonoidalmonads}]
$(\,\CatT{\DiagT}, \oplus, \otimes, \zero, \uno \,)$ is a right strict rig category.
\end{theorem}

\begin{table}[t]
    \begin{center}
    {
        \hfill {\tiny
  \[\begin{array}{c}
  \toprule
        \def\arraystretch{1.2}
        \begin{array}{cc}
            \begin{array}{@{}l}
                \dl{P}{Q}{R} \colon P \per (Q\piu R)  \to (P \per Q) \piu (P\per R) \vphantom{\symmt{P}{Q}} \\
                \midrule
                \dl{\zero}{Q}{R} \defeq \id{\zero} \vphantom{\symmt{P}{\zero} \defeq \id{\zero}} \qquad
                \dl{U \piu P'}{Q}{R} \defeq (\id{U\per (Q \piu R)} \piu \dl{P'}{Q}{R});(\id{U\per Q} \piu \symmp{U\per R}{P'\per Q} \piu \id{P'\per R}) \vphantom{\symmt{P}{V \piu Q'} \defeq \dl{P}{V}{Q'} ; (\Piu[i]{\tapesymm{U_i}{V}} \piu \symmt{P}{Q'})} \\
            \end{array}
            &
            \begin{array}{@{}l}
                \symmt{P}{Q} \colon P\per Q \to Q \per P, \text{ with } P = \Piu[i]{U_i} \\
                \midrule
                \symmt{P}{\zero} \defeq \id{\zero} \qquad
                \symmt{P}{V \piu Q'} \defeq \dl{P}{V}{Q'} ; (\Piu[i]{\tapesymm{U_i}{V}} \piu \symmt{P}{Q'})
                \phantom{\quad}
            \end{array}
        \end{array}\\
        \midrule
        \begin{array}{rclrcl|rclrcl}
            \midrule
            \LW U {\id\zero} &\defeq& \id\zero& \LW U {\t_1 \piu \t_2} &\defeq& \LW U {\t_1} \piu \LW U {\t_2} &  \RW U {\id\zero} &\defeq& \id\zero & \RW U {\t_1 \piu \t_2} &\defeq& \RW U {\t_1} \piu \RW U {\t_2} \\
            \LW U {\tape{c}} &\defeq& \tape{\id U \per c} &    \LW U {\t_1 ; \t_2} &\defeq& \LW U {\t_1} ; \LW U {\t_2}&  \RW U {\tape{c}} &\defeq& \tape{c \per \id U} & \RW U {\t_1 ; \t_2} &\defeq& \RW U {\t_1} ; \RW U {\t_2} \\
            \LW U {\symmp{V}{W}} &\defeq& \symmp{UV}{UW}       &       && & \RW U {\symmp{V}{W}} &\defeq& \symmp{VU}{WU}&  &&  \\
                        \LW U {\diagp V} &\defeq& \diagp{UV} &  \LW U {\bang V} &\defeq& \bang{UV}& \RW U {\diagp V} &\defeq& \diagp{VU} & \RW U {\bang V} &\defeq& \bang{VU} \\
            \LW U {\codiag V} &\defeq& \codiag{UV} &  \LW U {\cobang V} &\defeq& \cobang{UV}& \RW U {\codiag V} &\defeq& \codiag{VU} & \RW U {\cobang V} &\defeq& \cobang{VU} \\
            \hline \hline
            \LW{\zero}{\t} &\defeq& \id{\zero} &  \LW{W\piu S'}{\t} &\defeq& \LW{W}{\t} \piu \LW{S'}{\t} & \RW{\zero}{\t} &\defeq& \id{\zero} &
            \RW{W \piu S'}{\t} &\defeq& \dl{P}{W}{S'} ; (\RW{W}{\t} \piu \RW{S'}{\t}) ; \Idl{Q}{W}{S'} \\
            \hline \hline
            \multicolumn{12}{c}{
                \t_1 \per \t_2 \defeq \LW{P}{\t_2} ; \RW{S}{\t_1}   \quad \text{ ( for }\t_1 \colon P \to Q, \t_2 \colon R \to S   \text{ )}
            }
            \\
            \bottomrule
        \end{array}
    \end{array}        
      \]
      }
      \hfill
      \caption{Inductive definitions of  left distributor $\dl{P}{Q}{R}$ and $\otimes$-symmetry $\symmt{P}{Q}$ (top); monomial and polynomial whiskerings (center); definition of $\per$ (bottom).}\label{tab:producttape}
       }
    \end{center}
\end{table}

By means of the isomorphism in Corollary~\ref{cor:isotapematrices}, we can equip $\stmat{\DiagT^+}$ with an additional product $\otimes$. This is described as follows. On objects $P\otimes Q$ is defined as in \eqref{def:productPolynomials}. On arrows, first define, for all $d_1\in \DiagT^+[U_1,V_1]$ and $d_2\in \DiagT^+[U_2,V_2]$,  $d_1\otimes^+d_2 \in \DiagT^+[U_1U_2,V_1V_2]$ as
\[d_1\otimes^+ d_2 (c) \defeq \sum_{\{(c_1,c_2) \mid c_1\otimes c_2 =c\}} d_1(c_1)\cdot d_2(c_2) \qquad \text{ for }c\in \DiagT[U_1U_2,V_1V_2]\text{.}\]
(recall from Section~\ref{sec:pca} that $ \DiagT^+[U,V]=\Dis(\DiagT[U,V])$). Then for matrices, it is defined as the Kronecker product using $\otimes^+$ as multiplication. 
More precisely, let $M\colon\bigoplus_{i=1}^nU_i \to \bigoplus_{i'=1}^{n'}U_{i'}$ and $N\colon\bigoplus_{j=1}^m V_j \to \bigoplus_{j'=1}^{m'} V_{j'}$, $M\otimes N$ is a matrix of size $n'm'\times nm$ whose rows are indexed by pairs $(i',j')\in \{1,\dots, n'\}\times \{1,\dots m'\}$ and columns by pairs $(i,j)\in \{1,\dots, n\}\times \{1,\dots m\}$. Assuming that $M_{i',i}=p_{i',i} \cdot d_{i',i}$
and $N_{j',j}=q_{j',j} \cdot e_{j',j}$, $M\otimes N$ is defined for all indexes $(i',j'),(i,j)$ as
$(M\otimes N)_{(i',j'),(i,j)} \defeq (p_{i',i}\cdot q_{j',j}) \cdot (d_{i',i}\otimes^+ e_{j',j}) $.
\begin{corollary}
$\CatT{\DiagT}$ and $\stmat{\DiagT^+}$ are isomorphic as rig categories.
\end{corollary}

Obviously, the results illustrated in this section also hold for $\CatT{\Diag{\Sigma}}$, namely the category of tapes over string diagrams generated by the monoidal signature $\Sigma$, rather than over the theory $\mathbb{T}$. It is often convenient to use both  $\CatT{\Diag{\Sigma}}$ and  $\CatT{\Diag{\mathbb{T}}}$. These convex biproduct categories are related by the morphism
\[\CatT{Q_\mathbb{T}} \colon \CatT{\Diag{\Sigma}} \to \CatT{\Diag{\mathbb{T}}}\]
obtained by applying the functor $\CatT{-}$ to $Q_\mathbb{T}\colon \Diag{\Sigma} \to \Diag{\mathbb{T}}$, namely the symmetric monoidal functor mapping each string diagram into its $=_{\mathbb{T}}$-equivalence class. The equivalence induced by $\CatT{Q_\mathbb{T}}$ can be easily characterised as the congruence $\tapeeqT$ generated by the following rules.

\begin{equation}\label{eq:congr2}
        \scalebox{0.85}{$
        \centering
        \begin{array}{c}
        \begin{array}{c@{\qquad\qquad}c@{\qquad\qquad}c@{\qquad\qquad}c}
            \inferrule*[right=($\mathbb{T}$)]{c =_{\mathbb{T}} d}{\tape{c} \tapeeqT \tape{d}}
            &
            \inferrule*[right=($\textsc{R}$)]{-}{\t \tapeeqT \t}
            &
            \inferrule*[right=($\textsc{S}$)]{\t_1 \tapeeqT \t_2}{\t_2 \tapeeqT \t_1}
            &    
            \inferrule*[right=($\textsc{T}$)]{\t_1 \tapeeqT \t_2 \quad \t_2 \tapeeqT \t_3}{\t_1 \tapeeqT \t_3}
        \end{array}
        \\[10pt]
        \begin{array}{c@{\qquad}c@{\qquad}c}
            \inferrule*[right=($;$)]{\t_1 \tapeeqT \t_2 \quad \s_1 \tapeeqT \s_2}{\t_1;\s_1 \tapeeqT \t_2;\s_2}
            &
            \inferrule*[right=($\piu$)]{\t_1 \tapeeqT \t_2 \quad \s_1 \tapeeqT \s_2}{\t_1\piu\s_1 \tapeeqT \t_2 \piu \s_2}
            &
            \inferrule*[right=($\per $)]{\t_1 \tapeeqT \t_2 \quad \s_1 \tapeeqT \s_2}{\t_1\per \s_1 \tapeeqT \t_2 \per \s_2}       
        \end{array}
        \end{array}
        $}
\end{equation}

\begin{proposition}\label{prop:quotienttheory}
For all $\s , \t \in \CatT{\Diag{\Sigma}}$, $\s \tapeeqT \t$ iff $\CatT{Q_{\mathbb{T}}}(\s) = \CatT{Q_{\mathbb{T}}}(\t)$.
\end{proposition}

\begin{figure}
\centering
\begingroup
\setlength{\tabcolsep}{4pt}
\renewcommand{\arraystretch}{1.05}
\small
\begin{tabular}{c c c}
\toprule
Symbol & Diagram & Matrix\\
\midrule

$\diagp{U}$
&

    \InputIfFileExists{/tapes/cipriano/pcomonoid.tikz}{}{\input{./tikz//tapes/cipriano/pcomonoid.tikz}}

&
\raisebox{-0.4ex}{$
\begin{pNiceMatrix}[first-col,first-row]
\rotatebox{90}{$\Lsh$} & U \\
U & p\cdot \wire{U} \\
U & (1-p)\cdot \wire{U}
\end{pNiceMatrix}
$}
\\[0.2em]

$\bangp{U}$
&
\Tcounit{U}
&
\raisebox{-0.4ex}{$
\begin{pNiceMatrix}[first-col,first-row]
\rotatebox{90}{$\Lsh$} \\
U
\end{pNiceMatrix}
$}
\\[0.2em]

$\tapeFunct{c}$
&
\Tcirc{c}{U}{V}
&
\raisebox{-0.4ex}{$
\begin{pNiceMatrix}[first-col,first-row]
\rotatebox{90}{$\Lsh$} & U \\
V & 1\cdot \Cgen{c}{U}{V}
\end{pNiceMatrix}
$}
\\[0.2em]

$\cobang{U}$
&
\Tunit{U}
&
\raisebox{-0.4ex}{$
\begin{pNiceMatrix}[first-col,first-row]
\rotatebox{90}{$\Lsh$} & U
\end{pNiceMatrix}
$}
\\[0.2em]

$\codiag{U}$
&
\Tmonoid{U}
&
\raisebox{-0.4ex}{$
\begin{pNiceMatrix}[first-col,first-row]
\rotatebox{90}{$\Lsh$} & U & U \\
U & 1\cdot \wire{U} & 1\cdot \wire{U}
\end{pNiceMatrix}
$}
\\[0.2em]

$\id{U}$
&
\Twire{U}
&
\raisebox{-0.4ex}{$
\begin{pNiceMatrix}[first-col,first-row]
\rotatebox{90}{$\Lsh$} & U \\
U & 1\cdot \wire{U}
\end{pNiceMatrix}
$}
\\[0.2em]

$\id{\zero}$
&

    \InputIfFileExists{empty.tikz}{}{\input{./tikz/empty.tikz}}

&
\raisebox{-0.4ex}{$
\begin{pNiceMatrix}[first-col,first-row]
\rotatebox{90}{$\Lsh$}
\end{pNiceMatrix}
$}
\\[0.2em]

$\sigma_{U,V}^{\piu}$
&
\Tsymmp{U}{V}
&
\raisebox{-0.4ex}{$
\begin{pNiceMatrix}[first-col,first-row]
\rotatebox{90}{$\Lsh$} & U & V \\
V & 0\cdot \star_{U,V} & 1\cdot \wire{V} \\
U & 1\cdot \wire{U} & 0\cdot \star_{V,U}
\end{pNiceMatrix}
$}
\\[0.2em]

$\t ; \t$
&

    \InputIfFileExists{tapes/seq_comp.tikz}{}{\input{./tikz/tapes/seq_comp.tikz}}

&
matrix multiplication
\\[0.2em]

$\t \piu \t$
&

    \InputIfFileExists{tapes/par_comp.tikz}{}{\input{./tikz/tapes/par_comp.tikz}}

&
direct sum of matrices
\\

\bottomrule
\end{tabular}
\endgroup
\caption{Dictionary of correspondences}\label{fig:dictionary}
\end{figure}

%% file: sections/probbooleancircuits.tex
\section{Probabilistic Boolean Circuits}\label{sec:probboolcircuits}
Probabilistic Boolean circuits were introduced in \cite{piedeleu2025boolean} as a string-diagrammatic language for finite probabilistic programming. In this section, we recall their syntax and their semantics. In the following section, we provide a complete axiomatization by encoding them into tape diagrams.

Our presentation is staged. We begin with Boolean circuits (Section~\ref{ssec:boolean}), then move to partial Boolean circuits (Section~\ref{ssec:partialboolean}), and finally illustrate probabilistic Boolean circuits (Section~\ref{ssec:probabilisticboolean}). We also provide a complete axiomatization for partial Boolean circuits, which, to the best of our knowledge, has not previously appeared in the literature.

\subsection{Boolean Circuits}\label{ssec:boolean}
Consider the monoidal signature with a single sort, $\sort = \{ A\}$, and  generators 
\[ \SigB \defeq \{\Andgate  , \Notgate , \Flip{1},   \CBcopier, \CBdischarger \} \text{.}\]  
Arities and coarities are determined by the number of ports on the left and on the right: for instance, $\Andgate$ has arity $A^2=AA$ and coarity $A$, while $\CBdischarger $ has arity $A$ and coarity $A^0 =1$. 

The arrows of $\Diag{\SigB}$, namely the category of string diagrams generated by the monoidal signature $\SigB$ (see Section~\ref{sec:tapediagrams}), are interpreted as Boolean functions. 
The first three generators of $\SigB$ correspond to the operations and constants of Boolean algebras, namely $\wedge$, $\neg$, and $1$. 
The generator \emph{copy}, denoted by $\CBcopier$, takes a Boolean signal as input and produces two identical outputs. 
The generator \emph{discard}, denoted by $\CBdischarger$, takes a Boolean signal as input and discards it.

Formally, the semantics of these diagrams is given by the monoidal functor  $\osemB{-}\colon \Diag{{\SigB}} \to \Sets$ defined on objects as $A^n \mapsto 2^n$ where $2$ is the set of Booleans  $\{0,1\}$; for arrows it is defined by first interpreting the generators as functions in $\Sets$
  \[
    \begin{array}{rclrclrcl}
      \osemB{\Andgate}\colon 2 \times 2& \to & 2  & \osemB{\Notgate}\colon  2& \to & 2 & \osemB{\Flip{1}} \colon 1 & \to & 2 \\
      (x,y)&\mapsto& x\wedge y &  x & \mapsto & \neg x & \bullet &\mapsto&1\\[4pt]
       \osemB{\CBcopier} \colon 2  & \to & 2\times 2 & \osemB{\CBdischarger} \colon 2 & \to & 1 & 
       \\
      x &\mapsto&(x,x) & x &\mapsto&{\bullet} & 
    \end{array}
  \]
and then inductively by means of the structure of the monoidal category $(\Sets, \otimes ,\uno)$: for instance, $\osemB{c_1\otimes c_2}=\osemB{c_1}\otimes \osemB{c_2}$ where the second $\otimes$ is the standard cartesian product of functions. 

The monoidal theory $\ThB$ consists of the signature $\SigB$ together with the axioms listed in Figure~\ref{tab:booleanalgebra}. 
The equalities in the first block are the standard axioms of Boolean algebras. 
The constant $\Flip{0}$ and the $\Orgate$ gate are defined as expected (see Figure~\ref{table:multiplexer}). 
Note that, in the string-diagrammatic setting, the structural maps $\CBcopier$ and $\CBdischarger$ must be made explicit.
The equalities in the second block state that $\CBcopier$ and $\CBdischarger$ form a commutative comonoid. 
The equalities in the third and fourth blocks impose naturality of $\CBcopier$ and $\CBdischarger$, respectively.


The congruence with respect to $;$ and $\otimes$ generated by the axioms in Figure~\ref{tab:booleanalgebra} is denoted by $\eqB$. 
We write $\Diag{\ThB}$ for the quotient of $\Diag{\SigB}$ by $\eqB$.

As expected, the equalities in Figure~\ref{tab:booleanalgebra} provide a sound and complete axiomatization of the equality induced by $\osemB{-}$.

\begin{theorem}\label{thm:completenessbooleancircuits}
For all  $c,d\in \Diag{\SigB}$, $\osemB{c}=\osemB{d}$ iff $c\eqB d$.
\end{theorem}
\begin{proof}
See e.g. \cite[Theorem 3.2]{piedeleu2025boolean}.
\end{proof}

\begin{figure}
  \scalebox{0.8}{%
  \begin{tabular}{ccccc}


  \mylabelbis{ax:B1}{B1}%
  $
    \InputIfFileExists{tapes/cipriano/axB1left.tikz}{}{\input{./tikz/tapes/cipriano/axB1left.tikz}}
\raisebox{-1.5ex}{\,$\stackrel{\eqref{ax:B1}}{=}$\,}
    \InputIfFileExists{tapes/cipriano/axB1right.tikz}{}{\input{./tikz/tapes/cipriano/axB1right.tikz}}
$
  & &
  \mylabelbis{ax:B2l}{B2l}\mylabelbis{ax:B2r}{B2r}%
  $
    \InputIfFileExists{tapes/cipriano/axB2left.tikz}{}{\input{./tikz/tapes/cipriano/axB2left.tikz}}
\raisebox{-1.5ex}{\,$\stackrel{\eqref{ax:B2l}}{=}$\,}
    \InputIfFileExists{tapes/cipriano/axB2middle.tikz}{}{\input{./tikz/tapes/cipriano/axB2middle.tikz}}
\raisebox{-1.5ex}{\,$\stackrel{\eqref{ax:B2r}}{=}$\,}
    \InputIfFileExists{tapes/cipriano/axB2right.tikz}{}{\input{./tikz/tapes/cipriano/axB2right.tikz}}
$
  & &
  \mylabelbis{ax:B3}{B3}%
  $
    \InputIfFileExists{tapes/cipriano/axB3left.tikz}{}{\input{./tikz/tapes/cipriano/axB3left.tikz}}
\raisebox{-1.5ex}{\,$\stackrel{\eqref{ax:B3}}{=}$\,}
    \InputIfFileExists{tapes/cipriano/axB3right.tikz}{}{\input{./tikz/tapes/cipriano/axB3right.tikz}}
$
  \\

  \mylabelbis{ax:B4}{B4}%
  $
    \InputIfFileExists{tapes/cipriano/axB4left.tikz}{}{\input{./tikz/tapes/cipriano/axB4left.tikz}}
\raisebox{-1.5ex}{\,$\stackrel{\eqref{ax:B4}}{=}$\,}
    \InputIfFileExists{tapes/cipriano/axB4right.tikz}{}{\input{./tikz/tapes/cipriano/axB4right.tikz}}
$
  & &
  \mylabelbis{ax:B5}{B5}%
  $
    \InputIfFileExists{tapes/cipriano/axB5left.tikz}{}{\input{./tikz/tapes/cipriano/axB5left.tikz}}
\raisebox{-1.5ex}{\,$\stackrel{\eqref{ax:B5}}{=}$\,}
    \InputIfFileExists{tapes/cipriano/axB5right.tikz}{}{\input{./tikz/tapes/cipriano/axB5right.tikz}}
$
  & &
  \mylabelbis{ax:B6}{B6}%
  $
    \InputIfFileExists{tapes/cipriano/axB6left.tikz}{}{\input{./tikz/tapes/cipriano/axB6left.tikz}}
\raisebox{-1.5ex}{\,$\stackrel{\eqref{ax:B6}}{=}$\,}
    \InputIfFileExists{tapes/cipriano/axB6right.tikz}{}{\input{./tikz/tapes/cipriano/axB6right.tikz}}
$
  \\

  & &
  \mylabelbis{ax:B7}{B7}%
  $
    \InputIfFileExists{tapes/cipriano/axB7left.tikz}{}{\input{./tikz/tapes/cipriano/axB7left.tikz}}
\raisebox{-1.5ex}{\,$\stackrel{\eqref{ax:B7}}{=}$\,}
    \InputIfFileExists{tapes/cipriano/axB7right.tikz}{}{\input{./tikz/tapes/cipriano/axB7right.tikz}}
$
  & &
  \\

  \midrule


  \mylabelbis{ax:A1}{B8}%
  $
    \InputIfFileExists{tapes/cipriano/axA1left.tikz}{}{\input{./tikz/tapes/cipriano/axA1left.tikz}}
\raisebox{-1.5ex}{\,$\stackrel{\eqref{ax:A1}}{=}$\,}
    \InputIfFileExists{tapes/cipriano/axA1right.tikz}{}{\input{./tikz/tapes/cipriano/axA1right.tikz}}
$
  & &
  \mylabelbis{ax:A2l}{B9l}\mylabelbis{ax:A2r}{B9r}%
  $
    \InputIfFileExists{tapes/cipriano/axA2left.tikz}{}{\input{./tikz/tapes/cipriano/axA2left.tikz}}
\raisebox{-1.5ex}{\,$\stackrel{\eqref{ax:A2l}}{=}$\,}
    \InputIfFileExists{tapes/cipriano/axA2middle.tikz}{}{\input{./tikz/tapes/cipriano/axA2middle.tikz}}
\raisebox{-1.5ex}{\,$\stackrel{\eqref{ax:A2r}}{=}$\,}
    \InputIfFileExists{tapes/cipriano/axA2right.tikz}{}{\input{./tikz/tapes/cipriano/axA2right.tikz}}
$
  & &
  \mylabelbis{ax:A3}{B10}%
  $
    \InputIfFileExists{tapes/cipriano/axA3left.tikz}{}{\input{./tikz/tapes/cipriano/axA3left.tikz}}
\raisebox{-1.5ex}{\,$\stackrel{\eqref{ax:A3}}{=}$\,}
    \InputIfFileExists{tapes/cipriano/axA3right.tikz}{}{\input{./tikz/tapes/cipriano/axA3right.tikz}}
$
  \\

  \midrule


  \mylabelbis{ax:C1}{B11}%
  $
    \InputIfFileExists{tapes/cipriano/axC1left.tikz}{}{\input{./tikz/tapes/cipriano/axC1left.tikz}}
\raisebox{-0.6ex}{\,$\stackrel{\eqref{ax:C1}}{=}$\,}
    \InputIfFileExists{tapes/cipriano/axC1right.tikz}{}{\input{./tikz/tapes/cipriano/axC1right.tikz}}
$
  & &
  \mylabelbis{ax:C2}{B12}%
  $
    \InputIfFileExists{tapes/cipriano/axC2left.tikz}{}{\input{./tikz/tapes/cipriano/axC2left.tikz}}
\raisebox{-0.6ex}{\,$\stackrel{\eqref{ax:C2}}{=}$\,}
    \InputIfFileExists{tapes/cipriano/axC2right.tikz}{}{\input{./tikz/tapes/cipriano/axC2right.tikz}}
$
  & &
  \mylabelbis{ax:C3}{B13}%
  $
    \InputIfFileExists{tapes/cipriano/axC3left.tikz}{}{\input{./tikz/tapes/cipriano/axC3left.tikz}}
\raisebox{-0.6ex}{\,$\stackrel{\eqref{ax:C3}}{=}$\,}
    \InputIfFileExists{tapes/cipriano/axC3right.tikz}{}{\input{./tikz/tapes/cipriano/axC3right.tikz}}
$
  \\

  \midrule


  \mylabelbis{ax:D1}{B14}%
  $
    \InputIfFileExists{tapes/cipriano/axD1left.tikz}{}{\input{./tikz/tapes/cipriano/axD1left.tikz}}
\raisebox{-0.6ex}{\,$\stackrel{\eqref{ax:D1}}{=}$\,}
    \InputIfFileExists{tapes/cipriano/axD1right.tikz}{}{\input{./tikz/tapes/cipriano/axD1right.tikz}}
$
  & &
  \mylabelbis{ax:D2}{B15}%
  $
    \InputIfFileExists{tapes/cipriano/axD2left.tikz}{}{\input{./tikz/tapes/cipriano/axD2left.tikz}}
\raisebox{-0.6ex}{$\stackrel{\eqref{ax:D2}}{=}$}
    \InputIfFileExists{tapes/cipriano/axD2right.tikz}{}{\input{./tikz/tapes/cipriano/axD2right.tikz}}
$
  & &
  \mylabelbis{ax:D3}{B16}%
  $
    \InputIfFileExists{tapes/cipriano/axD3left.tikz}{}{\input{./tikz/tapes/cipriano/axD3left.tikz}}
\raisebox{-0.6ex}{$\stackrel{\eqref{ax:D3}}{=}$}
  \raisebox{-0.5em}{
    \InputIfFileExists{empty.tikz}{}{\input{./tikz/empty.tikz}}
}$
  \\

  \end{tabular}}%
  \caption{The monoidal theory of Boolean algebras.}
  \label{tab:booleanalgebra}
\end{figure}
Moreover, one can characterise exactly the image of $\Diag{\SigB}$ through $\osemB{-}$. Let $\Sets_2$ be the full subcategory of sets where objects are powers of the set $2$, i.e.,  $2^n$ for all $n \in \mathbb{N}$.

\begin{proposition}\label{prop:foolboolean}
The monoidal functor  $\osemB{-}\colon \Diag{{\SigB}} \to \Sets$ factors as
\[\xymatrix{ \Diag{\SigB} \ar@{->>}[r]& \Diag{\ThB} \ar[r]^{\cong}& \Sets_2 \ar@{^{(}->}[r]& \Sets }\]
where the leftmost and rightmost functor are the obvious quotients and injections and the central arrow is a monoidal isomorphism.
\end{proposition}
\begin{proof}
Faithfulness of the central functor is a direct consequence of Theorem~\ref{thm:completenessbooleancircuits}. Fullness follows from the fact that every function $f\colon 2^n \to 2^m$ corresponds to a vector of $m$ Boolean formulas in $n$ variables, which can be easily encoded as a string diagram in $\Diag{\ThB}$ using Boolean operators, $\CBcopier$ and $\CBdischarger$.
\end{proof}

We conclude this section by fixing some syntactic sugar, collected in Figure~\ref{table:multiplexer}. 
The \emph{xnor} gate $\xnor$  takes two Boolean inputs and outputs $1$ if they are equal and $0$ otherwise. 
The \emph{multiplexer} $
    \InputIfFileExists{tapes/cipriano/multi.tikz}{}{\input{./tikz/tapes/cipriano/multi.tikz}}
$ takes three inputs and returns one output: if the first input is $1$, then it outputs the second input, otherwise the third input. 
The $m$-ary multiplexer $\Ifgatem$ works similarly, but it takes $2m+1$ inputs: if the first input is $1$, then it outputs the first $m$ inputs, otherwise it outputs the second $m$ inputs. 
Figure~\ref{table:multiplexer} also defines $m$-ary discard $\CBdischargerm$ and copier $\mcopier$. Note that, for the sake of readability of string diagrams, we label wires by  some natural number $n$ as an abbreviation for $A^n$.

\begin{lemma}\label{lemma:multiplexer}
 Let $c$ be a diagram in $\Diag{\SigB}[n,m]$. The following derived laws hold:

\mylabelbis{ax:M1}{D1}
\mylabelbis{ax:M2}{D2}
\mylabelbis{ax:M3}{D3}
\mylabelbis{ax:N1}{D4}
\mylabelbis{ax:N2}{D5}
{\setlength{\abovedisplayskip}{0.3em}
 \setlength{\belowdisplayskip}{0.3em}
\[
\begin{array}{c@{\qquad}c@{\qquad}c}
\raisebox{-0.3em}{\scalebox{0.8}{
    \InputIfFileExists{tapes/cipriano/D1left.tikz}{}{\input{./tikz/tapes/cipriano/D1left.tikz}}
}} 
\stackrel{\eqref{ax:M1}}{\eqB}

    \InputIfFileExists{tapes/cipriano/D1right.tikz}{}{\input{./tikz/tapes/cipriano/D1right.tikz}}

&
\raisebox{-0.3em}{\scalebox{0.8}{
    \InputIfFileExists{tapes/cipriano/D2left.tikz}{}{\input{./tikz/tapes/cipriano/D2left.tikz}}
}}
\stackrel{\eqref{ax:M2}}{\eqB}

    \InputIfFileExists{tapes/cipriano/D2right.tikz}{}{\input{./tikz/tapes/cipriano/D2right.tikz}}

&
\raisebox{-0.3em}{\scalebox{0.8}{
    \InputIfFileExists{tapes/cipriano/D3left.tikz}{}{\input{./tikz/tapes/cipriano/D3left.tikz}}
}}
\stackrel{\eqref{ax:M3}}{\eqB}

    \InputIfFileExists{tapes/cipriano/D3right.tikz}{}{\input{./tikz/tapes/cipriano/D3right.tikz}}

\\[0.5em]

    \InputIfFileExists{tapes/cipriano/D4left.tikz}{}{\input{./tikz/tapes/cipriano/D4left.tikz}}

\stackrel{\eqref{ax:N1}}{\eqB}

    \InputIfFileExists{tapes/cipriano/D4right.tikz}{}{\input{./tikz/tapes/cipriano/D4right.tikz}}

&

    \InputIfFileExists{tapes/cipriano/D5left.tikz}{}{\input{./tikz/tapes/cipriano/D5left.tikz}}

\stackrel{\eqref{ax:N2}}{\eqB}

    \InputIfFileExists{tapes/cipriano/D5right.tikz}{}{\input{./tikz/tapes/cipriano/D5right.tikz}}

&
\end{array}
\]}
\end{lemma}

\begin{figure}[t]
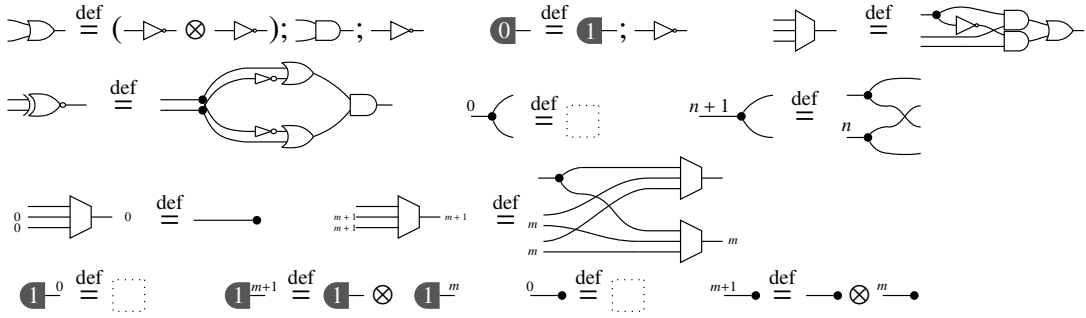

\centering
\resizebox{0.96\linewidth}{!}{%
\begin{tabular}{@{}l@{}}

$\Orgate \defeq (\Notgate \otimes \Notgate); \Andgate; \Notgate 
\qquad 
\Flip{0} \defeq \Flip{1}; \Notgate
\qquad

    \InputIfFileExists{tapes/cipriano/multi.tikz}{}{\input{./tikz/tapes/cipriano/multi.tikz}}
 
\defeq 

    \InputIfFileExists{tapes/cipriano/multiexpl.tikz}{}{\input{./tikz/tapes/cipriano/multiexpl.tikz}}
$ \\[4pt]

$
    \InputIfFileExists{tapes/cipriano/xnorgate.tikz}{}{\input{./tikz/tapes/cipriano/xnorgate.tikz}}

\qquad
\raisebox{-1em}{
    \InputIfFileExists{tapes/cipriano/ncopierzero.tikz}{}{\input{./tikz/tapes/cipriano/ncopierzero.tikz}}
} 
\raisebox{-0.5em}{$\defeq$}
\raisebox{-0.5em}{
    \InputIfFileExists{empty.tikz}{}{\input{./tikz/empty.tikz}}
}
\qquad 
\raisebox{-1em}{
    \InputIfFileExists{tapes/cipriano/ncopier.tikz}{}{\input{./tikz/tapes/cipriano/ncopier.tikz}}
}$ \\[6pt]

$\raisebox{-1.2em}{
    \InputIfFileExists{tapes/cipriano/mmultiplexerleftzero.tikz}{}{\input{./tikz/tapes/cipriano/mmultiplexerleftzero.tikz}}
} 
\defeq 

    \InputIfFileExists{tapes/cipriano/mmultiplexerrightzero.tikz}{}{\input{./tikz/tapes/cipriano/mmultiplexerrightzero.tikz}}

\qquad 
\raisebox{-1.2em}{
    \InputIfFileExists{tapes/cipriano/mmultiplexerleft.tikz}{}{\input{./tikz/tapes/cipriano/mmultiplexerleft.tikz}}
} 
\defeq 

    \InputIfFileExists{tapes/cipriano/mmultiplexerright.tikz}{}{\input{./tikz/tapes/cipriano/mmultiplexerright.tikz}}
$ \\[6pt]

$\mFlip{1}{0} \defeq 
    \InputIfFileExists{empty.tikz}{}{\input{./tikz/empty.tikz}}

\qquad
\mFlip{1}{m+1} \defeq \Flip{1} \otimes \mFlip{1}{m}
\qquad
\mCBdischarger{0} \defeq 
    \InputIfFileExists{empty.tikz}{}{\input{./tikz/empty.tikz}}

\qquad
\mCBdischarger{m+1} \defeq \CBdischarger\otimes \mCBdischarger{m}$

\end{tabular}
}
\caption{Definitions of OR gate, $\Flip{0}$ and multiplexer (top); definitions of XNOR gate and inductive definition of $n$-ary copier (middle); inductive definitions of $m$-ary multiplexer (third row); definitions of $\Flipzero{1}$, $\Flipm{1}$ and $n$-ary dischargers (bottom).}
\label{table:multiplexer}
\end{figure}

%
%
%
%
\subsection{Partial Boolean Circuits}\label{ssec:partialboolean}
Now consider the monoidal signature obtained by adding to $\SigB$ the generator \emph{cocopy} $\CBcocopier$ of arity $A^2$ and coarity $A$. 
\[ PB\defeq \SigB \cup\{\CBcocopier\}\text{}\] 
Intuitively, $\CBcocopier$ is the dual of $\CBcopier$: it compares two inputs and, if these are equal, it outputs that value; otherwise, it produces no output. Thus, string diagrams in $\Diag{PB}$ represent \emph{partial} Boolean functions. 

To formally define their semantics, let us first fix $(\Cat{Par}, \otimes, 1)$ to be the symmetric monoidal category of sets and partial functions, where the monoidal product $\otimes$ is the cartesian product of sets and, for all partial functions $f_1\colon X_1 \to Y_1$ and $f_2\colon X_2 \to Y_2$, $f_1\otimes f_2 \colon X_1 \times X_2 \to Y_1 \times Y_2$ is defined as
\begin{equation}\label{eq:productpar}
f_1\otimes f_2 (x_1,x_2) \defeq \begin{cases} (\,f_1(x_1), \, f_2(x_2)\,) & \text{ if }f_1(x_1)\neq \bot\text{, } f_2(x_2)\neq \bot \\
 \bot & \text{ otherwise}\end{cases}
\end{equation}
for all $(x_1,x_2)\in X_1\times X_2$. Above, and in the rest of the paper, we write $f(x)=\bot$ to say that a partial function $f\colon X \to Y$ is undefined for some $x\in X$. 
One can readily check that the obvious injection $\Pa\colon \Sets \to \Cat{Par}$ is a symmetric monoidal functor.

Now, the semantics of diagrams in $\Diag{PB}$ is given by the monoidal functor $\osemPB{-}\colon \Diag{PB} \to \Cat{Par}$ defined on objects as $A^n \mapsto 2^n$, like $\osemB{-}$; 
for arrows it is defined by first interpreting generators $c \in \SigB$ as $\osemPB{c}\defeq \Pa(\osemB{c})$, the generator $\CBcocopier $ as
  \begin{equation}\label{eq:semcocopier}
    \begin{array}{rcl}
      \osemPB{\CBcocopier}\colon 2 \times 2& \to & 2   \\
      (x,y)&\mapsto& \begin{cases} x &\text{if } x=y\\ \bot & \text{else}\end{cases}    \end{array}
  \end{equation}
and then inductively by means of the structure of the monoidal category $(\Cat{Par}, \otimes ,\uno)$: for instance $\osemPB{c_1\otimes c_2}=\osemPB{c_1}\otimes \osemPB{c_2}$ where the second $\otimes$ is the one in \eqref{eq:productpar}. 
For example, the semantics of the following diagrams
\begin{equation}\label{def:bottom}
  \raisebox{0cm}{
    \InputIfFileExists{tapes/cipriano/defcoflip.tikz}{}{\input{./tikz/tapes/cipriano/defcoflip.tikz}}
}
\end{equation}
is illustrated below.
  \[
    \begin{array}{rclrclrcl}
      \osemPB{\coflip{1} }\colon 2 & \to & 1  & \osemPB{\coflip{0} }\colon 2 & \to & 1 & \osemPB{\Flip{\bot}}\colon 1 & \to & 2 \\
      x&\mapsto& \begin{cases} \bullet &\text{if } x=1\\ \bot & \text{else}\end{cases}    
      &
      x&\mapsto& \begin{cases} \bullet &\text{if } x=0\\ \bot & \text{else}\end{cases}    
      &
      \bullet &\mapsto & \bot     
      \end{array}
  \]

In this section, we present a complete axiomatization of the equality induced by $\osemPB{-}$, which, to the best of our knowledge, has not been previously established.
The axioms are illustrated in Figure~\ref{tab:partialbooleanalgebra}: the equalities in the first row assert that $\CBcocopier$ is associative, commutative and idempotent; in the second row, there are the usual Frobenius equalities ruling the interaction of $\CBcocopier$ with $\CBcopier$; the two axioms in the third row are more peculiar: the rightmost states that $x\wedge y = 1$ iff $x=y=1$; the leftmost provides a decomposition of  $\CBcocopier$ in terms of $\Andgate$, $\coflip{1}$ and the xnor gate 
    \InputIfFileExists{tapes/cipriano/xnorgatealone.tikz}{}{\input{./tikz/tapes/cipriano/xnorgatealone.tikz}}
.

\begin{figure}
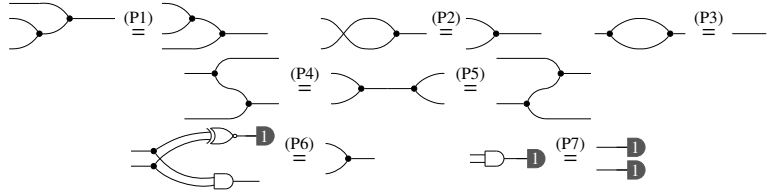

\centering
\scalebox{0.8}{
{
\setlength{\tabcolsep}{12pt}      
\setlength{\extrarowheight}{9pt}  
\renewcommand{\arraystretch}{1.0} 

\begin{tabular}{ccc}

\mylabelbis{ax:F1}{P1}%
$
    \InputIfFileExists{tapes/cipriano/axF1left.tikz}{}{\input{./tikz/tapes/cipriano/axF1left.tikz}}

 \raisebox{-0.6em}{\,\,$\stackrel{\eqref{ax:F1}}{=}$\,\,}
 
    \InputIfFileExists{tapes/cipriano/axF1right.tikz}{}{\input{./tikz/tapes/cipriano/axF1right.tikz}}
$
&
\mylabelbis{ax:F2}{P2}%

$
    \InputIfFileExists{tapes/cipriano/axF2left.tikz}{}{\input{./tikz/tapes/cipriano/axF2left.tikz}}

 \raisebox{-0.6em}{\,$\stackrel{\eqref{ax:F2}}{=}$\,}
 
    \InputIfFileExists{tapes/cipriano/axF2right.tikz}{}{\input{./tikz/tapes/cipriano/axF2right.tikz}}
$
&
\mylabelbis{ax:F3}{P3}%
$
    \InputIfFileExists{tapes/cipriano/axF3left.tikz}{}{\input{./tikz/tapes/cipriano/axF3left.tikz}}

 \raisebox{-0.6em}{\,\,$\stackrel{\eqref{ax:F3}}{=}$\,\,}
 
    \InputIfFileExists{tapes/cipriano/axF3right.tikz}{}{\input{./tikz/tapes/cipriano/axF3right.tikz}}
$
\\[1.0em] 

\multicolumn{3}{c}{%
  \mylabelbis{ax:F4}{P4}%
  \mylabelbis{ax:F5}{P5}%
  $
    \InputIfFileExists{tapes/cipriano/axF4left.tikz}{}{\input{./tikz/tapes/cipriano/axF4left.tikz}}

    \,\stackrel{\eqref{ax:F4}}{=}\,
    
    \InputIfFileExists{tapes/cipriano/axF4right.tikz}{}{\input{./tikz/tapes/cipriano/axF4right.tikz}}

    \,\stackrel{\eqref{ax:F5}}{=}\,
    
    \InputIfFileExists{tapes/cipriano/axF5.tikz}{}{\input{./tikz/tapes/cipriano/axF5.tikz}}
$%
} \\[1.0em] 

\multicolumn{3}{c}{%
  \mylabelbis{ax:F6}{P6}%
  $
    \InputIfFileExists{tapes/cipriano/axF6left.tikz}{}{\input{./tikz/tapes/cipriano/axF6left.tikz}}

   \, \stackrel{\eqref{ax:F6}}{=}\,
    
    \InputIfFileExists{tapes/cipriano/axF6right.tikz}{}{\input{./tikz/tapes/cipriano/axF6right.tikz}}
$%
  \hspace{4em}%
  \mylabelbis{ax:F7}{P7}%
  $
    \InputIfFileExists{tapes/cipriano/axF7left.tikz}{}{\input{./tikz/tapes/cipriano/axF7left.tikz}}

    \,\stackrel{\eqref{ax:F7}}{=}\,
    
    \InputIfFileExists{tapes/cipriano/axF7right.tikz}{}{\input{./tikz/tapes/cipriano/axF7right.tikz}}
$%
}

\end{tabular}
}
}
\caption{The monoidal theory of partial Boolean algebras.}
\label{tab:partialbooleanalgebra}
\end{figure}

We write $\ThPB$ for the monoidal theory consisting of the signature $PB$ and the axioms in Figures~\ref{tab:booleanalgebra} and~\ref{tab:partialbooleanalgebra};  $\eqPB$ for the generated congruence; $\Diag{\ThPB}$ for the quotient of $\Diag{PB}$ by $\eqPB$ and $Q_\ThPB\colon \Diag{ \SigPB} \to \Diag{\ThPB}$ the functor mapping each diagram into its $\eqPB$ equivalence class. Simple computations confirm that:

\begin{proposition}[Soundness]\label{prop:soundnesspartialb}
For all $c,d\in \Diag{PB}$, if $c\eqPB d$, then $\osemPB{c}=\osemPB{d}$.
\end{proposition}





The remainder of this section is devoted to proving the converse implication, namely completeness. 
Although this fact is not required for our proof, it is worth noting that the axiom \textsc{F8} in~\cite{piedeleu2025boolean}--corresponding to \eqref{ax:F8} in the following lemma--can be derived within $\ThPB$. 

\begin{lemma}\label{lemma:failureequalities}
  The following derived laws hold:
  \begin{center}
    \mylabelbis{ax:F9}{D6}%
    $
    \InputIfFileExists{tapes/cipriano/axF9left.tikz}{}{\input{./tikz/tapes/cipriano/axF9left.tikz}}

      \stackrel{\eqref{ax:F9}}{\eqPB}
      
    \InputIfFileExists{tapes/cipriano/axF9right.tikz}{}{\input{./tikz/tapes/cipriano/axF9right.tikz}}
$
    \qquad
    \mylabelbis{ax:F8}{D7}%
    $
    \InputIfFileExists{tapes/cipriano/axF8left.tikz}{}{\input{./tikz/tapes/cipriano/axF8left.tikz}}

      \stackrel{\eqref{ax:F8}}{\eqPB}
      
    \InputIfFileExists{tapes/cipriano/axF8right.tikz}{}{\input{./tikz/tapes/cipriano/axF8right.tikz}}
$
    \qquad
    \mylabelbis{ax:F10l}{D8l}%
    \mylabelbis{ax:F10r}{D8r}%
    $
    \InputIfFileExists{tapes/cipriano/axF10left.tikz}{}{\input{./tikz/tapes/cipriano/axF10left.tikz}}

      \stackrel{\eqref{ax:F10l}}{\eqPB}
      
    \InputIfFileExists{tapes/cipriano/axF10middle.tikz}{}{\input{./tikz/tapes/cipriano/axF10middle.tikz}}

      \stackrel{\eqref{ax:F10r}}{\eqPB}
      
    \InputIfFileExists{tapes/cipriano/axF10right.tikz}{}{\input{./tikz/tapes/cipriano/axF10right.tikz}}
$
  \end{center}
\end{lemma}
\begin{proof}
 \eqref{ax:F9} follows from the equalities  
 \begin{center} 
    \InputIfFileExists{tapes/cipriano/lemmafail1.tikz}{}{\input{./tikz/tapes/cipriano/lemmafail1.tikz}}

  \end{center}
  \begin{center}
  
    \InputIfFileExists{tapes/cipriano/lemmafail2.tikz}{}{\input{./tikz/tapes/cipriano/lemmafail2.tikz}}

  \end{center}
   \eqref{ax:F8} follows from the equalities
  \begin{center}
  
    \InputIfFileExists{tapes/cipriano/lemmafail3.tikz}{}{\input{./tikz/tapes/cipriano/lemmafail3.tikz}}

  \end{center}
  \begin{center}
  
    \InputIfFileExists{tapes/cipriano/lemmafail4.tikz}{}{\input{./tikz/tapes/cipriano/lemmafail4.tikz}}

  \end{center}
  \begin{center}
  
    \InputIfFileExists{tapes/cipriano/lemmafail5.tikz}{}{\input{./tikz/tapes/cipriano/lemmafail5.tikz}}

  \end{center}
  \begin{center}
  
    \InputIfFileExists{tapes/cipriano/lemmafail5mezzo.tikz}{}{\input{./tikz/tapes/cipriano/lemmafail5mezzo.tikz}}

  \end{center}
   \eqref{ax:F10l} follows from the equalities
  \begin{center}
  
    \InputIfFileExists{tapes/cipriano/lemmafail6.tikz}{}{\input{./tikz/tapes/cipriano/lemmafail6.tikz}}

  \end{center}
  \eqref{ax:F10r} can be proved with a similar argument.
\end{proof}

The strategy for proving completeness is as follows. 
First, we show that any diagram $c$ in $\Diag{PB}$ can be decomposed into two diagrams $D_c$ and $T_c$ in $\Diag{\SigB}$ (Proposition~\ref{prop:decompositionpartialcircuits}). 
We then appeal to the completeness result for Boolean circuits (Theorem~\ref{thm:completenessbooleancircuits}).

Intuitively, $D_c$ and $T_c$ represent, respectively, the domain and the total component of $c$. 
The diagram $D_c$ returns $1$ if the input of $c$ lies in the domain of definition of $c$, and $0$ otherwise. 
The diagram $T_c$ returns the output of $c$ when the input lies in its domain, and a vector of $0$ otherwise.

\begin{definition}
For all $c\in \Diag{PB}[A^n,A^m]$, $D_c\in \Diag{\SigB}[A^n,A]$ and $T_c\in\Diag{\SigB}[A^n,A^m]$ are inductively defined as
\begin{equation}\label{eq:defTD}  \begin{array}{rcl@{\hspace{1.5cm}}rcl}
    D_c &\defeq& \CBdischargern\, \Flip{1} &       T_c &\defeq& c \qquad \text{for all $c\in \SigB \cup \{\id{1},\id{A},\symmt{A}{A}\}$;}\\
    D_{\scalebox{0.6}{$\CBcocopier$}} &\defeq& 
    \InputIfFileExists{tapes/cipriano/xnorgatealone.tikz}{}{\input{./tikz/tapes/cipriano/xnorgatealone.tikz}}
 &     T_{\scalebox{0.6}{$\CBcocopier$}} &\defeq& \Andgate\\
    D_{c;d} &\defeq &\quad 
    \InputIfFileExists{tapes/cipriano/Dfg.tikz}{}{\input{./tikz/tapes/cipriano/Dfg.tikz}}
 &  T_{c;d} &\defeq &\quad 
    \InputIfFileExists{tapes/cipriano/Tfg.tikz}{}{\input{./tikz/tapes/cipriano/Tfg.tikz}}
\\
    D_{c\otimes d} &\defeq & \quad 
    \InputIfFileExists{tapes/cipriano/Dfperg.tikz}{}{\input{./tikz/tapes/cipriano/Dfperg.tikz}}
 & T_{c\otimes d} &\defeq& \quad 
    \InputIfFileExists{tapes/cipriano/Tfperg.tikz}{}{\input{./tikz/tapes/cipriano/Tfperg.tikz}}
\\
  \end{array}
  \end{equation}
\end{definition}
Observe that the domain of $\CBcocopier$ is the $\mathrm{xnor}$ gate: it returns $1$ if and only if the two inputs are equal. 
Its total part is given by the $\wedge$ gate: it returns $1$ precisely when both inputs are $1$.
Similarly, the domain of $c \otimes d$ is given by the conjunction of the domains of $c$ and $d$. 
Its total part returns the total parts of $c$ and $d$ when the domain evaluates to $1$, and $0$ otherwise.

\begin{proposition}\label{prop:decompositionpartialcircuits}
  For all $c\in \Diag{BP}$, it holds that
  \begin{center}
  
    \InputIfFileExists{tapes/cipriano/decompositionpartial.tikz}{}{\input{./tikz/tapes/cipriano/decompositionpartial.tikz}}

  \end{center}
\end{proposition}
\begin{proof}
  We proceed by induction on the structure of $c$. For the base case, if $c\in \SigB \cup \{\id{1},\id{A},\symmt{A}{A}\}$, then the statement follows by
  \begin{center}
    
    \InputIfFileExists{tapes/cipriano/propdecomposition1.tikz}{}{\input{./tikz/tapes/cipriano/propdecomposition1.tikz}}

  \end{center} 
  If $c$ is $\CBcocopier$, then the statement is exactly axiom \eqref{ax:F6} in Figure~\ref{tab:partialbooleanalgebra}. Now, in the case $c;d$, then assuming that the statement holds for $c$ and $d$, we have
  \begin{center}
    
    \InputIfFileExists{tapes/cipriano/propdecomposition2.tikz}{}{\input{./tikz/tapes/cipriano/propdecomposition2.tikz}}

    \end{center}
    \begin{center}
    
    \InputIfFileExists{tapes/cipriano/propdecomposition3.tikz}{}{\input{./tikz/tapes/cipriano/propdecomposition3.tikz}}

    \end{center}
    \begin{center}
    
    \InputIfFileExists{tapes/cipriano/propdecomposition4.tikz}{}{\input{./tikz/tapes/cipriano/propdecomposition4.tikz}}

    \end{center}
    \begin{center}
    
    \InputIfFileExists{tapes/cipriano/propdecomposition5.tikz}{}{\input{./tikz/tapes/cipriano/propdecomposition5.tikz}}

    \end{center}
    where now the statement follows by inductive hypothesis. Finally, in the case $c\otimes d$, then assuming that the statement holds for $c$ and $d$, we have 
    \begin{center}
    
    \InputIfFileExists{tapes/cipriano/propdecomposition6.tikz}{}{\input{./tikz/tapes/cipriano/propdecomposition6.tikz}}

    \end{center}
    \begin{center}
    
    \InputIfFileExists{tapes/cipriano/propdecomposition7.tikz}{}{\input{./tikz/tapes/cipriano/propdecomposition7.tikz}}

    \end{center}
    \begin{center}
    
    \InputIfFileExists{tapes/cipriano/propdecomposition8.tikz}{}{\input{./tikz/tapes/cipriano/propdecomposition8.tikz}}

    \end{center}
    now applying the inductive hypothesis for $c$ and $d$, we obtain the statement.
\end{proof}

Before proceeding, recall that Theorem~\ref{thm:completenessbooleancircuits} implies that every $b\in \Diag{\SigB}[1,A^n]$ is $\eqB$-equal (hence $\eqPB$-equal) to a Boolean vector, i.e., a circuit of the form $\bigotimes_{i=1}^n b_i$ for $b_i \in \{\Flip{0}, \Flip{1}\}$ if $n>0$, and $
    \InputIfFileExists{empty.tikz}{}{\input{./tikz/empty.tikz}}
$ if $n=0$. 

\begin{lemma}\label{lemma:complete1}
  Let $b\in \Diag{\SigB}[1,A^n]$ and $c\colon n\to m$ be a partial Boolean circuit. If $b;D_c \eqPB \Flip{0}$, then $b;T_c\eqPB\Flipm{0}$.
\end{lemma}
\begin{proof}
  We proceed by induction on the structure of $c$. For the base case, if $c\in \SigB \cup \{\id{1},\id{A},\symmt{A}{A}\}$, then $b;D_c=b;\CBdischargern\, \Flip{1}=\Flip{1}$, hence the statement is trivial. If $c$ is $\CBcocopier$, then, since $D_{\scalebox{0.6}{$\CBcocopier$}}$ is the xnor gate, $b$ must be of the form $\Flip{1}\otimes \Flip{0}$ or $\Flip{0}\otimes \Flip{1}$, hence $b;T_{\scalebox{0.6}{$\CBcocopier$}}$ is $\Flip{0}$, by axioms \eqref{ax:B2l} and \eqref{ax:B2r} in Figure~\ref{tab:booleanalgebra}. Now, in the case $c;d$, we have
  \begin{center}
    
    \InputIfFileExists{tapes/cipriano/lemma1fg.tikz}{}{\input{./tikz/tapes/cipriano/lemma1fg.tikz}}

  \end{center}
  \begin{center}
    
    \InputIfFileExists{tapes/cipriano/lemma1fgpt2.tikz}{}{\input{./tikz/tapes/cipriano/lemma1fgpt2.tikz}}

  \end{center}
   Finally, in the case $c\otimes d$, since  $b\eqPB b_1\otimes b_2$ for some $b_1\in \Diag{\SigB}[1,A^{n'}]$ and $b_2\in \Diag{\SigB}[1,A^{n''}]$ such that $n'+n''=n$, we have 
  \begin{center}
    
    \InputIfFileExists{tapes/cipriano/lemma1fperg.tikz}{}{\input{./tikz/tapes/cipriano/lemma1fperg.tikz}}

  \end{center}
  \begin{center}
    
    \InputIfFileExists{tapes/cipriano/lemma1fpergpt2.tikz}{}{\input{./tikz/tapes/cipriano/lemma1fpergpt2.tikz}}

  \end{center}
\end{proof}

\begin{lemma}\label{lemma:complete3}
  For all $b\in \Diag{\SigB}[1,A^n]$ and $c,d\in \Diag{PB}[A^n,A^m]$, 
  \begin{center} if $\osemPB{b;c}=\osemPB{b;d}$, then $\osemPB{b;D_c}=\osemPB{b;D_d}$ and $\osemPB{b;T_c}=\osemPB{b;T_d}$.\end{center}
\end{lemma}

\begin{theorem}\label{thm:completenesspartialcircuits}[Completeness]
  For all $c,d\in \Diag{PB}[A^n,A^m]$, if $\osemPB{c}=\osemPB{d}$ then $c=_{\mathbb{\SigPB }}d$.
\end{theorem}
\begin{proof}
  For every $b\in \Diag{\SigB}[1,A^n]$, 
  \begin{align}
    \osemPB{b;c}&=\osemPB{b};\osemPB{c} \notag\\
    &=\osemPB{b};\osemPB{d} \tag{Hp.}\\
    &=\osemPB{b;d}. \notag
  \end{align}
  Lemma~\ref{lemma:complete3} implies that, for every $b\in \Diag{\SigB}[1,A^n]$, $\osemPB{b;D_c}=\osemPB{b;D_d}$ and $\osemPB{b;T_c}=\osemPB{b;T_d}$, hence $\osemPB{D_c} = \osemPB{D_d}$ and $\osemPB{T_c} = \osemPB{T_d}$. Now, since $D_c,D_d,T_c$ and $T_d$ are Boolean circuits in $\Diag{\SigB}$, then by Theorem~\ref{thm:completenessbooleancircuits} it holds that $D_c\eqB D_d$ and $T_c\eqB T_d$. Finally, by Proposition~\ref{prop:decompositionpartialcircuits}, since $c$ and $d$ decompose through the same domain and total part, we have that $c\eqPB d$.
\end{proof}
Moreover, one can characterise exactly the image of $\Diag{\SigPB}$ through $\osemPB{-}$. Let $\Cat{Par}_2$ be the full subcategory of $\Cat{Par}$ where objects are powers of the set $2$, i.e.,  $2^n$ for all $n \in \mathbb{N}$.

\begin{proposition}\label{prop:isopartialboolean}
The monoidal functor  $\osemPB{-}\colon \Diag{{\SigPB }} \to \Cat{Par}$ factors as
\[\xymatrix{ \Diag{\SigPB } \ar@{->>}[r]& \Diag{\mathbb{\SigPB }} \ar[r]^{\cong}& \Cat{Par}_2 \ar@{^{(}->}[r]& \Cat{Par} }\]
where the leftmost and rightmost functors are the obvious quotients and injections and the central arrow is a monoidal isomorphism.
\end{proposition}

Theorem \ref{thm:completenesspartialcircuits} also provides a useful characterisation of diagrams in $\Diag{\SigPB}[1,A^n]$. First consider the diagram $\ncopierbis\colon A \to A^n$ defined inductively as
\begin{center}
  
    \InputIfFileExists{tapes/cipriano/ncopierbis.tikz}{}{\input{./tikz/tapes/cipriano/ncopierbis.tikz}}

\end{center}
then define the diagram $\Flipn{\bot}\colon 1 \to n$ as $\Flip{\bot};\ncopierbis$.
\begin{lemma}\label{lemma:caratterizzazionecircuiti0n}
  For all $c\in \Diag{\SigPB}[1,A^n]$, either $c\eqPB \Flipn{\bot}$ or $c\eqPB b$ for some $b\in \Diag{\SigB}[1,A^n]$.
\end{lemma}
\begin{proof}
Recall that $\osemPB{c}$ is a partial function of type $1 \to 2^n$. Thus, it is either $\bot$ or  of the form $\osemPB{b}$ for some $b\in \Diag{\SigB}[1,A^n]$. In the first case, simple computations confirm that  $\osemPB{\Flipn{\bot}}=\bot$, thus $\osemPB{c}=\osemPB{\Flipn{\bot}}$. By Theorem~\ref{thm:completenesspartialcircuits}, $c\eqPB \Flipn{\bot}$. In the second case, $\osemPB{c}=\osemPB{b}$ and, again by Theorem~\ref{thm:completenesspartialcircuits}, $c\eqPB b$.
\end{proof}
%
%

%
%
%
%
%
\subsection{Probabilistic Boolean circuits}\label{ssec:probabilisticboolean}
Finally, we illustrate probabilistic Boolean circuits from \cite{piedeleu2025boolean}.
Consider the monoidal signature  \[ \SigPRB \defeq \SigPB \cup\{ \Flip{p} \mid p\in (0,1)\}\]  
 where $\Flip{p}$ denotes a probabilistic gate that receives no input and emits a Boolean signal on the right, which is $1$ with probability $p$ and $0$ with probability $1-p$. Thus, string diagrams in $\Diag{\SigPRB}$ represent arrows of $\KlD$.

To formally define their semantics, recall the symmetric monoidal category $(\KlD, \otimes, \uno)$ and observe that there is an obvious monoidal embedding $\ParJ\colon \Cat{Par} \to \KlD$ which is the identity on objects and associates to a partial function $f\colon X \to Y$ the function $\ParJ(f)\colon X \to \Dis(Y)$, mapping $x$ to $\delta_{f(x)}$, if $f(x)$ is defined and to the null distribution $\star$ otherwise. 

The semantics of $\SigPRB$ is given by the monoidal functor $\osem{-}\colon \Diag{\SigPRB} \to \KlD$ defined on objects as $A^n \mapsto 2^n$. On arrows, it is defined by first interpreting generators  $c\in \SigPB$ as $\ParJ(\osemPB{c})$, i.e.
\[
    \begin{array}{rclrclrcl}
      \osem{\Andgate}\colon 2 \times 2& \to & 2  & \osem{\Notgate}\colon  2& \to & 2 & \osem{\Flip{1}} \colon 1 & \to & 2 \\
      (x,y)&\mapsto&\delta_{x\wedge y} &  x & \mapsto & \delta_{\neg x} & \bullet &\mapsto&\delta_{1}\\[4pt]
       \osem{\CBcopier} \colon 2  & \to & 2\times 2 & \osem{\CBdischarger} \colon 2 & \to & 1 & \osem{\CBcocopier} \colon 2 \times 2 & \to & 2  \\
      x &\mapsto&\delta_{(x,x)} & x &\mapsto&\delta_{\bullet} & (x,y) &\mapsto& \begin{cases}\delta_{x} &\text{if }x=y\\ \star & \text{else}\end{cases}, 
    \end{array}
  \]

 the generator $\Flip{p}$ as 
\begin{equation}\label{eq:semanticsFlip}
\begin{array}{rcl}
      \osem{\Flip{p}}\colon 1& \to & 2   \\
      \bullet&\mapsto& \delta_1 +_p \delta_0  \end{array}
    \end{equation}
and then inductively on the structure of the monoidal category $(\KlD, \otimes, \uno)$. Simple computations confirm that $\osem{-}$ extends both $\osemB{-}$ and $\osemPB{-}$, in the sense that the following diagrams commute.
\[
\xymatrix{\Diag{\SigB} \ar@{^{(}->}[r] \ar[d]_{\osemB{-}}& \Diag{\SigPB} \ar@{^{(}->}[r] \ar[d]|{\osemPB{-}} & \Diag{\SigPRB} \ar[d]^{\osem{-}}\\
\Sets \ar@{^{(}->}[r]_{\Pa}& \Cat{Par} \ar@{^{(}->}[r]_{\ParJ}& \KlD
}
\]

\medskip
Figures~4 and~5 in \cite{piedeleu2025boolean} provide a sound and complete equational axiomatisation of probabilistic Boolean circuits. The system essentially consists of the axioms of Boolean algebras (reported in Figure~\ref{tab:booleanalgebra}), together with four non-trivial axioms governing the behaviour of $\Flip{p}$ and additional axioms for $\CBcocopier$, which differ from those presented in Figure~\ref{tab:partialbooleanalgebra}.

Importantly, the equivalence induced by these axioms is strictly coarser than the one given by $\osem{-}$. In fact, two diagrams $c$ and $d$ are equivalent under the axioms of \cite{piedeleu2025boolean} if and only if $\osem{c} \propto \osem{d}$. The relation $\propto$ is defined on arrows of $\KlD$ by setting $f \propto g$ whenever there exists a real number $\lambda > 0$ such that $f(y\mid x) = \lambda \cdot g(y\mid x)$ for all $x \in X$ and $y \in Y$. In the next section, we provide a complete  axiomatisation for the equivalence induced by $\osem{-}$ by means of probabilistic tape diagrams.

%% file: tikz/tapes/cipriano/D4right.tikz
\begin{tikzpicture}
	\begin{pgfonlayer}{nodelayer}
		\node [style=none] (11) at (0, 0) {};
		\node [style=none] (12) at (3, 0) {};
		\node [style=black] (13) at (3, 0) {};
		\node [style=none, scale=0.6] (19) at (2.5, 0.5) {$n$};
	\end{pgfonlayer}
	\begin{pgfonlayer}{edgelayer}
		\draw (11.center) to (12.center);
	\end{pgfonlayer}
\end{tikzpicture}

%% file: tikz/tapes/cipriano/axF9left.tikz
\begin{tikzpicture}
	\begin{pgfonlayer}{nodelayer}
		\node [style=none] (176) at (-3.25, 0) {};
		\node [style=none] (178) at (-1.25, 0) {$\coflip{1}$};
		\node [style=none] (179) at (1.5, 0) {$\Flip{1}$};
	\end{pgfonlayer}
	\begin{pgfonlayer}{edgelayer}
		\draw (176.center) to (178.center);
	\end{pgfonlayer}
\end{tikzpicture}

%% file: tikz/tapes/cipriano/axF10left.tikz
\begin{tikzpicture}
	\begin{pgfonlayer}{nodelayer}
		\node [style=none] (0) at (-0.75, 0) {$\Flip{1}$};
		\node [style=none] (1) at (1, 0) {$\coflip{1}$};
	\end{pgfonlayer}
\end{tikzpicture}

%% file: tikz/tapes/cipriano/axF10right.tikz
\begin{tikzpicture}
	\begin{pgfonlayer}{nodelayer}
		\node [style=none] (2) at (-0.75, 0) {$\Flip{0}$};
		\node [style=none] (3) at (1, 0) {$\coflip{0}$};
	\end{pgfonlayer}
\end{tikzpicture}

%% file: sections/probpartialbooltapes-new.tex
\section{Probabilistic Boolean tapes}\label{sec:probbooltapes}
In \cite[Example 30]{bonchi2025tapediagramsmonoidalmonads}, an encoding of probabilistic Boolean circuits into tape diagrams is presented. 
In this section, we first recall such encoding (Section~\ref{ssec:probbooltapes:encoding}), we then illustrate an axiomatisation (Section~\ref{ssec:probbooltapes:axiomatisation}) and we prove its completeness (Section~\ref{ssec:probbooltapes:completeness}).

\subsection{Encoding probabilistic Boolean circuits into tapes}\label{ssec:probbooltapes:encoding}
Example 30 in \cite{bonchi2025tapediagramsmonoidalmonads} encodes probabilistic Boolean circuits from $\Diag{\SigPRB}$ into $\CatT{\Diag{\SigPB}}$, namely probabilistic tape diagrams of \emph{partial} Boolean circuits.
The encoding is defined as the unique (symmetric strict) monoidal functor $\encoding{-} \colon (\Diag{\SigPRB},\otimes, \uno) \to (\CatT{\Diag{\SigPB}},\otimes, \uno)$ mapping
\[
\resizebox{0.95\linewidth}{!}{$
\Andgate \mapsto \Andgate[t] \qquad
\Notgate \mapsto \Notgate[t] \qquad
\Flip{1} \mapsto \Flip{1}[t] \qquad
\CBcopier \mapsto  \CBcopier[t] \qquad
\CBdischarger \mapsto \CBdischarger[t] \qquad
\CBcocopier \mapsto \CBcocopier[t] 
$}
\]
and $\Flip{p} \mapsto \Flip{p}[t]$ where 
\begin{equation}\label{eq:exprobabilistictapes}
\Flip{p}[t] \defeq 
    \InputIfFileExists{tapes/examples/1p0.tikz}{}{\input{./tikz/tapes/examples/1p0.tikz}}
\end{equation}

The abstract theory developed in \cite{bonchi2025tapediagramsmonoidalmonads} guarantees the existence of a unique morphism of rig categories $\dsem{-}\colon \CatT{\Diag{\SigPB}} \to \KlD$ which assigns to each tape its semantics in $\KlD$. For convenience of the reader the definition of $\dsem{-}$  is reported in Table~\ref{eq:SEMANTICA}. For example, $\dsem{-}$ assigns to  the following tape 
\begin{equation*}
        
    \InputIfFileExists{tapes/examples/ANDpOR.tikz}{}{\input{./tikz/tapes/examples/ANDpOR.tikz}}
 
\end{equation*}
the arrow $2 \times 2 \to 2$ of \(\KlD\) which maps any pair of Booleans $(x,y)$ into $x\wedge y$ with probability $p$ and $x \vee y$ with probability $(1-p)$.
Instead, $\Flip{p}[t]$  is mapped by $\dsem{-}$ into the arrow $1 \to 2$ assigning to $\bullet$ the distribution $1 \mapsto p$, $0 \mapsto (1-p)$. Observe that this is exactly $\osem{\Flip{p}}$ as defined in \eqref{eq:semanticsFlip}. More generally, a simple inductive argument confirms that $\dsem{\encoding{c}}=\osem{c}$ for all $c$ in $\Diag{ PB}$, i.e., the encoding preserves the semantics.

\begin{proposition}\label{prop:encoding} For all $c\in \Diag{\SigPRB}$, $\osem{c}=\CBdsem{\encoding{c}}$.
\end{proposition}
\medskip

Interestingly, tapes in $\CatT{\Diag{\SigPB}}$ are strictly more expressive than probabilistic Boolean circuits in $\Diag{\SigPRB}$, in the sense that they denote a strictly larger class of arrows in $\KlD$. This difference becomes apparent by recalling, from~\cite{bonchi2025tapediagramsmonoidalmonads}, how probabilistic control is implemented in $\Diag{\SigPRB}$ and in $\CatT{\Diag{\SigPB}}$.

In~\cite{piedeleu2025boolean}, probabilistic control is realised via a multiplexer, represented by the tape $\scalebox{0.7}{\Ifgate[t]}$. Intuitively, when $\Flip{p}[t]$, $\Tcirc{c}{\!\!\!\!}{\!\!\!\!}$ and $\Tcirc{d}{\!\!\!\!}{\!\!\!\!}$ are connected, respectively, to the first, second, and third inputs of the multiplexer, the resulting output coincides with that of $c$ with probability $p$ and with that of $d$ with probability $1-p$.
Formally, this behaviour is captured by the composite
\[
(\Flip{p}[t] \;\per\; \Tcirc{c}{\!\!\!\!}{\!\!\!\!} \;\per\; \Tcirc{d}{\!\!\!\!}{\!\!\!\!}) \; ; \; \scalebox{0.7}{\Ifgate[t]},
\]
which, by the definition of $\Flip{p}[t]$ in~\eqref{eq:exprobabilistictapes} and of $\per$ in Table~\ref{tab:producttape}, corresponds to the tape shown on the left below:
\begin{equation*}
    
    \InputIfFileExists{tapes/examples/multiplexT.tikz}{}{\input{./tikz/tapes/examples/multiplexT.tikz}}

    \qquad \qquad
    
    \InputIfFileExists{tapes/examples/pchoice.tikz}{}{\input{./tikz/tapes/examples/pchoice.tikz}}

\end{equation*}

Although the two diagrams above exhibit similar behaviour, a crucial difference emerges when $d$ (or symmetrically $c$) is instantiated as $\Flip{\bot}[t]$, which denotes the null subdistribution and intuitively represents failure. In this case, the circuit on the left always fails (Lemma~4.4 in~\cite{piedeleu2025boolean}). By contrast, the circuit on the right still produces the output of $c$ (respectively $d$) with probability $p$ (respectively $1-p$).

\begin{table}[t]
    \renewcommand{\arraystretch}{1.5}
\!\!\!\begin{tabular}{l@{\;\;\;\;}l@{\;\;\;\;}l@{\;\;\;\;}l@{\;\;\;\;}l}
$\CBdsem{\id{A}} \defeq \id{2} $&$ \CBdsem{\id{1}} \defeq \id{1} $&$ \CBdsem{\symmt{A}{A}}  \defeq \symmt{2}{2}  $&$ \CBdsem{c;d}  \defeq \CBdsem{c} ; \CBdsem{d}   $ & $ \CBdsem{c\per d}  \defeq \CBdsem{c}  \per \CBdsem{d} $\\
$\CBdsem{s}  \defeq \osem{s} $& $\CBdsem{\, \tapeFunct{c} \,} \defeq \CBdsem{c}  $ & $\CBdsem{\diagp{A^n}}  \defeq \; \diagp{2^n}$& $\CBdsem{\bang{A^n}} \defeq \bang{2^n}$  \; $\CBdsem{\cobang{A^n}}  \defeq \cobang{2^n}$    & $\CBdsem{\codiag{A^n}} \defeq \codiag{2^n} $ \\
$\CBdsem{\id{A^n}} \defeq \id{2^n} $&$ \CBdsem{\id{\zero}}  \defeq \id{\zero}  $&$ \CBdsem{\symmp{A^n}{A^m}} \defeq \symmp{2^n}{2^m} $&$ \CBdsem{\s;\t}  \defeq \CBdsem{\s}  ; \CBdsem{\t}  $&$ \CBdsem{\s \piu \t}  \defeq \CBdsem{\s}  \piu \CBdsem{\t}  $ 
\end{tabular}
\caption{The semantics $\dsem{-}\colon \CatT{\Diag{\SigPB}} \to \KlD$. Here $s$ is a generator in $\SigPB$.}\label{eq:SEMANTICA}
\end{table}

\medskip

\subsection{Axiomatisation}\label{ssec:probbooltapes:axiomatisation}
In order to provide a complete axiomatisation for the equivalence induced by $\dsem{-}$, we reuse $\eqPB$ which, as proved in Theorem \ref{thm:completenesspartialcircuits}, provides a complete axiomatisation for partial Boolean circuits. 
Since the axioms in $\ThPB$ are sound with respect to $\osem{-}$ and since $\osem{-}=\dsem{\encoding{-}}$, these axioms are also sound for $\dsem{-}$. However, they are not complete, as explained below.

\begin{example}
Consider the congruence $\tapeeqPB$ on $ \CatT{\Diag{\SigPB}}$ generated by $\eqPB$ as in \eqref{eq:congr2}. By Proposition \ref{prop:quotienttheory}, the quotient of $ \CatT{\Diag{\SigPB}}$ by  $\tapeeqPB$ is exactly $ \CatT{\Diag{\ThPB}}$ which, 
by Corollary~\ref{cor:isotapematrices} and Proposition~\ref{prop:isopartialboolean} is isomorphic to $\stmat{\Cat{Par}_2^+}$. Now consider only $1\times 1$ matrices: these are arrows in $\Cat{Par}_2^+$, i.e., subdistributions of partial functions. These are not in one-to-one correspondence with Kleisli arrows. Take for instance the null subdistribution $\star \in \Dis(\Cat{Par}[1, 2])$ and $\delta_\bot \in \Dis(\Cat{Par}[1, 2])$ which are distinct arrows in $\Cat{Par}_2^+[1,2]$. The former is drawn as the tape \scalebox{0.6}{
    \InputIfFileExists{tapes/cipriano/AXPBP2right.tikz}{}{\input{./tikz/tapes/cipriano/AXPBP2right.tikz}}
} and the latter as $\Flip{\bot}[t]$. It is easy to see that $\dsem{-}$ maps both diagrams into the $\KlD$-arrow $\star_{1,2}$.
\end{example}

To obtain a complete axiomatisation, we need to add to $\ThPB$ the axioms in Figure~\ref{ax:BooleanTAPES}. The tape diagram on the right-hand side of \eqref{ax:AXPBP1} with probability $\frac{1}{2}$ tests if the input is $1$ and returns $1$ and, with probability $\frac{1}{2}$, tests if the input is $0$ and returns $0$. The equality \eqref{ax:AXPBP1} forces such tape to be equal to $\frac{1}{2} \cdot \id{A}$; Axiom \eqref{ax:AXPBP2} forces $\Flip{\bot}[t]$ to be the null subdistribution; Axiom \eqref{ax:canc} is an implication which asserts cancellativity: see Section~\ref{ssec:pca}. 

\begin{figure}[H]
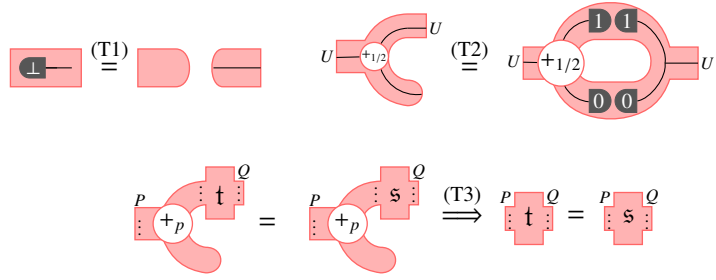

  \centering

  \mylabelbis{ax:AXPBP1}{T2}
  \mylabelbis{ax:AXPBP2}{T1}
  \mylabelbis{ax:canc}{T3}
  
  \[
  \begin{array}{c@{\qquad}c@{\qquad}c}

    \InputIfFileExists{tapes/cipriano/AXPBP2left.tikz}{}{\input{./tikz/tapes/cipriano/AXPBP2left.tikz}}

    \stackrel{\eqref{ax:AXPBP2}}{=}
    
    \InputIfFileExists{tapes/cipriano/AXPBP2right.tikz}{}{\input{./tikz/tapes/cipriano/AXPBP2right.tikz}}

    &
    
    \InputIfFileExists{tapes/cipriano/AXPBP1left.tikz}{}{\input{./tikz/tapes/cipriano/AXPBP1left.tikz}}

    \stackrel{\eqref{ax:AXPBP1}}{=}
    
    \InputIfFileExists{tapes/cipriano/AXPBP1right.tikz}{}{\input{./tikz/tapes/cipriano/AXPBP1right.tikz}}

    &
    {}
    \\[30pt]
    \multicolumn{3}{c}{
      
    \InputIfFileExists{tapes/cipriano/Axcancleft.tikz}{}{\input{./tikz/tapes/cipriano/Axcancleft.tikz}}

      \stackrel{\eqref{ax:canc}}{\implies}
      
    \InputIfFileExists{tapes/cipriano/Axcancright.tikz}{}{\input{./tikz/tapes/cipriano/Axcancright.tikz}}

    }
  \end{array}
  \]
  
\caption{Axioms for tapes of partial Boolean circuits.}\label{ax:BooleanTAPES}
\end{figure}

From $\ThPB$ and the axioms in Figure~\ref{ax:BooleanTAPES}, one generates the congruence $\dot{\sim}$ on $\CatT{\Diag{\SigPB}}$ by means of the following inference rules.
\begin{equation}\label{eq:congr3}
        \scalebox{0.85}{$
        \centering
        \begin{array}{c}
        \begin{array}{c@{\qquad\qquad}c@{\qquad\qquad}c@{\qquad\qquad}c@{\qquad\qquad}c}
            \inferrule*[right=\eqref{ax:AXPBP1}]{\t_1 \stackrel{\eqref{ax:AXPBP1}}{=} \t_2}{\t_1 \dot{\sim} \t_2}
            &
            \inferrule*[right=\eqref{ax:AXPBP2}]{\t_1 \stackrel{\eqref{ax:AXPBP2}}{=} \t_2}{\t_1 \dot{\sim} \t_2}
            &
            \inferrule*[right=($\textsc{R}$)]{-}{\t \dot{\sim} \t}
            &
            \inferrule*[right=($\textsc{S}$)]{\t_1 \dot{\sim} \t_2}{\t_2 \dot{\sim} \t_1}
            &    
            \inferrule*[right=($\textsc{T}$)]{\t_1 \dot{\sim} \t_2 \quad \t_2 \dot{\sim} \t_3}{\t_1 \dot{\sim} \t_3}
        \end{array}
        \\[10pt]
        \begin{array}{c@{\qquad}c@{\qquad}c@{\qquad}c@{\qquad}c}
        \inferrule*[right=($\ThPB$)]{c \eqPB d}{\tape{c} \dot{\sim} \tape{d}}
            &
            \inferrule*[right=\eqref{ax:canc}]{p\cdot \t \dot{\sim} p\cdot \s }{\t \dot{\sim} \s}
            &
            \inferrule*[right=($;$)]{\t_1 \dot{\sim} \t_2 \quad \s_1 \dot{\sim} \s_2}{\t_1;\s_1 \dot{\sim} \t_2;\s_2}
            &
            \inferrule*[right=($\piu$)]{\t_1 \dot{\sim} \t_2 \quad \s_1 \dot{\sim} \s_2}{\t_1\piu\s_1 \dot{\sim} \t_2 \piu \s_2}
            &
            \inferrule*[right=($\per $)]{\t_1 \dot{\sim} \t_2 \quad \s_1 \dot{\sim} \s_2}{\t_1\per \s_1 \dot{\sim} \t_2 \per \s_2}       
        \end{array}
        \end{array}
        $}
\end{equation}

Simple computations confirm that the axioms are sound.
\begin{proposition}[Soundness]\label{prop:soundnesstapes}
For all $\s,\t \in \CatT{\Diag{\SigPB}}$, if $\s \eqsynbis \t$ then $\dsem{\s}=\dsem{\t}$.
\end{proposition}

To prove completeness, it is convenient to consider diagrams in $\CatT{\Diag{\ThPB}}$. Let $\eqsyn$ be the congruence on $\CatT{\Diag{\ThPB}}$ defined as $\dot{\sim}$ but omitting the rule ($\ThPB$). By means of Proposition~\ref{prop:quotienttheory}, one can easily prove the following result.

\begin{proposition}\label{prop:twoequivalence}
For all $\s , \t \in \CatT{\Diag{\SigPB}}$, $\s \dot{\sim} \t$ iff $\CatT{Q_{\ThPB}}(\s) \eqsyn \CatT{Q_{\ThPB}}(\t)$.
\end{proposition}

Thanks to the above proposition, we can work with $\eqsyn$ rather than $\dot{\sim}$. We conclude this section by illustrating three key properties of $\eqsyn$. First, axiom \eqref{ax:AXPBP2} easily entails the following.

\begin{lemma}\label{lemma:bottomstar2}
 $\star_{A^0,A^n} \sim \Flipn{\bot}[t]$.  
\end{lemma}

Hereafter, we write $\Cat{B}[1,A^n]$ for the set of $n$-ary Boolean vectors, i.e., all those tapes in $\CatT{\Diag{\ThPB}}[1,A^n]$ of the form $\overrightarrow{b}\defeq\bigotimes_{i=1}^n b_i$ for $b_i \in \{\Flip{0}[t], \Flip{1}[t]\}$ if $n\not=0$, and $
    \InputIfFileExists{tapes/empty.tikz}{}{\input{./tikz/tapes/empty.tikz}}
$ if $n=0$. If $\overrightarrow{b}\in\Cat{B}[1,A^n]$, we write $\overleftarrow{b}$ for the tape $\bigotimes_{i=1}^n b'_i\in \CatT{\Diag{\ThPB}}[A^n,1]$ where $b'_i$ is $\coflip{1}[t]$ if $b_i$ is $\Flip{1}[t]$ and $\coflip{0}[t]$ if $b_i$ is $\Flip{0}[t]$.

The second property crucially relies on axiom \eqref{ax:AXPBP1}.

\begin{lemma}\label{lemma:axiomsninput}
 For all $n\in \mathbb{N}$, $\frac{1}{2^n}\cdot \id{A^n} \sim \sum_{\overrightarrow{b}\in \Cat{B}[1,A^n]}^{} \frac{1}{2^n}\cdot\overleftarrow{b};\overrightarrow{b}$;
\end{lemma}
\begin{proof}
By induction  on $n$. The base case is trivial. For the inductive step, we have
 \begin{align}
  \frac{1}{2^{n+1}}\cdot \id{A^{n+1}}&= \frac{1}{2^n}\frac{1}{2}\cdot( \id{A}\otimes \id{A^n}) \notag\\
  &= \frac{1}{2^n}\cdot\left(\frac{1}{2}\cdot \id{A}\otimes \id{A^n} \right) \tag{\eqref{eq:monoidalenrichment}}\\
  &= (\frac{1}{2}\cdot \id{A})\otimes (\frac{1}{2^n}\cdot \id{A^n}) \tag{\eqref{eq:monoidalenrichment}}\\
  &\sim (\frac{1}{2}\cdot \id{A})\otimes \left(\sum_{\overrightarrow{b}\in \Cat{B}[1,A^n]}^{} \frac{1}{2^n}\cdot\overleftarrow{b};\overrightarrow{b}\right) \tag{Inductive hypothesis}\\
  &\sim (\coflip{1}[t];\Flip{1}[t] +_{\frac{1}{2}} \coflip{0}[t];\Flip{0}[t]) \otimes \left(\sum_{\overrightarrow{b}\in \Cat{B}[1,A^n]}^{} \frac{1}{2^n}\cdot\overleftarrow{b};\overrightarrow{b}\right) \tag{Axiom \ref{ax:AXPBP1}}\\
  &= \sum_{\overrightarrow{b}\in \Cat{B}[1,A^n]}^{} \frac{1}{2^n}\cdot\left((\coflip{1}[t];\Flip{1}[t] +_{\frac{1}{2}} \coflip{0}[t];\Flip{0}[t]) \otimes (\overleftarrow{b};\overrightarrow{b})\right) \tag{\eqref{eq:monoidalenrichment}}\\
  &=\sum_{\overrightarrow{b}\in \Cat{B}[1,A^n]}^{} \frac{1}{2^n}\cdot\left(((\coflip{1}[t];\Flip{1}[t]) \otimes (\overleftarrow{b};\overrightarrow{b})) +_{\frac{1}{2}} ((\coflip{0}[t];\Flip{0}[t]) \otimes (\overleftarrow{b};\overrightarrow{b}))\right) \tag{SMC}\\
  &= \sum_{\overrightarrow{b}\in \Cat{B}[1,A^n]}^{} \frac{1}{2^n}\cdot\left((\coflip{1}[t]\otimes \overleftarrow{b});(\Flip{1}[t]\otimes \overrightarrow{b}) +_{\frac{1}{2}} (\coflip{0}[t]\otimes \overleftarrow{b});(\Flip{0}[t]\otimes \overrightarrow{b})\right) \tag{\eqref{eq:monoidalenrichment}}
 \end{align}
 Now, since the set $\{\Flip{1}[t]\otimes \overrightarrow{b}, \Flip{0}[t]\otimes \overrightarrow{b}\}$ as $\overrightarrow{b}\in \Cat{B}[1,A^n]$ is exactly $\Cat{B}[1,A^{n+1}]$, then, rearranging the above sum we obtain $\frac{1}{2^{n+1}}\cdot \id{A^{n+1}}= \sum_{\overrightarrow{b}\in \Cat{B}[1,A^{n+1}]}^{} \frac{1}{2^{n+1}}\cdot\overleftarrow{b};\overrightarrow{b}$, as desired.
\end{proof}

The  lemma above and axiom \eqref{ax:canc} give us the third key result.

\begin{lemma}\label{lemma:PBP2} 
Let $\s,\t\in \CatT{\Diag{\ThPB}}[A^n,A^m]$. If, for all $\overrightarrow{b}\in \Cat{B}[1,A^n]$, $\overrightarrow{b};\s \eqsyn\overrightarrow{b}; \t$, then $\s \eqsyn \t$.
\end{lemma}
\begin{proof}
Consider the following derivation: 
\begin{align}
  \frac{1}{2^n} \cdot \t &= \frac{1}{2^n}\cdot (\id{A^n};\t) \notag\\
  &= (\frac{1}{2^n}\cdot \id{A^n});\t \tag{Lemma \ref{lemma:initialproperties2}}\\ 
  &\sim (\sum_{\overrightarrow{b}\in \Cat{B}[1,A^n]}^{} \frac{1}{2^n}\cdot\overleftarrow{b};\overrightarrow{b});\t \tag{Lemma \ref{lemma:axiomsninput}}\\
  &= \sum_{\overrightarrow{b}\in \Cat{B}[1,A^n]}^{} \frac{1}{2^n}\cdot\overleftarrow{b};\overrightarrow{b};\t \tag{\eqref{eq:enr}}\\
  &\sim \sum_{\overrightarrow{b}\in \Cat{B}[1,A^n]}^{} \frac{1}{2^n}\cdot\overleftarrow{b};\overrightarrow{b};\s \tag{Hp.}\\
  &=(\sum_{\overrightarrow{b}\in \Cat{B}[1,A^n]}^{} \frac{1}{2^n}\cdot\overleftarrow{b};\overrightarrow{b});\s \tag{\eqref{eq:enr}}\\
  &\sim (\frac{1}{2^n}\cdot \id{A^n});\s \tag{Lemma \ref{lemma:axiomsninput}}\\
  &= \frac{1}{2^n} \cdot \s. \tag{Lemma \ref{lemma:initialproperties2}} 
\end{align}
Hence, applying axiom \eqref{ax:canc}, we conclude that $\t \eqsyn \s$.
\end{proof}

\subsection{Completeness}\label{ssec:probbooltapes:completeness}
We write $\eqsynq{\CatT{\Diag{\ThPB}}}$ for the category obtained as the quotient of $\CatT{\Diag{\ThPB}}$ by $\eqsyn$ and by $Q_\eqsyn\colon \CatT{\Diag{\ThPB}} \to \eqsynq{\CatT{\Diag{\ThPB}}}$ the functor mapping each tape into its $\eqsyn$ equivalence class. By Propositions~\ref{prop:soundnesstapes} and \ref{prop:twoequivalence}, $\dsem{-}\colon \CatT{\Diag{\SigPB}} \to \KlD$ can be factored as
\[\xymatrix{\CatT{\Diag{\SigPB}} \ar[r]^{\CatT{Q_\ThPB}} & \CatT{\Diag{\ThPB}} \ar[r]^{Q_\eqsyn} & \eqsynq{\CatT{\Diag{\ThPB}}} \ar@{.>}[r]^I& \KlD}\]
for some functor $I$. Proving completeness amounts to proving that  $I$  is faithful. 

\medskip

By Corollary~\ref{cor:isotapematrices} and Proposition~\ref{prop:isopartialboolean}  we know that 
\begin{equation}\label{eq:i2so}
  \CatT{\Diag{\ThPB}} \cong \stmat{\Diag{\ThPB}^+}\cong \stmat{\Cat{Par}_2^+}\text{.}
\end{equation}
Now, observe that the functor $\ParJ\colon \,\Cat{Par} \to \KlD$ from Section~\ref{ssec:probabilisticboolean} restricts to $\ParJ_2\colon \Cat{Par}_2 \to \KlD$. Since $\KlD$ is a convex biproduct category, by the two adjunctions in Theorems~ \ref{thm:freeenriched} and \ref{thm:matfree}, there exists a functor \[\ParJ_2'\colon \stmat{\Cat{Par}_2^+} \to \KlD\text{.} \]

Consider the diagram below, where the vertical isomorphisms are those in \eqref{eq:i2so}.

\[
\qquad \xymatrix@R=2ex@C=3ex{
\CatT{\Diag{\ThPB}} \ar[r]^{\eqsynq{Q}} \ar[d]_{\cong} 
  & \eqsynq{\CatT{\Diag{\ThPB}}} \ar[dd]^{I} \\
{\stmat{\Diag{\ThPB}^+}} \ar[d]_{\cong} 
  & {} \\
\stmat{\Cat{Par}_2^+} \ar[r]_{\ParJ_2'} 
  & \KlD
}
\]
In the following result we prove that the above diagram commutes, where $\comp$ denotes the composition of the left vertical isomorphisms with $\ParJ_2'$.

\begin{lemma}\label{lemma:diagramma-commutativo}
$I\circ \eqsynq{Q}=\comp$.
\end{lemma}

\begin{lemma}\label{lemma:Ffaith}
For all $\t,\t'\in \CatT{\Diag{\ThB}}[A^0,A^m]$, if $\comp(\t)=\comp(\t')$, then $\t=\t'$.
\end{lemma}

Now, by relying on Lemma~\ref{lemma:caratterizzazionecircuiti0n}, we have the following normal forms for tapes $\t \colon A^0 \to A^m$.

\begin{lemma}\label{lemma:PBP0}
For all  $\t\in \CatT{\Diag{\ThPB}}[A^0,A^m]$, there exist $n\in\mathbb{N}$, $p_1,\dots,p_n\in (0,1)$, $\overrightarrow{b}_1,\dots, \overrightarrow{b}_n\in \Cat{B}[1,A^m]$ and $p,q\in [0,1]$, such that $\sum_{i=1}^{n}p_i= 1$, and $$\t=(\sum_{i=1}^{n}p_i\cdot \overrightarrow{b}_i)+_p(\star_{A^0,A^m} +_q \Flipm{\bot}[t])\text{.}$$ 
  \end{lemma}

The previous two results, combined with Lemma~\ref{lemma:bottomstar2} allow us to prove the following key property.

\begin{lemma}\label{lemma:PBP1}
  For all $\t,\t'\in \CatT{\Diag{\ThPB}}[A^0,A^m]$, if $\comp(\t)=\comp(\t')$, then $\t\sim \t'$.
\end{lemma}

\begin{theorem}\label{thm:completenessPBPtapes}
  $I\colon \eqsynq{\CatT{\Diag{\ThPB}}} \to \stmat{\KlD_2}$ is faithful.
\end{theorem}
\begin{proof}
  Let $[\t]_\sim, [\s]_\sim \in \eqsynq{\CatT{\Diag{\ThPB}}}[A^n,A^m]$ be such that $I([\t]_\sim)=I([\s]_\sim)$. Then: 
  \begin{align}
    I([\t]_\sim)=I([\s]_\sim) &\implies \comp(\t)=\comp(\s) \tag{Lemma \ref{lemma:diagramma-commutativo}}\\
    &\implies \text{for all } \overrightarrow{b}\in \Cat{B}[1,A^n], \comp(\overrightarrow{b});\comp(\t)=\comp(\overrightarrow{b});\comp(\s) \notag \\
    &\implies \text{for all } \overrightarrow{b}\in \Cat{B}[1,A^n], \comp(\overrightarrow{b};\t)=\comp(\overrightarrow{b};\s) \tag{Functoriality}\\
    &\implies \text{for all } \overrightarrow{b}\in \Cat{B}[1,A^n], \overrightarrow{b};\t \eqsyn \overrightarrow{b};\s \tag{Lemma \ref{lemma:PBP1}}\\
    &\implies \t \eqsyn \s \tag{Lemma \ref{lemma:PBP2}}
  \end{align}
  Hence, $[\t]_\sim=[\s]_\sim$. The above derivation proves that $I$ is faithful for arrows of type $A^n \to A^m$.  This fact and convex products easily entail that $I$ is faithful for arrows $\s,\t \colon A^k \to \bigoplus_{i=1}^nA^{m_i}$. 
Indeed, by Lemma~\ref{decomposition} and Theorem~\ref{thm:TCconvexbiproductcategory}, 
\begin{equation}\label{eq:boh}\s=\langle \s_1, \dots, \s_n \rangle_{\vec{q}} \quad \t=\langle \t_1, \dots, \t_n \rangle_{\vec{p}}\end{equation} 
for some $\vec{q}=q_1,\dots, q_n$, $\vec{p}=p_1,\dots, p_n$, $\s_i \colon A^k \to A^{m_i}$ and $\t_i \colon A^k \to A^{m_i}$. Thus
\begin{align*}
I([\s]_\eqsyn) = I([\t]_\eqsyn) & \Rightarrow \comp(\s)=\comp(\t) \tag{Lemma \ref{lemma:diagramma-commutativo}}\\
& \Rightarrow \text{for all }i \text{, } \comp(\s);\comp(\pi_i)=\comp(\t);\comp(\pi_i) \\
& \Rightarrow \text{for all }i\text{, } \comp(\s;\pi_i)=\comp(\t;\pi_i) \tag{Functoriality} \\
& \Rightarrow  \text{for all }i\text{, }  \s;\pi_i\sim\t;\pi_i \tag{Previous implication}\\
& \Rightarrow  \text{for all }i\text{, }  q_i \cdot \s_i\sim p_i \cdot \t_i \tag{\ref{eq:boh}}\\
& \Rightarrow  \text{for all }i\text{, }  [q_i \cdot \s_i]_\sim = [p_i \cdot \t_i]_\sim 
\end{align*}
Axiom \ref{ax:canc} allows us to apply Proposition~\ref{cor: quotient category is convex biproduct }, which states that $\CatT{\Diag{\ThPB}}_\sim$ is a convex biproduct category and that the quotient functor $Q_\eqsyn\colon \CatT{\Diag{\ThPB}} \to \CatT{\Diag{\ThPB}}_\sim$ is a morphism of convex biproduct categories.  Hence, Proposition~\ref{prop:functor 1 e mezzo} implies that $(\bigoplus_{i=1}^nA^{m_i}, [\pi_i]_\sim)$ is the $n$-ary convex product of $A^{m_1}, \dots, A^{m_n}$ in $\CatT{\Diag{\ThPB}}_\sim$ and $ [\s]_\sim$ is the unique arrow such that for all $i=1,\dots,n$,
\begin{align}
  [\s]_\sim;[\pi_i]_\sim &= [\s;\pi_i]_\sim \tag{Functoriality of $Q_\eqsyn$}\\
   &=[q_i \cdot \s_i]_\sim \tag{\ref{eq:boh}}\\
  &=q_i\cdot [\s_i]_\sim; \tag{pca-enrichment of $\CatT{\Diag{\ThPB}}_\sim$}
\end{align}
while $[\t]_\sim$ is the unique arrow such that for all $i=1,\dots,n$,
\begin{align}
  [\t]_\sim;[\pi_i]_\sim&= [\t;\pi_i]_\sim \tag{Functoriality of $Q_\eqsyn$}\\
  &= [p_i \cdot \t_i]_\sim \tag{\ref{eq:boh}}\\
  &=p_i\cdot [\t_i]_\sim. \tag{pca-enrichment of $\CatT{\Diag{\ThPB}}_\sim$}
\end{align}
Hence, since for all $i=1,\dots,n$, $[q_i \cdot \s_i]_\sim = [p_i \cdot \t_i]_\sim$, we have $[\s]_\sim = [\t]_\sim$.


For arrows of arbitrary type $\bigoplus_{j=1}^o A^{k_j} \to \bigoplus_{i=1}^nA^{m_i}$, one can easily rely on the universal property of coproducts and the case that we just proved.
\end{proof}

\begin{corollary}\label{cor:completenessPBPtapes}
  For all $\s,\t \in \CatT{\Diag{\SigPB}}$, if $\dsem{\s}=\dsem{\t}$ then $\s \eqsynbis \t$.
\end{corollary}
\begin{proof}
By Theorem~\ref{thm:completenessPBPtapes}, $I$ is faithful. The following derivation concludes the proof.
  \begin{align}
    \dsem{\s}=\dsem{\t} &\Longleftrightarrow (\CatT{Q_{\ThPB}}; Q_\eqsyn ; I) (\s)= (\CatT{Q_{\ThPB}}; Q_\eqsyn ; I) (\t) \tag{$\dsem{-}= \CatT{Q_{\ThPB}}; Q_\eqsyn ; I$}\\
    &\implies   (\CatT{Q_{\ThPB}}; Q_\eqsyn ) (\s)= (\CatT{Q_{\ThPB}}; Q_\eqsyn )(\t) \tag{$I$ is faithful} \\
    &\Longleftrightarrow \CatT{Q_{\ThPB}}  (\s) \eqsyn \CatT{Q_{\ThPB}}(\t) \tag{Def. $Q_\eqsyn$}\\
    &\Longleftrightarrow \s \dot{\sim} \t \tag{Proposition \ref{prop:twoequivalence}}
  \end{align}
\end{proof}

\begin{corollary}\label{cor:finale}
  For all $c,d \in \Diag{\SigPRB}$, $\osem{c}=\osem{d}$ iff $\encoding{c} \eqsynbis \encoding{d}$.
\end{corollary}
\begin{proof}
  \begin{align}
    \osem{c}=\osem{d} &\Longleftrightarrow  \dsem{\encoding{c}}=\dsem{\encoding{d}} \tag{Proposition \ref{prop:encoding}}\\
    &\Longleftrightarrow \encoding{c} \eqsynbis \encoding{d}. \tag{Proposition \ref{prop:soundnesstapes}, Corollary \ref{cor:completenessPBPtapes}}
  \end{align}
\end{proof}

%% file: sections/finpac.tex
\section{Relating convex biproduct categories to finitely partially additive categories}\label{app:eff}
Probabilistic tape diagrams are close in the spirit to effectuses \cite{introductioneffectus} that are founded on finitely partially additive categories \cite{chophd,manes}.  In this section, we discuss how these structures are related to convex biproduct categories.
%
We begin by recalling the notion of partial commutative monoid and finitely partially additive categories with zero object.

A \emph{partial commutative monoid}, pcm for short, consists of a set $X$, together with a partial function $\circledvee\colon X \times X \to X$ and an element $\zeropcm \in X$ satisfying the following laws, where we write $x\bot y$ to mean that $x\circledvee y$ is defined:
\begin{itemize}
\item[A)] if $x\bot y$ and $x\circledvee y \bot z$, then $y \bot z$, $x \bot y\circledvee z$ and $(x\circledvee y) \circledvee z = x\circledvee (y \circledvee z)$;
\item[C)] if $x\bot y$, then $y\bot x$ and $x \circledvee y = y \circledvee x$;
\item[U)] $\zeropcm\bot x$ and $\zeropcm\circledvee x =x$. 
\end{itemize}



Before recalling the  notion of finitely partially additive category with zero object $0$ from \cite[Definition 3.1.2]{chophd}, we need the following definition.
\begin{definition}
Let  $\Cat{C}$ be a category with finite coproducts in which $0$ is also final. Let\, $\bangp{X}\colon X \to 0$ be the unique morphism to the final object $0$ and $\pi_1\colon X\oplus X \to X$ and $\pi_2\colon X\oplus X \to X$ be the projections given respectively by $(\id{X}\oplus\, \bangp{X});\runit{X}$ and $(\,\bangp{X}\oplus \id{X});\lunit{X}$.
Two parallel arrows $f_1,f_2 \colon A \to X$ in $\Cat{C}$ are \emph{compatible} if there exists some $h\colon A\to X \oplus X$ such that $h;\pi_i=f_i$ for $i=1,2$.
\end{definition}
\begin{definition}\label{dfn:finpac}
 A \emph{finitely partially additive category} with zero object is a  pcm-enriched category $\Cat{C}$ with finite coproducts such that $0$ is also final and the following two conditions hold:
    \begin{itemize}
        \item if $f_1,f_2\colon A \to X$ are compatible, then they are summable in the pcm $\Cat{C}[A,X]$;
        \item if $f_1,f_2\colon A \to X$ are summable, then $f_1;\iota_1\colon A\to X \oplus X$ and $f_2;\iota_2\colon A \to X \oplus X$ are summable too.
    \end{itemize}
\end{definition}

Our starting observation is that in any pca $(X,+_p, \star)$ 
one can define a partial binary operation $\circledvee$  with unit $\zeropcm\defeq \star$ as follows: 
\begin{equation}\label{eq:defbot}
x\bot y \text{ iff } \exists x',y'\in X, \, p,q\in[0,1] ,\; x=p\cdot x', \; y=q\cdot y', \; p+q\leq 1
\end{equation}
and in this case 
\begin{equation}\label{eq:defcircledvee}
x \circledvee y \defeq
\begin{cases}
x' +_p \bigl(y' +_{\frac{q}{1-p}} \star\bigr), & \text{if } p \neq 1, \\
x & \text{if } p = 1.
\end{cases}
\end{equation}

We now investigate under which conditions the structure described above yields a well-defined pcm. As the following example shows, not every pca gives rise to a well-defined pcm, even when restricted to cancellative pcas.
\begin{example}
   Consider the interval $[-1,1]$ equipped with the following structure of pca: for $x,y\in [-1,1]$ and $p\in [0,1]$, let $x+_p y \defeq p x + (1-p) y$ and $\star \defeq 0$. It is straightforward to verify that this defines a cancellative pca. Now consider the partial operation $\circledvee$ defined as above. We show that the associativity of $\circledvee$ may fail. Indeed, let $x=-0.1$, $y=0.5$ and $z=0.6$. Then $x\bot y$ via $x'=-1$, $y'=1$, $p=0.1$ and $q=0.5$, and $x\circledvee y = -0.1 + 0.5 = 0.4$. Moreover, $(x\circledvee y)\bot z$ via $u=1$, $z'=1$, $r=0.4$ and $s=0.6$, and $(x\circledvee y)\circledvee z = 0.4 + 0.6 = 1$. However, suppose that $y\bot z$, so that $y = r\cdot y''$ and $z = s\cdot z''$ for some $y'',z''\in [-1,1]$ and $r,s\in [0,1]$ with $r+s\le 1$. Since $y=0.5$ and $z=0.6$, we must have $0.5 \le r$ and $0.6 \le s$, which implies $r+s > 1$, a contradiction.
\end{example}

In order to obtain a well-defined pcm, we restrict our attention to free pcas. Recall that the free pca over a set $S$ is given by $\Dis(S)$ as defined in Section \ref{ssec:pca}. 
As shown in the following lemma, freeness is crucial to prove that $\circledvee$ is associative.
\begin{lemma}\label{lemma:pca-pcm}
Any free pca gives rise to a pcm with $\circledvee, \zeropcm$ and $\bot$ defined as in \eqref{eq:defbot} and \eqref{eq:defcircledvee}.
\end{lemma}
\begin{proof}


We focus on the associativity of $\circledvee$, as the other axioms follow from general properties of free pcas. In particular, any free pca is cancellative, and this directly implies that $\circledvee$ is well defined. Commutativity and unitality axioms follow from symmetry and idempotency of $+_p$, respectively. More details can be found in Appendix~\ref{app:app:eff}.

Associativity is the main point where we use the freeness of the pca. To prove it, consider $x,y,z\in X$ such that $x\bot y$ via $x',y'$ and $p,q\in [0,1]$ and $(x\circledvee y) \bot z$ via $w,z'$ and $t,u\in [0,1]$, we have to prove that $ y\bot z$ and $x \bot (y\circledvee z)$ and that $(x\circledvee y) \circledvee z = x\circledvee (y \circledvee z)$. Hence, we have that the sum of the subdistributions $p\cdot x'$, $q\cdot y'$ and $u\cdot z'$ is a subdistribution $(p\cdot x' + q\cdot y' ) + u\cdot z'= t\cdot w + u\cdot z'$. In particular, we have that $p\cdot\sum_{s\in S} x'(s) + q\cdot \sum_{s\in S} y'(s) + u\cdot \sum_{s\in S} z'(s) \leq 1$. Then, we have $y\bot z$ via $y'',z'$ and $q',u'$ where 
\[ q'\defeq q\cdot \sum_{s\in S} y'(s) \quad u' \defeq u\cdot \sum_{s\in S} z'(s)\]
and
\[y''\defeq \frac{y'}{\sum_{s\in S} y'(s)} \quad z''\defeq \frac{z'}{\sum_{s\in S} z'(s)}\text.\]
Indeed, we have that $q'+u' \leq 1$ and that $y=q'\cdot y''$ and $z=u'\cdot z''$. Similarly, we have that $x \bot (y\circledvee z)$ via $x'',w'$ and $p'',t'$, where
\[ p''\defeq p\cdot \sum_{s\in S} x'(s) \quad t' \defeq q\cdot \sum_{s\in S} y'(s) + u\cdot \sum_{s\in S} z'(s)\]
and
\[x''\defeq \frac{x'}{\sum_{s\in S} x'(s)} \quad w'\defeq \frac{q'}{t'}\cdot y'' + {\frac{u'}{t'}}\cdot z''\text.\]
Indeed, we have that $p''+t' \leq 1$ and that $x=p''\cdot x''$ and that $y\circledvee z = t' \cdot w'$. Finally, the equality $(x\circledvee y) \circledvee z = x\circledvee (y \circledvee z)$ is obvious.
\end{proof}

The above result allows us to prove that any $\Cat{PCA}$-enriched category, additionally satisfying the hypothesis of freeness, is pcm-enriched.

\begin{proposition}\label{prop:pca-pcm enrichment}
Let $\Cat{C}$ be a $\Cat{PCA}$-enriched category. If $\Cat{C}[X,Y]$ is a free pca for all objects $X,Y$, then $\Cat{C}$ is pcm-enriched.
\end{proposition}
\begin{proof}
By Lemma \ref{lemma:pca-pcm}, any homset $\Cat{C}[X,Y]$ is a pcm. To conclude that  $\Cat{C}$ is pcm-enriched it is enough to prove  for all properly typed arrows
\begin{itemize}
\item if $g\bot h$, then $f;(g \circledvee h) = (f;g) \circledvee (f;h) $;
\item if $f\bot g$, then  $(f \circledvee g); h = (f;h) \circledvee (g;h)$;
\item $f;\zeropcm=\zeropcm;f$.
\end{itemize}
The last point is obvious by the pca-enrichment: see \eqref{eq:enr}. Hereafter, we illustrate the proof of $(f \circledvee g); h = (f;h) \circledvee (g;h)$.
Assume that $f\bot g$, then by \eqref{eq:defbot}, there exist arrows $f',g'$ and  $p,q\in [0,1]$ such that
\[f=p\cdot f' \quad g=q\cdot g' \qquad p+q\leq 1\text{.}\]
Thus, for $p\neq 1$, we have that
\begin{align*}
(f\circledvee g); h & = (f' +_{p} (g'+_{\frac{q}{1-p}} \star));h \tag{\eqref{eq:defcircledvee}} \\
&  = (f';h +_{p} (g';h +_{\frac{q}{1-p}} \star)) \tag{$\Cat{C}$ is $\Cat{PCA}$-enriched} \\
&= f;h \circledvee g;h
\end{align*}
where the last step holds by \eqref{eq:defcircledvee} and the fact that,  by Lemma \ref{lemma:initialproperties2}.\ref{lemma:p;}, $p\cdot (f';h) = (p\cdot f'); h = f;h$ and $q\cdot (g';h) = (q\cdot g'); h = g;h$. In case $p=1$ and $q=0$, we have that $f\circledvee g = f$ and $g=0\cdot g' = \star=\zeropcm$. Hence $(f\circledvee g); h = f;h =f;h \circledvee \zeropcm = f;h \circledvee g;h$.
\end{proof}

We can finally state our main result about the relationship between convex bicategories and finitely partially additive categories.

\begin{proposition}\label{prop:fpac}
Let $\Cat{C}$ be a  $\Cat{PCA}$-enriched category such that $\Cat{C}[X,Y]$ is a free pca for all objects $X,Y$. If $\Cat{C}$  is a finitely partially additive category with zero object (with the pcm-enrichment given by Proposition~\ref{prop:pca-pcm enrichment}), then it is a convex biproduct category.
\end{proposition}
\begin{proof}
Since $\Cat{C}$ is a finitely partially additive category with zero object, then  it has finite coproducts, zero object and the equality in \eqref{eq:delta} holds.

Now we show that $X_1\oplus X_2$, with the projections $\pi_1$ and $\pi_2$,
 is a convex product in $\Cat{C}$. Take two arrows $f\colon A\to X_1$ and $g\colon A \to X_2$ and $p_1,p_2\in (0,1)$ such that $p_1+p_2\leq 1$.
Consider now the arrows $(p_1\cdot f) ; \iota_1$ and $(p_2 \cdot g) ;\iota_2$ in the pcm  $\Cat{C}[A , X_1\oplus X_2]$. Since by Lemma \ref{lemma:initialproperties2}.\ref{lemma:p;}  $(p_1\cdot f) ; \iota_1 = p_1 \cdot(f;\iota_1)$ and 
$(p_2\cdot g) ; \iota_2 = p_2 \cdot(g;\iota_2)$ and since $p_1+p_2\leq 1$, then by definition of $\bot$ in \eqref{eq:defbot}, it holds that 
\[(p_1\cdot f) ; \iota_1 \; \bot \; (p_2 \cdot g ) ;\iota_2 \text{.}\]
We thus can take $\langle f,g\rangle_{p_1,p_2}$ as $((p_1\cdot f) ; \iota_1)  \circledvee ((p_2 \cdot g ) ;\iota_2)\colon A \to X_1\oplus X_2$ which, by  \cite[Proposition 3.1.9]{chophd}, is the unique arrow such that \[\langle f,g\rangle_{p_1,p_2} ; \pi_1 = p_1\cdot f \quad \text{ and } \quad \langle f,g\rangle_{p_1,p_2} ; \pi_2 = p_2\cdot g\text{.}\]
\end{proof}

\begin{remark}
The reader may now wonder whether any convex biproduct category is also a finitely partially additive category. This is not the case in general. 

The definition of finitely partially additive category requires that compatible arrows are summable. In a convex biproduct category, there might exist compatible arrows $f_1$ and $f_2$ that are not summable in the sense of (\ref{eq:defcircledvee});  namely, there are no arrows $f_1',f_2'$ such that $f_1=p_1\cdot f_1'$ and $f_2=p_2\cdot f_2'$, for some scalar $p_1+p_2 \leq 1$. For instance, in $\CatTapeC$, the projections $\pi_1,\pi_2\colon X\oplus X\to X$ are compatible through the identity morphism, but they are not summable.

The essential difference can be understood as follows: summability in convex biproduct categories is fixed a priori by the structure of pca, while in partially additive categories is left to the structural properties of arrows. For instance, while in convex biproduct categories $\oplus$ behaves as a convex product, in partially additive categories $\oplus$ acts as a \emph{partial} biproduct: the product property only holds for compatible arrows.   

\end{remark}

%% file: sections/conclusion.tex
\section{Conclusion}\label{sec:conclusion} 

In this work, we introduced the notions of convex products (Def.~\ref{def:convprod}) and convex biproduct categories (Def.~\ref{def:convbicat}). We studied their associated monoidal algebra (Prop.~\ref{lemma: copca objects in convbicat}) and established a near analogue of Fox's theorem (Prop.~\ref{prop:A}). We then presented the construction $\stmat{\Cat{C}}$ of stochastic matrices over a $\Cat{PCA}$-enriched category $\Cat{C}$, and showed that $\stmat{\Cat{C}}$ is the free convex biproduct category generated by $\Cat{C}$ (Thm.~\ref{thm:matfree}).

We also developed a syntactic construction: for any category $\Cat{C}$, the category $\CatT{\Cat{C}}$ is obtained by freely adding a monoidal structure equipped with natural monoids and co-pointed convex algebras. We proved that $\CatT{\Cat{C}}$ is the free convex biproduct category over $\Cat{C}$ (Thm.~\ref{thm:syntacticadjunction}). Combining this result with a known theorem (Thm.~\ref{thm:freeenriched}), we derived the isomorphism
$\CatT{\Cat{C}} \cong \stmat{\Cat{C}^+},$
(Cor.~\ref{cor:isotapematrices}), where arrows in $\Cat{C}^+$ are subdistributions over arrows in $\Cat{C}$.

This construction becomes particularly meaningful when $\Cat{C} = \DiagS$, a category of string diagrams. In this case, $\CatT{\DiagS}$ yields the category of probabilistic tape diagrams introduced in~\cite{bonchi2025tapediagramsmonoidalmonads}. The isomorphism
$
\CatT{\DiagS} \cong \stmat{\DiagS^+}
$
thus provides a concrete interpretation of tape diagrams as stochastic matrices whose entries are probability subdistributions over string diagrams.

We applied this framework to probabilistic Boolean circuits with explicit conditioning from~\cite{piedeleu2025boolean}. As illustrated in~\cite[Ex.~30]{bonchi2025tapediagramsmonoidalmonads}, such circuits can be encoded as probabilistic tape diagrams over \emph{partial} Boolean circuits. To obtain a complete axiomatisation of these tapes, we first axiomatised partial Boolean circuits (Thm.~\ref{thm:completenesspartialcircuits}) and then introduced three additional axioms for tapes (Figure~\ref{ax:BooleanTAPES}). Thanks to the isomorphism
$\CatT{\DiagS} \cong \stmat{\DiagS^+}$, the resulting completeness theorem (Thm.~\ref{thm:completenessPBPtapes}) admits a conceptually simple and structurally transparent proof.

\paragraph{Related and future work.} The construction of stochastic matrices over $\Cat{PCA}$-enriched categories is conceptually akin to the classical construction of matrices over categories enriched in commutative monoids~\cite{mac_lane_categories_1978,coecke2017two}. However, while the latter is straightforward enough to be left as an exercise in~\cite{mac_lane_categories_1978}, the former involves significantly more work. In particular, finite biproduct categories enjoy a bijective correspondence between arrows \( f \colon A \oplus B \to C \oplus D \) and quadruples of morphisms \( (f_{A,C}, f_{B,C}, f_{A,D}, f_{B,D}) \), where each \( f_{X,Y} \colon X \to Y \). This correspondence breaks down in convex biproduct categories.

As illustrated in Section~\ref{sec:probboolcircuits}, our probabilistic tape diagrams are closely related to finitely partially additive categories~\cite{manes}, which are the basis of \emph{effectuses}~\cite{introductioneffectus,chophd}. The main difference is that in an arbitrary finitely partially additive category, one cannot define $\diagp{X}\colon X \to X\oplus X$ unless a notion of scalar is explicitly required. Scalars appear in {effectuses} as arrows of type \( 1 \to 1 \), where \( 1 \) is a distinguished object. As future work, we plan to explore the relationship between our framework and effectuses, which may require considering more structure for tapes, such as their \emph{copy-discard} variants introduced in~\cite{bonchi2025tapediagramsmonoidalmonads}.



In~\cite{villoria2024enriching}, a general framework is introduced to enrich string diagrams over arbitrary algebraic theories. This approach fundamentally differs from the tape-based construction in~\cite{bonchi2025tapediagramsmonoidalmonads}, in that it does not account for the structure induced by the coproduct \( \oplus \). As a consequence, the correspondence with matrix-based semantics 
does not arise in~\cite{villoria2024enriching}.

We plan to extend probabilistic tape diagrams with traces for \( \oplus \), following the approach developed in~\cite{bonchi2024diagrammatic}. That work showed that a specific property of monoidal traces--called \emph{uniformity}--corresponds to reasoning by invariants in Hoare-style program logics~\cite{hoare1969axiomatic}. Interestingly, as demonstrated in~\cite{jacobs2010coalgebraic}, uniformity holds for traces over \( \oplus \) in \( \KlD \). Preliminary results suggest that, in the probabilistic setting, invariants correspond to sub-martingales, thus offering a promising link between diagrammatic semantics and probabilistic program verification.

Moreover, probabilistic tape diagrams extended with traces are expressive enough to encode probabilistic regular expressions from~\cite{DBLP:conf/lics/RozowskiS24}, which provide a complete axiomatisation for the regular behaviours of generative probabilistic systems--that is, coalgebras in \( \KlD \)~\cite{hasuo2007generic}. We expect that tapes may again offer a more principled and modular axiomatisation, akin to Kleene algebra in~\cite{Kozen94acompleteness}. Interestingly, the completeness proof in~\cite{Kozen94acompleteness} relies heavily on matrix representations over languages; as we hint in Example~\ref{ex:matrices}, our tape formalism naturally supports stochastic matrices over probabilistic languages.

%% file: appendices/coherences.tex
\section{Additional Figures and Tables}\label{app:coherence axioms}

\begin{figure}[H]
    \begin{equation}\label{ax:Mon1}\tag{$\codiag{}$-as}
        \input{tikz-cd/monoid_assoc.tikz}
    \end{equation}
    \begin{minipage}[t]{0.50\textwidth}
        \begin{equation}\label{ax:Mon2}\tag{$\codiag{}$-un}
            \input{tikz-cd/monoid_unit.tikz}
        \end{equation}
    \end{minipage}
    \hfill
    \begin{minipage}[t]{0.46\textwidth}
        \begin{equation}\tag{$\codiag{}$-sym}
            \input{tikz-cd/monoid_comm.tikz}
        \end{equation}
    \end{minipage}
    \caption{Commutative monoid axioms}
    \label{fig:monoidax}
\end{figure}

 \begin{figure}[H]		
    \begin{equation}\label{eq:coherence codiag}\tag{$\codiag{}$-coh}
        \begin{tikzcd}[column sep=4.5em,baseline=(current  bounding  box.center)]
        (X \piu Y) \piu (X \piu Y) \ar[dd,"\assoc X Y {X \piu Y}"'] \ar[r,"\codiag{X \piu Y}"] & X \piu Y \\
        & (X \piu X) \piu (Y \piu Y) \ar[u,"\codiag X \piu \codiag Y"']\\
            X \piu (Y \piu (X \piu Y)) \ar[d,"\id X \piu \Iassoc Y X Y"'] & X \piu (X \piu (Y \piu Y)) \ar[u,"\Iassoc X X {Y \piu Y}"']  \\
        X \piu ((Y \piu X) \piu Y) \ar[r,"\id X \piu ( \symm{Y}{X}^{\piu} \piu \id Y)"] & X \piu (( X \piu Y) \piu Y) \ar[u,"\id X \piu \assoc X Y Y"']
        \end{tikzcd}
    \end{equation}
    \\
\begin{minipage}[b]{0.33\textwidth}
	\begin{equation}\tag{\,$\cobang{}$-coh}
        \begin{tikzcd}[baseline=(current  bounding  box.center)]
        \zero \ar[r,"\cobang{X \piu Y}"] \ar[d,"\Ilunit \zero"']  & X \piu Y \\
        \zero \piu \zero \ar[ur,"\cobang X \piu \cobang Y"']  
        \end{tikzcd}
	\end{equation}
\end{minipage}
\hfill
\begin{minipage}[b]{0.30\textwidth}
	\begin{equation}\tag{\,$\codiag{0}$-coh}
        \begin{tikzcd}
        \zero \piu \zero \ar[r,shift left=2,"\codiag \zero"] \ar[r,shift right=2,"\lunit \zero"'] &  \zero
        \end{tikzcd}
	\end{equation}
\end{minipage}
\hfill
\begin{minipage}[b]{0.26\textwidth}
	\begin{equation}\label{eq:cobang I = id I}\tag{\,$\cobang{0}$-coh}
        \begin{tikzcd}
        \zero \ar[r,shift left=2,"\cobang \zero"] \ar[r,shift right=2,"\id \zero"'] & \zero
        \end{tikzcd}
	\end{equation}
\end{minipage}
\begin{minipage}[b]{0.45\textwidth}
		\begin{equation}\label{eq:nat monoid1}\tag{\,$\codiag{}$-nat}
\begin{tikzcd}
	{X\piu X} & X \\
	{Y\piu Y} & Y
	\arrow["{\codiag{X}}", from=1-1, to=1-2]
	\arrow["{f\piu f}"', from=1-1, to=2-1]
	\arrow["f", from=1-2, to=2-2]
	\arrow["{\codiag{Y}}"', from=2-1, to=2-2]
\end{tikzcd}
		\end{equation}
	\end{minipage}
	\hfill
	\begin{minipage}[b]{0.45\textwidth}
		\begin{equation}\label{eq:nat monoid2}\tag{\,$\cobang{}$-nat}
\begin{tikzcd}
	\zero & Y & X
	\arrow["{\cobang{Y}}"', from=1-1, to=1-2]
	\arrow["{\cobang{X}}", bend left =30pt, from=1-1, to=1-3]
	\arrow["f"', from=1-2, to=1-3]
\end{tikzcd}
		\end{equation}
	\end{minipage}
\caption{Coherence and naturality axioms for commutative monoids}
\label{fig:fccoherence}
\end{figure}

\begin{figure}[H]
    \begin{equation}\label{ax:PCA1}\tag{$\diagp{}$-as}
\begin{tikzcd}[column sep=large]
	X && {X\piu X} \\
	{X\piu X} & {X\piu (X\piu X)} & {(X\piu X)\piu X}
	\arrow["{\diagp{}}", from=1-1, to=1-3]
	\arrow["{\diagptilde{}}"', from=1-1, to=2-1]
	\arrow["{\diagq{}\piu \id{X}}", from=1-3, to=2-3]
	\arrow["{\id{X}\piu\, \diagqtilde{}}"', from=2-1, to=2-2]
	\arrow["{\Iassoc{X}{X}{X}}"', from=2-2, to=2-3]
\end{tikzcd} \qquad \tilde{p} = pq, \quad \tilde{q} = \frac{p(1-q)}{1-pq}
    \end{equation}
    \begin{minipage}[t]{0.50\textwidth}
        \begin{equation}\label{ax:PCA2}\tag{$\diagp{}$-idem}
\begin{tikzcd}
	X & {X\piu X} \\
	& X
	\arrow["{\diagp{}}", from=1-1, to=1-2]
	\arrow["{\id{X}}"', from=1-1, to=2-2]
	\arrow["{\codiag{X}}", from=1-2, to=2-2]
\end{tikzcd}
        \end{equation}
    \end{minipage}
    \hfill
    \begin{minipage}[t]{0.46\textwidth}
        \begin{equation}\label{ax:PCA3}\tag{$\diagp{}$-sym}
\begin{tikzcd}
	& X \\
	{X\piu X} && {X\piu X}
	\arrow["{\diagp{}}"', from=1-2, to=2-1]
	\arrow["{\diagpbar{}}", from=1-2, to=2-3]
	\arrow["{\symm{X}{X}^{\piu}}"', from=2-1, to=2-3]
\end{tikzcd}
        \end{equation}
    \end{minipage}
    \caption{co-pca axioms}
    \label{fig:co-pca axioms}
\end{figure}

\begin{figure}[H]		
	\begin{equation}\label{eq:coherence conv sum}\tag{$\diagp{}$-coh}
		\begin{tikzcd}[column sep=4.5em,baseline=(current  bounding  box.center)]
				{X\piu Y} && {(X\piu Y)\piu (X\piu Y)} \\
				{(X\piu X)\piu (Y\piu Y)} && {X\piu (Y\piu (X\piu Y))} \\
				{X \piu (X \piu (Y\piu Y))} \\
				{X \piu ((X \piu Y)\piu Y)} && {X \piu ((Y \piu X)\piu Y)}
				\arrow["{\diagp{}}", from=1-1, to=1-3]
				\arrow["{\diagp{}\piu \diagp{}}"', from=1-1, to=2-1]
				\arrow["{\assoc{X}{X}{Y\piu Y}}"', from=2-1, to=3-1]
				\arrow["{\Iassoc{X}{Y}{X\piu Y}}"', from=2-3, to=1-3]
				\arrow["{\id{X}\piu \Iassoc{X}{Y}{Y}}"', from=3-1, to=4-1]
				\arrow["{\id X \piu ( \symm{X}{Y}^{\piu} \piu \id Y)}"', from=4-1, to=4-3]
				\arrow["{\id{X}\piu \assoc{Y}{X}{Y}}"', from=4-3, to=2-3]
		\end{tikzcd}
	\end{equation}
	\\
	\begin{minipage}[b]{0.33\textwidth}
		\begin{equation}\label{eq:coherence cobang}\tag{$\,\bangp{}$-coh}
			\begin{tikzcd}[baseline=(current  bounding  box.center)]
					{X\piu Y } & \zero \\
					& {\zero \piu \zero}
					\arrow["{\bangp{X\piu Y}}", from=1-1, to=1-2]
					\arrow["{\bangp{X}\piu \bangp{Y}}"', from=1-1, to=2-2]
					\arrow["{\lunit{0}}"', from=2-2, to=1-2]
			\end{tikzcd}
		\end{equation}
	\end{minipage}
	\hfill
	\begin{minipage}[b]{0.26\textwidth}
		\begin{equation}\label{ax:coh3}\tag{$\,\bangp{0}$-coh}
		\begin{tikzcd}
			\zero & \zero
			\arrow["{\bangp{\zero}}", shift left, from=1-1, to=1-2]
			\arrow["{\id{\zero}}"', shift right, from=1-1, to=1-2]
		\end{tikzcd}
		\end{equation}
	\end{minipage}
	\hfill
	\begin{minipage}[b]{0.30\textwidth}
		\begin{equation}\label{eq:star I = id I}\tag{$\diagp{0}$-coh}
			\begin{tikzcd}
				\zero & {\zero\piu \zero}
				\arrow["{\diagp{}}", shift left, from=1-1, to=1-2]
				\arrow["{\Ilunit{\zero}}"', shift right, from=1-1, to=1-2]
			\end{tikzcd}
		\end{equation}
	\end{minipage}
	\begin{minipage}[b]{0.45\textwidth}
		\begin{equation}\label{eq:nat copca1}\tag{$\diagp{}$-nat}
\begin{tikzcd}
	X & {X\piu X} \\
	Y & {Y\piu Y}
	\arrow["{\diagp{X}}", from=1-1, to=1-2]
	\arrow["f"', from=1-1, to=2-1]
	\arrow["{f\piu f}", from=1-2, to=2-2]
	\arrow["{\diagp{Y}}"', from=2-1, to=2-2]
\end{tikzcd}
		\end{equation}
	\end{minipage}
	\hfill
	\begin{minipage}[b]{0.45\textwidth}
		\begin{equation}\label{eq:nat copca2}\tag{$\,\bangp{}$-nat}
\begin{tikzcd}
	X & Y & \zero
	\arrow["f", from=1-1, to=1-2]
	\arrow["{\bangp{X}}"', curve={height=30pt}, from=1-1, to=1-3]
	\arrow["{\bangp{Y}}", from=1-2, to=1-3]
\end{tikzcd}
		\end{equation}
	\end{minipage}
	\caption{Coherence and naturality axioms for co-pca objects}
	\label{fig:copcacoherence}
\end{figure}

%% file: appendices/appendicerefusiprimaversione.tex
\section{Substochastic Markov kernels on standard Borel spaces}\label{app:examples}

Consider the category of \textit{standard Borel spaces} and measurable functions denoted with $\mathbf{BorelMeas}$.
Given a {standard Borel space} $(X,\Sigma_X)$, denote with $\mathcal{G}_{\le}(X)$ the set of subprobability measures on $X$, i.e. measurable functions $\mu:(X,\Sigma_X)\to ([0,1],\mathcal{B}([0,1]))$ such that $\mu(X)\le 1$, where $\mathcal{B}([0,1])$ is the Borel $\sigma$-algebra on $[0,1]$.  This assignment extends to a functor $\mathcal{G}_{\le}:\mathbf{BorelMeas}\to \mathbf{BorelMeas}$ which corresponds to the subdistribution version of the \textit{Giry monad}~\cite{giry2006categorical}. The functor $\mathcal{G}_{\le}$ is a symmetric monoidal monad and its Kleisli category, $\Kl({\mathcal{G}_{\le}})$, is the category of standard Borel spaces and substochastic Markov kernels. A substochastic Markov kernel from $(X,\Sigma_X)$ to $(Y,\Sigma_Y)$ is a function $f\colon \Sigma_Y\times X\to [0,1]$ such that for all $x\in X$, $f(-|x)\colon \Sigma_Y\to [0,1]$ is a subprobability measure on $Y$ and for all $U\in \Sigma_Y$, $f(U|-)\colon X\to [0,1]$ is a measurable function.
The identity arrow on $(X,\Sigma_X)$ is the Markov kernel $\delta_X$ such that for all $x\in X$ and $U\in \Sigma_X$, $\delta_X(U|x)=1$ if $x\in U$ and $0$ otherwise. The composition of arrows in $\Kl({\mathcal{G}_{\le}})$ is given by the Chapman-Kolmogorov equation: for all  substochastic Markov kernels $f\colon (X,\Sigma_X)\to (Y,\Sigma_Y)$ and $g\colon (Y,\Sigma_Y)\to (Z,\Sigma_Z)$, $f;g\colon (X,\Sigma_X)\to (Z,\Sigma_Z)$ is defined as the Markov kernel such that for all $x\in X$ and $U\in \Sigma_Z$,
\begin{equation*}
(f;g)(U|x)\coloneqq \int_Y g(U|y)\, f(dy|x)
\end{equation*}
   
\begin{proposition}
    \label{prop: SubStoch is convex biproduct}
    $\Kl({\mathcal{G}_{\le}})$ is a convex biproduct category.
\end{proposition}
\begin{proof}
Coproducts in $\Kl({\mathcal{G}_{\le}})$ are given by the coproducts in $\mathbf{BorelMeas}$, i.e. the disjoint union of standard Borel spaces. The initial and terminal object is the empty standard Borel space. The $\mathbf{PCA}$-enrichment is given by the pointwise convex combination of Markov kernels: for all Markov kernels $f,g\colon (X,\Sigma_X)\to (Y,\Sigma_Y)$ and $p\in [0,1]$, $f +_p g\colon (X,\Sigma_X)\to (Y,\Sigma_Y)$ is defined  for all $x\in X$ and $U\in \Sigma_Y$ as  $f+_p g(U|x)\coloneqq p\cdot f(U|x)+(1-p)\cdot g(U|x)$. The zero arrow $\star_{(X,\Sigma_X),(Y,\Sigma_Y)}\colon (X,\Sigma_X)\to (Y,\Sigma_Y)$ is the Markov kernel such that for all $x\in X$ and $U\in \Sigma_Y$, $\star_{(X,\Sigma_X),(Y,\Sigma_Y)}(U|x)=0$. The axioms of $\mathbf{PCA}$ are satisfied since they are so in $[0,1]$. A simple computation shows that $+_p$ and $\star$ are compatible with composition. For instance, 
 we obtain that $(f+_p g);h=f;h +_p g;h$ for all Markov kernels $f,g\colon (X,\Sigma_X)\to (Y,\Sigma_Y)$, $h\colon (Y,\Sigma_Y)\to (Z,\Sigma_Z)$ and $p\in [0,1]$ in the following way
\begin{align}
\big((f+_{p}g);h\big)(C\mid x)
&= \int_{Y} h(C\mid y)\,(f+_{p}g)(dy\mid x)
\tag*{} \\
&= \int_{Y} h(C\mid y)\,\big(p\,f(dy\mid x)+(1-p)\,g(dy\mid x)\big)
\tag*{} \\
&= p\int_{Y} h(C\mid y)\,f(dy\mid x) \;+\; (1-p)\int_{Y} h(C\mid y)\,g(dy\mid x)
\tag{linearity of the integral}\\
&= p\,(f;h)(C\mid x) \;+\; (1-p)\,(g;h)(C\mid x)
\tag*{} \\
&= \big((f;h)+_{p}(g;h)\big)(C\mid x).
\tag*{}
\end{align}

Now we prove that $(X_1,\Sigma_{X_1})\overset{\pi_1}{\leftarrow}(X_1,\Sigma_{X_1})\oplus (X_2,\Sigma_{X_2})\overset{\pi_2}{\rightarrow} (X_2,\Sigma_{X_2})$ is a binary convex product in $\Kl({\mathcal{G}_{\le}})$. The Markov kernels $\pi_1$ and $\pi_2$ are defined as follows: for all $z\in X_1\oplus X_2$ and $U\in \Sigma_{X_1}$, $\pi_1(U|z)=1$ if $z\in \iota_1(U)$ and $0$ otherwise; for all $z\in X\oplus Y$ and $V\in \Sigma_{X_2}$, $\pi_2(V|z)=1$ if $z\in \iota_2(V)$ and $0$ otherwise. Given $f\colon (A,\Sigma_A)\to (X_1,\Sigma_{X_1})$, $g\colon (A,\Sigma_A)\to (X_2,\Sigma_{X_2})$ and $p_1,p_2\in [0,1]$ such that $p_1+p_2\leq 1$, we prove that there is a unique Markov kernel $h\colon (A,\Sigma_A)\to (X_1,\Sigma_{X_1})\oplus (X_2,\Sigma_{X_2})$ such that $h;\pi_1 = p_1\cdot f$ and $h;\pi_2 = p_2\cdot g$. Existence is provided by the Markov kernel $h$ defined as follows: for all $a\in A$, $U\in \Sigma_{X_1}$ and $V\in \Sigma_{X_2}$,
\[ h(\iota_1(U)|a)\coloneqq p_1\cdot f(U|a) \quad h(\iota_2(V)|a)\coloneqq p_2\cdot g(V|a) \]
where for a generic element of $\Sigma_{X_1\oplus X_2}$ $h$ is extended by countable additivity. Uniqueness is proved as follows: suppose that there is also another $h'$ with the same property. Then, for all $a\in A$, $U\in \Sigma_{X_1}$ and $V\in \Sigma_{X_2}$, 
\begin{align}
    h'(\iota_1(U)\mid a)
    &= \int_{X_1 \oplus X_2} \mathbf{1}_{\iota_1(U)}(z)\, h'(dz\mid a)
    \tag{def. of integration w.r.t.\ $h'(\cdot\mid a)$}\\
    &= \int_{X_1 \oplus X_2} \pi_1(U\mid z)\, h'(dz\mid a)
    \tag{def. of indicator function}\\
    &= (h';\pi_1)(U\mid a)
    \tag{kernel composition}\\
    &= (p_1\cdot f)(U\mid a)
    \tag{hp. $h';\pi_1 = p_1\cdot f$}.
\end{align}

Similarly, $h'(\iota_2(V)|a)=(p_2\cdot g)(V|a)$. Since every element of $\Sigma_{X_1\oplus X_2}$ is obtained by countable unions of elements of the form $\iota_1(U)$ and $\iota_2(V)$, $h'$ is completely determined by the above equalities, hence $h'=h$.
\end{proof}

%% file: appendices/appconvbiprodcat.tex
\section{Appendix to Section \ref{sec:pca}}

\begin{proof}[Proof of Lemma~\ref{lemma:initialproperties2}]
    \begin{enumerate}
        \item  Follows easily by the enrichment. For instance, the first equality is proved as follows.
\begin{align*} 
(p\cdot f) ; g & = (f+_p\star_{X,Y});g \tag{def of $p\cdot(-)$} \\
&= (f;g) +_p(\star_{X,Y};g) \tag{\ref{eq:enr}}\\
&= (f;g) +_p(\star_{X,Z}) \tag{\ref{eq:enr}}\\
&= p\cdot (f ; g) \tag{def.\ of $p\cdot(-)$} \\
\end{align*}
\item Follows easily by the laws of PCAs.
\begin{align*}
p\cdot (q \cdot f) &= (f+_q \star )+_p \star \tag{def.\ of $p\cdot(-)$}\\
&= f+_{pq} (\star +_{\frac{p\cdot(1-q)}{1-pq}} \star) \tag{\ref{eq:pca}} \\
&= f+_{pq} \star \tag{\ref{eq:pca}} \\
&= pq \cdot f \tag{def.\ of $p\cdot(-)$}
\end{align*}
\item Follows by induction. The case $n=0$ follows from the laws of PCA: $q\cdot \star= \star +_q \star= \star$. For $n+1$ consider 
\begin{align}
    q\cdot\sum_{i=1}^{n+1}p_i\cdot f_i &= q\cdot ( f_{1} +_{p_{1}}\sum_{j=1}^{n}q_j\cdot f_j) \tag{Def. of $\sum_{i=1}^{n+1}{p_i}\cdot f_i$ in (\ref{eq:def somma n pca})}\\
    &= q\cdot f_{1} +_{p_{1}} \sum_{j=1}^{n}qq_j\cdot f_j \tag{PCA laws + Ind. Hp.}\\
    &= \sum_{i=1}^{n+1}qp_i\cdot f_i \tag{Def. of $\sum_{i=1}^{n+1}{p_i}\cdot f_i$ in (\ref{eq:def somma n pca})}
\end{align}
        \end{enumerate}
\end{proof}

\section{Appendix to Section \ref{sec:cbproducts}}\label{app:sec:cbproducts}

\subsection{Proofs of Section \ref{ssec:convexproducts}}
\begin{proof}[Proof of Lemma~\ref{lemma:naryconvexproduct}]
By induction on $n$, we show that there exists an object $\bigotimes_{i=1}^n X_i$ which is a convex product of $X_1, \dots, X_n$. For $n=0$, take the final object of $\Cat{C}$ which exists by hypothesis. 
For $n+1$ assume by induction hypothesis that $\bigotimes_{i=1}^n X_i$ is a convex product of $X_1, \dots, X_n$ with projection $\pi_i\colon \bigotimes_{i=1}^n X_i \to X_i$. Since $\Cat{C}$ has binary convex products we can take the convex product $(\bigotimes_{i=1}^n X_i) \otimes X_{n+1}$ with projections $\pi'_1\colon (\bigotimes_{i=1}^n X_i) \otimes X_{n+1} \to \bigotimes_{i=1}^n X_i$ and $\pi'_2\colon (\bigotimes_{i=1}^n X_i) \otimes X_{n+1} \to X_{n+1}$. We claim that $(\bigotimes_{i=1}^n X_i) \otimes X_{n+1}$ with projections $\pi''_i\colon (\bigotimes_{i=1}^n X_i) \otimes X_{n+1} \to X_i$ defined for all $i= 1,\dots, n+1$ as 
\[\pi''_i = \begin{cases} \pi_1'; \pi_i & i\in 1,\dots, n \\ \pi_2' & i=n+1 \end{cases}\]
is a convex product of $X_1, \dots, X_{n+1}$. 
To check the universal property, let $\vec{p}=p_1,\dots, p_{n+1}$ such that $\sum_{i=1}^{n+1}p_i\leq 1$ and $f_i\colon A \to X_i$. Take $\vec{q}=q_1,\dots, q_n$ with $q_i=\frac{p_i}{1-p_{n+1}}$ if $p_{n+1}\not= 1$ and $q_i=0$ otherwise. Since  $\bigotimes_{i=1}^n X_i$ is an $n$-ary convex product by induction hypothesis, there exists an arrow $\langle f_1, \dots ,f_n\rangle_{\vec{q}} \colon A \to \bigotimes_{i=1}^n X_i$ such that $q_i\cdot f_i = \langle f_1, \dots ,f_n\rangle_{\vec{q}} ; \pi_i$. For the $n+1$-ary, we construct the mediating morphism $h\colon A \to (\bigotimes_{i=1}^n X_i) \otimes X_{n+1}$ as $\langle \, \langle f_1, \dots ,f_n\rangle_{\vec{q}} \,,\, f_{n+1} \, \rangle_{1-p_{n+1},p_{n+1}}$.
 \end{proof}

\begin{proof}[Proof of Proposition \ref{prop:productvsconvex}]
Suppose that $(Z,\pi_1,\pi_2)$ is a convex product of $X_1$ and $X_2$. For all arrows $f_1\colon A \to X_1$ and $f_2\colon A \to X_2$, take the pairing $\langle f_1,f_2\rangle\colon A \to Z$ to be $\langle f_1,f_2\rangle_{p_1,p_2}$ for some $p_1,p_2\in (0,1)$ with $p_1+p_2\leq 1$. Since $p_i\cdot f_i = f_i$ for $i\in\{1,2\}$, then $\langle f_1,f_2\rangle ; \pi_i =  \langle f_1,f_2\rangle_{p_1,p_2};\pi_i= p_i\cdot f_i =f_i$. To show uniqueness, let $h\colon A \to Z$ be such that $h;\pi_i =f_i$. Since $f_i = p_i\cdot f_i $, then $h;\pi_i =p_i\cdot f_i$. By the uniqueness provided by the convex product, $h$ must be $\langle f_1,f_2\rangle_{p_1,p_2}$, that is $\langle f_1,f_2\rangle$. Hence $(Z,\pi_1,\pi_2)$ is a product.

Vice versa, suppose that $(Z,\pi_1,\pi_2)$ is a product of $X_1$ and $X_2$. For all arrows $f_1\colon A \to X_1$ and $f_2\colon A \to X_2$ and $p_1,p_2\in[0,1]$, take $\langle f_1,f_2 \rangle_{p_1,p_2}\colon A \to Z$ to be the pairing $\langle f_1',f_2' \rangle$ where 
\[
f'_i = \begin{cases} 
f_i &\text{if }p_i \neq 0 \\
\star_{A,X_i} &\text{if }p_i=0\end{cases}
\]
Since $p_i\cdot f_i=f_i$ for all $p_i\in (0,1]$ and $p_i \cdot f_i =\star_{A,X_i}$ for $p_i=0$, it holds that for all $p_i\in[0,1]$, $p_i \cdot f_i =f_i'$.
Thus $\langle f_1,f_2 \rangle_{p_1,p_2}; \pi_i = \langle f_1',f_2' \rangle;\pi_i =f_i' =p_i \cdot f_i$. To prove uniqueness, let $h\colon A \to Z$ be such that $h;\pi_i =p_i\cdot f_i$. Since $p_i\cdot f_i = f_i'$, then $h;\pi_i = f_i '$. 
By the uniqueness provided by the product, $h$ must be $\langle f_1',f_2'\rangle$, that is $\langle f_1,f_2 \rangle_{p_1,p_2}$. Hence, $(Z,\pi_1,\pi_2)$ is a convex product.
\end{proof}

\subsection{Proof of Section \ref{ssec:convexbiproduct}}

\begin{proof}[Proof of Lemma~\ref{lemma:initialproperties}]
We prove below item by item.
\begin{enumerate}
\item Follows easily by enrichment and finality of $\zero$: 
\begin{align*}
\star_{X,Y} & = \star_{X,\zero} ; \cobang{Y} \tag{\ref{eq:enr}}\\
&= \bang{X};\cobang{Y} \tag{$\zero$ is final}
\end{align*}
\item Consider the arrows $f;\iota_1$ and $g;\iota_2$ from $X$ to $Y\oplus Y$ and their sum $f;\iota_1 +_p g;\iota_2$. The following computation shows that 
    \begin{align}
        (f;\iota_1 +_p g;\iota_2);\pi_1 & = f;\iota_1;\pi_1 +_p g;\iota_2;\pi_1 \tag{\ref{eq:enr}}\\
        & = f +_p \star_{X,Y} \tag{by \eqref{eq:delta}}\\
        & = p \cdot f \tag{def. of $p\cdot -$}
    \end{align} 
    and similarly $(f;\iota_1 +_p g;\iota_2);\pi_2 = (1-p)\cdot g$. By the universal property of convex products, we conclude that $f;\iota_1 +_p g;\iota_2 = \langle f,g \rangle_{p,1-p}$. Postcomposing both sides with $[\id{Y},\id{Y}]$ we obtain 
    \begin{align}
        \langle f,g \rangle_{p,1-p};[\id{Y},\id{Y}] & = (f;\iota_1 +_p g;\iota_2);[\id{Y},\id{Y}] \notag\\
        &= (f;\iota_1 ;[\id{Y},\id{Y}] +_p g;\iota_2 ;[\id{Y},\id{Y}]) \tag{\ref{eq:enr}}\\
        &=(f ;\id{Y} +_p g ; \id{Y}) \tag{properties of $[-,-]$}\\
        & = f +_p g.   \tag{category}
    \end{align}
\item We prove that $ \id{X_1}\oplus \cobang{X_2}=\runit{X_1}; \iota_1$ by universal property of $X_1\oplus \zero$ using the injections $\iota_1^{X_1\oplus\zero} \colon X_1 \to X_1\oplus \zero$ and 
$\iota_2^{X_1\oplus\zero} \colon \zero \to X_1\oplus \zero$. First observe that, by initiality of $\zero$,
 \[\iota_2^{X_1\oplus \zero}; (\id{X_1}\oplus \cobang{X_2}) = \iota_2^{X_1\oplus \zero}; \runit{X_1}; \iota_1 \text{.} \]
Moreover
\begin{align*}
\iota_1^{X_1\oplus \zero} ; (\id{X_1}\oplus \cobang{X_2}) &= \iota_1^{X_1\oplus \zero} ; [\id{X_1}; \iota_{1},\cobang{X_2};\iota_2 ] \tag{def. of $\oplus$}\\
&=\id{X_1};\iota_1 \tag{coproduct}\\
&=\iota_1^{X_1\piu\zero}; \runit{X_1} ; \iota_1. \tag{def. of $\runit{}$} 
\end{align*}
Thus $ \id{X_1}\oplus \cobang{X_2}=\runit{X_1}; \iota_1$.
\item Analogous to the point above.
\item Observe that, for $\iota_i\colon X_i \to X_1\oplus X_2$, 
\[\iota_1 ; (\id{X_1}\oplus \bang{X_2});\runit{X_1}=\id{X_1} \qquad \iota_2 ;  (\id{X_1}\oplus \bang{X_2});\runit{X_1} = \bang{X_2}; \cobang{X_1}= \star_{X_2,X_1}\]
namely $(\id{X_1}\oplus \bang{X_2});\runit{X_1} = [\id{X_1}, \star_{X_2,X_1}]$. Now, by \eqref{eq:delta}, one has that $[\id{X_1}, \star_{X_2,X_1}]=\pi_1$.
\item Analogous to the point above.
\item By the following derivation.
\begin{align*}
p \cdot f &= f+_p\star_{X,Y} \tag{def. of $p\cdot(-)$}\\
&=\, \diagp{X} ; (f \oplus \star_{X,Y}) ; \codiag{Y} \tag{Lemma~\ref{lemma:initialproperties}.\ref{eq: assioma aggiuntivo}} \\ 
&=\, \diagp{X} ; (f \oplus (\bang{X};\cobang{Y})) ; \codiag{Y} \tag{Lemma \ref{lemma:initialproperties}.\ref{lemma:starxy}} \\ 
&=\, \diagp{X} ; (f \oplus \bang{X}) ;(\id{Y} \oplus \cobang{Y}) ; \codiag{Y} \tag{Symmetric Monoidal Category} \\ 
&=\, \diagp{X} ; (f \oplus \bang{X}) ;\runit{Y}. \tag{\ref{ax:Mon2}} \\ 
 \end{align*}
 \item By the following derivation.
\begin{align*}
(1-p) \cdot f &= f+_{1-p}\star_{X,Y} \tag{def. of $p\cdot(-)$}\\
&= \star_{X,Y} +_{p} f \tag{\ref{eq:pca}}\\
&=\, \diagp{X} ; (\star_{X,Y} \oplus f) ; \codiag{Y} \tag{Lemma~\ref{lemma:initialproperties}.\ref{eq: assioma aggiuntivo}} \\ 
&=\, \diagp{X} ; ((\bang{X};\cobang{Y}) \oplus f) ; \codiag{Y} \tag{Lemma \ref{lemma:initialproperties}.\ref{lemma:starxy}} \\ 
&=\, \diagp{X} ; (\bang{X} \oplus f) ;(\id{Y} \oplus \cobang{Y}) ; \codiag{Y} \tag{Symmetric Monoidal Category} \\ 
&=\, \diagp{X} ; ( \bang{X} \oplus f) ;\lunit{Y}. \tag{\ref{ax:Mon2}} \\ 
 \end{align*} 
 
\end{enumerate}
\end{proof}

\begin{lemma}\label{lemma:naturality}
In a convex biproduct category, $\bang{}$ and $\diagp{}$ form natural transformations, namely 
$f;\bang{Y}=\bang{X}$  and $f;\diagp{Y} =\, \diagp{X};(f\oplus f)$ for all $f\colon X\to Y$.
\end{lemma}
\begin{proof}
The equation $f;\bang{Y}=\bang{X}$ for $f\colon X \to Y$ follows from the fact that $\zero$ is final. The equation
\[f;\diagp{Y} =\, \diagp{X};(f\oplus f)\]
follows from the universal property of convex products. Indeed, observe that
\begin{align*}
    \diagp{X};(f\oplus f);\pi_1 &=\, \diagp{X};(f\oplus f);(\id{Y}\piu \bang{Y});\runit{X} \tag{Lemma \ref{lemma:initialproperties}.\ref{lemma:pi1}}\\
    &=\, \diagp{X};(f\oplus (f;\bang{Y}));\runit{X}  \tag{Symmetric Monoidal Category}\\
    &=\, \diagp{X};(f\oplus \bang{X});\runit{X} \tag{Naturality of $\bang{X}$}\\
    &= p\cdot f \tag{Lemma \ref{lemma:initialproperties}.\ref{lemma:pf}}
\end{align*}
\begin{align*}
    \diagp{X};(f\oplus f);\pi_2 &=\, \diagp{X};(f\oplus f);(\bang{Y} \piu \id{Y} );\lunit{X} \tag{Lemma \ref{lemma:initialproperties}.\ref{lemma:pi2}}\\
    &=\, \diagp{X};( (f;\bang{Y})\oplus f);\lunit{X}  \tag{Symmetric Monoidal Category}\\
    &=\, \diagp{X};( \bang{X} \oplus f);\lunit{X} \tag{Naturality of $\bang{X}$}\\
    &= (1-p)\cdot f \tag{Lemma \ref{lemma:initialproperties}.\ref{lemma:1-pf}}
\end{align*}
and that
\begin{align*}
f;\diagp{Y};\pi_1 &= f ; \langle \id{X}, \id{X} \rangle_{p,(1-p)};\pi_1 \tag{def. of $\diagp{}$}\\
&= f; (p\cdot \id{X}) \tag{convex product} \\
&= p\cdot(f;\id{X}) \tag{Lemma \ref{lemma:initialproperties2}.\ref{lemma:p;}}\\
&=p\cdot f \tag{Category}
\end{align*}
\begin{align*}
f;\diagp{Y};\pi_2 &= f ; \langle \id{X}, \id{X} \rangle_{p,(1-p)};\pi_2 \tag{def. of $\diagp{}$}\\
&= f; (\,(1-p)\cdot \id{X}\,) \tag{convex product} \\
&= (1-p)\cdot(f;\id{X}) \tag{Lemma \ref{lemma:initialproperties2}.\ref{lemma:p;}}\\
&=(1-p)\cdot f \tag{Category}
\end{align*}
\end{proof}

\begin{lemma}\label{lemma:coherence}
In a convex biproduct category, the coherence axioms in Figure~\ref{fig:copcacoherence} hold.
\end{lemma}
\begin{proof}
(Coh2), (Coh3) and (Coh4) follow from the fact that $\zero$ is both initial and final. (Coh1) follows from the naturality and the universal property of coproducts. 

Let $s\defeq \assoc{X}{X}{Y\piu Y}; ( {\id{X}\piu \Iassoc{X}{Y}{Y}} ); ({\id X \piu ( \symm{X}{Y}^{\piu} \piu \id Y)} ); (\id{X}\piu \assoc{Y}{X}{Y}); \Iassoc{X}{Y}{X\oplus Y}$. Let
\[\iota_1^{(X\oplus X)\oplus(Y\oplus Y)}\colon X\piu X \to (X\piu X)\piu (Y\piu Y) \text{ and }\iota_2^{(X\oplus X)\oplus(Y\oplus Y)}\colon Y\piu Y \to (X\piu X)\piu (Y\piu Y)\] be the injections. By using the definitions of the structural isomorphisms induced by the coproducts, one can easily check that
\begin{equation}\label{eq:local}
\iota_1^{(X\oplus X)\oplus(Y\oplus Y)}; s= \iota_1 \oplus \iota_1 \qquad \text{and} \qquad \iota_2^{(X\oplus X)\oplus(Y\oplus Y)}; s= \iota_2 \oplus \iota_2\text{.}
\end{equation}
Then
\begin{align*}
\iota_1; (\diagp{X}\piu\, \diagp{Y}) ;s &= \iota_1; [\diagp{X}; \iota_1^{(X\oplus X)\oplus(Y\oplus Y)} \, , \, \diagp{Y}; \iota_2^{(X\oplus X)\oplus(Y\oplus Y)}] ;s \tag{def. of $\oplus$}\\
& =\, \diagp{X};\iota_1^{(X\oplus X)\oplus(Y\oplus Y)};s  \tag{coproduct}\\
&=\, \diagp{X}; (\iota_1 \oplus \iota_1) \tag{\ref{eq:local}} \\
&= \iota_1;\,\diagp{X\piu Y} \tag{Lemma \ref{lemma:naturality}}
\end{align*}
and 
\begin{align*}
\iota_2; (\diagp{X}\piu\, \diagp{Y}) ;s &= \iota_2; [\diagp{X}; \iota_1^{(X\oplus X)\oplus(Y\oplus Y)} \, , \, \diagp{Y}; \iota_2^{(X\oplus X)\oplus(Y\oplus Y)}] ;s \tag{def. of $\oplus$}\\
& =\, \diagp{Y};\iota_2^{(X\oplus X)\oplus(Y\oplus Y)};s  \tag{coproduct}\\
&=\, \diagp{Y}; (\iota_2 \oplus \iota_2) \tag{\ref{eq:local}} \\
&= \iota_2;\,\diagp{X\piu Y} \tag{Lemma \ref{lemma:naturality}}
\end{align*}
Hence, the universal property of coproducts implies that $\diagp{X\piu Y} = \diagp{X}\piu \diagp{Y};s$.
\end{proof}

\begin{lemma}\label{lemma:pdiag} \label{lemma:diagp}
In a convex biproduct category the following hold:
\begin{enumerate}
\item $p \cdot \diagq{X} = \langle \id{X},\id{X}\rangle_{pq, p(1-q)} $.
\item $\diagp{X}; (\id{X} \oplus q\cdot \id{X})=\langle \id{X},\id{X}\rangle_{p,(1-p)q}$ 
\end{enumerate}
\end{lemma}
\begin{proof}
We prove below item by item.
\begin{enumerate}
\item By the uniqueness of $\langle \id{X},\id{X}\rangle_{pq, p(1-q)}$ and the following two derivations.
\begin{align*}
(p \cdot \diagq{X})  ; \pi_1 &= p \cdot (\diagq{X};\pi_1) \tag{Lemma \ref{lemma:initialproperties}.4}\\
&=p\cdot (\langle \id{X},\id{X}\rangle_{q,1-q} ; \pi_1) \tag{def. of $\diagq{}$}\\
&= p\cdot (q \cdot \id{X}) \tag{convex product}\\
&= (p q) \cdot \id{X} \tag{Lemma~\ref{lemma:initialproperties2}.\ref{lemma:pq}}
\end{align*}
\begin{align*}
(p \cdot \diagq{X})  ; \pi_2 &= p \cdot (\diagq{X};\pi_2) \tag{Lemma \ref{lemma:initialproperties}.4}\\
&=p\cdot (\langle \id{X},\id{X}\rangle_{q,1-q} ; \pi_2) \tag{def. of $\diagq{}$}\\
&= p\cdot ((1-q) \cdot \id{X}) \tag{convex product}\\
&= (p (1-q)) \cdot \id{X} \tag{Lemma~\ref{lemma:initialproperties2}.\ref{lemma:pq}}
\end{align*}
\item By uniqueness of $\langle \id{X},\id{X}\rangle_{p,(1-p)q}$ and the following two derivations. 
\begin{align*}
\diagp{X}; (\id{X} \oplus q\cdot \id{X}) ; \pi_1 &=\, \diagp{X}; (\id{X} \oplus q\cdot \id{X}) ; (\id{X}\oplus \bang{X});\runit{X} \tag{Lemma \ref{lemma:initialproperties}.\ref{lemma:pi1}}\\
&=\, \diagp{X}; (\id{X} \oplus (q\cdot \id{X} ; \bang{X})); \runit{X} \tag{Symmetric Monoidal Category}\\
&=\, \diagp{X}; (\id{X} \oplus  \bang{X}); \runit{X} \tag{Naturality of $\bang{}$} \\
&= p\cdot \id{X} \tag{Lemma \ref{lemma:initialproperties}.\ref{lemma:pf}}
\end{align*}
\begin{align*}
\diagp{X}; (\id{X} \oplus q\cdot \id{X}) ; \pi_2 &=\, \diagp{X}; (\id{X} \oplus q\cdot \id{X}) ; (\bang{X}\oplus \id{X});\lunit{X} \tag{Lemma \ref{lemma:initialproperties}.\ref{lemma:pi2}}\\
&=\, \diagp{X}; ( \bang{X}\oplus( q\cdot \id{X}) ) ; \lunit{X} \tag{Symmetric Monoidal Category}\\
&= (1-p)\cdot (q \cdot \id{X}) \tag{Lemma \ref{lemma:initialproperties}.\ref{lemma:1-pf}}\\
&= ((1-p)q) \cdot \id{X}. \tag{Lemma \ref{lemma:initialproperties2}.\ref{lemma:pq}}
\end{align*}
\end{enumerate}
\end{proof}

\begin{proof}[Proof of Proposition \ref{lemma: copca objects in convbicat}]
The first point follows immediately by the Fox theorem \cite{fox1976coalgebras}.

For the second point, we proved naturality and coherence in Lemmas~\ref{lemma:naturality} and \ref{lemma:coherence}. Below we prove the axioms in Figure~\ref{fig:co-pca axioms}.

    (PCA1) in Figure~\ref{fig:co-pca axioms} follows from the universal property of convex products. Indeed, observe that 
    \begin{align*}
        \diagp{X} ; (\diagq{X} \oplus \id{X}) ; \pi_1 &=\, \diagp{X} ; (\diagq{X} \oplus \id{X}) ;( \id{X \oplus X} \oplus \bang{X}) ; \runit{X\oplus X} \tag{Lemma \ref{lemma:initialproperties}.\ref{lemma:pi1}}\\
        &=\, \diagp{X} ; (\diagq{X} \oplus \bang{X}) ; \runit{X\oplus X} \tag{Symmetric Monoidal Category}\\
        &= p \cdot \diagq{} \tag{Lemma \ref{lemma:initialproperties}.\ref{lemma:pf}}\\
        &= \langle \id{X}, \id{X} \rangle_{pq,\,p(1-q)}\tag{Lemma \ref{lemma:pdiag}}
    \end{align*}
    
        \begin{align*}
        \diagp{X} ; (\diagq{X} \oplus \id{X}) ; \pi_2 &=\, \diagp{X} ; (\diagq{X} \oplus \id{X}) ;( \bang{X \oplus X} \oplus \id{X}) ; \lunit{X} \tag{Lemma \ref{lemma:initialproperties}.\ref{lemma:pi2}}\\
        & =\, \diagp{X} ; ( \, (\diagq{X} ; (\bang{X} \oplus \bang{X}) ) \oplus \id{X} \,)  ; \lunit{X} \tag{Symmetric Monoidal Category}\\
        &=\, \diagp{X} ;  (\bang{X}  \oplus \id{X} )  ; \lunit{X} \tag{Lemma \ref{lemma:naturality}} \\
        &=  (1-p)\cdot \id{X} \tag{Lemma \ref{lemma:initialproperties}.\ref{lemma:1-pf}}
    \end{align*}
    
Similarly, by fixing  $\tilde{p}\defeq  pq$ and $\tilde{q} \defeq \frac{p(1-q)}{1-pq}$, it holds  that
\begin{align}
    \diagptilde{}; (\id{X}\oplus\, \diagqtilde{}) ; \Iassoc{X}{X}{X}; \pi_1 &=\, \diagptilde{}; (\id{X}\oplus\,  \diagqtilde{}) ; \Iassoc{X}{X}{X}; ( \id{X \oplus X} \oplus \bang{X}) ; \runit{X\oplus X} \tag{Lemma \ref{lemma:initialproperties}.\ref{lemma:pi1}}\\
     &=\, \diagptilde{}; (\id{X}\oplus\, \diagqtilde{}) ; ( \id{X} \oplus (\id{X} \oplus \bang{X})) ; \runit{X} \tag{Symmetric Monoidal Category}\\
 &=\, \diagptilde{}; (\,\id{X}\oplus ( \diagqtilde{} ;(\id{X} \oplus \bang{X})) \,); \runit{X} \tag{Symmetric Monoidal Category}\\
&=\, \diagptilde{}; (\,\id{X}\oplus \tilde{q}\cdot \id{X} \,) \tag{Lemma \ref{lemma:initialproperties}.\ref{lemma:pf}}\\
    &=\langle \id{X}, \id{X} \rangle_{\tilde{p},(1-\tilde{p})\tilde{q}} \tag{Lemma \ref{lemma:diagp}}\\
    &= \langle \id{X}, \id{X} \rangle_{pq,p(1-q)} \tag{def.\ of $\tilde{p}$ and $\tilde{q}$}
\end{align} 
    \begin{align}
    \diagptilde{}; (\id{X}\oplus\, \diagqtilde{}); \Iassoc{X}{X}{X};\pi_2&=\, \diagptilde{};( \id{X}\oplus\, \diagqtilde{}); \Iassoc{X}{X}{X}; (\bang{X\piu X}\piu \id{X});\lunit{X}  \tag{Lemma \ref{lemma:initialproperties}.\ref{lemma:pi2}}\\
    &=\, \diagptilde{};( \id{X}\oplus\, \diagqtilde{}); \Iassoc{X}{X}{X}; ((\bang{X}\piu \bang{X})\piu \id{X});\lunit{X}  \tag{Lemma \ref{lemma:coherence}}\\
    &=\, \diagptilde{}; (\bang{X} \oplus \id{X}); \lunit{X} ; \diagqtilde{X}; (\bang{X}\piu \id{X}) ;\lunit{X}  \tag{Symmetric Monoidal Category}\\
    &=\, \diagptilde{}; (\bang{X} \oplus \id{X}); \lunit{X} ; (1-\tilde{q})\cdot \id{X}  \tag{Lemma \ref{lemma:initialproperties}.\ref{lemma:1-pf}}\\
    &=\, \diagptilde{}; (\bang{X} \oplus ( (1-\tilde{q})\cdot \id{X} )) ;  \lunit{X} \tag{Symmetric Monoidal Category}\\
    &= (1-\tilde{p}) \cdot ((1-\tilde{q})\cdot \id{X} )  \tag{Lemma \ref{lemma:initialproperties}.\ref{lemma:1-pf}}\\
    &=  (1-\tilde{p}) (1-\tilde{q}) \cdot \id{X} \tag{Lemma \ref{lemma:initialproperties2}.\ref{lemma:pq}}\\
    &= (1-p) \cdot \id{X} \tag{def. $\tilde{p}$ and $\tilde{q}$}
\end{align}
    Hence, $\diagp{} ; (\diagq{} \oplus \id{X}) =\, \diagptilde{} ; \id{X}\oplus\, \diagqtilde{}; \Iassoc{X}{X}{X}$.\\

    (PCA2) in Figure~\ref{fig:co-pca axioms} follows from the coherence of the PCA-enrichment (i.e.,  Lemma~\ref{lemma:initialproperties}.\ref{eq: assioma aggiuntivo}):
  \begin{align*}
    \diagp{X}; \codiag{X} &=\,   \diagp{X}; (\id{X} \oplus \id{X}) \codiag{X}\\
    &= \id{X}+_p \id{X} \tag{Lemma~\ref{lemma:initialproperties}.\ref{eq: assioma aggiuntivo}}\\
    &=\id{X} \tag{\ref{eq:pca}}
    \end{align*}
    \\

    (PCA3) in Figure~\ref{fig:co-pca axioms} follows from the universal property of convex products. Indeed, observe that
    \begin{align*}
        \diagp{X};\symm{X}{X};\pi_1 &=\, \diagp{X};\symm{X}{X};(\id{X}\piu \bang{X});\runit{X} \tag{Lemma \ref{lemma:initialproperties}.\ref{lemma:pi1}}\\
        &=\, \diagp{X};(\bang{X} \piu \id{X});\lunit{X} \tag{Symmetric Monoidal Category}\\
        &= (1-p)\cdot \id{X}. \tag{Lemma \ref{lemma:initialproperties}.\ref{lemma:1-pf}}
    \end{align*}
    And similarly for $\pi_2$ we have that $\diagp{X};\symm{X}{X};\pi_2 = p\cdot \id{X}$. Hence, the universal property of convex products implies that $\diagp{X};\symm{X}{X} =\, \diagpbar{X}$.
 \end{proof}

\begin{proposition}\label{prop: functor1}
    Let $F\colon \Cat{C}\to \Cat{D}$ be a morphism between convex biproduct categories. 
     $F$ preserves convex products.
   \end{proposition}
   \begin{proof}[Proof of Proposition \ref{prop: functor1}]
    Preservation of finite coproducts is equivalent to preservation of binary coproducts and initial object. Hence, assume that $F(X)\overset{F(\iota_1)}{\to}F(X\piu Y)\overset{F(\iota_2)}{\leftarrow}$ is a coproduct of $F(X)$ and $F(Y)$, which is equivalent to requiring that the arrow $[F(\iota_1),F(\iota_2)]\colon F(X)\piu F(Y)\to F(X\piu Y)$ is an isomorphism (call $k$ the inverse). And assume also that $F(\zero_\Cat{C})$ is initial in $\Cat{D}$ (hence $\cobang{F(\zero_{\Cat{C}})}$ is an isomorphism). Now observe that
    \begin{align}
        F(\pi_1^{X\piu Y})&= F(\id{X}\piu\, \bangp{Y};\rho_X) \tag{Lemma \ref{lemma:initialproperties}.\ref{lemma:pi1}}\\
        &=F(\id{X}\piu\, \bangp{Y});F([\id{X}, \cobang{X}]) \tag{def. $\rho$}\\
        &= F([\id{X}, \bangp{Y};\cobang{X}]) \tag{Coproducts}\\
        &=k;[\id{F(X)}, F(\,\bangp{Y};\cobang{X})] \tag{coproducts and $k$ iso}\\
        &= k;(\id{F(X)}\oplus\, \bangp{F(Y)});[\id{F(X)},\cobang{F(X)}] \tag{$F(\zero_\Cat{C})\cong\zero_\Cat{D}$}\\
        &= k;\pi_1^{F(X)\piu F(Y)} \tag{def.\ $\rho$}
    \end{align}
    and similarly $F(\pi_2^{X\piu Y})=k;\pi_2^{F(X)\piu F(Y)}$. We can now prove that $F(X\piu Y)$ with projections $F(\pi_1^{X\piu Y})$ and $F(\pi_2^{X\piu Y})$ is a convex binary product of $F(X)$ and $F(Y)$. Let $f\colon A\to F(X)$ and $g\colon A\to F(Y)$ and $p\in (0,1)$ and take the arrow $\langle f,g \rangle_{p,1-p};[F(\iota_1), F(\iota_2)]\colon A \to F(X \piu Y)$, where $\langle f,g \rangle_{p,1-p}\colon A \to F(X)\piu F(Y)$ is the unique arrow induced by convex binary products in $\Cat{D}$. Now, since $F(\pi_1^{X\piu Y})= k; \pi_1^{F(X\piu F(Y))}$ it follows that 
    \begin{align}
        \langle f,g \rangle_{p,1-p};[F(\iota_1), F(\iota_2)];F(\pi_1^{X\piu Y})&= \langle f,g \rangle_{p,1-p};[F(\iota_1), F(\iota_2)];k;\pi_1^{F(X)\piu F(Y)}  \tag*{}\\
        &= \langle f,g \rangle_{p,1-p};\pi_1^{F(X)\piu F(Y)} \tag{$k$ inverse}\\
        &=p\cdot f \tag{convex products}
    \end{align}
    and similarly $\langle f,g \rangle_{p,1-p};[F(\iota_1), F(\iota_2)];F(\pi_2^{X\piu Y})=(1-p)\cdot g$. For the uniqueness, if $h:A \to F(X\piu Y)$ is such that $h;F(\pi_1^{X\piu Y})=p\cdot f$ and $h;F(\pi_2^{X\piu Y})=(1-p)\cdot g$ then $h;k= \langle f,g \rangle_{p,1-p}$. Post-composition with $[F(\iota_1),F(\iota_2)]$ gives $h= \langle f,g \rangle_{p,1-p};[F(\iota_1), F(\iota_2)]$.
   \end{proof}
   
   \begin{proof}[Proof of Proposition \ref{prop:functor 1 e mezzo}]
    It follows from Proposition~\ref{prop: functor1} and Lemma~\ref{lemma:naryconvexproduct} which shows how to produce n-ary convex products from binary ones.
\end{proof}

\begin{proof}[Proof of Proposition \ref{prop: monoidal functors}]\label{proof: prop monoidal functors}
By Fox's theorem we know that every finite coproduct category $\Cat{C}$ (hence any convex biproduct category) correspond to  monoidal category $(\Cat{C},\oplus_\Cat{C},\zero_\Cat{C})$ in which every object $X$ is equipped with a coherent and natural monoid structure $(X,\codiag{X},\cobang{X})$. Moreover, a functor $F\colon \Cat{C}\to \Cat{D}$ between finite coproduct  categories preserves finite coproducts if and only if seen as a monoidal functor $F\colon(\Cat{C},\oplus_\Cat{C},\zero_\Cat{C})\to (\Cat{D},\oplus_\Cat{D},\zero_\Cat{D})$ it is strong monoidal  and preserves monoids
    \begin{equation*}
        \psi_{X,X};F(\codiag{X})= \codiag{F(X)} \qquad \psi_\zero;F(\cobang{X})=\cobang{F(X)}
    \end{equation*}
    where  $\psi\colon  \piu_\Cat{C}\circ (F\times F) \overset{\cong}{\Rightarrow} F\circ \piu_{\Cat{D}}$ and $\psi_\zero\colon \zero_\Cat{D}\overset{\cong}{\to} F(\zero_\Cat{C})$ are given by the strong monoidal structure. For any convex biproduct category,  Proposition~\ref{lemma: copca objects in convbicat} implies that every object $X$ is equipped with a natural and coherent co-pca structure $(X,\diagp{X}, \bangp{X})$. In this case, a strong monoidal functor that preserves monoids between monoidal categories, in which every object has the structure of a coherent and natural monoid and co-pca, also preserves $\,\bangp{X}$. Indeed,
     \begin{align}
        \psi_0^{-1}&= \psi_0^{-1};\,\bangp{0} \tag{\ref{ax:coh3} in Figure~\ref{fig:copcacoherence} }\\
        &=\bangp{F(0)} \tag{\ref{eq:nat copca1}}
    \end{align}
    Hence, $\,\bangp{X}$ is preserved: $F(\,\bangp{X});\psi_\zero^{-1}=F(\,\bangp{X});\bangp{F(0)}=\bangp{F(X)}$. 
    
    Moreover, observe that a convex biproduct category $\Cat{C}$ corresponds to a monoidal category 
    $(\Cat{C},\oplus_\Cat{C},\zero_\Cat{C})$ in which every object is equipped with coherent and natural monoid and co-pca structures and moreover it holds that given $f:X\to A$ and $g:X\to B$ then for every $h\colon X\to A\piu B$
    \begin{equation}\label{eq:ax aggiuntivo per fox convex biprod}
         \begin{cases}
    h;(\id{A}\piu\, \bangp{B})=\, \diagp{X};(f\piu \,\bangp{B})\\
    h;(\,\bangp{A}\piu \id{B})=\,\diagp{X};(\,\bangp{A}\piu g)
\end{cases}
\qquad
\Rightarrow 
\qquad 
h=\,\diagp{X};(f\piu g)
    \end{equation}
    since $\Cat{C}$ has convex binary products.
    Now if $F$ is PCA-enriched, then 
    \begin{align}
        F(\diagp{X});\psi_{X,X}^{-1};(\id{F(X)}\piu\, \bangp{F(X)})&=F(\diagp{X});\psi_{X,X}^{-1};(\id{F(X)}\piu\, \bangp{F(X)}); \rho_{F(X)};\rho^{-1}_{F(X)} \tag*{}\\
        &=F(\diagp{X});\psi_{X,X}^{-1};(\id{F(X)}\piu\, \bangp{F(X)}); (\id{F(X)}\piu \cobang{F(X)});\codiag{F(X)};\rho_{F(X)}^{-1} \tag{Def. $\rho$}\\
        &=F(\diagp{X});F(\id{X}\piu (\,\bangp{X};\cobang{X})); \psi_{X,X}^{-1};\codiag{F(X)};\rho^{-1}_{F(X)}\tag{Nat. $\psi$}\\
        &=F(\diagp{X});F(\id{X}\piu (\,\bangp{X};\cobang{X}));F(\codiag{X});\rho^{-1}_{F(X)} \tag{$\psi_{X,X};F(\codiag{X})= \codiag{F(X)} $}\\
        &= F(p\cdot \id{X});\rho^{-1}_{F(X)}  \tag{Lemma \ref{lemma:initialproperties}.\ref{eq: assioma aggiuntivo}}\\
        &= p\cdot \id{F(X)};\rho^{-1}_{F(X)}\tag{PCA-enrichment}\\
        &=\, \diagp{F(X)};(\id{F(X)}\piu\,\bangp{F(X)});\rho_{F(X)};\rho^{-1}_{F(X)}\tag{Lemma \ref{lemma:initialproperties}.\ref{eq: assioma aggiuntivo} + Def. $\rho$ + Lemma \ref{lemma:initialproperties}.\ref{lemma:starxy}}\\
        &=\, \diagp{F(X)};(\id{F(X)}\piu\,\bangp{F(X)}) \tag*{}
    \end{align}
and similarly one can prove that $F(\diagp{X});\psi_{X,X}^{-1};(\,\bangp{F(X)}\piu\id{F(X)})= \, \diagp{F(X)};(\,\bangp{F(X)}\piu \id{F(X)})$. Hence, by property (\ref{eq:ax aggiuntivo per fox convex biprod}) it follows that $F(\diagp{X});\psi_{X,X}^{-1}=\, \diagp{F(X)}$.
Vice versa, if $F\colon(\Cat{C},\oplus_\Cat{C},\zero_\Cat{C})\to (\Cat{D},\oplus_\Cat{D},\zero_\Cat{D})$ also preserves diagonals $\diagp{}$ then it is PCA-enriched since
\begin{align}
    F(f+_pg)&= F(\diagp{X};(f\piu g);\codiag{X}) \tag{Lemma \ref{lemma:initialproperties}.\ref{eq: assioma aggiuntivo}}\\
    &=\, \diagp{F(X)};\psi_{X,X};F(f\piu g; \codiag{X}) \tag*{}\\
&=\, \diagp{F(X)};F(f)\piu F(g);\codiag{F(X)}\tag*{}\\
&= F(f)+_p F(g) \tag{Lemma \ref{lemma:initialproperties}.\ref{eq: assioma aggiuntivo}}
\end{align}
 Similarly, the preservation of $\,\bangp{X}$ and $\cobang{Y}$ implies that $ F(\star_{X,Y})= \star_{F(X),F(Y)}$.

\end{proof}

   \subsection{Proofs of Section \ref{ssec:almost}}
   
   \begin{proof}[Proof of Lemma \ref{lemma:enrichmentpca}]
We first prove that for all objects $X,Y$ in $\Cat{C}$, the structure defined above is a pca, namely that the equalities in \eqref{eq:pca} hold.
   For idempotency, we have the following derivation. 
   \begin{align}
        f +_p f & =\, \diagp{X};(f \oplus f);\codiag{Y} \tag{def. $+_p$}\\
        & =\, \diagp{Y};\codiag{Y};f \tag{\ref{eq:nat monoid1}}\\
        & = f. \tag{\ref{ax:PCA2}}
   \end{align} 
For commutativity, we have the following derivation.
    \begin{align}
        f +_p g & =\, \diagp{X};(f \oplus g);\codiag{Y} \tag{def. $+_p$}\\
        & =\, \diagp{X};\sigma_{X,X};(g \oplus f);\codiag{Y} \tag{SMC}\\
        & =\, \diaggen{1-p}{X} ;(g \oplus f);\codiag{Y} \tag{\ref{ax:PCA3}}\\
        & = g +_{1-p} f. \tag{def. $+_{1-p}$}
    \end{align}
    For associativity, consider $f,g,h\colon X \to Y$ and $p,q\in[0,1]$. We have:
    \begin{align}
        (f +_q g) +_p h & =\, \diagp{X};((f +_q g) \oplus h);\codiag{Y} \tag{def. $+_p$}\\
        & =\, \diagp{X};(\diagq{X} \oplus \id{X});((f \oplus g) \oplus h);(\codiag{Y}\oplus \id{Y});\codiag{Y} \tag{def. $+_q$}\\
        & =\, \diagptilde{X};(\id{X} \oplus \diagqtilde{X});(f \oplus (g \oplus h));( \id{X}\oplus\codiag{Y} );\codiag{Y} \tag{\ref{ax:PCA1}+\ref{ax:Mon1}}\\
        & =\, \diagptilde{X};(f \oplus (g +_{\tilde{q}} h));\codiag{Y} \tag{def. $+_{\tilde{q}}$}\\
        & = f +_{\tilde{p}} (g +_{\tilde{q}} h). \tag{def. $+_{\tilde{p}}$}
    \end{align}

    We pass now to the equations in \eqref{eq:enr}. For every $e\colon Z \to X$ and $f,g\colon X \to Y$, we have:  
    \begin{align}
        e;(f +_p g) & = e;\,\diagp{X};(f \oplus g);\codiag{Y} \tag{def. $+_p$}\\
        & =\, \diagp{X};(e \oplus e);(f \oplus g);\codiag{Y} \tag{\ref{eq:nat copca1}}\\
        & =\, \diagp{X};((e;f) \oplus (e;g));\codiag{Y} \tag{SMC}\\
        & = (e;f) +_p (e;g). \tag{def. $+_{p}$}
    \end{align}
The second equation is proved similarly. The third equation is proved below for $f\colon Z\to X$.
    \begin{align}
        f;\star_{X,Y} & = f;\,\bangp{X};\cobang{Y} \tag{def. $\star$}\\
        & =\, \bangp{Z};\cobang{Y} \tag{\ref{eq:nat copca2}}\\
        & = \star_{Z,Y}. \tag{def. $\star$}
    \end{align}
\end{proof}

\begin{proof}[Proof of Lemma \ref{lemma:coproducts}]
The first two points follow immediately by the Fox theorem \cite{fox1976coalgebras}. Naturality of $\,\bangp{X}$ implies that $0$ is also a terminal object, hence a zero object. 
  To prove that axiom \eqref{eq:delta} holds, consider the following equalities:
    \begin{align}
        \iota_1;\pi_1 & = \Irunit{X_1};(\id{X_1}\oplus \cobang{X_2});(\id{X_1}\oplus \bang{X_1});\runit{X_1} \tag{def. of $\iota_1$ and $\pi_1$}\\
        & = \Irunit{X_1};(\id{X_1}\oplus (\cobang{X_2};\bang{X_1}));\runit{X_1} \tag{SMC}\\
        & = \Irunit{X_1};(\id{X_1}\oplus \id{0});\runit{X_1} \tag{\ref{eq:nat copca2} + \ref{ax:coh3}}\\
        & = \id{X_1} \tag{SMC}
    \end{align}
    Similarly, one can prove that $\iota_2;\pi_2 = \id{X_2}$. Now, for $i\neq j$, we have: 
    \begin{align}
        \iota_1;\pi_2 & = \Irunit{X_1};(\id{X_1}\oplus \cobang{X_2});(\bang{X_1}\oplus \id{X_1});\lunit{X_1} \tag{def. of $\iota_1$ and $\pi_2$}\\
        & = \Irunit{X_1};((\id{X_1};\bang{X_1})\oplus (\cobang{X_2};\id{X_1}));\lunit{X_1} \tag{SMC}\\
        & = \Irunit{X_1};(\bang{X_1}\oplus \cobang{X_2});\lunit{X_1} \tag{SMC}\\
        & = \Irunit{X_1};(\,\bangp{X_1}\oplus \id{0});(\id{0}\oplus \cobang{X_2});\lunit{X_1}\tag{SMC}\\  
        & = \bang{X_1};\Irunit{0};\lunit{0};\cobang{X_2}    \tag{Nat. $\Irunit{},\lunit{}$}     \\
        & = \bang{X_1};\cobang{X_2} \tag{$\runit{0}=\lunit{0}$} \\
        & =\star_{X_1,X_2} \tag{Lemma~\ref{lemma:enrichmentpca}}
    \end{align}
    and similarly $\iota_2;\pi_1 = \bang{X_2};\cobang{X_1}$.
\end{proof}

\begin{proof}[Proof of Lemma \ref{lemma:diagpq}] 
We prove the statement only for $p,q\in(0,1)$ such that $p+q<1$, the other cases are similar and left to the reader.
Point \eqref{lemma:diagpq1} is proved by the following derivation.

    \begin{align}
        \diaggen{p,q\,}{X};\pi_1 
        &=\, \diaggen{p,q\,}{X};(\id{X}\oplus \bang{X});\runit{X} \tag{\eqref{eq:defpi}}\\
         & =\, \diagp{X};(\id{X}\oplus \diaggen{\frac{q}{1-p}}{X});(\id{X}\oplus (\id{X} \oplus\bang{X}));(\id{X}\oplus \runit{X});(\id{X}\oplus \bang{X});\runit{X} \tag{\eqref{diagpq generalizzata}}\\
         &= \, \diagp{X};(\id{X}\oplus \diaggen{\frac{q}{1-p}}{X});(\id{X}\oplus (\id{X} \oplus\bang{X}));(\id{X}\oplus (\runit{X};\,\bangp{X}));\runit{X} \tag{SMC}\\
         &= \, \diagp{X};(\id{X}\oplus \diaggen{\frac{q}{1-p}}{X});(\id{X}\oplus (\id{X} \oplus\bang{X}));(\id{X}\oplus ((\,\bangp{X}\oplus \id{0});\runit{0}));\runit{X} \tag{Nat. $\runit{}$}\\
         &= \, \diagp{X};(\id{X}\oplus \diaggen{\frac{q}{1-p}}{X});(\id{X}\oplus (\id{X} \oplus\bang{X}));(\id{X}\oplus ((\,\bangp{X}\oplus\, \bangp{0});\runit{0}));\runit{X} \tag{\ref{ax:coh3}}\\
        & =\, \diagp{X};(\id{X}\oplus \diaggen{\frac{q}{1-p}}{X});(\id{X}\oplus (\id{X} \oplus\bang{X}));(\id{X}\oplus \bang{X\oplus 0});\runit{X} \tag{\ref{eq:coherence cobang}}\\
        & =\, \diagp{X};(\id{X}\oplus \diaggen{\frac{q}{1-p}}{X});(\id{X}\oplus ((\id{X} \oplus\bang{X});\bang{X\oplus 0}));\runit{X} \tag{SMC}\\
        & =\, \diagp{X};(\id{X}\oplus \diaggen{\frac{q}{1-p}}{X});(\id{X}\oplus \bang{X\oplus X}); \runit{X} \tag{\ref{eq:nat copca2}}\\   
        & =\, \diagp{X};(\id{X}\oplus \diaggen{\frac{q}{1-p}}{X});(\id{X}\oplus ((\bang{X} \oplus \bang{X});\Ilunit{0})); \runit{X} \tag{\ref{eq:coherence cobang}}\\
        & =\, \diagp{X};(\id{X}\oplus (\diaggen{\frac{q}{1-p}}{X};(\bang{X} \oplus \bang{X});\Ilunit{0})); \runit{X} \tag{SMC}\\
        & =\, \diagp{X};(\id{X}\oplus (\,\bangp{X};\diaggen{\frac{q}{1-p}}{0};\Ilunit{0})); \runit{X} \tag{\ref{eq:nat copca1}}\\
        & =\, \diagp{X};(\id{X}\oplus\, \bang{X});\runit{X} \tag{\ref{eq:star I = id I} }\\
        & = \,\diagp{X};(\id{X}\oplus \bang{X});(\id{X}\oplus\cobang{X});\codiag{X} \tag{\ref{ax:Mon2}}\\
        & =\, \diagp{X};(\id{X}\oplus \bang{X};\cobang{X});\codiag{X} \tag{SMC}\\
        & =\, \diagp{X};(\id{X}\oplus \star_{X,X});\codiag{X} \tag{\eqref{eq:enrichment}}\\
        & = \id{X}+_p \star \tag{\eqref{eq:enrichment}}\\
        & = p\cdot \id{X}. \tag{def. of $p\cdot-$}
    \end{align}
    Point \eqref{lemma:diagpq2} is proved similarly. Point~\eqref{lemma:diagpq4} is proved by the following derivation.
    \begin{align}
         \diaggen{p,q\,}{X};(f_1\oplus f_2);(\id{Y_1}\oplus \bang{Y_1});\runit{Y_1} 
        & =\, \diaggen{p,q\,}{X};(\id{X}\oplus \bang{X}); (f_1 \oplus \id{0});\runit{Y_1} \tag{\ref{eq:nat copca2}}\\
        & =\, \diaggen{p,q\,}{X};(\id{X}\oplus \bang{X});\runit{X};f_1 \tag{Nat. $\rho$}\\
        & =\, (p\cdot\id{X});f_1 \tag{Lemma~\ref{lemma:diagpq}.\ref{lemma:diagpq1}}\\
        & = p\cdot f_1 \text{.} \tag{pca-enrichment}
    \end{align}
    Similarly, one can prove point~\eqref{lemma:diagpq5}. 
    \end{proof}

\begin{proof}[Proof of Lemma \ref{lemma:diagpq3}]
By Proposition~\ref{lemma: copca objects in convbicat}, every object in $\Cat{C}$ is equipped with natural and coherent monoid and co-pca structures. By the last two items of Lemma \ref{lemma:diagpq}, $\diaggen{p,q\,}{Z};(f_1\oplus f_2); \pi_1= p\cdot f_1$ and $\diaggen{p,q\,}{Z};(f_1\oplus f_2); \pi_2= q\cdot f_2$. Since $\Cat{C}$ has convex products, the statement holds by uniqueness.
\end{proof}

%% file: appendices/appcmatrices.tex
\section{Appendix to Section \ref{sec:cmatrix}}\label{app:sec:cmatrix}


\begin{proof}[Proof of Lemma~\ref{lemma:stmat coproduct}]
We proceed by illustrating that, for every object $P$, $(\codiag{P},\cobang{P})$ defined as in \eqref{eq:matmonpca}, is a natural and coherent comonoid.


Observe that, by the definitions in  \eqref{eq:matmonpca}, $\codiag{P}$ where $P=\bigoplus_{k=1}^{n} U_k$ is the matrix
\[\begin{pNiceMatrix}
1\cdot{\id{U_1}}  	& \emptyset & \Cdots & \emptyset & 1\cdot{\id{U_1}}  	& \emptyset & \Cdots & \emptyset \\
\emptyset  &   & \Ddots & \Vdots & \emptyset  &   & \Ddots & \Vdots \\	
\Vdots & \Ddots &   & \emptyset & \Vdots & \Ddots &   & \emptyset \\
\emptyset  & \Cdots & \emptyset  & 1\cdot{\id{U_{n}}} & \emptyset  & \Cdots & \emptyset  & 1\cdot{\id{U_{n}}} 
\CodeAfter
\line{1-1}{4-4}
\line{1-5}{4-8}
\end{pNiceMatrix}\]
while $\cobang{U}$ is the unique $n\times 0$ matrix.

Axioms in Table~\ref{fig:freestrictfccat} follow by matrix multiplication. For instance, the associativity axiom \ref{eq:codiag assoc} is obtained by the following computation:
\[ (\id{ P}\piu \codiag{ P}) ; \codiag{ P}=  \begin{pmatrix}
 	 \id{P} &  \id{P} &   \id{P}
 \end{pmatrix}= (\codiag{ P}\piu \id{P});\codiag{P}
 \]
and the naturality  \ref{eq:codiag nat} by:
\[ (M\piu M);\codiag{Q} 
=\begin{pmatrix}
    M & M\end{pmatrix}=
\codiag{P};M\]
where $M\colon P\to Q$. Finally, 
\ref{eq:cobang nat} is obvious since the multiplication of an $m\times n$ matrix with the $n\times 0$ empty matrix is the empty $m\times 0$ matrix. 

We can collect the above observations and conclude that $\stmat{\Cat{C}}$ is a strict symmetric monoidal category in which every object $U$ is equipped with a coherent and natural monoid structure $(U,\codiag{U},\cobang{U})$, hence, by Fox's theorem $\stmat{\Cat{C}}$ is a finite coproduct category.
\end{proof}

\begin{proof}[Proof of Lemma~\ref{lemma:stmat copca}]
    Recall from \eqref{eq:matmonpca} the structure $(\diagp{P},\bang{P})$. 
For $P=\bigoplus_{k=1}^{n} U_k$, $\diagp{P}$ is the $(2n\times n)$ matrix illustrated below
\[\begin{pNiceMatrix}
p\cdot{\id{U_1}}  	& \emptyset & \Cdots & \emptyset   \\
\emptyset  &   & \Ddots & \Vdots   \\	
\Vdots & \Ddots &   & \emptyset   \\
\emptyset  & \Cdots & \emptyset  & p\cdot{\id{U_{n}}}   \\
(1-p)\cdot{\id{U_1}}  	& \emptyset & \Cdots & \emptyset \\
\emptyset  &   & \Ddots & \Vdots \\
\Vdots & \Ddots &   & \emptyset\\
\emptyset  & \Cdots & \emptyset  & (1-p)\cdot{\id{U_{n}}}
\CodeAfter
\line{1-1}{4-4}
\line{5-1}{8-4}
\end{pNiceMatrix}\]
while $\,\bangp{P}$ is the empty $0\times n$ matrix.
Axioms in Table~\ref{fig:freecopcacat} follow by matrix multiplication. For instance, in the case $U\in\Cat{C}$, the associativity axiom \ref{eq:diagp assoc} which states that  $\diagp{U};(\diagq{U}\piu \id{U})=\, \diagptilde{U};(\id{U}\piu \diagqtilde{U})$ where $\tilde{p}= pq$ and $ \tilde{q}= \frac{p(1-q)}{1-pq}$
follows, for $U\in\Cat{C}$, by the following computation

\[ \diagp{U};(\diagq{U}\piu \id{U}) = \begin{pmatrix}
    q\cdot \id{U} &  \emptyset \\
    (1-q)\cdot \id{U} &  \emptyset\\
    \emptyset & 1\cdot \id{U}
\end{pmatrix}
\begin{pmatrix}
    p\cdot \id{U} \\
    (1-p)\cdot \id{U}
\end{pmatrix}= \begin{pmatrix}
    pq\cdot \id{U} \\
    p(1-q)\cdot \id{U} \\
    (1-p)\cdot \id{U}
\end{pmatrix} = \begin{pmatrix}
    \tilde{p}\cdot \id{U} \\
    \tilde{q}(1-\tilde{p})\cdot \id{U} \\
    (1-\tilde{p})(1-\tilde{p})\cdot \id{U}
    \end{pmatrix} \]
which is equal to
\[ \diagptilde{U};(\id{U}\piu\, \diagqtilde{U}) = \begin{pmatrix}
    1\cdot \id{U} &  \emptyset\\
    \emptyset & \tilde{q}\cdot \id{U} \\
    \emptyset & (1-\tilde{q})\cdot \id{U}
\end{pmatrix}
\begin{pmatrix}
    \tilde{p}\cdot \id{U} \\
    (1-\tilde{p})\cdot \id{U}
\end{pmatrix}= \begin{pmatrix}
    \tilde{p}\cdot \id{U} \\
    \tilde{q}(1-\tilde{p})\cdot \id{U} \\
    (1-\tilde{p})(1-\tilde{p})\cdot \id{U}
\end{pmatrix} \]
The general case $P=\bigoplus_{k=1}^{n} U_k$ is obtained similarly. Axiom~\ref{eq:diagp idempotency} is obvious since $p\cdot \id{U_i}+ (1-p)\cdot \id{U_i}= \id{U_i}$ by PCA axiom~\ref{eq:pca}. 
Axiom~\ref{eq:diagp nat} is obtained as follows: for $M\colon U\to V$

\[M;\diagp{V} = \begin{pmatrix}
    p\cdot M\\
    (1-p)\cdot M
\end{pmatrix} 
M=\, \diagp{U};M\piu M\]
where $p\cdot M$ denotes the matrix whose $(j,i)$-entries are given by $p\cdot M_{ji}$. 
Finally, axiom \ref{eq:bangp nat} is obvious since the multiplication of an $m\times n$ matrix with the $n\times 0$ empty matrix is the empty $m\times 0$ matrix. Hence, we can collect the above observations and conclude that every object $U\in\stmat{\Cat{C}}$ is equipped with a co-pca structure $(U,\diagp{U},\bangp{U})$ which satisfies the axioms in Table~\ref{fig:freecopcacat}.
    \end{proof}

\begin{proof}[Proof of Proposition \ref{prop:stmatfun}]
    The functor $\stmat{F}$ is strict monoidal and preserves monoids and co-pca objects, hence by Proposition~\ref{prop: monoidal functors} it is a morphism of convex biproduct categories.
\end{proof}

\begin{proof}[Proof of Theorem \ref{thm:matfree}]
    Let $\Cat{C}$ be a $\Cat{PCA}$-enriched category and $\Cat{D}$ a convex biproduct category, and consider a $\Cat{PCA}$-enriched functor $F\colon \Cat{C}\to U(\Cat{D})$. The unit $\eta:\Cat{C}\to \stmat{\Cat{C}}$ is the identity-on-objects functor mapping an arrow $f\in\Cat{C}[U,V]$ into the $(1\times 1)$ matrix with entry $1\cdot f$. 
    
   The functor $F^\sharp:\stmat{\Cat{C}}\to \Cat{D}$ sends an object $\bigoplus_{k=1}^n U_k$ to $\bigoplus_{k=1}^n F(U_k)$ and an arrow $M\colon\bigoplus_{k=1}^n U_k\to \bigoplus_{k=1}^m V_k$, which corresponds to a matrix with entries $M_{ji}= p_{ji}\cdot f_{ji}$, into the arrow of $\Cat{D}$ obtained as follows:
     we first define the image of a column of the matrix, then we observe that $M$ is the copairing of its columns, denoted with $\column{M}{i}$ for $i=1,\dots,n$, and define $F^\sharp(M)$ as the copairing of the images of the columns. The interpretation of the $i$-column $\column{M}{i}$ is defined as the arrow induced by the $m$-ary convex product through the vector $\vec{p}=(p_{1i},\dots,p_{mi})$ and the arrows $F(f_{ji})$ in $\Cat{D}$, which exists by Lemma~\ref{lemma:naryconvexproduct}. $F^\sharp$ is well-defined since, for $M\equiv M'$, it holds that $M_{ji}=p_{ji}\cdot f_{ji}= p'_{ji}\cdot f'_{ji}= M'_{ji}$. Now, for $k=1,\dots, m$
    \begin{align} 
        F^\sharp(\column{M'}{i});\pi_k&=p'_{ki}\cdot F(f'_{ki})\tag{Def. $F^\sharp(\column{M'}{i})$ }\\
        &=F(p'_{ki}\cdot f'_{ki}) \tag{$F$ $\Cat{PCA}$-enriched}\\
        &= F(p_{ki}\cdot f_{ki}) \tag{$M\equiv M'$}\\
        &=p_{ki}\cdot F(f_{ki}) \tag{$F$ $\Cat{PCA}$-enriched}
        \end{align}   
        and the universal property of convex products in $\Cat{D}$ implies that $F^\sharp(\column{M'}{i})= F^\sharp(\column{M}{i})$ and, hence, $F^\sharp(M')=F^\sharp(M)$ by the universal property of coproducts in $\Cat{D}$. Moreover, a simple computation shows that for every object $U\in\stmat{\Cat{C}}$, $F^\sharp$ preserves monoids and co-pca structures, i.e.\ it holds: 
    
    \begin{itemize}
    \item $F^\sharp(\,\diagp{U})=\, \diagp{F(U)}$ and $F^\sharp(\,\bangp{U})=\, \bangp{F(U)}$;
    \item $F^\sharp(\,\codiag{U})= \codiag{F(U)}$ and $F^\sharp(\,\cobang{U})= \cobang{F(U)}$;
    \item $F^\sharp(1\cdot f)= F(f)$, for every $f\colon U\to V$ in $\Cat{C}$;
    \end{itemize}
     Hence, by Proposition~\ref{prop: monoidal functors} it follows that $F^\sharp$ is a morphism of convex biproduct categories. It is also the unique morphism such that $\eta;F^\sharp = F$ since any other morphism $G\colon\stmat{\Cat{C}}\to \Cat{D}$ such that $\eta;G=F$ must preserve coproducts and convex products by Proposition~\ref{prop: functor1} and \ref{prop:functor 1 e mezzo}. Hence, $F^\sharp= G$ since every column of $M$ is the arrow induced by the m-ary convex product through $\vec{p}=(p_{1i},\dots,p_{mi})$ and the arrows $f_{ji}$ for a fixed $i$. 
    \end{proof}

\begin{proof}[Proof of Proposition \ref{prop:counit fullfaithful}]\label{proof:counit fullfaithful}
    Recall that the counit $\epsilon$ is obtained as $\id{\Cat{C}}^\sharp$, following the notation in the proof of Theorem \ref{thm:matfree}. Fullness follows from the fact that every arrow $f\colon A\to B$ in $\Cat{C}$ is equal to $\epsilon(\eta(f))$ by construction. Faithfulness is obtained by observing that if $\epsilon(M)=\epsilon(M')$ for two matrices in $\stmat{U(\Cat{C})}$, $M,M'\colon\bigoplus_{k=1}^n U_k\to \bigoplus_{k=1}^m V_k$ with entries $M_{ji}= p_{ji}\cdot f_{ji}$ and $M'_{ji}= p'_{ji}\cdot f'_{ji}$ , then the copairing of the image through $\epsilon$ of the columns of $M$ is equal to the copairing of the images through $\epsilon$ of the columns of $M'$. The universal property of the coproducts in $\Cat{C}$ then ensures that for every $i$-th column of $M$ and $M'$, denoted with $\column{M}{i}$ and $\column{M'}{i}$, for $i=1,\dots,n$, it holds $\epsilon(\column{M}{i})=\epsilon(\column{M'}{i})$. Now, since $\epsilon(\column{M}{i})$ is given by the arrow induced by the $m$-ary convex product through $f_{1i},\dots,f_{mi}$, with $\vec{p}=(p_{1i},\dots,p_{mi})$, and $\epsilon(\column{M'}{i})$ is given by the arrow induced by the $m$-ary convex product through $f'_{1i},\dots,f'_{mi}$, with $\vec{p}'=(p'_{1i},\dots,p'_{mi})$, it follows that $\epsilon(\column{M}{i});\pi_j=p_{ji}\cdot f_{ji}=p'_{ji}\cdot f'_{ji} =\epsilon(\column{M'}{i});\pi_{j}$, for $j=1,\dots,m$. Hence, $M_{ji}\equiv M'_{ji}$ and therefore they are equal in $\stmat{U(\Cat{C})}$.  
    \end{proof} 

%% file: appendices/apptc.tex
\section{Appendix to Section \ref{sec:syntactic}}\label{app:sec:syntactic}

\subsection{Preliminary Results}
We collect below several simple results that we will frequently use.

\begin{lemma}\label{lemma:pcdot}
Let $\t\colon P \to Q$ be an arrow of $\CatTapeC$ . Then, for all $p\in(0,1)$, $p\cdot \t =\, \diagp{P};(\t \oplus \bang{P})$.
\end{lemma}
\begin{proof}
\begin{align*}
p\cdot \t &= \t+_p \star_{P,Q} \tag{def. of $p\cdot -$}\\
&=\, \diagp{P} ; (\t \oplus \star_{P,Q}) ; \codiag{Q} \tag{\ref{eq:enrichment}}\\
&=\, \diagp{P} ; (\t \oplus (\bang{P}; \cobang{Q})); \codiag{Q} \tag{\ref{eq:enrichment}}\\
&=\, \diagp{P} ; (\t \oplus  \bang{P}) ;(\id{Q} \oplus \cobang{Q}) ; \codiag{Q} \tag{Symmetric Monoidal Category}\\
&=\, \diagp{P} ; (\t \oplus  \bang{P}) ;\id{Q} \tag{\ref{eq:codiag unital}} \\
&=\, \diagp{P} ; (\t \oplus  \bang{P})  \tag{Category} 
\end{align*}
\end{proof}

\begin{lemma}\label{lemma:pcopairing}
Let $\s = (\s_1 \oplus \s_2) ; \codiag{}$ be an arrow of $\CatTapeC$ . 
Then $r \cdot \s = (r\cdot \s_1 \oplus r\cdot \s_2); \codiag{}$. 
\end{lemma}
\begin{proof}
Assume that $\s$ has type $P_1\oplus P_2 \to Q$ and observe that
\begin{align*}
r \cdot \s &=\, \diagpX{r}{P_1\oplus P_2} ; (\s \oplus \bang{}) \tag{Lemma \ref{lemma:pcdot}}\\
&= (\diagpX{r}{P_1} \oplus\, \diagpX{r}{P_2} ) ; (\id{P_1}\oplus \symm{P_1}{P_2} \oplus \id{P_2}) ; (\s \oplus \bang{P_1}\oplus \bang{P_2}) \tag{\ref{eq:diagp coherence}}\\
&= (\diagpX{r}{P_1} \oplus\, \diagpX{r}{ P_2} ) ; (\id{P_1}\oplus \symm{P_1}{P_2} \oplus \id{P_2}) ; (\, ((\s_1 \oplus \s_2) ; \codiag{Q}) \oplus \bang{P_1}\oplus \bang{P_2}\,) \tag{Hypothesis}\\
&= (\diagpX{r}{P_1} \oplus\, \diagpX{r}{P_2} ) ; (\id{P_1}\oplus \symm{P_1}{P_2} \oplus \id{P_2}) ; ( \s_1 \oplus \s_2  \oplus \bang{P_1}\oplus \bang{P_2}); \codiag{Q} \tag{Sym. Mon. Cat.}\\
&= (\diagpX{r}{P_1} \oplus\, \diagpX{r}{ P_2} ) ;( \s_1 \oplus \bang{P_1} \oplus \s_2   \oplus \bang{P_2}); \codiag{Q} \tag{Symmetric Monoidal Category}\\
&= (\diagpX{r}{P_1} ; ( \s_1 \oplus \bang{P_1})) \oplus (\diagpX{r}{ P_2}  ;( \s_2   \oplus \bang{P_2})); \codiag{Q} \tag{Symmetric Monoidal Category}\\
&= (r\cdot \s_1 ) \oplus (r\cdot  \s_2 ); \codiag{Q} \tag{Lemma \ref{lemma:pcdot}}
\end{align*}
\end{proof}

\begin{lemma}\label{lemma:basicTC}
Let $\t=\,\diagp{} ; (\t_1 \oplus \t_2)$ be an arrow of $\CatTapeC$. Then:
\begin{enumerate}
\item $\t; (\id{} \oplus \bang{}) = p\cdot \t_1$;
\item $\t ; (\bang{} \oplus \id{}) = (1-p) \cdot \t_2$;
\item $r \cdot \t =\, \diagp{} ; (r\cdot \t_1 \oplus r\cdot \t_2)$ for all $r\in (0,1)$.
\end{enumerate}
\end{lemma}
\begin{proof}
Points 1 and 2 follow immediately from the definition of the enrichment.

We prove below point 3 for $\t_1\colon:P\to Q_1$ and $\t_2\colon P\to Q_2$.
\begin{align*}
r \cdot \t &=\, \diagpX{r}{} ; (\t \oplus \bang{Q_1\piu Q_2}) \tag{Lemma \ref{lemma:pcdot}}\\
&=\, \diagpX{r}{} ; (\diagp{} ; (\t_1 \oplus \t_2) \oplus \bang{{Q_1\piu Q_2}}) \tag{Hypothesis}\\
&=\, \diagpX{r}{} ; ( \, (\diagp{} ; (\t_1 \oplus \t_2)) \oplus (\diagp{}; (\bang{Q_1} \oplus \bang{Q_2}))\, ) \tag{\ref{eq:bangp coherence} + \ref{eq:diagp idempotency}}\\
&=\, \diagpX{r}{} ; (\diagp{} \oplus \diagp{}) ; (\,  (\t_1 \oplus \t_2) \oplus (\bang{Q_1} \oplus \bang{Q_2}) \, ) \tag{Symmetric Monoidal Category}\\
&=\, \diagp{} ; (\diagpX{r}{} \oplus \diagpX{r}{}) ; (\id{} \oplus \symm \oplus \id{}) ; (\,  (\t_1 \oplus \t_2) \oplus (\bang{Q_1} \oplus \bang{Q_2}) \, ) \tag{\ref{eq:diagp nat} + Sym. Mon. Cat.}\\
&=\, \diagp{} ; (\diagpX{r}{} \oplus \diagpX{r}{}) ;  (\,  (\t_1 \oplus \bang{Q_1}) \oplus (\t_2 \oplus \bang{Q_2}) \, ) \tag{Symmetric Monoidal Category}\\
&=\, \diagp{} ; ( \,(\diagpX{r}{} ; (\t_1 \oplus \bang{Q_1})) \oplus (\diagpX{r}{} ;(\t_2 \oplus \bang{Q_2})) \, ) \tag{Symmetric Monoidal Category}\\
&=\, \diagp{} ; (r\cdot \t_1 \oplus r\cdot \t_2)  \tag{Lemma \ref{lemma:pcdot}}
\end{align*}
 
\end{proof}

\begin{lemma}\label{lemma:fractions}
Let $\t=\,\diagp{} ; (\t_1 \oplus \t_2)$ and $\s=\, \diagq{} ; (\s_1 \oplus \s_2)$. If $\t_1=\frac{q}{p}\cdot \s_1$ and $\frac{1-p}{1-q} \cdot \t_2 = \s_2$, then $\t = \s$.
\end{lemma}
\begin{proof}
\begin{align}
\t &=\,  \diagp{};( \frac{q}{p} \cdot \s_1 \oplus \t_2  )\tag{Lemma \ref{lemma:basicTC}.1}\\
&=\, \diagp{}; (\diagpX{\frac{q}{p}\,\,}{} \oplus \id{}) ; (\s_1 \oplus \bang{} \oplus \t_2) \tag{Lemma \ref{lemma:pcdot}}\\
&=\, \diagq; (\id{} \oplus\, \diagpX{\frac{p-q}{1-q}\,\,}{}) ; (\s_1 \oplus \bang{} \oplus \t_2) \tag{\ref{eq:diagp assoc}} \\
&=\, \diagq; (\id{} \oplus ( \diagpX{\frac{1-p}{1-q}\,\,}{}; \symm)) ; (\s_1 \oplus (\bang{} \oplus \t_2)) \tag{\ref{eq:diagp symmetry}} \\
&=\, \diagq; (\id{} \oplus ( \diagpX{\frac{1-p}{1-q}\,\,}{})) ; (\s_1 \oplus  \t_2 \oplus \bang{}) \tag{Symmetric Monoidal Category}\\
&=\,  \diagq{};( \s_1 \oplus \frac{1-p}{1-q} \cdot \t_2  ) \tag{Lemma \ref{lemma:basicTC}.2}\\
&=\s \tag{Hypothesis}
\end{align}
\end{proof}

\subsection{Normal forms}\label{app:normal forms}
We are now going to prove several normal form results. For these results we need to fix some notation.



%
%
For all natural numbers $n \in \mathbb{N}$ one can define $\codiag{U}^n\coloneqq \bigoplus_{i=1}^n U\to U$ as follows:
\begin{equation}\label{eq:codiagn in appendice}
    (n=0):\  \codiag{U}^0\defeq\cobang{U}\qquad (n+1):\ \codiag{U}^{n+1}\defeq(\id{U} \oplus \codiag{U}^n  );\codiag{U}.
\end{equation}

\begin{lemma}[Pre normal form]\label{lemma:prenormalform}
Every arrow $\t\colon \oplus_{i=1}^n U_i\to \oplus_{i=1}^m V_i$ in $\CatTapeC$ can be written as $\t_1;\t_2;\t_3;\t_4$ where:
\begin{itemize}
    \item $\t_1$ is obtained through the fragment \setlength{\tabcolsep}{4pt}\begin{tabular}{rc ccccccccccccccccccccc}
   $\!\!\! \mid \!\!\!$ & $\diagp{U}$  & $\!\!\! \mid \!\!\!$  & $\id{\zero}$ & $\!\!\! \mid \!\!\!$ &  $\id{U}$ & $\!\!\! \mid \!\!\!$   & $   \t ; \t   $ & $\!\!\! \mid \!\!\!$  & $  \t \piu \t  $ & $\!\!\! \mid \!\!\!$  &  $\sigma_{U,V}^{\piu}$
\end{tabular}
\item $\t_2$ is obtained through the fragment \setlength{\tabcolsep}{4pt}\begin{tabular}{rc ccccccccccccccccccccc}
           $\!\!\! \mid \!\!\!$  &  $\bangp{U}$& $\!\!\! \mid \!\!\!$&  $\id{\zero}$ & $\!\!\! \mid \!\!\!$ & $\id{U}$    & $\!\!\! \mid \!\!\!$   & $   \t ; \t   $ & $\!\!\! \mid \!\!\!$  & $  \t \piu \t  $ & $\!\!\! \mid \!\!\!$  &  $\sigma_{U,V}^{\piu}$
    \end{tabular}
    \item $\t_3$ is obtained through the fragment \begin{tabular}{rc ccccccccccccccccccccc}\setlength{\tabcolsep}{0.0pt}
        $\!\!\! \mid \!\!\!$ & $\tapeFunct{c}$  & $\!\!\! \mid \!\!\!$  &  $\id{\zero}$ & $\!\!\! \mid \!\!\!$ & $\id{U}$&  $\!\!\! \mid \!\!\!$   & $   \t ; \t   $ & $\!\!\! \mid \!\!\!$  & $  \t \piu \t  $ & $\!\!\! \mid \!\!\!$  &  $\sigma_{U,V}^{\piu}$
    \end{tabular}
    \item $\t_4$ is obtained through the fragment \begin{tabular}{rc ccccccccccccccccccccc}
        $\!\!\! \mid \!\!\!$ & $\codiag{U}$ &$\!\!\! \mid \!\!\!$ & $\cobang{U}$ & $\!\!\! \mid \!\!\!$  &  $\id{\zero}$ & $\!\!\! \mid \!\!\!$ & $\id{U}$ & $\!\!\! \mid \!\!\!$   & $   \t ; \t   $ & $\!\!\! \mid \!\!\!$  & $  \t \piu \t  $ & $\!\!\! \mid \!\!\!$  &  $\sigma_{U,V}^{\piu}$
    \end{tabular}
\end{itemize}
\end{lemma}
\begin{proof}
Using the naturality, one can first move all the occurrences of $\diagp{}$ to the left. Then move all the occurrences of $\cobang{}$ to the left until an occurrence of $\diagp{}$ is met (note that since $\cobang{}$ is not the unit of $\diagp{}$, $\cobang{}$ remains at the right of $\diagp{}$). Similarly, one can move all occurrences of $\codiag{}$ and $\cobang{}$ on the right.
\end{proof}

\begin{lemma}\label{lemmaABform}
Let $\t\colon U \to V$ be an arrow of $\CatTapeC$. Then one of the following holds: 
\begin{itemize}
\item $\t=\star_{U,V}$;
\item there exist  $c \in \Cat{C}[U,V]$ such that $\t=\tapeFunct{c}$;
\item there exist distinct $c_1, \dots, c_n \in \Cat{C}[U,V]$ and $\vec{p}=p_1, \dots, p_{n}\in (0,1)$ such that
\[\t=\, \diagpn{\vec{p}\;}{U}{n}
;  (\bigoplus_{i=1}^n  \tapeFunct{c_i})  ; \codiag{V}^n\text{.}\]
\end{itemize}
\end{lemma}
\begin{proof}
By Lemma~\ref{lemma:prenormalform}, $\t= \t_1;\t_2;\t_3;\t_4$.
Using the naturality of $\sigma_{U,V}^{\piu}$ (see Table~\ref{fig:freestricmmoncatax}), one can move all the occurrences of $\sigma_{U,V}^{\piu}$ to the right until an occurrence of $\codiag{V}$ is met. 
Using the associativity and the symmetry of $\diagp{U}$,  all multiple occurrences of every $c_i$ and $\star_{U,V}$ can be collected in the form $\diagp{U};(c_i\otimes c_i)$ and then collapsed to one occurrence of $c_i$ by using the idempotency axiom and axioms for $\codiag{U}$. Then again using the associativity and the symmetry of $\diagp{U}$ one can move all the occurrences of $\diagp{U}$ in the form of $\diagpn{\vec{p}\;}{U}{n}$ for some vector $\vec{p}=(p_1,\dots,p_n)$, and similarly all occurrences of $\codiag{V}$ in the form of $\codiag{V}^n$.
\end{proof}

\begin{lemma}\label{lemma: arrows A-B subdistributions}
    Arrows in $\CatTapeC$ of the form $\t\colon U\to V$ are in bijection with $\mathcal{D}_{\le}(\Cat{C}[U,V])$.
\end{lemma}
\begin{proof}
    By Lemma~\ref{lemmaABform} every arrow $\t\colon U\to V$ is in one of the following forms:
\begin{itemize}
\item $\t=\star_{U,V}$;
\item there exists $c \in \Cat{C}[U,V]$ such that $\t=\tapeFunct{c}$;
\item there exist distinct $c_1, \dots, c_n \in \Cat{C}[U,V]$ and $\vec{p}=p_1, \dots, p_{n}\in(0,1)$ such that 
\[\t=\,\diagpn{\vec{p}\;}{U}{n}
;  (\bigoplus_{i=1}^n  \tapeFunct{c_i})  ; \codiag{V}^n\text{.}\]
\end{itemize}
The first two cases correspond respectively to the zero subdistribution and the Dirac distribution on $c$. The third case corresponds to the subdistribution $\vec{p}$ on the set of arrows $\{c_1, \dots, c_n\}$, where $p_i$ is the probability of $c_i$. 
 If $\t$ and $\t'=\, \diagpn{\vec{q}\;}{U}{n}; (\bigoplus_{i=1}^{n'}  \tapeFunct{c'_i}) ; \codiag{V}^{n'}$ are sent to the same subdistribution, then one has that $n=n'$, $p_i=q_i$ and $c_i=c'_i$ for all $i=1,\dots,n$. 
\end{proof}

\begin{lemma}\label{lemma:division}
Let $\t \colon U \to \bigoplus_{i=1}^n V_i$ be an arrow of $\CatTapeC$. Then, 
\begin{itemize}
\item either $\t= \bang{U}; \cobang{\bigoplus_{i=1}^n V_i}$,
\item or there exists $j\in \{1,\dots, n\}$ and $\t_j \colon U \to V_j$ such that $\t = \cobang{\bigoplus_{i=1}^{j-1} V_i} \oplus \t_j \oplus \cobang{\bigoplus_{i=j+1}^{n} V_i}$,
\item or for all $j\in \{1,\dots, n\}$, there exists $\s_1\colon U \to \bigoplus_{i=1}^{j} V_i$ and $\s_2 \colon U \to \bigoplus_{i=j+1}^{n} V_i$ such that $\t=\,\diagp; (\s_1\oplus \s_2)$.
\end{itemize}
\end{lemma}
\begin{proof}
We have that $\t= \t_1;\t_2;\t_3;\t_4$ where the $\t_i$ are prescribed by Lemma~\ref{lemma:prenormalform}.
We consider two different cases: either $\t_1=\id{U}$ or $\t_1 \neq \id{U}$.

Assume that $\t_1=\id{U}\colon U \to U$. Then $\t_2$ is an arrow with domain $U$ and thus either $\t_2 = \bang{U}$ or $\t_2= \id{U}$. 

\begin{itemize}
\item If $\t_2 = \bang{U} \colon U \to \zero$, then $\t_3$ is an arrow with domain $\zero$. This entails that $\t_3$ should be $\id{\zero}\colon \zero \to \zero$ and $\t_4$ an arrow of type $\zero \to \bigoplus_{i=1}^n V_i$. The latter fact entails that $\t_4 = \cobang{\bigoplus_{i=1}^n V_i}$. Thus $\t= \bang{U}; \cobang{\bigoplus_{i=1}^n V_i}$.
\item If $\t_2 = \id{U} \colon U \to U$, then $\t_3$ should have type $U\to V_j$ and thus it must be either $\id{U}$ or $\tape{c}$. We proceed with the case for $\tape{c}$, the one for $\id{U}$ is identical. Since $\tape{c}\colon U\to V_j$ then $\t_4$ should have type $V_j \to \bigoplus_{i=1}^n V_i$. This entails that $\t_4= \cobang{\bigoplus_{i=1}^{j-1} V_i} \oplus \id{V_j} \oplus \cobang{\bigoplus_{i=j+1}^{n} V_i}$. Thus $\t= \tape{c};\cobang{\bigoplus_{i=1}^{j-1} V_i} \oplus \id{V_j} \oplus \cobang{\bigoplus_{i=j+1}^{n} V_i}$ which, by the laws of symmetric monoidal categories, is equal to $\cobang{\bigoplus_{i=1}^{j-1} V_i} \oplus \tape{c} \oplus \cobang{\bigoplus_{i=j+1}^{n} V_i}$. Thus $\t$ is in the second shape.
\end{itemize}

Let us consider now the case $\t_1 \neq \id{U}$. For any $j\in \{1,\dots, n\}$, $\t_4$ can be written as $\codiag{P}^x \oplus \codiag{Q}^{y}$ where 
$x,y\in \mathbb{N}$, $P=\bigoplus_{i=1}^{j} V_i$ and $Q=\bigoplus_{i=j+1}^{n} V_i$
and  $\codiag{P}^x\colon \bigoplus_{k=1}^x P \to P$ and $\codiag{Q}^y \colon  \bigoplus_{k=1}^yQ\to Q$ are the arrows defined in (\ref{eq:codiagn in appendice}).

Now $\t_3$ should have as codomain $ \bigoplus_{k=1}^x P \oplus \bigoplus_{k=1}^yQ$. Since $\t$ has domain $U$, $\t_1; \t_2$ has type $U \to \bigoplus_{i=1}^z U$ for some $z\in \mathbb{N}$. Since the arrows in $\t_3$ preserve the size of domain and codomains we have that, the domain of $\t_3$ can be written as $ \bigoplus_{k=1}^x(\bigoplus_{i=1}^{j} U) \oplus \bigoplus_{k=1}^y(\bigoplus_{i=j+1}^{n} U)$. This fact entails that $\t_3= \t_3^1 \oplus \t_3^2$ for some $\t_3^1 \colon  \bigoplus_{k=1}^x(\bigoplus_{i=1}^{j} U) \to  \bigoplus_{k=1}^x(\bigoplus_{i=1}^{j} V_i)$ and  $\t_3^2 \colon  \bigoplus_{k=1}^y(\bigoplus_{i=j+1}^{n} U) \to  \bigoplus_{k=1}^y(\bigoplus_{i=j+1}^{n} V_i)$. 

Now $\t_2$ should have as codomain $ \bigoplus_{k=1}^x(\bigoplus_{i=1}^{j} U) \oplus \bigoplus_{k=1}^y(\bigoplus_{i=j+1}^{n} U)$. Since the arrows in the shape of $\t_2$ can only decrease the size of codomain, $\t_2$ has domain $ \bigoplus_{k=1}^x(\bigoplus_{i=1}^{x'} U) \oplus \bigoplus_{k=1}^y(\bigoplus_{i=1}^{y'} U)$ for some $x'\geq j$ and $y'\geq n-j+1$.
This fact entails that $\t_2=\t_2^1 \oplus \t_2^2$ for some $\t_2^1 \colon \bigoplus_{k=1}^x(\bigoplus_{i=1}^{x'} U) \to  \bigoplus_{k=1}^x(\bigoplus_{i=1}^{j} U)$ and $\t_2^2 \colon \bigoplus_{k=1}^y(\bigoplus_{i=1}^{y'} U) \to  \bigoplus_{k=1}^y(\bigoplus_{i=j+1}^{n} U)$.

Now $\t_1$ has type $U \to \bigoplus_{k=1}^x(\bigoplus_{i=1}^{x'} U) \oplus \bigoplus_{k=1}^y(\bigoplus_{i=1}^{y'} U)$. Note that the latter object is different from $U$ as $\t_1\neq \id{U}$. By using \eqref{eq:diagp assoc} and \eqref{eq:diagp symmetry}, $\t_1$ can be written in the form $\diagp{}; (\t_1^1\oplus \t_1^2)$ for some $\t_1^1\colon U \to \bigoplus_{k=1}^x(\bigoplus_{i=1}^{x'} U)$ and $\t_1^2\colon U \to \bigoplus_{k=1}^y(\bigoplus_{i=1}^{y'} U)$.

In summary, $\t=\diagp{};(\, (\t_1^1;\t_2^1; \t_3^1; \codiag{P}^x) \oplus (\t_1^2;\t_2^2; \t_3^2; \codiag{Q}^y) \,)$ is in the third form.

\end{proof}
\begin{proof}[Proof of Lemma~\ref{decomposition}]
Let $\t\colon U \to \bigoplus_{i=1}^n Q_i$ be an arrow of $\CatTapeC$ and observe that every $Q_i=\bigoplus_{j=1}^{m_i} U_j^i$ for some $m_i\in \mathbb{N}$ and $U_j^i$ objects of $\Cat{C}$. Hence, $\bigoplus_{i=1}^n Q_i = \bigoplus_{i=1}^n \bigoplus_{j=1}^{m_i} U_j^i= \bigoplus_{l=1}^s V_l$ for $s=m_1+\dots + m_n$ and $V_l = U_j^i$ if $l=m_1+\dots + m_{i-1}+j$ and  $j\in \{1,\dots,m_i\}$. Now we proceed by induction on $n$. The case $n=0$ is trivial since $\bigoplus_{i=1}^n Q_i$ is the final object $\zero$ and $\t= \bangp{U}=\,\diagpn{\vec{p}\;}{U}{0}$. For the induction case $n+1$, consider $\t\colon U \to Q_1 \oplus \bigoplus_{j=1}^n R_j$ for $R_j=Q_{j+1}$. By Lemma~\ref{lemma:division}, $\t$ is in one of the forms indicated in the statement of Lemma~\ref{lemma:division}. If $\t= \bang{U}; \cobang{\bigoplus_{l=1}^s V_l}= \bang{U}; \cobang{\bigoplus_{i=1}^{n+1} Q_i}$, then taking all $p_i=0$ it follows that $\t=\diagpn{\;\vec{p}}{U}{n+1} ; \bigoplus_{i=1}^{n+1} \star_{U,Q_i}$. Now consider the case in which there exists $k\in \{1,\dots,s\}$ and $\t_k\colon U \to V_k$ such that $\t = \cobang{\bigoplus_{l=1}^{k-1} V_l} \oplus \t_k \oplus \cobang{\bigoplus_{l=k+1}^{s} V_l}$. Assume that $k=m_1+\dots + m_{h-1}+j$ for $j\in \{1,\dots,m_h\}$ (then $V_k= U_j^h$). Hence, $\t= \cobang{\bigoplus_{i=1}^{h-1} Q_i} \oplus \t' \oplus \cobang{\bigoplus_{i=h+1}^{n+1} Q_i}$ where $\t' = \cobang{\bigoplus_{i=1}^{j-1} U_i^h} \oplus \t_k \oplus \cobang{\bigoplus_{i=j+1}^{m_h} U_i^h}$. Taking $p_i=0$ for $i\neq h$ and $p_h=1$, it follows that $\t=\diagpn{\;\vec{p}}{U}{n+1} ; \bigoplus_{i=1}^{h-1} \star_{U,Q_i} \oplus \t' \oplus \bigoplus_{i=h+1}^{n+1} \star_{U,Q_i}$. Now consider the last case and assume that $\t=\,\diagp{U};(\s_1\oplus \s_2)$, with $\s_1\colon U \to Q_1$ and $\s_2\colon U \to \bigoplus_{j=1}^n R_j$. By induction hypothesis, $\s_2$ can be written as $\s_2 = \diagpn{\;\vec{p}}{U}{n} ; \bigoplus_{i=1}^n \s_2^j$ where $\s_2^j\colon U \to Q_{j+1}$ for $j=1,\dots, n$. Hence, $\t =\, \diagp{U};(\s_1\oplus \s_2) = \diagp{U};(\id{U} \oplus \diagpn{\;\vec{p}}{U}{n}) ; (\s_1 \oplus \bigoplus_{i=1}^n \s_2^j)$. Taking $t_1=p$ and $t_i = (1-p)\cdot p_{i-1}$ for $i=2,\dots, n+1$, it follows that $\t=\diagpn{\;\vec{t}}{U}{n+1} ; \bigoplus_{i=1}^{n+1}\s'_i$, where $\s'_1 = \s_1$ and $\s'_i = \s_2^{i-1}$ for $i=2,\dots, n+1$.

%
%
\end{proof}

\subsection{Cancellativity}\label{app:cancellativity}
In order to prove cancellativity, we first consider the case of arrows of type $U \to V$ where $U,V$ are objects of $\Cat{C}$.
\begin{lemma}\label{lemma:cancellativity11}
For all $r\in (0,1)$, for all $\s,\t \colon U \to V$, 
if $r\cdot \s = r\cdot \t$ then $\s = \t$.
\end{lemma}
\begin{proof}
By Lemma \ref{lemma: arrows A-B subdistributions}, we know $\CatTapeC[U,V]= \subdistr(\Cat{C}[U,V])$. It is well known (see e.g. \cite{sokolova2018termination}) that, for all sets $X$, $\subdistr(X)$ is a cancellative pca.
\end{proof}

Then, we consider the case of arrows of type $U \to \bigoplus_{i=1}^n V_i$ where $U,V_i$ are objects of $\Cat{C}$.
\begin{lemma}\label{lemma:cancellativityhalf}
For all $n \in \mathbb{N}$, for all $r\in (0,1)$, for all $\s,\t \colon U \to \bigoplus_{i=1}^n V_i$, 
if $r\cdot \s = r\cdot \t$ then $\s = \t$.
\end{lemma}
\begin{proof}
We proceed by induction on $n$. 

The base case $n=0$ is trivial since $\bigoplus_{i=1}^n V_i$ is the final object $\zero$.

For the induction case $n+1$, we exploit Lemma~\ref{lemma:division} and we consider the case where
\begin{equation}\label{eq:st}
\s=\, \diagq{} ; (\s_1 \oplus \s_2)   \qquad \text{ and } \qquad \t =\, \diagp{}; (\t_1 \oplus \t_2)
\end{equation}
for $\s_1, \t_1 \colon U \to \bigoplus_{i=1}^n V_i$ and $\s_2,\t_2 \colon U \to V_{n+1}$. The other cases are trivial.

By Lemma \ref{lemma:basicTC}.3, we have that
\begin{equation}\label{eq:simple}
r\cdot \s=\, \diagq{} ; (r\cdot \s_1 \oplus r\cdot \s_2)   \qquad \text{ and } \qquad r\cdot  \t =\, \diagp{}; (r\cdot \t_1 \oplus r\cdot \t_2)
\end{equation}
Since by hypothesis $r\cdot \s = r \cdot \t$, it holds both  
\[r\cdot \s ; (\id{} \oplus \bang{}) = r \cdot \t ; (\id{} \oplus \bang{}) \qquad \text{ and } \qquad r\cdot \s ; ( \bang{} \oplus \id{}) = r \cdot \t ;( \bang{} \oplus \id{})\text{.}\]
Thus, by \eqref{eq:simple} and Lemma~\ref{lemma:basicTC}.1 and 2, one obtains that
\begin{equation}\label{eq:twoeq}q\cdot r \cdot \s_1 = p \cdot r \cdot \t_1 \qquad \text{ and } \qquad (1-q) \cdot r \cdot \s_2 = (1-p) \cdot r \cdot \t_2\text{.}\end{equation}

There are now three cases: $p=q$, $p>q$ and $p<q$.

If $p=q$, then  $\s_1=\t_1$ by induction hypothesis and  $\s_2=\t_2$ by Lemma~\ref{lemma:cancellativity11}. Then, by \eqref{eq:st}, $\s=\t$.

If $p>q$, we can rewrite \eqref{eq:twoeq} as
\begin{equation*}p\cdot r \cdot \frac{q}{p} \cdot \s_1 = p \cdot r \cdot \t_1 \qquad \text{ and } \qquad (1-q) \cdot r \cdot \s_2 = (1-q) \cdot r \cdot \frac{1-p}{1-q} \cdot \t_2\text{.}\end{equation*}
Thus, by induction hypothesis $ \frac{q}{p} \cdot \s_1 = \t_1$ and, by Lemma~\ref{lemma:cancellativity11}, $\s_2 =\frac{1-p}{1-q} \cdot \t_2$. By \eqref{eq:st} and Lemma~\ref{lemma:fractions}, we have that $\s=\t$.
The case $p<q$ is symmetrical.
\end{proof}

We can finally consider the general case.

\begin{proof}[Proof of Lemma~\ref{cancellativity}]
Let $P=\bigoplus_{i=1}^n U_i$. We proceed by induction on $n$. The base case $n=0$ is trivial since $\bigoplus_{i=1}^n U_i$ is the initial object $\zero$.

For the induction case $n+1$, since $\CatTapeC$ has finite coproducts, one can assume without loss of generality that
\begin{equation}\label{eq:s1s2coproduct}
\s = (\s_1 \oplus \s_2) ; \codiag{Q} \qquad \text{ and } \t = (\t_1 \oplus \t_2) ; \codiag{Q}
\end{equation}
for some $\s_2,\t_2 \colon U_{n+1} \to Q$ and  $\s_1,\t_1\colon P' \to Q$ where $P'=\bigoplus_{i=1}^n U_i$. Note moreover that the coproduct universal property entails that
\begin{equation}\label{eq:ifflocal}
\s=\t \text{ iff } \s_1=\t_1 \text{ and } \s_2=\t_2\text{.}
\end{equation}

By Lemma~\ref{lemma:pcopairing}, we have that
$r \cdot \s = (r\cdot \s_1 ) \oplus (r\cdot  \s_2 ); \codiag{Q}$ and $r \cdot \t = (r\cdot \t_1 ) \oplus (r\cdot  \t_2 ); \codiag{Q} $. Thus, by \eqref{eq:ifflocal} we have that
 \[r\cdot \s_1= r\cdot \t_1 \qquad \text{ and }\qquad r\cdot \s_2= r\cdot \t_2 \]
 From the rightmost equality above and Lemma~\ref{lemma:cancellativityhalf}, we derive that $\s_2=\t_2$. From the leftmost equality above and induction hypothesis, we derive that $\s_1=\t_1$. By \eqref{eq:ifflocal}, it holds that $\s=\t$.
 \end{proof}

\subsection{Jointly Monic Projections}

\begin{lemma}\label{keyLemma001}
Let $\t\colon U \to Q \oplus R$ be an arrow in $\CatTapeC$. 
\begin{enumerate} 
\item If $\t; (\id{Q} \oplus\, \bangp{R}) =\s$
and $\t; (\,\bang{Q} \oplus \id{R}) = \bang{U}; \cobang{R}$ then $\t =\s \oplus \cobang{R}$;
\item If $\t; (\id{Q} \oplus\, \bangp{R}) =\bang{U}; \cobang{Q}$  
and $\t; (\,\bang{Q} \oplus \id{R}) = \s$ then $\t = \cobang{Q} \oplus \s$;
\end{enumerate}
\end{lemma}
\begin{proof}
We prove the first point: the second is symmetrical.

First, observe that $\t$ should be in one of the forms of Lemma~\ref{lemma:division}. The only challenging case is the third form (iii) $\t=\,\diagp{U}; (\t_1 \oplus \t_2)$ for some $\t_1\colon U \to Q$ and $\t_2\colon U \to R$. Since
\begin{align*}
\bang{U}; \cobang{R} &= \t; (\,\bang{Q} \oplus \id{R}) \tag{Hypothesis}\\
&=\, \diagp{U}; (\t_1 \oplus \t_2) ; (\,\bang{Q} \oplus \id{R}) \tag{iii}\\
&=\, \diagp{U}; ((\t_1 ; \bang{Q}) \oplus \t_2) \tag{Symmetric Monoidal Category}\\
&=\, \diagp{U} ; (\bang{U} \oplus \t_2) \tag{\ref{eq:bangp nat}}\\
&= (1-p)\cdot \t_2 \tag{Lemma \ref{lemma:basicTC}.2}
\end{align*}
since also $(1-p)\cdot (\bang{U}; \cobang{R})=\bang{U}; \cobang{R}$, by cancellativity (Lemma~\ref{cancellativity}), we have that $\t_2 = \bang{U}; \cobang{R}$.
Thus
\begin{align*}
\s \oplus \cobang{R} & = (\t; (\id{Q} \oplus\, \bangp{R})) \oplus \cobang{R} \tag{Hypothesis}\\
&= (\diagp{U}; (\t_1 \oplus \t_2) ; (\id{Q} \oplus\, \bangp{R})) \oplus \cobang{R} \tag{iii}\\
&= (\diagp{U}; (\t_1 \oplus (\t_2 ; \bang{R})))  \oplus\, \cobang{R} \tag{Symmetric Monoidal Category}\\
&= (\diagp{U}; (\t_1 \oplus ( \bang{U}; \cobang{R} ; \bang{R})))  \oplus\, \cobang{R} \tag{$\t_2 = \bang{U}; \cobang{R}$}\\
&= (\diagp{U}; (\t_1 \oplus  \bang{U}) ) \oplus \cobang{R} \tag{\ref{eq:bangp nat}}\\
&=\, \diagp{U}; (\t_1 \oplus  (\bang{U} ; \cobang{R}))\tag{Symmetric Monoidal Category}\\
&=\, \diagp{U}; (\t_1 \oplus  \t_2)\tag{$\t_2 = \bang{U}; \cobang{R}$}\\
&= \t \tag{iii}\\
\end{align*}
\end{proof}

\begin{lemma}\label{lemma:comvexproductpreliminary}
Let $\t_1\colon U\to P$ and $\t_2\colon U \to Q$ be arrows of $\CatTapeC$ and $p\in (0,1)$. Then $\diagp{U}; (\t_1 \oplus \t_2) \colon U \to P\oplus Q$ is the unique arrow $\t$ such that
$p\cdot \t_1 = \t; (\id{P} \oplus \bang{Q})$ and $(1-p)\cdot \t_2 = \t; (\bang{P} \oplus \id{Q})$.
\end{lemma}
\begin{proof}
First observe that by Lemma~\ref{lemma:basicTC}.1 and 2 we have that $\diagp{U}; (\t_1 \oplus \t_2) ; (\id{P} \oplus \bang{Q}) = p\cdot \t_1$ and that $\diagp{U}; (\t_1 \oplus \t_2) ; (\bang{P} \oplus \id{Q}) = (1-p) \cdot \t_2$.

We now need to prove that $\diagp{U}; (\t_1 \oplus \t_2)$ is the unique arrow with such property.  We prove that if (a) $p\cdot \t_1 = \t; (\id{P} \oplus \bang{Q})$ and (b) $(1-p)\cdot \t_2 = \t; (\bang{P} \oplus \id{Q})$ then $\t =\, \diagp{U}; (\t_1 \oplus \t_2)$.

Note that $\t$ can be in one of the forms of Lemma~\ref{lemma:division}. The only relevant case is the last one: $\t =\, \diagpX{p'}{U} ; (\t_1' \oplus 	\t_2')$. By (a), (b) Lemma~\ref{lemma:basicTC}.1 and 2, it holds that 
\begin{center}$p\cdot \t_1 = p' \cdot \t_1'$ and  $(1-p)\cdot \t_2 = (1-p') \cdot \t_2'$. \end{center}

There are now three cases:
\begin{itemize}
\item If $p=p'$, then by Lemma~\ref{cancellativity}, $\t_1=\t_1'$ and $\t_2=\t_2'$. Thus, $\t =\, \diagp{U}; (\t_1 \oplus \t_2) $.
\item If $p>p'$, then we can rewrite the above equalities as  $p\cdot \t_1 = p \cdot \frac{p'}{p} \cdot \t_1'$ and $(1-p')\cdot\frac{(1-p)}{1-p'}\cdot \t_2 = (1-p') \cdot \t_2'$. By Lemma~\ref{cancellativity}, we have that $\t_1=\frac{p'}{p}\cdot\t_1'$ and $\frac{1-p}{1-p'} \cdot \t_2 =\t_2'$. Thus, by Lemma~\ref{lemma:fractions}, $\t=\, \diagp{U}; (\t_1 \oplus \t_2)$.
\item The case $p<p'$ is symmetrical to the one above.
\end{itemize}
\end{proof}

\begin{corollary}\label{cor:finaleTC}
    Let $\t_1\colon U\to P$ and $\t_2\colon U \to Q$ be arrows of $\CatTapeC$ and $p,q\in [0,1]$ such that $p+q\le 1$. Then $\diaggen{p,q}{U};(\t_1 \oplus \t_2) \colon U \to P\oplus Q$ is the unique arrow $\t$ such that  $p\cdot \t_1 = \t; (\id{P} \oplus \bang{Q})$ and $q\cdot \t_2 = \t; (\bang{P} \oplus \id{Q})$.
\end{corollary}
\begin{proof}
    The case where $p$ or $q$ is equal to $1$ follows from Lemma~\ref{keyLemma001} and the fact that $\diaggen{1,0}{U} =\, \diaggen{1}{U}= \id{U}\oplus \cobang{U}$ and $\diaggen{0,1}{U} =\, \diaggen{0}{U} = \cobang{U}\oplus \id{U}$. The case $p,q\in (0,1)$ and $p+q=1$ follows from Lemma~\ref{lemma:comvexproductpreliminary} and the fact that $\diaggen{p,q}{U} =\, \diagp{U}$. The case where $p,q\in (0,1)$ and $p+q<1$ follows from the previous case and the fact that $\diaggen{p,q}{U};(\t_1\oplus \t_2) =\, \diaggen{p,1-p}{U}; (\t_1 \oplus \frac{q}{1-p}\cdot \t_2)$.
\end{proof}

\begin{proof}[Proof of Lemma~\ref{lemma:TC jmono}]
     First, we consider two arrows $\t,\s\colon U \to Q\oplus R$ where $U$ is an object of $\Cat{C}$ and assume that $\t;\pi_1 = \s;\pi_1$ and $\t;\pi_2 = \s;\pi_2$. By Lemma~\ref{decomposition}, there exist $\t_1,\s_1\colon U \to Q$, $\t_2,\s_2\colon U \to R$ and $p,q,p',q'\in [0,1]$ where $p+q\le 1$ and $p'+q'\le 1$ such that
     \[\t=\, \diaggen{p,q}{};(\t_1 \oplus \t_2) \text{ and }\s=\,\diaggen{p',q'}{};(\s_1 \oplus \s_2)\text{.}\]
The last two items of Lemma~\ref{lemma:diagpq} guarantee that $\t;\pi_1 = p\cdot \t_1$ and $\t;\pi_2 = q\cdot \t_2$ and $\s;\pi_1 = p'\cdot \s_1$ and $\s;\pi_2 = q'\cdot \s_2$. Since, by hypothesis, $\t;\pi_1 = \s;\pi_1$ and $\t;\pi_2 = \s;\pi_2$, we have that $p\cdot \t_1 = p'\cdot \s_1$ and $q\cdot \t_2 = q'\cdot \s_2$. 
    Finally, Corollary~\ref{cor:finaleTC} implies that $\t=\s$.

    For the general case $\t,\s\colon P\to Q\oplus R$ where $P=\bigoplus_{i=1}^n U_i$, we rely on the fact that $\CatTapeC$ has coproducts. Indeed, by the universal property of coproducts, we can write $\t = [\t_1,\dots,\t_n]$ and $\s = [\s_1,\dots,\s_n]$ for some $\t_i,\s_i\colon U_i \to Q\oplus R$. Now, from $\t;\pi_1 = \s;\pi_1$ and $\t;\pi_2 = \s;\pi_2$ it follows that $\t_i;\pi_1 = \s_i;\pi_1$ and $\t_i;\pi_2 = \s_i;\pi_2$ for all $i=1,\dots,n$. By the previous case, we have that $\t_i=\s_i$ for all $i=1,\dots,n$ and hence $\t=\s$.
\end{proof}

\subsection{Proof of Proposition~\ref{cor: quotient category is convex biproduct }}
We conclude this appendix with a proof of Proposition~\ref{cor: quotient category is convex biproduct } that turns out to be useful when considering diagrammatic languages (see Section~\ref{sec:probbooltapes}). First, we need the following result.

\begin{lemma}\label{lemma: convex products quoziente}
Let $\Cat{C}$ be a category and let $\sim$  be a congruence relation on $\CatTapeC$ (w.r.t. $;$ and $\oplus$) that is cancellative: for every $p\in (0,1)$, if $p\cdot f\sim p\cdot g$ then $f\sim g$. For all $f_1,g_1\colon Z\to X_1$, $f_2,g_2\colon Z\to X_2$ and $p_1,p_2,q_1,q_2\in[0,1]$ such that $p_1+p_2\le 1$ and $q_1+q_2\le 1$
\begin{center}if $p_1\cdot f_1 \sim q_1\cdot g_1$ and $p_2\cdot f_2\sim q_2\cdot g_2$  then $\langle f_1,f_2\rangle_{(p_1,p_2)}\sim\langle g_1,g_2\rangle_{(q_1,q_2)}$.\end{center}
\end{lemma}

\begin{proof}\label{proof-lemma: convex products quoziente}
We first consider the case where $p_1,p_2,q_1,q_2\in(0,1)$. We assume that $p_1\leq q_1$  and $q_2\leq p_2$.  The cases of the other inequalities are proved similarly.

First observe that
\begin{align*}
(1-p_1)\cdot (\frac{p_2}{1-p_1}\cdot f_2) &= p_2\cdot f_2 \notag \\
& \sim q_2\cdot g_2 \tag{Hypothesis}\\
&= (1-p_1)\cdot (\frac{q_2}{1-p_1}\cdot g_2) \notag
\end{align*}
Thus, by cancellativity of $\sim$, it holds that $\frac{p_2}{1-p_1}\cdot f_2 \sim \frac{q_2}{1-p_1}\cdot g_2$ which, by Lemma \ref{lemma:initialproperties}, means
\begin{equation}\tag{$\dagger$}\label{eq:internal1}
\diagpX{\frac{p_2}{1-p_1}\,}{};(f_2\piu\, \bangp{}) \sim \diagpX{\frac{q_2}{1-p_1}\,}{};(g_2\piu\, \bangp{}) 
\end{equation}
One can similarly prove that 
\begin{equation}\tag{$\ddagger$}\label{eq:internal2}
\diagpX{\frac{p_1}{1-q_2}\,}{};(f_1\piu\, \bangp{}) \sim \diagpX{\frac{q_1}{1-q_2}\,}{};(g_1\piu\, \bangp{}) 
\end{equation}

We conclude with the following derivation. 
  \begin{align*}
        \langle f_1,f_2\rangle_{(p_1,p_2)}&=\, \diagpX{p_1}{}; (f_1\piu\, (\diagpX{\frac{p_2}{1-p_1}\,}{};(f_2\piu\, \bangp{})) ) \tag{Lemma~\ref{lemma:diagpq3}}\\
        &\sim \, \diagpX{p_1}{}; (f_1\piu\, (\diagpX{\frac{q_2}{1-p_1}\,}{};(g_2\piu\, \bangp{})) ) \tag{\eqref{eq:internal1}}\\
& = \, \diagpX{q_2}{};(g_2\piu\, (\diagpX{\frac{p_1}{1-q_2}\,}{};(f_1\piu\, \bangp{})) ) \tag{\ref{eq:diagp assoc} + \ref{eq:diagp symmetry}}  \\
&\sim \, \diagpX{q_2}{};(g_2\piu\, (\diagpX{\frac{q_1}{1-q_2}\,}{};(g_1\piu\, \bangp{})) ) \tag{\eqref{eq:internal2}}\\
 &=\,  \diagpX{q_1}{}; (g_1\piu\, (\diagpX{\frac{q_2}{1-q_1}\,}{};(g_2\piu\, \bangp{})) ) \tag{\ref{eq:diagp assoc} + \ref{eq:diagp symmetry}}\\
 &=  \langle g_1,g_2\rangle_{(q_1,q_2)}  \tag{Lemma~\ref{lemma:diagpq3}}
  \end{align*}

We now prove the case where $p_1=1$ and $p_2=0$. We have two sub-cases: either $q_1=1$ or $q_1\neq 1$. 
 If $q_1=1$, then by Lemma~\ref{lemma:diagpq3} easily follows that $\langle f_1,f_2\rangle_{(1,0)}=f_1;\iota_1$ and $\langle g_1,g_2\rangle_{(1,0)}=g_1;\iota_1$. Since by hypothesis 
 $f_1 = 1\cdot f_1= p_1\cdot f_1 \sim q_1\cdot g_1 = 1\cdot g_1 =g_1$, then $f_1;\iota_1 \sim g_1;\iota_1$.

If $q_1\not=1$, by hypothesis, we have that $q_2\cdot g_2\sim p_2\cdot f_2=0 \cdot f_2 = \star = q_2\cdot\star$. If $q_2\not=0$, by cancellativity we have, 
\begin{equation}\tag{$\heartsuit$}\label{eq:heart}
g_2\sim\star \text{.}\end{equation}
The following derivation
    \begin{align}
        \langle g_1,g_2\rangle_{(q_1,q_2)}&= \, \diagpX{q_1}{}; (g_1\piu\, (\diagpX{\frac{q_2}{1-q_1}\,}{};(g_2\piu\, \bangp{})) ) \tag{Lemma~\ref{lemma:diagpq3}}\\
        &\sim \, \diagpX{q_1}{}; (g_1\piu\, (\diagpX{\frac{q_2}{1-q_1}\,}{};(\star \piu\, \bangp{})) ) \tag{\eqref{eq:heart}}\\
        &= \, \diagpX{q_1}{}; (g_1\piu\, (\diagpX{\frac{q_2}{1-q_1}\,}{};(\,\bangp{};\cobang{} \piu\, \bangp{})) ) \tag{Def. $\star$}\\
        &=\, \diagpX{q_1}{}; (g_1\piu\, (\,\bangp{};(\,\cobang{}\piu \id{0}) ))\tag{\ref{eq:diagp nat}}\\
        &= \, (\diagpX{q_1}{}; (g_1\piu\, \,\bangp{})); (\id{} \oplus \cobang{} )\tag{SMC} \\
        &= (q_1\cdot g_1) ; \iota_1 \tag{Lemma~\ref{lemma:initialproperties}}
    \end{align}
proves that $\langle g_1,g_2\rangle_{(q_1,q_2)}\sim q_1\cdot g_1;\iota_1$. By means of Lemma~\ref{lemma:diagpq3} and Lemma~\ref{lemma:initialproperties}, one can easily prove that $\langle f_1,f_2\rangle_{(p_1,p_2)}=\iota_1^{Z\oplus Z};(f_1\oplus f_2) =f_1;\iota_1 $. Now, by hypothesis, $f_1= 1\cdot f_1 = p_1\cdot f_1 \sim q_1\cdot g_1$. We can thus conclude that $\langle g_1,g_2\rangle_{(q_1,q_2)}\sim \langle f_1,f_2\rangle_{(p_1,p_2)}$. The remaining cases are proved similarly.
\end{proof}

\begin{proof}[Proof of Proposition~\ref{cor: quotient category is convex biproduct }]
   By construction, every object in $\CatTapeC$ carries a commutative and coherent monoid and co-pca structure. Since $\sim$ is a congruence relation, the equational axioms in Table~\ref{fig:freestrictfccat} and \ref{fig:freecopcacat}  hold in $\CatTapeC_\sim$. Thus, thanks to Proposition~\ref{prop:A}, to prove that $\CatTapeC_\sim$ is a convex biproduct category, it is enough to prove that it has jointly monic projections.

First, we consider two arrows of the form $[\t]_\sim,[\s]_\sim\colon U\to Q\oplus R$ where $U$ is an object of $\Cat{C}$. We assume that $[\t]_\sim;[\pi_1]_\sim = [\s]_\sim;[\pi_1]_\sim$ and $[\t]_\sim;[\pi_2]_\sim = [\s]_\sim;[\pi_2]_\sim$. By Lemma~\ref{decomposition}, $\t=\,\diaggen{p,q}{};(\t_1\oplus \t_2)$ and $\s=\,\diaggen{p',q'}{};(\s_1\oplus \s_2)$. By the assumption and Lemma~\ref{lemma:diagpq}, it holds that $p\cdot \t_1 \sim p'\cdot \s_1$ and $q\cdot \t_2 \sim q'\cdot \s_2$. Hence, by Lemma~\ref{lemma: convex products quoziente}, $\langle \t_1,\t_2\rangle_{(p,q)}\sim\langle \s_1,\s_2\rangle_{(p',q')}$ and thus $[\t]_\sim=[\s]_\sim$.
    
 The general case $[\t]_\sim,[\s]_\sim\colon P\to Q\oplus R$ where $P=\bigoplus_{i=1}^n U_i$ is proved similarly by relying on the fact that $\CatTapeC_\sim$  also has coproducts thanks to Fox's Theorem. Indeed, by the universal property of coproducts, we can write $\t = [\t_1,\dots,\t_n]$ and $\s = [\s_1,\dots,\s_n]$ for some $\t_i,\s_i\colon U_i \to Q\oplus R$. Now, from $[\t]_\sim;[\pi_1]_\sim = [\s]_\sim;[\pi_1]_\sim$ and $[\t]_\sim;[\pi_2]_\sim = [\s]_\sim;[\pi_2]_\sim$ it follows that $[\t_i]_\sim;[\pi_1]_\sim = [\s_i]_\sim;[\pi_1]_\sim$ and $[\t_i]_\sim;[\pi_2]_\sim = [\s_i]_\sim;[\pi_2]_\sim$ for all $i=1,\dots,n$. By the previous case, we have that $[\t_i]_\sim=[\s_i]_\sim$ for all $i=1,\dots,n$ and hence $[\t]_\sim=[\s]_\sim$, given that $[[\t_1,\dots,\t_n]]_\sim=[[\t_1]_\sim,\dots,[\t_n]_\sim]$ and $[[\s_1,\dots,\s_n]]_\sim=[[\s_1]_\sim,\dots,[\s_n]_\sim]$.

 The functor $Q_\sim\colon \CatTapeC\to \CatTapeC_\sim$ is a morphism of convex biproduct categories since it preserves finite coproducts.
\end{proof}

%% file: appendices/appprobtapes.tex
\section{Appendix to Section \ref{sec:tapediagrams}}\label{app:sec:tapediagrams}

\begin{proof}[Proof of Proposition~\ref{prop:quotienttheory}]
Let $\CatT{\Diag{\Sigma}}_{\tapeeqT}$ be the quotient of $\CatT{\Diag{\Sigma}}$ by $\tapeeqT$: objects are those of  $\CatT{\Diag{\Sigma}}$, arrows are $\tapeeqT$-equivalence classes of arrows of $\CatT{\Diag{\Sigma}}$, hereafter denoted by $[\t]_{\tapeeqT}$. We prove that $\CatT{\Diag{\Sigma}}_{\tapeeqT}$ is isomorphic to $\CatT{\Diag{\mathbb{T}}}$.

We first define the functor $F \colon \CatT{\Diag{\Sigma}}_{\tapeeqT} \to \CatT{\Diag{\mathbb{T}}}$ inductively as
\[\begin{array}{ccccc}
F(\diagp{U})\defeq \; \diagp{U} &  F(\,\bangp{U})\defeq \bangp{U} & F(\tapeFunct{c})\defeq \tapeFunct{[c]_{\mathbb{T}}} &  F(\cobang{U})\defeq \cobang{U} &
F(\codiag{U})\defeq\codiag{U}  \\
 F(\id{U})\defeq \id{U} & F(\id{\zero})\defeq\id{\zero} & F(\sigma_{U,V}^{\piu})\defeq \sigma_{U,V}^{\piu}& F(\t_1 ; \t_2)\defeq F(\t_1);F(\t_2) & F(\t_1 \piu \t_2) \defeq F(\t_1)\piu F(\t_2)
\end{array}\]
Since arrows in $\CatT{\Diag{\Sigma}}_{\tapeeqT}$ are $\tapeeqT$-equivalence classes, we need to prove that the functor is well-defined, namely that if $\s \tapeeqT  \t$ then $F(\s)=F(\t)$. This is done by induction on the rules in \eqref{eq:congr2}. For the case of the rule ($\mathbb{T}$), observe that $\s=\tape{c}$, $\t=\tape{d}$ and $c=_{\mathbb{T}}d$. By the definition of $F$, $F(\s)=F(\t)$. All the other cases are trivial, with the only exception of the rule $(\otimes)$ where one observes that, since  $=_{\mathbb{T}}$ is a congruence w.r.t. $\otimes$, then $F$ is a morphism of rig categories.

The functor  $G \colon  \CatT{\Diag{\mathbb{T}}} \to \CatT{\Diag{\Sigma}}_{\tapeeqT} $ is defined on $\tapeFunct{[c]_{\mathbb{T}}}$ as $G(\tapeFunct{[c]_{\mathbb{T}}})\defeq [\tapeFunct{c}]_{\tapeeqT}$ and then inductively in the same style of the functor $F$ above. Observe that $F(G(\tapeFunct{[c]_{\mathbb{T}}}))=\tapeFunct{[c]_{\mathbb{T}}}$ and $G(F([\tapeFunct{c}]_{\tapeeqT})) = [\tapeFunct{c}]_{\tapeeqT}$. A simple inductive argument confirms that $F$ and $G$ are inverses of each other.
\end{proof}

%% file: appendices/appprobboolcircuits.tex
\section{Appendix to Section \ref{sec:probboolcircuits}}\label{app:sec:probbooleancircuits}

\begin{proof}[Proof of Lemma~\ref{lemma:multiplexer}] 
  \eqref{ax:M1} and \eqref{ax:M2} can be proved by a simple inductive argument on $m$. \eqref{ax:M3} can be proved by a direct computation with the AND and OR gates. \eqref{ax:N1} and \eqref{ax:N2} follow by induction on the structure of $c$. For the base case, if $c$ is a Boolean gate or one of the structural maps $\id{1},\id{A},\symmt{A}{A}$, then the statement follows from axioms \eqref{ax:C1}, \eqref{ax:C2}, \eqref{ax:C3}  and \eqref{ax:D1}, \eqref{ax:D2}, \eqref{ax:D3} in Figure~\ref{tab:booleanalgebra}. In the case $c;d$, then assuming that the statement holds for $c$ and $d$, we have
  \begin{center}
    
    \InputIfFileExists{tapes/cipriano/booleanscopyable.tikz}{}{\input{./tikz/tapes/cipriano/booleanscopyable.tikz}}

  \end{center}
  \begin{center}
    
    \InputIfFileExists{tapes/cipriano/booleansdiscardable.tikz}{}{\input{./tikz/tapes/cipriano/booleansdiscardable.tikz}}

  \end{center}
  Finally, in the case $c\otimes d$, assuming that the statement holds for $c$ and $d$, we have
  \begin{center}
    
    \InputIfFileExists{tapes/cipriano/booleanscopyable2.tikz}{}{\input{./tikz/tapes/cipriano/booleanscopyable2.tikz}}

  \end{center}
  \begin{center}
    
    \InputIfFileExists{tapes/cipriano/booleansdiscardable2.tikz}{}{\input{./tikz/tapes/cipriano/booleansdiscardable2.tikz}}

   \end{center}
\end{proof}

\begin{proof}[Proof of Proposition~\ref{prop:soundnesspartialb}]
It is enough to check that the statement holds for the axioms in Figure~\ref{tab:partialbooleanalgebra}. Those in the first two lines are standard. For the third row, simple computations confirm that $\osemPB{-}$ maps the left hand side of \eqref{ax:F6} into $\osemPB{\CBcocopier}$ as defined in \eqref{eq:semcocopier}. For \eqref{ax:F7}, one readily checks that $\osemPB{-}$ maps both the left and the right hand side to the partial function $2\times 2 \to 1$ defined as
\[(x,y)\mapsto \begin{cases}\bullet &\text{if }x=y=1\\ \bot & \text{else}\end{cases}\]
\end{proof}

\begin{lemma}\label{lemma:complete2}
  Let $b\in \Diag{\SigB}[1,A^n]$ and $c\in \Diag{PB}[A^n,A^m]$.
  \begin{itemize}
    \item[i.] If $\osemPB{b;c}=\bot$ , then $\osemPB{b;D_c} = 0$;
    \item[ii.] If $\osemPB{b;c}=\osemPB{b'}$ for some $b'\in \Diag{\SigB}[1,A^m]$, then $\osemPB{b;D_c}=1$ and $\osemPB{b;T_c}=\osemPB{b'}$.
  \end{itemize}
\end{lemma}
\begin{proof}
Since $D_c \in \Diag{\SigB}[A^n,1]$, then $\osemPB{b;D_c} = \osemB{b;D_c}$, i.e., it is a (total) function of type $1 \to 2$. Hence $\osemPB{b;D_c}$ is either $0$ or $1$. 

Suppose that $\osem{b;D_c}=1$, or, equivalently, $b;D_c\eqPB\Flip{1}$ thanks to Theorem~\ref{thm:completenessbooleancircuits}. Then
  \begin{align}
    \osemPB{b;c}&=  \osemPB{b;\ncopier;((D_c;\coflip{1})\otimes T_c)} \tag{Prop. \ref{prop:decompositionpartialcircuits}}\\
    &= \osemPB{(b\otimes b);((D_c;\coflip{1})\otimes T_c)} \tag{\eqref{ax:N2}}\\
    &= \osemPB{(b;D_c;\coflip{1})\otimes (b;T_c)} \tag{SMC}\\
    &= \osemPB{(\Flip{1};\coflip{1})\otimes (b;T_c)} \tag{$b;D_c\eqPB\Flip{1}$}\\
    &= \osemPB{b;T_c}. \tag{\eqref{ax:F10l}}
\end{align}
Instead, suppose that $\osem{b;D_c}=0$, or, equivalently, $b;D_c\eqPB\Flip{0}$ thanks to Theorem~\ref{thm:completenessbooleancircuits}. Then, 
\begin{align}
    \osemPB{b;c}  &= \osemPB{b;\ncopier;((D_c;\coflip{1})\otimes T_c)} \tag{Prop. \ref{prop:decompositionpartialcircuits}}\\
    &= \osemPB{(b\otimes b);((D_c;\coflip{1})\otimes T_c)} \tag{\eqref{ax:N2}}\\
    &= \osemPB{(b;D_c;\coflip{1})\otimes (b;T_c)} \tag{SMC}\\
    &= \osemPB{(\Flip{0};\coflip{1})\otimes (b;T_c)} \tag{$b;D_c\eqPB\Flip{0}$}\\
    &= \bot \notag
\end{align}

To prove i.,  we assume that $\osemPB{b;c}=\bot$ and $\osem{b;D_c}=1$ and  immediately obtain a contradiction: by the first derivation above $\osemPB{b;T_c}=\osemPB{b;c}=\bot$, but this is not possible since $b;T_c \in \Diag{\SigB}[1,A^m]$, i.e., it denotes a total Boolean function. Hence, if  $\osemPB{b;c}=\bot$, then $\osem{b;D_c}$ should be $0$.

Now, to prove ii., we first assume that $\osemPB{b;c}=\osemPB{b'}$ and $\osem{b;D_c}=0$ and immediately obtain a contradiction: by the second derivation above $\osemPB{b;c}=\bot$. Hence, if  $\osemPB{b;c}=\osemPB{b'}$, then $\osem{b;D_c}$ should be $1$. By the first derivation $\osemPB{b;T_c}=\osemPB{b;c}=\osemPB{b'}$.

\end{proof}

\begin{proof}[Proof of Lemma~\ref{lemma:complete3}]
There are two cases: either $\osemPB{b;c}=\osemPB{b;d}=\bot$, or $\osemPB{b;c}=\osemPB{b;d}\neq \bot$. In the first case, by Lemma~\ref{lemma:complete2} (i), we have $\osemPB{b;D_c}=0=\osemPB{b;D_d}$, and by Theorem~\ref{thm:completenessbooleancircuits} we have $b;D_c\eqPB b;D_d\eqPB\Flip{0}$. Hence, Lemma~\ref{lemma:complete1} implies ${b;T_c}\eqPB\Flipm{0}={b;T_d}$. In the second case, by Lemma~\ref{lemma:complete2} (ii), we have $\osemPB{b;D_c}=1=\osemPB{b;D_d}$ and $\osemPB{b;T_c}=\osemPB{b'}=\osemPB{b;T_d}$ for some $b'\in \Diag{\SigB}[1,A^m]$.
  \end{proof}

  \begin{proof}[Proof of Proposition~\ref{prop:isopartialboolean}]
Recall that $\osemPB{-}\colon \Diag{{\SigPB}} \to \Cat{Par}$ is defined on objects as $\osemPB{A^n}=2^n$. Hence for all $c\in \Diag{{\SigPB}}[A^n,A^m]$, $\osemPB{c}\in \Cat{Par}[2^n,2^m]$, namely 
 $\osemPB{c}$ is an arrow of $\Cat{Par}_2$. Thus  $\osemPB{-}\colon \Diag{{\SigPB}} \to \Cat{Par}$ factors as
 \[\xymatrix{ \Diag{\SigPB} \ar[r] & \Cat{Par}_2 \ar@{^{(}->}[r]& \Cat{Par} }\text{.}\]
 Moreover, by soundness (Proposition~\ref{prop:soundnesspartialb}),  it factors through $\Diag{{\ThPB}}$: 
\[\xymatrix{ \Diag{\SigPB} \ar@{->>}[r]& \Diag{\ThPB} \ar[r]& \Cat{Par}_2 \ar@{^{(}->}[r]& \Cat{Par} }\]
The central functor is faithful by completeness. To conclude that it is an isomorphism, it is enough to prove fullness, i.e., that for any partial function $f\colon 2^n\to2^m$, there exists some diagram $d\colon A^n\to A^m$ such that $\osemPB{d}=f$. But this is trivial by first observing that for any such partial function, there exist total functions $D_f\colon2^n \to 2 $ and $T_f\colon 2^n \to 2^m$ that decompose $f$ as in Proposition~\ref{prop:decompositionpartialcircuits} and then make use of Proposition~\ref{prop:foolboolean}.
\end{proof}

%% file: appendices/appprobpartialbooleantapes-new.tex
\section{Appendix to Section \ref{sec:probbooltapes}}\label{app:sec:probbooltapes}

\begin{proof}[Proof of Proposition \ref{prop:encoding}] 
By induction on $c$. The base case $c=\Flip{p}$ is proved before the statement of the proposition. For the base case $c\in \SigPB\cup\{\id{A},\id{1},\symmt{A}{A}\}$,  $\CBdsem{\encoding{c}} = \CBdsem{\tapeFunct{c} }$ by the definition in Section~\ref{ssec:probbooltapes:encoding} and $\CBdsem{\tapeFunct{c} } = \ParJ(\osemPB{c})$ by Table~\ref{eq:SEMANTICA}. The latter is exactly $\osem{c}$.
The inductive cases follow immediately from the fact that $\CBdsem{-}$ is a morphism of rig categories and hence preserve $\per$.
\end{proof}

\begin{proof}[Proof of Proposition \ref{prop:soundnesstapes}]
Since the axioms in $\ThPB$ are sounds with respect to $\osem{-}$ and since $\osem{-}=\dsem{\encoding{-}}$, these axioms are also sound for $\dsem{-}$.
For \eqref{ax:AXPBP1} and \eqref{ax:AXPBP2}, it enough to check that the left and the right hand sides are mapped into the same $\KlD$-arrows by $\dsem{-}$.
For \eqref{ax:canc}, it is enough to observe that cancellativity holds in $\KlD$.
\end{proof}

\begin{proof}[Proof of Proposition \ref{prop:twoequivalence}]
  We first prove by induction on the rules in \eqref{eq:congr3} that if $\s \dot{\sim} \t$ then $\CatT{Q_{\ThPB}}(\s) \eqsyn \CatT{Q_{\ThPB}}(\t)$.
Consider the case of the rule $(\ThPB)$: in this case $\s=\tapeFunct{c}$, $\t=\tapeFunct{d}$ and $c\eqPB d$. Thus $\CatT{Q_{\ThPB}}(\s)= \CatT{Q_{\ThPB}}(\t)$.
All the other cases trivially follow from the fact that $\CatT{Q_{\ThPB}}$ is a morphism of rig categories: see \cite{bonchi2025tapediagramsmonoidalmonads}.

Again, by induction on the rules, we prove that if $\CatT{Q_{\ThPB}}(\s) \eqsyn \CatT{Q_{\ThPB}}(\t)$ then $\s \dot{\sim} \t$.
Here the only non-trivial cases are $(R)$, $\eqref{ax:AXPBP1}$ and $\eqref{ax:AXPBP2}$. For $(R)$, if $\CatT{Q_{\ThPB}}(\s) = \CatT{Q_{\ThPB}}(\t)$, then by Proposition \ref{prop:quotienttheory}, $\s \tapeeqPB \t$. Hence, $\s \dot{\sim} \t$. For $\eqref{ax:AXPBP1}$, if $\CatT{Q_{\ThPB}}(\s) \overset{\eqref{ax:AXPBP1}}{=} \CatT{Q_{\ThPB}}(\t)$, then $\s$ is $\tapeeqPB$-equivalent to the left hand side of \eqref{ax:AXPBP1} and $\t$ is $\tapeeqPB$-equivalent to the right hand side of \eqref{ax:AXPBP1}. Hence, they are $\dot{\sim}$-equivalent, and applying transitivity of $\dot{\sim}$, we conclude that $\s \dot{\sim} \t$. The case of $\eqref{ax:AXPBP2}$ is analogous.
\end{proof}

\begin{proof}[Proof of Lemma \ref{lemma:bottomstar2}]
 \begin{align}
  \Flipn{\bot}[t]&=\Flip{\bot}[t] ;\ncopierbis[t]  \tag{Def.\ of $\Flipn{\bot}$}\\
  &\sim \Tcounitempty{};\Tunit{}; \ncopierbis[t]  \tag{Axiom \ref{ax:AXPBP2}}\\
  &=  \Tcounitempty{};\Tunitm{n} \tag{\eqref{ax:tapes:cobangnat}}\\
  &= \star_{A^0,A^n}. \tag{Def. of $\star$}
 \end{align}
\end{proof}

Denote with $G^\flat: \CatT{\Cat{C}}\to \Cat{D}$ the arrow obtained by a functor $G\colon \Cat{C}\to U(\Cat{D})$ through the adjunction $\CatT{-}\dashv U$ in Theorem~\ref{thm:syntacticadjunction}. By construction, it is given by the following inductive definition on the structure of arrows in $\CatT{\Cat{C}}$:
\begin{equation}\label{eq:bemolle}
\begin{aligned}
G^\flat(\id{U}) &\defeq \id{G(U)} 
&\qquad G^\flat(\tapeFunct{c}) &\defeq G(c) \\
G^\flat(\diagp{P}) &\defeq\, \diagp{G(P)}  
&\qquad G^\flat(\bang{P}) &\defeq \bang{G(P)} \\
G^\flat(\codiag{P}) &\defeq \codiag{G(P)} 
&\qquad G^\flat(\cobang{P}) &\defeq \cobang{G(P)} \\
G^\flat(\t;\s) &\defeq\, G^\flat(\t);G^\flat(\s)
&\qquad G^\flat(\t \oplus \s) &\defeq\, G^\flat(\t)\oplus G^\flat(\s) \\
G^\flat(\symmp{P}{Q}) &\defeq\, \symmp{G(P)}{G(Q)}
&\qquad G^\flat(\id{\zero}) &\defeq \id{G(\zero)}
\end{aligned}
\end{equation}

Denote with $\jcompiso$ the composition of the isomorphism in Proposition~\ref{prop:isopartialboolean} and the functor $\ParJ_2$:
\[\xymatrix{  \Diag{\mathbb{\SigPB }} \ar[r]^{\cong}& \Cat{Par}_2 \ar@{->}[r]^{\ParJ_2}& \KlD }\]

\begin{lemma}\label{lemma:J_2}
$\comp=\jcompiso^\flat$.
\end{lemma}
\begin{proof}
  It follows from the adjunctions in Theorems~\ref{thm:freeenriched} and \ref{thm:matfree}, and applying Theorem~\ref{thm:syntacticadjunction}  and Proposition~\ref{prop:isopartialboolean}.
\end{proof}

\begin{proof}[Proof of Lemma \ref{lemma:diagramma-commutativo}]
    In order to prove the statement we exploit Proposition~\ref{prop:quotienttheory} which provides an isomorphism $F\colon\CatT{\Diag{\SigPB}}_{\tapeeqPB} \to \CatT{\Diag{\ThPB}}$, such that $Q_{\tapeeqPB};F= {\CatT{Q_\ThPB}}$, where $Q_{\tapeeqPB}\colon \CatT{\Diag{\SigPB}} \to \CatT{\Diag{\SigPB}}_{\tapeeqPB}$ is the quotient functor.
     Then, since $\dsem{-}$ is clearly sound with respect to ${\tapeeqPB}$, it is enough to prove that $\comp\circ F\circ Q_{\tapeeqPB} = \dsem{-}$. Since $\dsem{-}$ is defined inductively, it is enough to check that $\comp\circ F\circ Q_{\tapeeqPB}$ and $\dsem{-}$ coincide on generators.
    For instance consider the generator $\diagp{A^n}$:
    \begin{align*}
      \comp\circ F\circ Q_{\tapeeqPB}(\diagp{A^n})&= \comp\circ F(\diagp{A^n}) \tag{Def. $Q_{\tapeeqPB}$}\\
      &= \comp(\diagp{A^n}) \tag{Def. $F$}\\
      &=\jcompiso^\flat(\diagp{A^n}) \tag{Lemma \ref{lemma:J_2}}\\
      &=\, \diagp{2^n} \tag{\ref{eq:bemolle}}\\
      &=\dsem{\diagp{A^n}}. \tag{Def. $\dsem{-}$}
    \end{align*}
    The cases for the generators \,$\bangp{A^n}$, $\cobang{A^n}$ and $\codiag{A^n}$ are analogous. For the generator $\tapeFunct{c}$, we have
    \begin{align}
      \comp\circ F\circ Q_{\tapeeqPB}(\tapeFunct{c})&= \comp\circ F([\tapeFunct{c}]_{\tapeeqPB}) \tag{Def. $Q_{\tapeeqPB}$}\\
      &= \comp(\tapeFunct{[c]_{\ThPB}}) \tag{Def. $F$}\\
      &=\jcompiso^\flat(\tapeFunct{[c]_{\ThPB}}) \tag{Lemma \ref{lemma:J_2}}\\
      &= \jcompiso([c]_{\ThPB}) \tag{\ref{eq:bemolle}}\\
      &=\ParJ_2(\osemPB{c}) \tag{Proposition \ref{prop:isopartialboolean}}\\
      &= \dsem{\tapeFunct{c}}. \tag{Def. $\dsem{-}$}
\end{align}
The remaining cases are trivial.
Hence, by the universal property of
$\CatT{\Diag{\SigPB}}_{\tapeeqPB}$ we obtain that $\comp\circ F= I\circ \eqsynq{Q}\circ F$, and since $F$ is an isomorphism, we obtain the statement.
\end{proof}

\begin{proof}[Proof of Lemma \ref{lemma:Ffaith}]
   Denote with $\ParJ_2^\sharp\colon \Cat{Par}_2^+\to \KlD$ the arrow obtained by $\ParJ_2$ through the adjunction $(-)^+\dashv U$ in Theorem~\ref{thm:freeenriched}. By construction, it is the identity-on-objects functor sending a subdistribution $d\in \Cat{Par}_2^+[1,2^m]$ into the subdistribution $\ParJ_2^\sharp(d)= \sum_{f\in \Cat{Par}_2[1,2^m]} d(f)\cdot{\ParJ_2(f)}$. Hence, if $d\in \Sets_2^+[1,2^m]$, then $\ParJ_2^\sharp(d)= \sum_{f\in \Sets_2[1,2^m]} d(f)\cdot{\ParJ_2(f)}=\sum_{f\in \Sets_2[1,2^m]} d(f)\cdot{\delta_{f(\bullet)}}$, where $f(\bullet)\in 2^m$. Then, $\ParJ_2^\sharp$ restricted to $\Sets_2^+[1,2^m]$ provides the obvious bijection $$\Dis(\Sets_2[1,2^m])\cong \Dis(2^m) \cong   \KlD_2[1,2^m]\text{.}$$ 

The statement now follows from the fact that by construction $\eta_{\Cat{Par}_2^+};\ParJ_2'=\ParJ_2^\sharp$, where $\eta$ is the unit of the adjunction $\stmat{-}\dashv U$ in Theorem~\ref{thm:matfree}, and that fact that the isomorphisms 
\[\CatT{\Diag{\ThPB}} \cong \stmat{\Diag{\ThPB}^+}\cong \stmat{\Cat{Par}_2^+}\] 
restricted to $[A^0,A^m]$ factor through the composition
\[\CatT{\Diag{\ThPB}}[A^0,A^m] \cong \Diag{\ThPB}^+[A^0,A^m]\cong \Cat{Par}_2^+[2^0,2^m]\overset{\eta_{\Cat{Par}_2^+}}{\to }\stmat{\Cat{Par}_2^+}[2^0,2^m],\]
where the first isomorphism is given by Corollary~\ref{cor:isotapematrices} and the obvious isomorphism between $\stmat{\Diag{\ThPB}^+}[A^0,A^m]$ and $ \Diag{\ThPB}^+[A^0,A^m]$, and the second isomorphism is induced by Proposition~\ref{prop:isopartialboolean}.
\end{proof}

\begin{proof}[Proof of Lemma~\ref{lemma:PBP0}]
Thanks to Lemma~\ref{lemma: arrows A-B subdistributions},  $\t$ corresponds to  a subdistribution on $\Diag{\ThPB}[A^0,A^m]$. By Lemma~\ref{lemma:caratterizzazionecircuiti0n}, each partial Boolean circuit in $\Diag{\ThPB}[A^0,A^m]$ is either a Boolean circuit in $\Diag{\SigB}[1,A^m]$ or it is of the form $\Flipm{\bot}$. Hence, we can rearrange the subdistribution corresponding to $\t$ into a tape of the desired form.
\end{proof}

\begin{proof}[Proof of Lemma \ref{lemma:PBP1}]
  By Lemma~\ref{lemma:PBP0}, we can write 
  \begin{equation}\label{eq:BLABLA}
  \t=(\sum_{i=1}^{n}p_i\cdot \overrightarrow{b}_i)+_p(\star_{A^0,A^m} +_q \Flipm{\bot}[t]) \text{ and }\t'=(\sum_{j=1}^{n'}p'_j\cdot \overrightarrow{b}'_j)+_{p'}(\star_{A^0,A^m} +_{q'} \Flipm{\bot}[t])\end{equation} 
  for some $p_i,p'_j\in (0,1)$, for $i=1,\dots,n$ and $j=1,\dots,n'$, $p,p',q,q'\in [0,1]$  such that $\sum_{i=1}^{n}p_i= 1$ and $\sum_{j=1}^{n'}p'_j= 1$, where $\overrightarrow{b}_i\colon A^0 \to A^m\in \Cat{B}[1,A^m]$ and $\overrightarrow{b}'_j\colon A^0 \to A^m\in \Cat{B}[1,A^m]$. Then,
  \begin{align}
    \comp(\t)&= \comp((\sum_{i=1}^{n}p_i\cdot \overrightarrow{b}_i) +_p(\star_{A^0,A^m} +_q \Flipm{\bot}[t])) \tag{\ref{eq:BLABLA}}\\
    &= \comp((\sum_{i=1}^{n}p_i\cdot \overrightarrow{b}_i)) +_p(\comp(\star_{A^0,A^m}) +_q \comp(\Flipm{\bot}[t])) \tag{$\comp$ pca-enriched}\\
    &= \comp((\sum_{i=1}^{n}p_i\cdot \overrightarrow{b}_i)) +_p (\star_{1,2^m} +_q \star_{1,2^m}) \tag{$\comp(\star_{A^0,A^m})=\star_{1,2^m}=\comp(\Flipm{\bot}[t])$}\\
    &= p\cdot \comp(\sum_{i=1}^{n}p_i\cdot \overrightarrow{b}_i) \tag{Idemp.\ in (\ref{eq:pca})}\\
    &=\comp(p\cdot \sum_{i=1}^{n}p_i\cdot \overrightarrow{b}_i)\tag{$\comp$ pca-enriched}
  \end{align}
  and similarly $\comp(\t')=\comp(p'\cdot \sum_{j=1}^{n'}p'_j\cdot \overrightarrow{b}'_j)$. Hence, $\comp(p\cdot \sum_{i=1}^{n}p_i\cdot \overrightarrow{b}_i)=\comp(p'\cdot \sum_{j=1}^{n'}p'_j\cdot \overrightarrow{b}'_j)$, and Lemma~\ref{lemma:Ffaith} implies that  
  \begin{equation}\label{BLA2}p\cdot \sum_{i=1}^{n}p_i\cdot \overrightarrow{b}_i=p'\cdot \sum_{j=1}^{n'}p'_j\cdot \overrightarrow{b}'_j\text{.}\end{equation} Thus,
  \begin{align}
    \t&=(\sum_{i=1}^{n}p_i\cdot \overrightarrow{b}_i) +_p(\star_{A^0,A^m} +_q \Flipm{\bot}[t]) \tag{\eqref{eq:BLABLA}}\\
    &\sim  \sum_{i=1}^{n}p_i\cdot \overrightarrow{b}_i +_p(\star_{A^0,A^m} +_q \star_{A^0,A^m}) \tag{Lemma \ref{lemma:bottomstar2}}\\
    &= \sum_{i=1}^{n}p_i\cdot \overrightarrow{b}_i +_p \star_{A^0,A^m} \tag{Idemp.\ in (\ref{eq:pca})}\\
    &= p\cdot \sum_{i=1}^{n}p_i\cdot \overrightarrow{b}_i\tag{Def. $p\cdot-$}\\
    &= p'\cdot \sum_{j=1}^{n'}p'_j\cdot \overrightarrow{b}'_j \tag{\eqref{BLA2}}\\
    &= \sum_{j=1}^{n'}p'_j\cdot \overrightarrow{b}'_j +_{p'} \star_{A^0,A^m} \tag{Def. $p\cdot-$}\\
    &\sim \sum_{j=1}^{n'}p'_j\cdot \overrightarrow{b}'_j +_{p'}(\star_{A^0,A^m} +_{q'} \star_{A^0,A^m}) \tag{Idemp.\ in (\ref{eq:pca})}\\
    &= \sum_{j=1}^{n'}p'_j\cdot \overrightarrow{b}'_j +_{p'}(\star_{A^0,A^m} +_{q'} \Flipm{\bot}[t]) \tag{Lemma \ref{lemma:bottomstar2}}\\
    &= \t'. \tag{\eqref{eq:BLABLA}}
  \end{align}
\end{proof}

%% file: appendices/appfinpac.tex
\section{Appendix to Section \ref{app:eff}}\label{app:app:eff}

\begin{proof}[Additional details for the proof of Lemma \ref{lemma:pca-pcm}]
    Assume that $X$ is the free pca generated by a set $S$ and recall that free pcas are cancellative. We first prove that for the structure $(X,\circledvee,\zeropcm)$ defined above $\circledvee$ is well-defined, then we prove that the axioms of pcm are satisfied.
    
    Hence, assume that $x\bot y$, which means that there exist $x',y'\in X$ and $p,q\in [0,1]$ such that $x=p\cdot x'$, $y=q\cdot y'$ and $p+q\leq 1$. Assume also that $x\bot y$ via $x'',y''\in X$ and $r,s\in [0,1]$ such that $x=r\cdot x''$, $y=s\cdot y''$ and $r+s\leq 1$. First consider the case where $p,q,r,s\in(0,1)$, then we need to prove the equality
\begin{equation}\label{eq:pcmbendefinita}
    x'+_p(y'+_{\frac{q}{1-p}}\star) = x''+_r(y''+_{\frac{s}{1-r}}\star)\text{.}\end{equation}
In order to do that, we assume that $p\leq r$ and that $s\leq q$ (the other cases can be treated similarly).
First observe that 
\begin{align*}
    (1-p)\cdot (\frac{q}{1-p}\cdot y') & = q\cdot y' \\
    &= s\cdot y'' \tag{Hypothesis}\\
    & = (1-p)\cdot (\frac{s}{1-p}\cdot y'')\text{.}
    \end{align*}
    Thus, by cancellativity, we obtain that 
    \begin{equation}\tag{$\dagger$}\label{pcmtemp1}
    \frac{q}{1-p}\cdot y' = \frac{s}{1-p}\cdot y''.
    \end{equation}
     One can similarly show that 
     \begin{equation}\tag{$\ddagger$}\label{pcmtemp2}
     \frac{p}{1-s}\cdot x'=\frac{r}{1-s}\cdot x''.
     \end{equation}
    Now we show that both sides of \eqref{eq:pcmbendefinita} are equal to $x'+_p(y''+_{\frac{s}{1-p}}\star)$. We have 
\begin{align*}
    x'+_p(y'+_{\frac{q}{1-p}}\star) & = x'+_p((\frac{q}{1-p}\cdot y')) \tag{Def.\ of $p\cdot x$}\\
    & = x'+_p((\frac{s}{1-p}\cdot y'')) \tag{\ref{pcmtemp1}}\\
    & = x'+_p(y''+_{\frac{s}{1-p}}\star) \tag{Def.\ of $p\cdot x$}
    \end{align*}
    Similarly,
    \begin{align*}
 x''+_r(y''+_{\frac{s}{1-r}}\star) & = y''+_s (x''+_{\frac{r}{1-s}}\star) \tag{using axioms in \ref{eq:pca}}\\
    & = y''+_s(\frac{r}{1-s}\cdot x'') \tag{Def.\ of $p\cdot x$}\\
    & = y''+_s(\frac{p}{1-s}\cdot x') \tag{\ref{pcmtemp2}}\\
    & = y''+_s(x'+_{\frac{p}{1-s}}\star) \tag{Def.\ of $p\cdot x$}\\
    & = x'+_p(y''+_{\frac{s}{1-p}}\star ). \tag{using axioms in \ref{eq:pca}}
\end{align*}
Now assume that $p=1$ and $q=0$, then $x\circledvee y = x$. We distinguish two subcases: either $r=1$ or $r\not=1$. In the first case, the claim holds trivially. In the second case, we have that $q=0$ and hence, if $s\not= 0$, we have $s\cdot y''= q\cdot y'= 0 \cdot y' = \star = s\cdot \star$ and then, by cancellativity, $y''=\star$. Thus, $x''+_r(y''+_{\frac{s}{1-r}}\star) = x''+_r \star = r\cdot x'' = x$. The remaining cases can be treated similarly. Hence the equality in \eqref{eq:pcmbendefinita} holds and $\circledvee$ is well defined.

Now we prove the axioms of pcm. To prove commutativity, assume that $x\bot y$, and hence there exist $x',y'\in X$ and $p,q\in [0,1]$ such that $x=p\cdot x'$, $y=q\cdot y'$, $p+q\leq 1$. Then $y\bot x$ via $y',x'$ and $q,p$. 
Moreover, if $p,q\not=1$, then
\begin{align*}
    x\circledvee y & = x'+_p(y'+_{\frac{q}{1-p}}\star) \tag{\eqref{eq:defcircledvee}}\\
    & = (y'+_{\frac{q}{1-p}}\star) +_{1-p} x' \tag{symmetry in \eqref{eq:pca}}\\
    & = y'+_q(\star +_{\frac{1-p-q}{1-q}} x') \tag{associativity in \eqref{eq:pca}}\\
    & = y'+_q(x'+_{\frac{p}{1-q}}\star) \tag{symmetry in \eqref{eq:pca}}\\
& = y\circledvee x. \tag{\eqref{eq:defcircledvee}}
\end{align*}
If $p=1$ and $q=0$, then $x\circledvee y = x = y\circledvee x$. The case $p=0$ and $q=1$ is similar.

To prove that $\zeropcm \bot x$ and $\zeropcm \circledvee x =x$, we observe that $\zeropcm=\star = 0\cdot \star$ and $x=1\cdot x$, and hence by \eqref{eq:defbot}, $\zeropcm \bot x$. Obviously, $\zeropcm \circledvee x = x +_1 \star = x$. 

Associativity has been already proved in the main text.
\end{proof}

%% file: references.bib
@inproceedings{completenessprobbooltaperCONCUR26,
  author    = {Filippo Bonchi and Cipriano Junior Cioffo},
  title     = {Completeness for Probabilistic Boolean Tapes},
  booktitle = {Proceedings of the 37th International Conference on Concurrency Theory (CONCUR 2026)},
  year       = {2026},
  note       = {To appear}
}

@InProceedings{probbooltapesfossacs,
author="Bonchi, Filippo
and Cioffo, Cipriano Junior",
editor="Bertrand, Nathalie
and Milius, Stefan",
title="Tapes as Stochastic Matrices of String Diagrams",
booktitle="Foundations of Software Science and Computation Structures",
year="2026",
publisher="Springer Nature Switzerland",
address="Cham",
pages="88--109",
isbn="978-3-032-22730-0"
}

@article{manes,
  title={Partially additive categories and flow-diagram semantics},
  author={Arbib, Michael A. and Manes, Ernest G.},
  journal={Journal of Algebra},
  volume={62},
  number={1},
  pages={203--227},
  year={1980},
  publisher={Elsevier}
}

@article{chophd,
  title={Effectuses in categorical quantum foundations},
  author={Cho, Kenta},
  journal={arXiv preprint arXiv:1910.12198},
  year={2019}
}

@inproceedings{giry2006categorical,
  title={A categorical approach to probability theory},
  author={Giry, Michele},
  booktitle={Categorical Aspects of Topology and Analysis: Proceedings of an International Conference Held at Carleton University, Ottawa, August 11--15, 1981},
  pages={68--85},
  year={2006},
  organization={Springer}
}

@inproceedings{DBLP:conf/lics/RozowskiS24,
  author       = {Wojciech Rozowski and
                  Alexandra Silva},
  editor       = {Pawel Sobocinski and
                  Ugo Dal Lago and
                  Javier Esparza},
  title        = {A Completeness Theorem for Probabilistic Regular Expressions},
  booktitle    = {Proceedings of the 39th Annual {ACM/IEEE} Symposium on Logic in Computer
                  Science, {LICS} 2024, Tallinn, Estonia, July 8-11, 2024},
  pages        = {66:1--66:14},
  publisher    = {{ACM}},
  year         = {2024},
  url          = {https://doi.org/10.1145/3661814.3662084},
  doi          = {10.1145/3661814.3662084},
  bibsource    = {dblp computer science bibliography, https://dblp.org}
}

@article{DBLP:journals/lmcs/BonchiSV22,
  author       = {Filippo Bonchi and
                  Ana Sokolova and
                  Valeria Vignudelli},
  title        = {The Theory of Traces for Systems with Nondeterminism, Probability,
                  and Termination},
  journal      = {Log. Methods Comput. Sci.},
  volume       = {18},
  number       = {2},
  year         = {2022},
  url          = {https://doi.org/10.46298/lmcs-18(2:21)2022},
  doi          = {10.46298/LMCS-18(2:21)2022},
  bibsource    = {dblp computer science bibliography, https://dblp.org}
}

@inproceedings{DBLP:conf/lics/MioSV21,
  author       = {Matteo Mio and
                  Ralph Sarkis and
                  Valeria Vignudelli},
  title        = {Combining Nondeterminism, Probability, and Termination: Equational
                  and Metric Reasoning},
  booktitle    = {36th Annual {ACM/IEEE} Symposium on Logic in Computer Science, {LICS}
                  2021, Rome, Italy, June 29 - July 2, 2021},
  pages        = {1--14},
  publisher    = {{IEEE}},
  year         = {2021},
  url          = {https://doi.org/10.1109/LICS52264.2021.9470717},
  doi          = {10.1109/LICS52264.2021.9470717},
  bibsource    = {dblp computer science bibliography, https://dblp.org}
}

@inproceedings{fritz2023weakly,
title={Weakly Markov categories and weakly affine monads},
author={Fritz, Tobias and Gadducci, Fabio and Perrone, Paolo and Trotta, Davide},
editor       = {Paolo Baldan and
                  Valeria de Paiva},
  booktitle    = {10th Conference on Algebra and Coalgebra in Computer Science, {CALCO}
                  2023, Indiana University Bloomington, IN, USA, June 19-21, 2023},
  series       = {LIPIcs},
  volume       = {270},
  pages        = {16:1--16:17},
  publisher    = {Schloss Dagstuhl - Leibniz-Zentrum f{\"{u}}r Informatik},
  year         = {2023},
  url          = {https://doi.org/10.4230/LIPIcs.CALCO.2023.16},
  doi          = {10.4230/LIPICS.CALCO.2023.16},
  bibsource    = {dblp computer science bibliography, https://dblp.org}
}

@article{fritz2023d,
title={The d-separation criterion in categorical probability},
author={Fritz, Tobias and Klingler, Andreas},
journal={Journal of Machine Learning Research},
volume={24},
number={46},
pages={1--49},
year={2023}
}

@article{perrone2023markov,
title={Markov categories and entropy},
author={Perrone, Paolo},
journal={IEEE Transactions on Information Theory},
volume={70},
number={3},
pages={1671--1692},
year={2023},
publisher={IEEE}
}

@article{stein2024probabilistic,
  title={Probabilistic programming with exact conditions},
  author={Stein, Dario and Staton, Sam},
  journal={Journal of the ACM},
  volume={71},
  number={1},
  pages={1--53},
  year={2024},
  publisher={ACM New York, NY}
}

@article{fritz2018bimonoidal,
  title={Bimonoidal structure of probability monads},
  author={Fritz, Tobias and Perrone, Paolo},
  journal={Electronic notes in theoretical computer science},
  volume={341},
  pages={121--149},
  year={2018},
  publisher={Elsevier}
}

@article{fritz2023dilations,
  title={Dilations and information flow axioms in categorical probability},
  author={Fritz, Tobias and Gonda, Tom{\'a}{\v{s}} and Houghton-Larsen, Nicholas Gauguin and Lorenzin, Antonio and Perrone, Paolo and Stein, Dario},
  journal={Mathematical Structures in Computer Science},
  volume={33},
  number={10},
  pages={913--957},
  year={2023},
  publisher={Cambridge University Press}
}

@article{fritz2021finetti,
  title={De {F}inetti's Theorem in Categorical Probability},
  author={Fritz, Tobias and Gonda, Tom{\'a}{\v{s}} and Perrone, Paolo},
  journal={arXiv preprint arXiv:2105.02639},
  year={2021}
}

@article{moss2023category,
  title={A category-theoretic proof of the ergodic decomposition theorem},
  author={Moss, Sean and Perrone, Paolo},
  journal={Ergodic Theory and Dynamical Systems},
  volume={43},
  number={12},
  pages={4166--4192},
  year={2023},
  publisher={Cambridge University Press}
}

@book{jacobs2021logical,
  title={The logical essentials of Bayesian reasoning},
  author={Jacobs, Bart and Zanasi, Fabio and others},
  year={2021},
  publisher={Cambridge University Press}
}

@inproceedings{jacobs2019causal,
  title={Causal inference by string diagram surgery},
  author={Jacobs, Bart and Kissinger, Aleks and Zanasi, Fabio},
  booktitle={International conference on foundations of software science and computation structures},
  pages={313--329},
  year={2019},
  organization={Springer}
}

@article{cho2019disintegration,
  title={Disintegration and Bayesian inversion via string diagrams},
  author={Cho, Kenta and Jacobs, Bart},
  journal={Mathematical Structures in Computer Science},
  volume={29},
  number={7},
  pages={938--971},
  year={2019},
  publisher={Cambridge University Press}
}

@article{fritz2009presentation,
  title={A presentation of the category of stochastic matrices},
  author={Fritz, Tobias},
  journal={arXiv preprint arXiv:0902.2554},
  year={2009}
}

@article{DBLP:journals/corr/abs-2301-12989,
  author       = {Di Lavore, Elena and Rom{\'{a}}n, Mario},
  title        = {Evidential Decision Theory via Partial Markov Categories},
  journal      = {CoRR},
  volume       = {abs/2301.12989},
  year         = {2023},
  url          = {https://doi.org/10.48550/arXiv.2301.12989},
  doi          = {10.48550/ARXIV.2301.12989},
  eprinttype    = {arXiv},
  eprint       = {2301.12989},
  bibsource    = {dblp computer science bibliography, https://dblp.org}
}

@article{DBLP:journals/corr/abs-2502-03477,
  author       = {Di Lavore, Elena and Rom{\'{a}}n, Mario and Sobocinski, Pawel},
  title        = {Partial Markov Categories},
  journal      = {CoRR},
  volume       = {abs/2502.03477},
  year         = {2025},
  url          = {https://doi.org/10.48550/arXiv.2502.03477},
  doi          = {10.48550/ARXIV.2502.03477},
  eprinttype    = {arXiv},
  eprint       = {2502.03477},
  bibsource    = {dblp computer science bibliography, https://dblp.org}
}

@inproceedings{DBLP:journals/corr/abs-2410-10627,
  author       = {Bonchi, Filippo and Di Lavore, Elena and Rom{\'{a}}n, Mario},
  title        = {Effectful Mealy Machines: Bisimulation and Trace},
 booktitle    = {40th Annual {ACM/IEEE} Symposium on Logic in Computer Science, {LICS}
                  2025, Singapore, June 23-26, 2025},
  pages        = {541--554},
  publisher    = {{IEEE}},
  year         = {2025},
  url          = {https://doi.org/10.1109/LICS65433.2025.00047},
  doi          = {10.1109/LICS65433.2025.00047},
  bibsource    = {dblp computer science bibliography, https://dblp.org}
}

@article{stone1949postulates,
  title={Postulates for the barycentric calculus},
  author={Stone, Marshall Harvey},
  journal={Annali di Matematica Pura ed Applicata},
  volume={29},
  number={1},
  pages={25--30},
  year={1949},
  publisher={Springer}
}

@article{sokolova2018termination,
  title={Termination in convex sets of distributions},
  author={Sokolova, Ana and Woracek, Harald},
  journal={Logical Methods in Computer Science},
  volume={14},
  year={2018},
  publisher={Episciences. org}
}

@inproceedings{bonchi2017power,
  title={The power of convex algebras},
  author={Bonchi, Filippo and Silva, Alexandra and Sokolova, Ana},
  booktitle={28th International Conference on Concurrency Theory (CONCUR 2017)},
  pages={23--1},
  year={2017},
  organization={Schloss Dagstuhl--Leibniz-Zentrum f{\"u}r Informatik}
}

@misc{bonchi2025tapediagramsmonoidalmonads,
      title={Tape Diagrams for Monoidal Monads}, 
      author={Bonchi, Filippo and Cioffo, Cipriano Junior and Di Giorgio, Alessandro and Di Lavore, Elena},
      booktitle =	{11th Conference on Algebra and Coalgebra in Computer Science (CALCO 2025)},
  pages =	{11:1--11:24},
  series =	{Leibniz International Proceedings in Informatics (LIPIcs)},
  ISBN =	{978-3-95977-383-6},
  ISSN =	{1868-8969},
  year =	{2025},
  volume =	{342},
  editor =	{C\^{i}rstea, Corina and Knapp, Alexander},
  publisher =	{Schloss Dagstuhl -- Leibniz-Zentrum f{\"u}r Informatik},
  address =	{Dagstuhl, Germany},
  URL =		{https://drops.dagstuhl.de/entities/document/10.4230/LIPIcs.CALCO.2025.11},
  URN =		{urn:nbn:de:0030-drops-235703},
  doi =		{10.4230/LIPIcs.CALCO.2025.11} 
}

@inproceedings{bonchi2024diagrammatic,
  author={Bonchi, Filippo and Di Giorgio, Alessandro and Di Lavore, Elena},
  title={A Diagrammatic Algebra for Program Logics},
  editor="Abdulla, Parosh Aziz
and Kesner, Delia",
booktitle="Foundations of Software Science and Computation Structures",
year="2025",
publisher="Springer Nature Switzerland",
address="Cham",
pages="308--330",
isbn="978-3-031-90897-2"
}

@article{hasuo2007generic,
  title={Generic trace semantics via coinduction},
  author={Hasuo, Ichiro and Jacobs, Bart and Sokolova, Ana},
  journal={Logical Methods in Computer Science},
  volume={3},
  year={2007},
  publisher={Episciences. org}
}

@article{jacobs2010coalgebraic,
  title={From coalgebraic to monoidal traces},
  author={Jacobs, Bart},
  journal={Electronic Notes in Theoretical Computer Science},
  volume={264},
  number={2},
  pages={125--140},
  year={2010},
  publisher={Elsevier}
}

@article{introductioneffectus,
  title={An Introduction to Effectus Theory},
  author={Kenta Cho and Bart Jacobs and Bas Westerbaan and Abraham Westerbaan},
  journal={ArXiv},
  year={2015},
  volume={abs/1512.05813},
  url={https://api.semanticscholar.org/CorpusID:14013233}
}

@inproceedings{villoria2024enriching,
  title={Enriching diagrams with algebraic operations},
  author={Villoria, Alejandro and Basold, Henning and Laarman, Alfons},
  booktitle={International Conference on Foundations of Software Science and Computation Structures},
  pages={121--143},
  year={2024},
  organization={Springer}
}

@article{DBLP:journals/pacmpl/LiellCockS25,
  author       = {Jack Liell{-}Cock and
                  Sam Staton},
  title        = {Compositional Imprecise Probability: {A} Solution from Graded Monads
                  and Markov Categories},
  journal      = {Proc. {ACM} Program. Lang.},
  volume       = {9},
  number       = {{POPL}},
  pages        = {1596--1626},
  year         = {2025},
  url          = {https://doi.org/10.1145/3704890},
  doi          = {10.1145/3704890},
  bibsource    = {dblp computer science bibliography, https://dblp.org}
}

@inproceedings{laplaza_coherence_1972,
	title = {Coherence for Distributivity},
	booktitle = {Coherence in {{Categories}}},
	author = {Laplaza, Miguel L.},
	editor = {Kelly, G. M. and Laplaza, M. and Lewis, G. and Mac Lane, Saunders},
	year = {1972},
	series = {Lecture {{Notes}} in {{Mathematics}}},
	pages = {29--65},
	publisher = {{Springer}},
	address = {{Berlin, Heidelberg}},
	doi = {10.1007/BFb0059555},
	isbn = {978-3-540-37958-4},
	langid = {english}
}

@book{mac_lane_categories_1978,
	title = {Categories for the {{Working Mathematician}}},
	author = {Mac Lane, Saunders},
	year = {1978},
	series = {Graduate {{Texts}} in {{Mathematics}}},
	edition = {Second},
	volume = {5},
	publisher = {{Springer-Verlag}},
	address = {{New York}},
	url = {//www.springer.com/gb/book/9780387984032},
	isbn = {978-0-387-98403-2},
	langid = {english}
}

@article{johnson2021bimonoidal,
  title={Bimonoidal Categories, $ E\_n $-Monoidal Categories, and Algebraic $ K $-Theory},
  author={Johnson, Niles and Yau, Donald},
  journal={arXiv preprint arXiv:2107.10526},
  year={2021}
}

@article{joyal1991geometry,
  title = {The Geometry of Tensor Calculus, {{I}}},
  author = {Joyal, Andr{\'e} and Street, Ross},
  year = {1991},
  month = jul,
  journal = {Advances in Mathematics},
  volume = {88},
  number = {1},
  pages = {55--112},
  issn = {0001-8708},
  doi = {10.1016/0001-8708(91)90003-P},
  langid = {english}
}

@incollection{selinger2010survey,
  title={A survey of graphical languages for monoidal categories},
  author={Selinger, Peter},
  booktitle={New structures for physics},
  pages={289--355},
  year={2010},
  publisher={Springer}
}

@article{fox1976coalgebras,
  title = {Coalgebras and Cartesian Categories},
  author = {Fox, Thomas},
  year = {1976},
  journal = {Communications in Algebra},
  volume = {4},
  number = {7},
  pages = {665--667},
  issn = {0092-7872},
  doi = {10.1080/00927877608822127}
}

@article{hoare1969axiomatic,
  title={An axiomatic basis for computer programming},
  author={Hoare, Charles Antony Richard},
  journal={Communications of the ACM},
  volume={12},
  number={10},
  pages={576--580},
  year={1969},
  publisher={ACM New York, NY, USA}
}

@article{fritz_2020,
  title={A Synthetic Approach to {{Markov}} Kernels, Conditional Independence and Theorems on Sufficient Statistics},
  author={Fritz, Tobias},
  journal={Advances in Mathematics},
  volume={370},
  pages={107239},
  year={2020},
  eprinttype = {arXiv},
  eprint    = {1908.07021},
  publisher={Elsevier}
}

@inproceedings{piedeleu2025boolean,
  author       = {Piedeleu, Robin and
                  Torres{-}Ruiz, Mateo and
                  Silva, Alexandra and
                  Zanasi, Fabio},
  editor       = {Viktor Vafeiadis},
  title        = {A Complete Axiomatisation of Equivalence for Discrete Probabilistic
                  Programming},
  booktitle    = {Programming Languages and Systems - 34th European Symposium on Programming,
                  {ESOP} 2025, Held as Part of the International Joint Conferences on
                  Theory and Practice of Software, {ETAPS} 2025, Hamilton, ON, Canada,
                  May 3-8, 2025, Proceedings, Part {II}},
  series       = {Lecture Notes in Computer Science},
  volume       = {15695},
  pages        = {202--229},
  publisher    = {Springer},
  year         = {2025},
  url          = {https://doi.org/10.1007/978-3-031-91121-7\_9},
  doi          = {10.1007/978-3-031-91121-7\_9},
  bibsource    = {dblp computer science bibliography, https://dblp.org}
}

@article{Kozen94acompleteness,
	Author = {Dexter Kozen},
	Journal = {Information and Computation},
	Pages = {366--390},
	Title = {A Completeness Theorem for Kleene Algebras and the Algebra of Regular Events},
	Volume = {110},
	Year = {1994}}

@article{coecke2017two,
	title = {Two {{Roads}} to {{Classicality}}},
	booktitle = {Electronic {{Proceedings}} in {{Theoretical Computer Science}}},
	author = {Coecke, Bob and Selby, John and Tull, Sean},
  journal={Electronic Proceedings in Theoretical Computer Science},
	year = {2018},
	month = feb,
	volume = {266},
	eprint = {1701.07400},
	eprinttype = {arxiv},
	pages = {104--118},
	address = {{Nijmegen, The Netherlands}},
	doi = {10.4204/EPTCS.266.7},
}

@article{bonchi2023deconstructing,
  title={Deconstructing the calculus of relations with tape diagrams},
  author={Bonchi, Filippo and Di Giorgio, Alessandro and Santamaria, Alessio},
  journal={Proceedings of the ACM on Programming Languages},
  volume={7},
  number={POPL},
  pages={1864--1894},
  year={2023},
  publisher={ACM New York, NY, USA}
}

@book {borceux2,
    AUTHOR = {Borceux, Francis},
     TITLE = {Handbook of categorical algebra. 2},
    SERIES = {Encyclopedia of Mathematics and its Applications},
    VOLUME = {51},
      NOTE = {Categories and structures},
 PUBLISHER = {Cambridge University Press, Cambridge},
      YEAR = {1994},
     PAGES = {xviii+443},
      ISBN = {0-521-44179-X},
   MRCLASS = {18-02 (18Exx)},
  MRNUMBER = {1313497},
MRREVIEWER = {Martin\ Hyland},
}
